\documentclass{article}

\usepackage{amssymb,amsmath,amsfonts,eurosym,geometry,ulem,graphicx,caption,color,setspace,sectsty,comment,footmisc,caption,natbib,pdflscape,subfigure,array,hyperref,upgreek,bbm}

\usepackage{tabularx}
\usepackage{amsthm}
\usepackage{amssymb}
\usepackage{graphicx}
\usepackage{amsmath}
\usepackage{verbatim}
\usepackage{setspace}
\usepackage{ulem}
\usepackage{textpos}
\usepackage{changepage}
\usepackage{url}
\usepackage{float}
\usepackage[para,online,flushleft]{threeparttable}
\usepackage{booktabs}
\usepackage{multirow}
\usepackage{palatino}
\usepackage{enumitem}
\usepackage[table,xcdraw]{xcolor}
\DeclareMathOperator*{\argmax}{arg\,max}

\usepackage{rotating}
\usepackage{adjustbox}
\usepackage{mathabx}

\def\bs{\boldsymbol}

\def\t{^{\top}}

\newtheorem{theorem}{Theorem}
\newtheorem{corollary}[theorem]{Corollary}
\newtheorem{proposition}{Proposition}
\newtheorem{lemma}[theorem]{Lemma}
\newtheorem{example}{Example}

\usepackage{natbib} 

\usepackage{etoolbox}
\newcounter{bibcount}
\makeatletter
\patchcmd{\@lbibitem}{\item[}{\item[\hfil\stepcounter{bibcount}{[\thebibcount]}}{}{}
\renewcommand\NAT@bibsetup%
[1]{\setlength{\leftmargin}{\bibhang}\setlength{\itemindent}{-\parindent}%
  \setlength{\itemsep}{\bibsep}\setlength{\parsep}{\z@}}
\makeatother

\newcommand{\R}{\mathbb{R}}
\newcommand{\E}{\mathbb{E}}

\newtheorem{assumption}{Assumption}

\usepackage{apptools}
\AtAppendix{\counterwithin{assumption}{section}}
\AtAppendix{\counterwithin{proposition}{section}}

\newcolumntype{L}[1]{>{\raggedright\let\newline\\arraybackslash\hspace{0pt}}m{#1}}
\newcolumntype{C}[1]{>{\centering\let\newline\\arraybackslash\hspace{0pt}}m{#1}}
\newcolumntype{R}[1]{>{\raggedleft\let\newline\\arraybackslash\hspace{0pt}}m{#1}}

\begin{document}

\begin{titlepage}
\title{\textbf{Estimating the Number of Components in Panel Data Finite Mixture Regression Models with an Application to Production Function Heterogeneity}
\author{Yu Hao\thanks{Address for correspondence: Yu (Jasmine) Hao, Faculty of Business and Economics, The University of Hong Kong. We thank the editor, Stéphane Bonhomme, and two anonymous referees whose insightful comments and suggestions substantially improved this paper. We are also grateful to Chun Pang Chow, Vadim Marmer, Kevin Song, Bin Chen, and participants at the 2019 IAAE Conference for their helpful discussions. Financial support from the Natural Sciences and Engineering Research Council of Canada (NSERC) is gratefully acknowledged. }\\
Faculty of Business and Economics\\
The University of Hong Kong\\
haoyu@hku.hk \and
Hiroyuki Kasahara\\
Vancouver School of Economics\\
The University of British Columbia\\
hkasahar@mail.ubc.ca
}}
\date{\today}
\maketitle
\begin{abstract}
\noindent  
This paper develops statistical methods for determining the number of components in panel data finite mixture regression models with regression errors independently distributed as normal or more flexible normal mixtures. We analyze the asymptotic properties of the likelihood ratio test (LRT) and information criteria (AIC and BIC) for model selection in both conditionally independent and dynamic panel settings. Unlike cross-sectional normal mixture models, we show that panel data structures eliminate higher-order degeneracy problems while retaining issues of unbounded likelihood and infinite Fisher information. Addressing these challenges, we derive the asymptotic null distribution of the LRT statistic as the maximum of random variables and develop a sequential testing procedure for consistent selection of the number of components. Our theoretical analysis also establishes the consistency of BIC and inconsistency of AIC. Empirical application to Chilean manufacturing data reveals significant heterogeneity in production technology, with substantial variation in output elasticities of material inputs and factor-augmented technological processes within narrowly defined industries, indicating plant-specific variation in production functions beyond Hicks-neutral technological differences. These findings contrast sharply with standard  practice of assuming homogeneous production function and highlight the necessity of accounting for unobserved plant heterogeneity in empirical production analysis.
  \\
\bigskip
\end{abstract}
\setcounter{page}{0}
\thispagestyle{empty}
\end{titlepage}

\section{Introduction}\label{sec:intro}

Finite mixture models provide a flexible and natural framework for representing unobserved heterogeneity across a finite number of latent classes. Due to their adaptability, these models have found extensive applications in empirical studies across various disciplines since the seminal introduction of a two-component normal mixture model by \cite{Pearson1894}. In economics, often utilizing panel data, finite mixture models have been particularly influential in capturing unobserved individual-specific effects in fields such as labor economics, health economics, and industrial organization.\footnote{For instance, \cite{Heckman1984} utilize finite mixture models as an alternative method for handling unobserved heterogeneity in unemployment duration analysis. Similarly, \cite{Keane1997} and \cite{Cameron1998} employ these models to analyze dynamic decision-making regarding schooling and occupational choices in the presence of unobserved heterogeneity in human capital. In health economics, \cite{Deb1997} propose a finite mixture negative binomial model to capture the unobserved dispersion in healthcare utilization among elderly populations. Additionally, finite mixture models are extensively used in consumer segmentation within industrial organization, as demonstrated by \cite{Kamakura1989} and \cite{Andrews2003}. \cite{Kasahara2009} and \cite{Hu2012} establish the conditions under which finite mixture models are non-parametrically identified using panel data.} The theoretical underpinnings and practical applications of finite mixture models have been extensively explored by \cite{Titterington1985}, \cite{Lindsay1995}, and \cite{McLachlan2000}.

A critical aspect in the application of finite mixture models is determining the appropriate number of components, which often corresponds to the number of latent individual types or abilities in economic models. Selecting an incorrect number of components can lead to estimation biases, excessive computational costs, or identification difficulties. Therefore, developing reliable statistical procedures for accurately selecting the number of components in finite mixture models is essential for empirical analysis.

This paper investigates the statistical properties and practical performance of methods for determining the number of components in \textit{panel data finite mixture regression models}, including the Likelihood Ratio Test (LRT) and the Akaike and Bayesian Information Criteria (AIC and BIC). Panel data finite mixture models offer significant advantages over their cross-sectional counterparts by utilizing repeated observations for each unit, thereby enhancing parameter identification and mitigating common identification issues associated with cross-sectional analyses. However, the existing literature on asymptotic properties of panel data finite mixture regression models is limited, particularly when the component-specific regression errors are normally distributed—a commonly assumed baseline model in regression analyses. This study addresses this gap by rigorously examining and comparing these estimation methods under normally distributed regression errors.

Specifically, this paper develops the LRT and BIC methods for consistently selecting the number of components in various panel data finite mixture regression models and applies these methods to examine production function heterogeneity. We first consider a baseline scenario where regression errors are conditionally independent and normally distributed across periods, conditional on the latent type. Acknowledging the limitations of the normality assumption, we also analyze cases where regression errors follow more flexible, component-specific normal mixtures. Moreover, we extend the analysis to \textit{dynamic panel data finite mixture models}, incorporating lagged outcomes as covariates under a Markovian assumption, which accommodates regression errors modeled by an autoregressive process. We also consider a nonparametric procedure to estimate a lower bound on the number of components for finite mixture models under conditional independence using the rank test of \cite{Kleibergen06} and \cite{Kasahara2014}.


Testing for the number of components in normal mixture regression models remains challenging, as standard asymptotic regularity conditions fail due to parameter non-identifiability, singular Fisher information, and boundary parameter issues. While numerous studies have addressed the likelihood ratio test (LRT) for finite mixture models \citep[e.g.,][]{ghoshsen85book, chernofflander95jspi, lemdanipons97spl, chenchen01cjstat, chenchen03sinica, cck04jrssb, garel01jspi, garel05jspi, Chen14joe}, deriving its asymptotic distribution through Gaussian processes \citep{dacunha99as, liushao03as, zhuzhang04jrssb, azais09esaim}, these approaches fail in \textit{cross-sectional} normal regression models due to: (i) infinite Fisher information for testing, (ii) unbounded log-likelihood function, and (iii) linear dependence between density derivatives for mean and variance parameters.\footnote{\cite{chenli09as}, \cite{Chen2012jasa}, and \cite{kasaharashimotsu15jasa} analyze the LRT asymptotics for univariate models, with \cite{Kasahara2019} extending to multivariate cases. Additionally, \cite{Amengual2022} develop a score-type test for cross-sectional normal mixtures.}

To the best of our knowledge, whether the aforementioned problems (i)--(iii) of the cross-sectional normal mixture still arise in panel data normal finite mixture models or their extensions remains unknown in the literature. Furthermore, no likelihood-based test has yet been developed for testing the null hypothesis of an $M_0$-component model against an alternative $(M_0+1)$-component model for $M_0\geq 1$ in panel normal regression mixture models with conditionally independent errors.

We show that the higher-order degeneracy of problem (iii) disappears in panel normal mixture models with conditional independence and dynamic panel normal mixture models with a Markov assumption, but problems related to (i) and (ii) arise. Consequently, the existing approaches in \cite{dacunha99as}, \cite{liushao03as}, \cite{zhuzhang04jrssb}, and \cite{azais09esaim} do not directly apply to this class of panel normal mixture models. We impose bounds on the component-specific variance parameters and mixing proportions to address the unboundedness and infinite Fisher information, respectively, and then analyze the asymptotic distribution of the LRT using reparameterization orthogonal to the direction in which the Fisher information matrix is singular. The likelihood ratio of an $(M_0+1)$-component model against the $M_0$-component model is approximated with local quadratic-form expansion with squares and cross-products of the reparameterized parameters. We demonstrate that the asymptotic null distributions of the LRT statistic are characterized by the maximum of $M_0$ random variables. Building on the LRT tests, we propose a sequential hypothesis testing approach for consistently estimating the number of components.

This paper contributes to the literature in several ways. First, we analyze the likelihood ratio test (LRT) for determining the number of components in panel data normal regression finite mixture models, covering both conditional independence and dynamic panel cases with lagged dependent variables. While \cite{kasaharashimotsu15jasa} and \cite{Kasahara2019} study the LRT for cross-sectional univariate and multivariate normal mixture regression models, respectively, we show that the asymptotic distribution of the LRT statistic in panel data differs due to the absence of higher-order dependencies when repeated outcome measurements are available under conditional independence or Markov assumptions.\footnote{From a technical perspective, our panel data finite mixture models constitute a special case of the multivariate normal framework studied by \cite{Kasahara2019}. The primary distinction lies in our assumption of either conditional independence or a Markov process within each mixture component, which facilitates improved identification of model parameters. We show that this structure yields asymptotic null distributions for the likelihood ratio test that differ from those in \cite{Kasahara2019}, mainly because our panel data setting does not involve higher-order singularities.}

Second, while it is known that the log-likelihood function of normal mixture models is unbounded \citep{Hartigan1985}, whether this issue occurs in panel data remains unexplored. We show that the likelihood ratio test (LRT) statistic is also unbounded in panel normal mixture models with conditionally independent errors when the time dimension is finite, though this issue diminishes as the time dimension grows. This unboundedness may lead to excessive rejection rates in the LRT, which we address by imposing explicit bounds on the component-specific variance parameters.


Third, we establish the consistency of the Bayesian Information Criterion (BIC) and the inconsistency of the Akaike Information Criterion (AIC) for selecting the true number of components in panel data finite mixture models. Our consistency result generalizes \cite{keribin00sankhya} by relaxing its restrictive condition (P2) through higher-order rank conditions and utilizing a generalized form of Le Cam's differentiability in quadratic mean (DQM) framework \citep{liushao03as,kasahara2018arXiv}.

Fourth, we empirically analyze production technology heterogeneity using panel data from Chilean manufacturing plants. Our findings reveal significant variation in output elasticities of material inputs and the stochastic processes governing factor-augmented technological changes within narrowly defined industries. This contrasts sharply with standard production function estimation methods, which typically impose homogeneous coefficients across plants \citep{olley1996dynamics, levinsohn2003estimating, Ackerberg2015}. Our results highlight the necessity of explicitly accounting for unobserved plant heterogeneity beyond Hicks-neutral technological differences in empirical analyses of production functions \citep{LiSasaki17arxiv, doraszelski2018measuring, Balat19mimeo, Kasahara2022esri}.

The identification and estimation of latent group structures in panel data has received attention in recent studies \citep{Kasahara2009, LinNg12jem, Bonhomme15ecma, AndoBai16jae, SuShiPhillips16ecma, LuSu17qe}. Finite mixture modeling provides a practical, model-based approach to determining unobserved group structures. Choosing the number of groups is often a prerequisite for classifying each individual's group membership. We can estimate the number of groups in panel data regression models by applying our proposed sequential hypothesis testing approach.

This paper closely relates to \cite{Kasahara2022esri}, which analyzes the nonparametric identification and estimation of finite mixture production models with unobserved heterogeneity, explicitly assuming a known number of mixture components. In contrast, our study emphasizes testing and estimating the number of technology types, focusing specifically on heterogeneity in output elasticities of material inputs. By leveraging first-order condition expressions, our approach enables more flexible finite mixture model specifications than those studied by \cite{Kasahara2022esri}.\footnote{Additionally, we analyze a production function specification with factor-augmented technological changes, explicitly testing for plant-level heterogeneity in labor-augmented technological processes, a specification not explored by \cite{Kasahara2022esri}.} However, as our analysis does not recover the entire production function, our framework cannot address heterogeneity in output elasticities of predetermined inputs such as capital, nor heterogeneity in the stochastic processes governing Hicks-neutral technological changes. Thus, our paper complements \cite{Kasahara2022esri} by addressing distinct yet related aspects of mixture modeling in production function analysis.

 
The rest of this paper is organized as follows. In Section~\ref{sec:model}, we introduce several classes of panel data finite mixture regression models studied in this paper, along with empirical examples. In Section~\ref{sec:regularity}, we discuss the failure of regularity conditions when analyzing panel data finite mixture normal regression models. Section~\ref{sec:consistency} analyzes the consistency of Maximum Likelihood Estimation (MLE). Section~\ref{sec:LRT1} analyzes the likelihood ratio test (LRT) for testing $H_0: M=1$ against $H_1: M=2$, while Section~\ref{sec:LRT2} considers the LRT for testing $H_0: M=m$ against $H_1: M=m+1$ for $m\geq 2$. Section~\ref{sec:sht} develops a sequential hypothesis testing procedure, and Section~\ref{sec:bic} analyzes the consistency of the Bayesian Information Criterion (BIC) for selecting the number of components. Section~\ref{sec:rank} discusses the estimation of a lower bound for the number of components using the rank test. Section~\ref{sec:simulation} presents simulation results, while Section~\ref{sec:empirics} presents empirical analysis.

In what follows, all limits are taken as $n \rightarrow \infty$ unless otherwise stated. Let $:=$ denote "equals by definition." For a $k\times 1$ vector $\bs{a}$ and a function $f(\bs{a})$, let $\nabla_{\bs{a}}f(\bs{a})$ denote the $k\times 1$ vector of partial derivatives $(\partial/\partial \bs{a})f(\bs{a})$, and let $\nabla_{\bs{a}\bs{a}^{\top}}f(\bs{a})$ denote the $k\times k$ matrix of second partial derivatives $(\partial/\partial \bs{a}\partial \bs{a}^{\top})f(\bs{a})$. Let $||\cdot||$ denote the Euclidean norm. We adopt the convention that capitalized letters, such as $\bs{W}$, represent random variables, whereas their lowercase counterparts, such as $\bs{w}$, denote evaluation points.

\section{Heteroskedastic finite mixture panel regression model}\label{sec:model}

We consider finite mixture regression models with panel data, where the panel length \( T \geq 2 \) is fixed, and the number of cross-sectional observations \( n \) tends to infinity. Given \( M \geq 2 \) components, we assume that \( y_t \) is conditionally independent over time. The conditional probability density function of \(\{y_t\}_{t=1}^T\) given \(\{\bs{x}_t\}_{t=1}^T\), where \( y_t \in \mathbb{R} \) and \(\bs{x}_t \in \mathbb{R}^q\), is expressed as:
\begin{equation}\label{eq:fm}
    g_M(\{y_t\}_{t=1}^T \mid \{\bs{x}_t\}_{t=1}^T; \bs{\vartheta}_M)
    = \sum_{j=1}^M \alpha_j f(\{y_t\}_{t=1}^T \mid \{\bs{x}_t\}_{t=1}^T;\bs{\theta}_j),
\end{equation}
where \( f(y_t \mid \bs{x}_t; \bs{\theta}_j) \) denotes the density function of the \( j \)-th component for \( y_t \) given \(\bs{x}_t\), belonging to a parametric class, and \(\alpha_j\) represents the population proportion of the \( j \)-th component. Let \(\bs{\vartheta}_M := (\bs{\alpha}^\top, \bs{\theta}_1^\top, \ldots, \bs{\theta}_M^\top)^\top \in \Theta_{\bs{\vartheta}_M}\), where \(\bs{\alpha}^\top := (\alpha_1, \ldots, \alpha_{M-1})^\top\) and \(\alpha_M = 1 - \sum_{j=1}^{M-1} \alpha_j\). The vector $\bs x_t$ may include both time-varying and time-invariant regressors.

We are interested in estimating the number of components $M$ when we have correctly specified parametric mixture models given in (\ref{eq:fm}).  As a baseline model, we specify  $f(\{y_t\}_{t=1}^T \mid \{\bs{x}_t\}_{t=1}^T;\bs{\theta}_j)$ using the normal density function as
\begin{equation}\label{eq:f1}
	  f(\{y_t\}_{t=1}^T \mid \{\bs{x}_t\}_{t=1}^T;\bs{\theta}_j) =	\prod_{t=1}^T  \frac{1}{\sigma_j} \phi\left(\frac{ y_{t} -\mu_{j} - \bs{x}^\top_{t}\bs{\beta}_j  }{\sigma_j } \right),
\end{equation}
where $\bs{\theta}_j := (\mu_j, \sigma_j^2, \bs{\beta}_j^\top)^\top  \in \Theta_{\bs\theta}$ with $\sigma_j^2 \in \Theta_{\sigma^2} :=\R_{++}$, and $\bs\beta_j \in \Theta_{\bs\beta}$ while  $\phi(t) := (2\pi)^{-1/2} \exp(-\frac{t^2}{2})$ is the standard normal probability density function.  Given that the assumption of a normal density function is restrictive, we also consider more flexible parametric models using a mixture of normal density functions:
\begin{equation}\label{eq:f1-mixture}
	  f(\{y_t\}_{t=1}^T \mid \{\bs{x}_t\}_{t=1}^T;\bs{\theta}_j) = \prod_{t=1}^T \left(	\sum_{k=1}^{{\cal{K}}} \tau_{jk}  \frac{1}{\sigma_j} \phi\left(\frac{ y_{t} -\mu_{jk} - \bs{x}^\top_{t}\bs{\beta}_j  }{\sigma_j } \right)\right),
\end{equation}
where  we assume $\mu_{j1}<\mu_{j2}< \cdots<\mu_{j{\cal{K}}}$ for $j=1,...,M$ while  the component-specific parameter is given by  $\bs{\theta}_j := (\bs\tau_j,\bs\mu_j, \sigma_j^2, \bs{\beta}_j^\top)^\top  \in \Theta_{\bs\theta}$,
$\bs\tau_j=(\tau_{j1},...,\tau_{j{\cal{K}}-1}) \in\Theta_{\bs\tau}$, $\tau_{j{\cal{K}}}:=1-\sum_{k=1}^{{\cal{K}}-1}\tau_{jk} $, $\bs\mu_j=(\mu_{j1},...,\mu_{{\cal{K}}})^\top \in \Theta_{\mu}^{{\cal{K}}}$ , $\sigma_j^2 \in \Theta_{\sigma^2}$,  $\bs\beta_j \in \Theta_{\bs\beta}$.\footnote{Testing procedures under other parametric classes of mixture models can be developed based on the results of this paper or existing results in the literature, provided the regularity conditions are satisfied \citep[e.g.,][]{zhuzhang04jrssb}.} In (\ref{eq:f1-mixture}), the mean parameter $\mu_{jk}$ varies across ${\cal{K}}$ components, but neither the variance parameter $\sigma_j$ nor the coefficient $\bs\beta_j$ varies across ${\cal{K}}$ components. The constant variance parameter assumption prevents another source of unboundedness in this context, while the constant coefficient $\bs\beta_j$ assumption ensures that the mixture structure arises solely from the mixture distribution of regression error terms.

To simplify our analysis, we assume that ${{\cal{K}}}\geq 2$ is known to the researcher and that all true mixing proportions $\tau_{jk}$ are non-zero.\footnote{Testing the number of components $K$ simultaneously with the number of components $M$ is certainly an important issue but beyond the scope of this paper, and left for future research.}
\begin{assumption}\label{assumption:K}
In (\ref{eq:f1-mixture}) and (\ref{eq:f1-dynamic-mixture}),
$\cal K$ is known to the researcher and $\tau^*_{jk}> 0$,  $\mu_{jk}\neq \mu_{j\ell}$, and  $\mu_{1,jk}\neq \mu_{1,j\ell}$ if $k\neq \ell$ for  $j=1,...,M$ and $k,\ell=1,...,{\cal{K}}$.
\end{assumption}

Here, $M$ and ${\cal{K}}$ serve distinct roles in our model specification. Specifically, $M$   captures the number of latent types representing permanent unobserved heterogeneity across units, whereas increasing ${\cal{K}}$ enhances the flexibility of the i.i.d. error term distributions. Our primary interest in this paper lies in analyzing the permanent unobserved heterogeneity represented by $M$.


In our empirical application, we will analyze plant heterogeneity beyond Hicks-neutral technology differences by estimating the number of latent technology types. The following example provides a finite mixture specification for analyzing plant heterogeneity in input elasticities.


\begin{example} [Production Function Heterogeneity  with Conditionally Independent Shocks] \label{example-1}
To motivate our analysis, consider the panel data of $n$ plants over $T$  years, $\{\{O_{it}, V_{it}, L_{it},K_{it},Z_{it}\}_{t=1}^T\}_{i=1}^N$, where $O_{it}$, $V_{it}$,  $L_{it}$, $K_{it}$, and $Z_{it}$ represent output, material input, labor,  capital, and other observed characteristics---such as import status and industry classifications---of plant $i$ in year $t$, respectively. 

We consider the possibility of plant heterogeneity in input elasticities by specifying plant's production function with finite mixture.
Define the latent technology type $D_i \in \{ 1,2,\ldots, M\}$ to represent the production technology type of plant $i$. If $D_i  = j$, then plant $i$ is of type $j$.  The population proportion of type $j$ is denoted by $\alpha_j=\Pr(D=j)$.  For technology type $D_i=j$, the output is related to inputs through the following  production function with heterogenous output elasticities as
\begin{align}
O_{it}&=e^{\omega_{it}}  F_{jt}(V_{it},L_{it},K_{it}, Z_{it},\epsilon_{it})\ \text{ with }\ \omega_{it} =h_j(\omega_{it-1})+ \eta_{it},\quad  \label{prod}
\end{align}
where  $\omega_{it}$ is the serially correlated Hicks-neutral productivity shock that follows the exogenous first-order Markov process given technology type $j$ while $\epsilon_{it}$ is an unobserved random variable that affects the output elasticity of material input.

Assume that  $V_{it}$ is flexibly chosen by plant $i$  after observing the serially correlated productivity shock $\omega_{it}$ and the elasticity shock $\epsilon_{it}$ given output and material input prices, $P_{O,t}$ and $P_{V,t}$, while the value of  $\bs x_{it}$ and $\bs z_{it}$ including $(K_{it},L_{it})$ is determined before $\omega_{it}$ and  $\epsilon_{it}$ are realized. Then, the first-order conditions of profit maximization implies that
\begin{equation}\label{y-foc}
\frac{P_{V,t} V_{it}}{P_{O,t} O_{it}}=\frac{\partial \log  F_{jt}(V_{it},L_{it},K_{it}, Z_{it},\epsilon_{it})}{\partial \log V_{it}}.
\end{equation}
Therefore, the ratios of material input to output capture the output elasticities with respect to material input.

Conditional on being the $j$-th type and a vector of observed plant characteristics $\mathbf x_{it}$  (e.g., capital, labor,  and other plant characteristics), we assume that the log of output elasticity of material input is related to $\bs x_{it}$ and $\epsilon_{it}$ as
\begin{equation}\label{eq:d-v}
\log\left(\frac{\partial \log  F_{jt}(V_{it},L_{it},K_{it}, Z_{it},\epsilon_{it})}{\partial \log V_{it}} \right)= \mathbf x_{it}^\top \bs\beta_j + \epsilon_{it},
\end{equation}
wheret $\epsilon_{it}$ independently follows a  normal distribution with mean $\mu_j$ and variance $\sigma_j^2$, expressed as $\epsilon_{it}\overset{iid}{\sim}  \mathcal{N}(\mu_{j},\sigma_j^2)$ or a mixture distribution $ \epsilon_{it} \overset{iid}{\sim} \sum_{k=1}^{{\cal{K}}}\tau_{jk} \mathcal{N}(\mu_{jk},\sigma_j^2)$. Denoting the logarithm of materials-to-output ratios by
 \[
 y_{it}:=\log\left(\frac{P_{V,t} V_{it}}{P_{O,t} O_{it}}\right),
 \]
 the first-order conditions of profit maximization (\ref{y-foc}) and equation  (\ref{eq:d-v})  imply
 \begin{equation}\label{eq:spec1}
y_{it} = \mathbf x_{it}^\top \bs\beta_j  +\epsilon_{it}, \ \text{  where } \ \epsilon_{it} \overset{iid}{\sim} \mathcal{N}(\mu_{j},\sigma_j^2) \ \text{ or }\ \epsilon_{it} \overset{iid}{\sim} \sum_{k=1}^{{\cal{K}}}\tau_{jk} \mathcal{N}(\mu_{jk},\sigma_j^2).
 \end{equation}
Then, the  density function of $\{y_t\}_{t=1}^T$ given $\{\bs x_t\}_{t=1}^T$  is given by the finite mixture model (\ref{eq:fm}) with component density functions (\ref{eq:f1}) or (\ref{eq:f1-mixture}).
 \end{example}

This specification can be extended to a finite mixture variant of the correlated random effects framework \citep{Mundlak1978, Chamberlain1984}, where the unobserved individual-specific effects have finite discrete support. In this extension, the covariate vector \( \bs{x}_t \) may include either the individual-specific averages of the time-varying regressors or even the complete sequence of regressors \(\{\bs x_t\}_{t=1}^T\) over the entire observation period.

The strict exogeneity assumption on \(\bs{x}_t\) can be relaxed by adopting a sequential exogeneity framework, where \(\bs{x}_t\) includes covariates observed up to period \(t-1\). Under sequential exogeneity, coefficients become period-specific, as the dimension of \(\bs{x}_t\) may vary over time. Although this approach has the advantage of relaxing strict exogeneity, it also substantially increases the number of parameters to estimate, particularly as the number of mixture components grows. To maintain analytical tractability, we focus below on the simpler scenario where the dimension of \(\bs{x}_t\) is constant and the associated coefficients are time-invariant but extending our theoretical results to accommodate sequential exogeneity is conceptually straightforward.


The assumption of conditional independence of unobserved idiosyncratic shocks over time in the model (\ref{eq:fm}) may be seen as restrictive. To address this limitation, we extend our analysis to finite mixture dynamic panel data models that incorporate lagged outcomes as covariates by considering models where \( \{y_t\}_{t=1}^T \)  follows a first-order Markov process, conditional on both the latent component type and the sequence of covariates \( \{\bs{x}_t\}_{t=1}^T \). Specifically, for the baseline dynamic panel data model, the conditional probability density of \( \{y_t\}_{t=1}^T \) given  \( \{\bs{x}_t\}_{t=1}^T \) is given by (\ref{eq:fm}) with
\begin{align}\label{eq:fm-dynamic}
f( \{y_t\}_{t=1}^T|  \{\bs{x}_t\}_{t=1}^T;\bs\theta_j) =   f_1(y_1 | \bs{x}_1; \boldsymbol{\theta}_j) \prod_{t=2}^T f(y_t | y_{t-1}, \bs{x}_t; \boldsymbol{\theta}_j)
\end{align}
with
\begin{align}\label{eq:f1-dynamic}
 f_1(y_1|\bs x_1;\bs{\theta}_j) =  \frac{1}{\sigma_{1,j}} \phi\left(\frac{ y_{1} -\mu_{1,j} - \bs{x}^\top_{1}{\bs{\beta}}_{1,j}  }{\sigma_{1,j} } \right) \text{ and }
 f(y_t|y_{t-1},\bs x_t;\bs{\theta}_j) =  \frac{1}{\sigma_j} \phi\left(\frac{ y_{t} -\mu_{j} - \rho_j y_{t-1}- \bs{x}^\top_{t}{\bs{\beta}}_j  }{\sigma_j } \right),
\end{align}
where $\bs\theta_j=(\mu_{1,j},\mu_j,\sigma_{1,j},\sigma_{j},\bs\beta_{1,j}\t,\bs\beta_j\t,\rho_j)\t$.
In an extended case,  the densities within each component are flexibly modeled as mixtures of normal densities:
\begin{align}
 f_1(y_1|\bs x_1;\bs{\theta}_j) = \sum_{k=1}^{{\cal{K}}} \tau_{jk}\frac{1}{\sigma_{1,j}} \phi\left(\frac{ y_{1} -\mu_{1,jk} - \bs{x}^\top_{1}{\bs{\beta}}_{1,j}  }{\sigma_{1,j} } \right) \text{ and }\nonumber\\
 f(y_t|y_{t-1},\bs x_t;\bs{\theta}_j) =\sum_{k=1}^{{\cal{K}}} \tau_{jk}  \frac{1}{\sigma_j} \phi\left(\frac{ y_{t} -\mu_{jk} - \rho_j y_{t-1}- \bs{x}^\top_{t}{\bs{\beta}}_j  }{\sigma_j } \right)\label{eq:f1-dynamic-mixture}
\end{align}
with $\bs\theta_j=(\bs\tau_j,\bs\mu_{1,j}\t,\bs\mu_j\t,\sigma_{1,j}^2,\sigma_{j}^2,\bs\beta_{1,j}\t,\bs\beta_j\t,\rho_j)\t$, where $\bs\tau_j=(\tau_{j1},...,\tau_{{j\cal{K}}-1})\t$ and $\tau_{j{\cal{K}}} = 1 - \sum_{k=1}^{{\cal{K}}-1} \tau_{jk}$, $\bs\mu_{1,j}=(\mu_{1,j1},...,\mu_{1,j\cal{K}})\t$, and $\bs\mu_j=(\mu_{j1},...,\mu_{{j\cal{K}}})\t$.
The following example illustrates that the finite mixture dynamic panel data model  incorporates a scenario  in which the regression errors follow a first-order autoregressive (AR(1)) process.

  \begin{example}   [Production Function Heterogeneity with AR(1) Process] \label{example-2}
 To capture a possibility that output elasticities of material input are subject to persistent shocks, suppose that the elasticity shock \(\epsilon_{it}\) in (\ref{eq:spec1}) follows a first-order autoregressive (AR(1)) process. Specifically, conditional on the technology type \(D_i = j\), we consider
\begin{equation}\label{y-ar1}
\epsilon_{it} = \rho_j \epsilon_{it-1} + \xi_{it},
\end{equation}
where the innovation term \(\xi_{it}\) is independently and identically distributed (iid) as:
\begin{equation}\label{xi}
\xi_{it} \overset{iid}{\sim} \mathcal{N}(\mu_{j},\sigma_j^2) \quad\text{or}\quad \xi_{it} \overset{iid}{\sim} \sum_{k=1}^{{\cal{K}}} \tau_{jk} \mathcal{N}(\mu_{jk},\sigma_j^2).
\end{equation}
Substituting \(\epsilon_{it}=y_{it} - \mathbf{x}_{it}^\top \bs{\beta}_j\) into (\ref{y-ar1}) for \(D_i = j\) and rearranging terms leads to the following specification:
\begin{align}
y_{it} &=   \rho_j y_{it-1} + \tilde{\mathbf{x}}_{it}^\top \tilde{\bs{\beta}}_j + \xi_{it}, \label{eq:spec2}
\end{align}
where  \(\tilde{\bs{\beta}}_j := (\bs{\beta}_j^\top, -\rho_j\bs{\beta}_j^\top)^\top\), and \(\tilde{\bs{x}}_{it} := (\bs{x}_{it}^\top, \bs{x}_{it-1}^\top)^\top\).
By further appropriately specifying the initial distribution of \(y_1\) given \(\bs{x}_1\) and \(D_i = j\), the conditional density of \(\{y_t\}_{t=1}^T\) and \(\{\bs{x}_{t}\}_{t=1}^T\) can be represented in the form (\ref{eq:fm}) combined with (\ref{eq:fm-dynamic}), where the component-specific density functions are described by either (\ref{eq:f1-dynamic}) or (\ref{eq:f1-dynamic-mixture}).

 \end{example}

Factor augmented technological changes are also a popular way to incorporate heterogeneity in production functions beyond the Hicks-neutral technological change \citep{doraszelski2018measuring,Zhang2019,Raval2019}. The heterogeneity in the stochastic process of factor augmented technological changes can be modelled using finite mixture dynamic panel data model as the following example illustrates.

\begin{example}[Factor Augmented Technological Change with AR(1) Process] \label{example-3}
We modify a production function (\ref{prod}) by incorporating labor-augmented technological change  $\delta_{it}$  as follows \citep[cf.,][]{Demirer2022}:
\[
O_{it}=e^{\omega_{it}}\bar F_{jt}( \chi_{jt}(V_{it},e^{\delta_{it}} L_{it}),K_{it}, Z_{it},\epsilon_{it})\ \text{ with }\
\chi_{jt}(V_{it}, e^{\delta_{it}} L_{it}) = \left[ \alpha_{V,jt} V_{it}^{1/\varsigma_j} + \alpha_{L,jt} \left(e^{\delta_{it}} L_{it}\right)^{1/\varsigma_j} \right]^{\varsigma_j},
\]
where $\delta_{it}$ follows an AR(1) process:
\begin{equation}\label{delta}
\delta_{it} = \rho^{\delta}_j \delta_{it-1} + \eta^{\delta}_{it}\ \text{ with }\ \text{${\eta}^{\delta}_{it}\overset{iid}{\sim}   \mathcal{N}(0,(\sigma^{\delta}_j)^2)$ \ or\ ${\eta}^{\delta}_{it}\overset{iid}{\sim}\sum_{k=1}^{{\cal{K}}}\tau^{\delta}_{jk} \mathcal{N}(\mu^{\delta}_{jk},(\sigma_j^{\delta })^2)$}.
\end{equation}
 Assume that both $V_{it}$ and $L_{it}$ are flexibly chosen after observing $(\omega_{it},\delta_{it},\epsilon_{it})$ given intermediate and labor input prices, $P_{V,t}$ and $P_{L,t}$. Then, taking the ratios of the first-order conditions of profit maximization with respect to $L_{it}$ and $V_{it}$, respectively, gives $
\delta_{it} = \varsigma_j \log\left( \frac{P_{L,t}}{P_{V,t}} \cdot \frac{\alpha_{V,jt}}{\alpha_{L,jt}} \right) + (1 - \varsigma_j) \log\left( {V_{it}}/{L_{it}} \right).$
Then, in view of (\ref{delta}),  we may express the dynamic process of the log input ratios $y_{it}^\delta:= \log\left( {V_{it}}/{L_{it}} \right)$ as
\begin{equation}\label{eq:spec3}
y^{\delta}_{it} = \mu^{\delta}_{jt} + \rho^{\delta}_j y^{\delta}_{it-1} + \tilde{\eta}^{\delta}_{it},
\end{equation}
where
$
\mu^{\delta}_{jt} := \frac{\rho^{\delta}_j \varsigma_j}{1 - \varsigma_j} A_{jt-1} - \frac{\varsigma_j}{1 - \varsigma_j} A_{jt}$,  $A_{jt} := \log\left( \frac{P_{L,t}}{P_{V,t}} \cdot \frac{\alpha_{V,jt}}{\alpha_{L,jt}} \right),
$
and
$
\tilde{\eta}^{\delta}_{it} := \frac{\eta^{\delta}_{it}}{1 - \varsigma_j}.
$ By appropriately specifying the density functions of \( y^{\delta}_{i1} \), the conditional density of \( \{ y^{\delta}_{it} \}_{t=1}^T \) can be represented as (\ref{eq:fm}) with (\ref{eq:fm-dynamic}), where the component-specific density function are given in either (\ref{eq:f1-dynamic}) or (\ref{eq:f1-dynamic-mixture}).

\end{example}

The number of components, denoted by $M_0$, is defined as the smallest integer $M$ such that the data density of $\bs{w}$ admits the representation (\ref{eq:fm}). Consider a random sample of $n$ with a panel length of $T$ independent observations $\{\bs{W}_{i}\}_{i=1}^n$, where  $\bs{W}_i = \{ (Y_{it},\bs{X}\t_{it})\}_{t=1}^T$ is drawn from a true $M_0$-component density $g_M(\{y_{t}\}_{t=1}^T|\{\bs{x}_{t}\}_{t=1}^T;\bs\vartheta_{M_0}^*)$   in Equation (\ref{eq:fm}), where $\bs\vartheta_{M_0}^*$ represents the true parameter value.  Because component distributions can be identified only up to permutation, we assume that $\mu_1^*<\mu_2^*<\cdots <\mu_{M_0}^*$ for identification.\footnote{More generally, we may consider a lexicographical order: $\bs\theta_1^*<\bs\theta_2^*<\cdots<\bs\theta_{M_0}^*$.}


In Examples \ref{example-1}–\ref{example-3}, the number of components corresponds to the number of latent technology types, reflecting plant-level heterogeneity in  output elasticities of material input or   factor-augmented technology stochastic processes. We propose a procedure based on the Likelihood Ratio Test (LRT) for the finite mixture model specified in equation (\ref{eq:fm}) to evaluate the null hypothesis of $H_0: M=1$ against $H_1: M=2$. This provides a systematic method to assess whether significant heterogeneity exists across plants in terms of input elasticities or factor-augmented technologies.  Moreover, we further develop the LRT  for testing
\[
\text{$H_0: M=m$ \ against \ $H_1: M=m+1$}
\]
where $m$ is some known integer. By sequentially testing $H_0: M=m$ against $H_1: M=m+1$ for $m=1,2,...$, this procedure consistently estimates the number of distinct plant types.


Additionally, we consider determining the number of components using the Akaike Information Criterion (AIC) and Bayesian Information Criterion (BIC), with BIC consistency analyzed in Section \ref{sec:bic}. We further propose estimating a lower bound on the number of components based on the nonparametric rank test approach of \cite{Kasahara2014}, thus avoiding parametric assumptions. These methods are empirically applied to analyze the plant heterogeneity in production function using Examples \ref{example-1}–\ref{example-3} in Section \ref{sec:empirics}.

For notational brevity, let $\bs w = \{y_{t},\bs{x}_{t}\}_{t=1}^T$ and write $g_M(\bs w;\bs\vartheta_M):=g_M(\{y_{t}\}_{t=1}^T|\{\bs{x}_{t}\}_{t=1}^T;\bs\vartheta_{M})$ in (\ref{eq:fm}), and let $f(\bs w;\bs\theta_j):=f( \{y_t\}_{t=1}^T|  \{\bs{x}_t\}_{t=1}^T;\bs\theta_j)$ be component-specific density function so that   (\ref{eq:fm})  is written as
$g_M(\bs w;\bs\vartheta_M)
	=
	  \sum_{j=1}^M \alpha_j   f(\bs w;\bs\theta_j).$

\section{The failure of the regularity conditions}\label{sec:regularity}

We consider a random sample of $n$ observations $\{\boldsymbol{W}_{i}\}_{i=1}^n$, where $\boldsymbol{W}_i = \{ (Y_{it},\boldsymbol{X}^\top_{it})^\top \}_{t=1}^T$ from an $M_0$-component density $g_{M_0}(\boldsymbol{w}; \boldsymbol{\vartheta}_{M_0})$ defined in equation (\ref{eq:fm0}):
\begin{equation}\label{eq:fm0}
g_{M_0}(\boldsymbol{w}; \boldsymbol{\vartheta}_{M_0}^*)  = \sum_{j=1}^{M_0} \alpha^*_j  f(\boldsymbol{w};\boldsymbol{\theta}^*_j),
\end{equation}
where $\boldsymbol{\vartheta}_{M_0}^* = (\boldsymbol{\theta}_1^*,\boldsymbol{\theta}_2^*,\ldots,\boldsymbol{\theta}_{M_0}^*,\alpha_1^*,\ldots,\alpha_{M_0 -1}^*) \in \Theta_{\boldsymbol{\vartheta}_{M_0}}$ and $\alpha_{M_0}^*=1-\sum_{j=1}^{M-1}\alpha_j^*$.
 We assume $\mu_{1}^* < \mu_{2}^*, \ldots, < \mu_{M_0}^*$ in the true parameters for identification.

To examine the failure of the regularity conditions of the LRTS in panel finite mixture models,  consider testing the null hypothesis $H_0: M=1$ against the alternative hypothesis $H_1: M=2$, where  $\boldsymbol{W}_i = \{ (Y_{it}, \boldsymbol{X}_{it}^{\top})^{\top} \}_{t=1}^T$, drawn from a true one-component density $f(\boldsymbol{w};  \boldsymbol{\theta}^*)=\prod_{t=1}^Tf(y_t|\bs x_t;\bs\theta^*)$ or $f_1(y_1|\bs x_1;\bs\theta^*)\prod_{t=2}^Tf(y_t|y_{t-1},\bs x_t;\bs\theta^*)$.
The two-component model
$g_2(\boldsymbol{w}; \boldsymbol{\vartheta}_2) = \alpha f(\boldsymbol{w};   \boldsymbol{\theta}_1 ) + (1 - \alpha) f(\boldsymbol{w};  \boldsymbol{\theta}_2)$
can generate the true one-component density in two cases: (1) $\boldsymbol{\theta}_1 = \boldsymbol{\theta}_2 = \boldsymbol{\theta}^*$ and (2) $\alpha = 0$ or $1$. Consequently,  $H_0: M=1$ can be partitioned into two sub-hypotheses: $H_{01}: \boldsymbol{\theta}_1 = \boldsymbol{\theta}_2$ and $H_{02}: \alpha (1-\alpha) = 0$. The regularity conditions of the LRTS fail in any finite mixture models: under $H_{01}$, $\alpha$ is not identified, and the Fisher information matrix for the other parameters becomes singular; under $H_{02}$, $\alpha$ is on the boundary of the parameter space, and either $\boldsymbol{\theta}_1$ or $\boldsymbol{\theta}_2$ is not identified.

Analyzing the asymptotic distribution of the LRTS for the \textit{cross-sectional} normal mixture,  i.e., with $T=1$, is even more challenging because \citep[cf.,][]{chenli09as}: (i) the Fisher information for testing $H_{02}$ is not finite, (ii) the log-likelihood function is unbounded \citep{Hartigan1985}, and (iii) the first-order derivative of $g_2(\boldsymbol{w}; \boldsymbol{\vartheta}_2)$ with respect to $\sigma_j^2$ is linearly dependent on its second-order derivative with respect to $\mu_j$. We now examine whether problems (i)--(iii) are present in panel normal mixture models with $T\geq 2$.

Regarding problem (i),  the issue of the infinite Fisher information for testing $H_{02}$ arises in the panel normal mixture model.  The score for testing $H_{02}: \alpha=0$ takes the form
 \[
\left.\frac{\partial  g_2(\boldsymbol{W} ; \alpha,\bs\theta_1,\bs\theta_2)}{\partial \alpha} \right|_{\alpha=0,\bs\theta_2=\bs\theta^*}=
 \frac{f(\boldsymbol{W};\bs\theta_1)}{f(\boldsymbol{W};\bs\theta^* )}-1,
 \]
where $f(\boldsymbol{W};\bs\theta)= \prod_{t=1}^T \phi((Y_t-\mu-\bs X_t\t\bs\beta)/\sigma)/\sigma$. Then, $\mathbb{E}[\{f(\boldsymbol{W};\bs\theta_1)/f(\boldsymbol{W};\bs\theta^*  )-1\}^2]=\infty$ when $\sigma_1^2> 2\sigma^{*2}$. This infinite Fisher information also arises for normal mixture density (\ref{eq:f1-mixture}) in place of (\ref{eq:f1})  and/or for dynamic panel models (\ref{eq:fm-dynamic}).  For more details, please refer to Proposition \ref{prop:infinite_fisher}. Because the infinite Fisher information causes difficulty in deriving the asymptotic distribution under $H_{02}$, throughout this paper, we  focus on testing $H_{01}$ by restricting the value of $\alpha$ to be away from $0$ and $1$ by assuming that $\alpha\in[c_{1}, 1-c_{1}]$ for some positive constant $c_{1}$.   Appendix \ref{sec:LRT2} discusses the asymptotic distribution under $H_{02}$ under the constraint $\sigma_1^2\leq 2\sigma^{*2}-c$ for some $c>0$ based on the framework of \cite{andrews01em}.

Because we focus on $H_{01}$, our test may not have power against the local alternatives with $\alpha_n\rightarrow 0$ as discussed in Appendix \ref{sec:local}. 

 Related to problem (ii), the LRTS in normal mixture models with panel data becomes unbounded as the sample size $n$ goes to $\infty$. Define the likelihood ratio statistic with respect to the true parameter under $H_0$ as
$
LR_n^*(\boldsymbol{\vartheta}_{2}) := 2\left\{ \sum_{i=1}^n \log g_2(\boldsymbol{W}_i;\boldsymbol{\vartheta}_{2}) - \sum_{i=1}^n\log f(\boldsymbol{W}_i;\boldsymbol{\theta}^*)\right\},
$
where $g_2$ is the density of the two-component finite mixture distribution in (\ref{eq:fm}) with $M=2$ and $\boldsymbol{\theta}^*$ is the true parameter value under $H_0$.  
 \begin{proposition}\label{prop:unbounded_likelihood}
Suppose that the true model is described by the one-component model with $\boldsymbol{\theta}=\boldsymbol{\theta}^*$. The parameter space for $\sigma_j^2$ is given by $\Theta_{\sigma^2}=\mathbb{R}_{++}$. Then, for any positive constant $M > 0$,  (i) when the component-specific density function is given by (\ref{eq:fm}) with (\ref{eq:f1}) or  (\ref{eq:f1-mixture}), $\Pr \Big( \sup_{\boldsymbol{\vartheta}_{2}\in \Theta_{\boldsymbol{\vartheta}_{2}}} LR_n^*(\boldsymbol{\vartheta}_{2}) < M \Big) =\exp(-C_T n^{1/T})$  as $n\rightarrow \infty$, where $C_T$ for $T=2,3,..$ are some positive constant that depends on $M$ and $T$; (ii) when the density function is given by (\ref{eq:fm-dynamic}) with (\ref{eq:f1-dynamic}) or (\ref{eq:f1-dynamic-mixture}), $\Pr \Big( \sup_{\boldsymbol{\vartheta}_{2}\in \Theta_{\boldsymbol{\vartheta}_{2}}} LR_n^*(\boldsymbol{\vartheta}_{2}) < M \Big) = 0$ for any $n$.
\end{proposition}

The assumption of non-compactness for the parameter space of $\sigma_j^2$ is crucial for Proposition \ref{prop:unbounded_likelihood}. Without allowing parameters to approach boundary points (i.e., if compactness were imposed), one could not demonstrate the divergence of the likelihood ratio statistic, thus highlighting the source of unboundedness.

For problem (iii), consider the two-component panel mixture model without covariates given by
\begin{align}
g_2(\{y_t\}_{t=1}^T;\alpha,\bs\theta_1,\bs\theta_2)&=\alpha f(\{y_t\}_{t=1}^T ;\bs{\theta}_1)+(1-\alpha)  f(\{y_t\}_{t=1}^T ;\bs{\theta}_2)\quad\text{with}\nonumber\\
  f(\{y_t\}_{t=1}^T ;\bs{\theta}_j) &= \prod_{t=1}^T  \frac{1}{\sigma_j} \phi\left(\frac{ y_{t} -\mu_{j}   }{\sigma_j } \right)\quad\text{for $j=1,2$},
  \end{align} where $\bs\theta_j=(\mu_j,\sigma_j^2)\t$.
Then, when evaluated under the null hypothesis $H_{01}: \bs{\theta}_1 = \bs{\theta}_2 = \bs{\theta}^*$ with $\alpha(1 - \alpha) \neq 0$, the first-order derivative of $g_2(\{y_t\}_{t=1}^T;\alpha,\bs\theta_1,\bs\theta_2)$ with respect to $\sigma_j^2$ is \textit{not} linearly dependent on its second-order derivative with respect to $\mu_j$.

\begin{proposition}\label{prop:linear_independence}
For any $j \in {1,2}$, any $a,b \in \mathbb{R}$, and any $\bar \alpha \in [c, 1-c]$ for some constant $c > 0$, we have:
$$
\Pr\left(\left. \frac{\partial^2 g_2(\{Y_t\}_{t=1}^T; \alpha, \bs{\theta}_1, \bs{\theta}_2)}{(\partial \mu_j)^2} \right|_{ (\alpha,\bs\theta_1,\bs\theta_2)=(\bar\alpha,\bs\theta^*,\bs\theta^*)} =a + b \left. \frac{\partial  g_2(\{Y_t\}_{t=1}^T; \alpha, \bs{\theta}_1, \bs{\theta}_2)}{\partial \sigma_j^2} \right|_{ (\alpha,\bs\theta_1,\bs\theta_2)=(\bar\alpha,\bs\theta^*,\bs\theta^*)}\right) = 0.
$$
when $T\geq 2$.
\end{proposition}

Consequently, the panel mixture model (\ref{eq:fm}) with component density function (\ref{eq:f1}) is strongly identifiable, and  the best rate of convergence for estimating the mixing distribution is $n^{-1/4}$ when the number of components is unknown \citep[cf.,][]{chen95as}; see also our analysis under local alternatives in Appendix \ref{sec:local}.  The key assumption underlying this result is that the component density function can be factored into a product of density functions, reflecting limited dependence in regression errors across periods under either the conditional independence or Markov assumption.  In contrast, the strong identifiability does not hold for the cross-sectional normal mixture or the multivariate normal mixture with general dependence, and its convergence rate becomes as slow as $n^{-1/8}$ when the number of components is over-specified  \citep[cf.,][]{kasaharashimotsu15jasa,Kasahara2019}.

As in any finite mixture models, however, the standard asymptotic analysis breaks down in testing $H_{01}: \bs{\theta}_1 = \bs{\theta}_2 = \bs{\theta}^*$  because $\alpha$ is not identified under $H_{01}$; in addition, the first-order derivative at the true value $\bs\vartheta^*_2 = (\alpha, (\bs\theta^*)\t, (\bs\theta^*)\t)\t$ is linear dependent as
\begin{align}\label{eq:dependence}
 \nabla_{\bs{\theta}_1} \log g_2(\bs{w};\bs\vartheta^*_2)
= \frac{\alpha}{1 - \alpha}   \nabla_{\bs{\theta}_2} \log g_2(\bs{w};\bs\vartheta^*_2).
\end{align}
Consequently, we derive the asymptotic distribution of the LRTS using a fourth-order Taylor series approximation for the log-likelihood function.

\section{Consistency of the MLE}\label{sec:consistency}

To address the issue of unboundedness, following  \cite{hathaway85as}, we restrict the parameter space to ensure that the ratio of each component-specific variance to the unconditional variance is bounded away from zero, i.e., $\min_{j} \left( \frac{\sigma_j^2}{\sum_{k=1}^M \alpha_k \sigma_k^2} \right) \geq c_{2},$
for some constant \( c_{2} > 0 \). This constraint enforces a positive lower bound on the component variances when the unconditional variance is strictly positive with $\sum_{k=1}^M \alpha_k \sigma_k^2\geq c_3>0$. We also assume that $\alpha_j$ is bounded away from $0$ and $1$ by a small constant $c_1>0$ so that $\bs\alpha\in\Theta_{\bs\alpha,c_1}:=\{\bs\alpha: \alpha_j\in [c_1,1-c_1], \sum_{j=1}^{M-1}\alpha_j\leq 1-c_1\}$.   This assumption prevents issues related to infinite Fisher information, allowing us to focus on $H_{01}$. 
Let $\bs c=(c_{1},c_{2},c_3)\t$, and define a restricted parameter space for $\bs\vartheta_M$ for a model with the component-specific density function (\ref{eq:fm}) by
 \[
 {\bar \Theta}_{\bs\vartheta_M}(\bs c) := \{ \bs\vartheta_M \in  \Theta_{\bs\vartheta_M} : \bs \alpha\in \Theta_{\bs\alpha,c_1}, \ \min_j \sigma_j^2\geq c_2 \sum_{k=1}^M  \alpha_k\sigma_k^2,\   \sum_{k=1}^M \alpha_k \sigma_k^2\geq c_3 \}.
 \]
 For the dynamic panel model (\ref{eq:f1-dynamic}), we define a restricted parameter space for $\bs\vartheta_M$ analogously, but with additional restrictions on $\sigma_{1,j}$: $\min_{j}  \sigma_{1,j}^2 \geq c_{2} \sum_{k=1}^M \alpha_k \sigma_{1,k}^2 $ and $\sum_{k=1}^M \alpha_k \sigma_{1,k}^2\geq c_3$.


We consider the MLE of $\bs\vartheta_M$ over a restricted parameter space  ${\bar \Theta}_{\bs\vartheta_M}(\bs c)$ for the M-component model:
\begin{equation}\label{eq:MLE}
\ell_n^M(\widehat{\bs\vartheta}_M) = \sup_{{\bs{\vartheta}}_M \in {\bar \Theta}_{\bs\vartheta_M}(\bs c)} \ell_n^M({\bs\vartheta}_M),
\end{equation}
where
\begin{equation}\label{eq:log-like}
\ell_n^M({\bs{\vartheta}}_M) := \sum_{i=1}^n \log g_M(\boldsymbol{W}_i; {\bs{\vartheta}}_M)
\end{equation}
is the the log likelihood function of ${\bs{\vartheta}}_M$. The MLE $\widehat{\bs\vartheta}_M$ is potentially one of many possible global maximum points of $\ell_n^M({\bs{\vartheta}}_M)$ when $M$ is greater than the true number of components.

For $M>M_0$, define a set of parameter values for the $M$-component density that generates the true $M_0$-component density by
\[
\Theta_{\bs\vartheta_M}^*:=\left\{\bs\vartheta_M\in \Theta_{\bs\vartheta_M}: g_M(\bs w;\bs\vartheta_M)=g_{M_0}(\bs w;\bs\vartheta_{M_0}^*)\quad\text{for all $\bs w\in \bs{\mathcal{W}}$}\right\}.
\]
For example, when  $M_0=1$ with the true density function $f(\bs w;\bs\theta^*)$, we have   \[
\Theta_{\bs\vartheta_2}^*:= \left\{ (\alpha,\boldsymbol{\theta}_1,\boldsymbol{\theta}_2) \in  \Theta_{\bs\vartheta_2}: \boldsymbol{\theta}_1 = \boldsymbol{\theta}_2 = \boldsymbol{\theta}^*; \alpha=1  \text{ and } \bs\theta_1=\bs\theta^*; \alpha=0 \text{ and } \bs\theta_2=\bs\theta^* \right\},\] where
$\boldsymbol{\theta}_1$ and $\boldsymbol{\theta}_2$ are component-specific parameters.

The following proposition establishes the consistency of the MLE for over-specified mixture models. 

\begin{assumption}\label{assumption:consistency} (a) $\boldsymbol{X}$  has finite second moments and $\text{Pr}(\boldsymbol{X}^\top \boldsymbol{\beta}_j \neq \boldsymbol{X}^\top \boldsymbol{\beta}_j^*  ) > 0 $ for $\boldsymbol{\beta}_j  \neq \boldsymbol{\beta}_j^*$ for $j=1,2,...,M_0$. (b) For finite $M> M_0$, $\Theta_{\bs\vartheta_M}^*\cap {\bar \Theta}_{\bs\vartheta_{M}}(\bs c)$ is non-empty. (c)  ${\bar \Theta}_{\bs\vartheta_{M}}(\bs c)$ is compact. 

\end{assumption}

\begin{proposition}\label{prop:consistency}
Suppose that Assumptions \ref{assumption:K} and \ref{assumption:consistency} holds and that the data is generated from $M_0$-component models with the component-specific density function  (\ref{eq:fm}) with (\ref{eq:f1}) or  (\ref{eq:f1-mixture}), or  the component-specific density function (\ref{eq:fm-dynamic}) with (\ref{eq:f1-dynamic}) or  (\ref{eq:f1-dynamic-mixture}). Then, for any  finite $M>M_0$, $\inf_{\boldsymbol{\vartheta}_M \in \Theta_{\bs\vartheta_M}^*} ||\widehat{\boldsymbol{\vartheta}}_M - \boldsymbol{\vartheta}_M|| \to 0$ almost surely.
\end{proposition}
Consequently, Proposition \ref{prop:consistency} suggests that the MLE $\widehat{\boldsymbol{\vartheta}}_M$ converges almost surely to a set of parameters for which the true density function $g_{M_0}(\boldsymbol{w}; \boldsymbol\theta_{M_0}^*)$ emerges within the space of  overspecified $M$-component density functions.

\section{Likelihood ratio test for $H_0: M=1$ against $H_1: M=2$}\label{sec:LRT1}


We first consider testing   $H_0: M=1$ against  $H_1: M=2$.  Let $\widehat{\bs{\theta}}_0$ be the one-component MLE that maximizes the one-component log likelihood function $\ell_n^1(\bs\theta):= \sum_{i=1}^n \log f(\bs{W}_i; \bs{\theta}) $. Define the LRTS  of testing $H_{01}$  as \begin{align}\label{eq:LR_def}
	LR_n &:=   2 \left\{\ell_n^2(\hat{\bs{\vartheta}}_2) - \ell_n^1(\hat{\bs\theta}_0)  \right\},
\end{align}
where $\hat{\bs\vartheta}_2$ is the MLE under a restricted parameter space as defined by (\ref{eq:MLE}).

To derive the asymptotic distribution of the LRTS, we consider the following one-to-one reparameterization of $\bs{\theta}_1$ and $\bs{\theta}_2$ given $\alpha$ \cite[cf.,][]{Kasahara2012}:
\begin{equation}\label{eq:m0_repar2}
\begin{pmatrix}
\bs{\lambda} \\
\bs{\nu}
\end{pmatrix} := \begin{pmatrix}
\bs{\theta}_1 - \bs{\theta}_2 \\
\alpha \bs{\theta}_1 + (1 - \alpha) \bs{\theta}_2
\end{pmatrix} \text{ so that }
\begin{pmatrix}
\bs{\theta}_1 \\
\bs{\theta}_2
\end{pmatrix} = \begin{pmatrix}
\bs{\nu} + ( 1- \alpha) \bs{\lambda} \\
\bs{\nu} - \alpha \bs{\lambda}
\end{pmatrix},
\end{equation}
where $\bs{\nu}$ and $\bs{\lambda}$ are both $q \times 1$ reparameterized parameter vectors.

For the model (\ref{eq:fm}) with density   (\ref{eq:f1}), we have $\bs{\nu}=(\nu_\mu,\nu_\sigma,\bs{\nu}_{\bs\beta}\t)\t$ and  $\bs\lambda = (\lambda_{\mu},\lambda_{\sigma}, \bs\lambda_{\beta}\t)\t= ( \mu_1 - \mu_2, \sigma_1^2 - \sigma_2^2, (\bs \beta_1 - \bs\beta_2)\t)\t$ while for the model  (\ref{eq:fm}) with  (\ref{eq:f1-mixture}),  $\bs{\nu}$ and   $\bs\lambda$ are given as $\bs{\nu}=(\bs\nu_{\bs\tau}\t,\bs\mu\t,\nu_\sigma,\bs{\nu}_{\bs\beta}\t)\t$ and  $\bs\lambda = (\bs\lambda_{\bs\tau}\t,\bs\lambda_{\bs\mu}\t,\lambda_{\sigma}, \bs\lambda_{\bs\beta}\t)\t= ( \bs\tau_1\t-\bs\tau_2\t,\bs\mu_1\t - \bs\mu_2\t, \sigma_1^2 - \sigma_2^2, \bs \beta_1\t - \bs\beta_2\t)\t$. We may also define $\bs\nu$ and $\bs\lambda$ appropriately fort the dynamic models (\ref{eq:fm-dynamic}) with (\ref{eq:f1-dynamic}) or  (\ref{eq:f1-dynamic-mixture}). Here,   the parameter $\bs{\lambda}$ captures a deviation from the one-component model, such that $\bs \nu=\bs\theta^*$ and $\bs\lambda=\bs 0$  under $H_0: M=1$.

Define the elements of $\bs{\theta}$ and $\bs{\lambda}$ as $\bs{\theta} =(\theta_1,\theta_2,...,\theta_{q})\t$ and $\bs{\lambda}=(\lambda_1,\lambda_2,...,\lambda_{q})\t$.

Taking the derivative with respect to $\bs\lambda$ reveals that $\nabla_{\bs{\lambda}} \log g_2(\bs{w};\alpha,\bs{\nu} + ( 1- \alpha) \bs{\lambda},\bs{\nu} - \alpha \bs{\lambda})|_{\bs\lambda=\bs 0}  = \bs{0}$.
Therefore, the Fisher information matrix is singular, and the standard quadratic approximation fails. Consequently, we use the second-order derivative with respect to  $\bs{\lambda}$ to identify $\bs{\lambda}$.

Let $f^*$ and $\nabla f^*$ denote $f(\bs{W}; \bs{\theta}^*)$ and $\nabla f(\bs{W}; \bs{\theta}^*)$, respectively.
Define the vector $\bs{s}(\bs{W})$ as
\begin{align}\label{eq:s_1}
	 &\bs{s}(\bs{W}) = \begin{pmatrix}
		\bs{s}_{\bs{\nu}}(\bs W)   \\
		\bs{s}_{\bs{\lambda} \bs\lambda}(\bs W)
\end{pmatrix},
 \text{ where }  {\bs{s}_{\bs{\nu} }(\bs W) }:= \frac{\nabla_{\bs{\theta}} f^*}{f^*}   \text{ and } {\bs{s}_{\bs{\lambda} \bs\lambda}(\bs W) } :=\frac{1}{2}\widetilde{\text{vech}}_2\left(\frac{ {\nabla}_{\bs\theta\bs\theta\t} f^*}{f^*}\right),
\end{align}
where $\widetilde{\text{vech}}_2\left({ {\nabla}_{\bs\theta\bs\theta\t} f^*}/{f^*}\right)$ is a $q(q+1)/2\times 1$ vector that contains all unique elements of $q$-dimensional symmetric array ${ {\nabla}_{\bs\theta\bs\theta\t} f^*}/{f^*}$, multiply by its frequency, and stacks them into a vector.
 The function $\bs{s}_{\bs{\lambda} \bs\lambda}(\bs w)$ essentially comprises the second-order derivatives of the log-likelihood function with respect to $\bs\lambda$, serving as a score function for identifying $\bs{\lambda}$. We refer to $\bs{s}(\bs{w})$  as a score function.
An explicit expression for the score function $\bs{s}(\bs{w})$ can be derived using Hermite polynomials, as shown in Appendix \ref{sec:appendix_score_1}, from which the non-singularity of $\bs{\mathcal{I}}:=\mathbb{E}[\bs s(\bs W)\bs s(\bs W)\t]$ follows.






Partition $\bs{\mathcal{I}}$ and define
\begin{equation}\label{eq:information}
\begin{split}
\bs{\mathcal{I}} =
\begin{pmatrix}
 \bs{\mathcal{I}}_{\bs\nu}  & \bs{\mathcal{I}}_{\bs\nu\bs\lambda}  \\
 \bs{\mathcal{I}}_{\bs\lambda\bs\nu}  & \bs{\mathcal{I}}_{\bs\lambda\bs\lambda}
\end{pmatrix},
\quad \bs{\mathcal{I}}_{\bs{\nu}}  = \mathbb{E}[\bs s_{\bs\nu}(\bs W)  \bs s_{\bs\nu}(\bs W)\t], \quad \bs{\mathcal{I}}_{\bs\lambda\bs\nu}   = \mathbb{E}[\bs s_{\bs\lambda\bs\lambda }(\bs W)  \bs s_{\bs\nu}(\bs W)\t ], \quad \bs{\mathcal{I}}_{\bs\nu\bs\lambda}  = \bs{\mathcal{I}}_{\bs\lambda\bs\nu} \t, \\
 \bs{\mathcal{I}}_{\bs\lambda\bs\lambda}  =\mathbb{E}[\bs s_{\bs\lambda\bs\lambda }(\bs W)  \bs s_{\bs\lambda\bs\lambda}(\bs W)\t ],  \quad \bs{\mathcal{I}}_{\bs\lambda,\bs\nu}  = \bs{\mathcal{I}}_{\bs\lambda\bs\lambda}  - \bs{\mathcal{I}}_{\bs\lambda\bs\nu}  \bs{\mathcal{I}}_{\bs\nu} ^{-1} \bs{\mathcal{I}}_{\bs\nu\bs\lambda} , \quad\text{and}\quad \bs{G}_{\bs{\lambda},\bs\nu} := ( \bs{\mathcal{I}}_{\bs\lambda,\bs\nu}  )^{-1} \bs{S}_{\bs\lambda,\bs\nu} ,
\end{split}
\end{equation}
where $\bs{G}_{\bs\lambda,\bs\nu}  \sim \mathcal{N}(\bs 0,(\bs{\mathcal{I}}_{\bs\lambda,\bs\nu})^{-1})$.


Define $\widehat{\bs{t}}_{\lambda}$ as
\begin{equation}\label{eq:t_lambda_def}
r_{\lambda}(\widehat{\bs t}_{\bs \lambda}) = \inf_{\bs{t}_{\bs{\lambda}} \in \Lambda_{\bs\lambda}} r_{\lambda}(\bs{t}_{\bs{\lambda}}), \quad r_{\lambda} (\bs{t}_{\bs{\lambda}}) : = (\bs{t}_{\bs{\lambda}} - \bs{G}_{\bs{\lambda},\bs\nu})\t \bs{\mathcal{I}}_{\bs\lambda,\bs\nu}  (\bs{t}_{\bs{\lambda}} - \bs{G}_{\bs{\lambda},\bs\nu} )
\end{equation}
such that $\widehat{\bs{t}}_{\bs\lambda}$ is a projection of a Gaussian random variable $\bs{G}_{\bs{\lambda}}$ on a cone $\Lambda_{\bs\lambda}: =  \Big\{    v(\bs{\lambda})  : \bs{\lambda} \in \mathbb{R}^q\Big\}$,
where $v(\bs{\lambda})$  is a vector of unique elements of $\bs\lambda\bs\lambda\t$, where the diagonal elements are divided by
2 :
\begin{equation}
\label{eq:v}
v(\bs{\lambda})   := (\lambda_1^2, \lambda_2\lambda_1,\lambda_2^2,..., \lambda_q\lambda_1,..., \lambda_q\lambda_{q-1},\lambda_q^2)\t.
\end{equation}

 The following proposition establishes the asymptotic  distribution of the LRTS  under the null hypothesis $H_0: M=1$.
\begin{assumption}\label{assumption:LRT1}
	(a) \(\bs{X}\) has finite $9$-th moments.  (b) $\mathbb{E}[\boldsymbol U \boldsymbol{U\t}]$ is non-singular, where $\boldsymbol U = [1, \boldsymbol{ X\t}]\t$. (c) $\bs{\mathcal{I}}$ is non-singular.  (d) The true parameter $\bs\theta^*$ is an interior of the parameter space $\Theta$.
\end{assumption}
\begin{proposition}\label{prop:t2_distribution}
Suppose that Assumptions \ref{assumption:K}-\ref{assumption:LRT1} hold. Suppose that the data is generated under the null hypothesis $H_0: M = 1$. Then, for the model with the component-specific distribution (\ref{eq:fm}) with (\ref{eq:f1}) or  (\ref{eq:f1-mixture}), or   (\ref{eq:fm-dynamic}) with (\ref{eq:f1-dynamic}) or  (\ref{eq:f1-dynamic-mixture}),  $LR_n  \overset{d}{\to} (\widehat{\bs{t}}_{\bs\lambda} )\t \bs{\mathcal{I}}_{\bs\lambda,\bs\nu}  \widehat{\bs{t}}_{\bs\lambda}$.
	\end{proposition}

For Assumption \ref{assumption:LRT1}(c), the non-singularity of $\bs{\mathcal{I}}$ is proved in Lemma \ref{lemma:expansion} for the model with component-specific density function (\ref{eq:fm}) with (\ref{eq:f1}). For other models, while analytically proving the non-singularity of $\bs{\mathcal{I}}$ is not straightforward, each element of the score function $\bs s(\bs W)$ is a linear transformation of products of posterior probabilities and Hermite polynomials of different degrees. Therefore, we expect $\bs{\mathcal{I}}=\mathbb{E}[\bs s(\bs W) \bs s(\bs W)^\top]$ to be finite and non-singular.

%

\section{Likelihood ratio test for $H_0: M=m$ against $H_1: M=m+1$ }\label{sec:LRT2}
If the data is generated from the $M_0$-component model (\ref{eq:fm0}),
the $(M_0 + 1)$-component model
\begin{equation}\label{eq:fm0_1}
		g_{M_0+1}(\boldsymbol{w}; \boldsymbol{\vartheta}_{M_0+1}) =\sum_{j=1}^{M_0+1} \alpha_j  f(\boldsymbol{w};\boldsymbol{\theta}_j)
\end{equation}
gives rise to the true density (\ref{eq:fm0}) in two cases: (i) two components have identical parameters or (ii) one component has zero mixing proportion. Accordingly, we partition the null hypothesis $H_0: M=M_0$ into two as $H_0 = H_{01}\cup H_{02}$, with $H_{01}: \boldsymbol{\theta}_h = \boldsymbol{\theta}_{h+1} = \boldsymbol{\theta}_{h}^*$ for some $h=1,\ldots,M_0$ and $\alpha_j>0$ for all $j$, and $H_{02}: \alpha_h = 0$ for some $h=1,\ldots,M_0+1$.

Partition  $H_{02}$ as $H_{02}=\cup_{h=1}^{M_0}H_{0,2h}$, where $H_{0,2h}: \alpha_h=0$.
Define the subset of $\Theta_{\vartheta_{M_0+1}}^*$  corresponding to $H_{0,2h}$ as
	$\Theta_{\bs \vartheta_{M_0+1},2h}^* = \{ \bs \vartheta_{M_0+1} \in \Theta_{\bs \vartheta_{M_0+1}} : \alpha_h = 0; (\alpha_j, \bs\theta_j)  = (\alpha_j^*, \bs\theta_j^*)\text{ for } j < h; (\alpha_j, \bs\theta_j)= (\alpha_{j-1}^*, \bs\theta_{j-1}^*)\text{ for } j > h\}$. 
Because $\bs\theta_h$ is not identified when $\alpha_h = 0$, we take the supremum of the variance of $\nabla_{\alpha_h} \log g_{M_0 + 1}(\bs W_i; \bs \vartheta_{M_0+1})$ over $\bs \vartheta_{M_0+1} \in \Theta_{\bs \vartheta_{M_0+1},2h}^*$ to examine the finiteness of the Fisher information matrix for testing $H_{0,2h}: \alpha_h = 0$.  The Fisher information is infinite unless there is an a priori restriction on the values of component-specific variance, $\sigma_j^2$.
\begin{proposition}\label{prop:infinite_fisher}
	$\sup_{\vartheta_{M_0+1} \in \Theta_{\bs \vartheta_{M_0+1},2h}^*} \E [  \{ \nabla_{\alpha_h}  \log g_{M_0+1}  (\bs W, \vartheta_{M_0+1} )\}^2  ]  < \infty$ if and only if $\max \{ \sigma^2: \sigma \in \Theta_{\sigma} \} < 2 \max \{ \sigma_1^{2*},\ldots,\sigma_{M_0}^{2*}\}$ for $g_{M_0+1}  (\bs w, \vartheta_{M_0+1} )$ in (\ref{eq:fm0_1}) with the component-specific density function  (\ref{eq:fm}) with (\ref{eq:f1}) or  (\ref{eq:f1-mixture}), or  the component-specific density function (\ref{eq:fm-dynamic}) with (\ref{eq:f1-dynamic}) or  (\ref{eq:f1-dynamic-mixture}).
\end{proposition}
Because the restriction on the values of $\sigma_j^2$ in Proposition \ref{prop:infinite_fisher} is difficult to justify and not easy to enforce in practice, we focus on testing $H_{01}$.

Partition  $H_{01}$ as $H_{01}=\cup_{h=1}^{M_0}H_{0,1h}$, where $H_{0,1h}: \bs\theta_{h}=\bs\theta_{h+1}$ with
$\mu_1<\cdots < \mu_h=\mu_{h+1}<\cdots < \mu_{M_0+1}$. We impose these inequality constraints on $\mu_j$ for component identification.
There are $M_0$ ways to describe the $M_0$ component null model in the space of $(M_{0}+1)$ component models, and each way corresponds to the null hypothesis of $H_{0,1h}: \bs\theta_{h}=\bs\theta_{h+1}$ for $h=1,2,..., M_0$. Testing $H_{0,1h}: \bs\theta_{h}=\bs\theta_{h+1}$ in the $M_0$-component null models is similar to  testing $H_{01}: \bs\theta_1=\bs\theta_2$ in the one-component null model in Section \ref{sec:LRT1}.

Define the subset of $\Theta_{\vartheta_{M_0+1}}$ corresponding to  $H_{0,1h}$ as
\begin{equation}\label{eq:Upsilon_1}
	\Theta_{\bs \vartheta_{M_0+1},1h}^*  : =
	 \left\{
\begin{split}\bs{\vartheta}_{M_0+1} \in \Theta_{\vartheta_{M_0+1}}: \alpha_h,\alpha_{h+1}>0; \alpha_{h} + \alpha_{h+1} = \alpha_{h}^* \text{ and } \bs{\theta}_{h} = \bs{\theta}_{h+1} = \bs{\theta}_{h}^*   ;  \alpha_j = \alpha_{j}^* 
	\\
 \text{ and } \bs{\theta}_j = \bs{\theta}_j^{*} \text{ for } 1 \le j < h ; \alpha_j = \alpha_{j-1}^* \text{ and } \bs{\theta}_j = \bs{\theta}_{j-1}^* \text{ for } h +1 \le j \le M_0 + 1
\end{split} \right\}
\end{equation}
for $h = 1,\ldots,M_0$.   The set $\Theta_{\bs \vartheta_{M_0+1},1}^*  := \cup_{h=1}^{M_0} \Theta_{\bs \vartheta_{M_0+1},1h}^*$ corresponds to $H_{01}=\cup_{h=1}^{M_0}H_{0,1h}$.

We first develop a test statistic for testing test $H_{0,1h}$ by restricting the estimators under the $(M_0 + 1)$-component  model to be in a neighbourhood of $\Theta_{\bs \vartheta_{M_0+1},1}^*$.

Recall that $\mu_1^{*} < \mu_2^* \ldots < \mu_{M_0}^*$. Let $ \underline{\Theta}_{\mu}$ and $\overline{\Theta}_{\mu}$ denote the lower and upper bounds of $\Theta_{\mu}$, respectively.  Define
\begin{equation}\label{theta-partition}
\Theta_{\bs\theta,h}^* =
\begin{cases}
  [\underline{\Theta}_\mu, \frac{\mu_1^* + \mu_2^*}{2}] \times \Theta_\beta \times \Theta_{\sigma^2}, & h = 1 \\
  [\frac{\mu_{h-1}^* + \mu_h^*}{2}, \frac{\mu_h^* + \mu_{h+1}^*}{2}] \times \Theta_\beta \times \Theta_{\sigma^2}, & 2 \leq h \leq M_0 - 1 \\
  [\frac{\mu_{M_0-1}^* + \mu_{M_0}^*}{2}, \overline{\Theta}_\mu] \times \Theta_\beta \times \Theta_{\sigma^2}, & h = M_0
\end{cases}
\end{equation}
  Then, $\{\Theta_{\bs\theta,h}^*\}_{h=1}^{M_0}$ is a partition of $\Theta_{\bs\theta}$ such that
  $\Theta_{\bs\theta,h}^*$ is a neighbourhood containing $\bs\theta_h^{*}$ but not $\bs\theta_j^{*}$ for $j \neq h$.
For $h=1,\ldots M_0$,  define a restricted parameter space  $\bs{\Psi}_h^* \subset \bar\Theta_{\bs\vartheta_{M_0 + 1}}(\bs c)$ as
\begin{equation}\label{Omega}
	\bs{\Psi}_h^* = \left\{ \begin{split} \bs\vartheta_{M_0+1}\in\bar\Theta_{\bs\vartheta_{M_0 + 1}}(\bs c):
		\sum_{j=1}^{M_0 + 1} \alpha_j = 1;  \bs{\theta}_j \in \Theta_{\bs\theta,j}^* \text{ for } j=1,\ldots,h-1;  \\
		 \bs{\theta}_h,\bs{\theta}_{h+1} \in \Theta^*_{\bs\theta,h};
        \bs{\theta}_j \in \Theta^*_{\bs\theta,j-1} \text{ for } j=h+2,\ldots,M_0+1.
	\end{split}  \right\}
\end{equation}
Note that $\bs{\Psi}_h^* \cap \Theta_{\bs \vartheta_{M_0+1},1h}^*  \neq \emptyset$ and $\bs{\Psi}_h^* \cap \Theta_{\bs \vartheta_{M_0+1},1l}^*  = \emptyset$ if $h \neq l$, and $\cup_{h=1}^{M_0}\bs{\Psi}_h^*= \bar\Theta_{\bs\vartheta_{M_0 + 1}}(\bs c)$.

Let $\widehat{\bs{\Psi}}_h^* $ and $\widehat{\Theta}_{\bs \theta,h}^* $ be consistent estimators of ${\bs{\Psi}}_h^* $ and ${\Theta}^*_{\bs\theta,h}$, which can be constructed from a consistent estimator of $\bs{\vartheta}_{M_0}^*$ in the $M_0$-component model. We test $H_{0,1h}: \bs\theta_{h}=\bs\theta_{h+1}$ by estimating the $(M_0+1)$-component model under the restriction that $\bs\vartheta^{M_0+1}\in \widehat{\bs{\Psi}}_h^*$.

For $h=1,2,..., M_0$,  define the \textit{local} MLE that maximizes the log-likelihood function of the $(M_0+1)$-component model under the constraint that $\bs\vartheta_{M_0+1}\in \widehat{\bs\Psi}_h^*$ in (\ref{Omega}) by
\[
 \ell^{M_0+1}_n (\widehat{\bs{\vartheta}}_{M_0 + 1}^h)   = \arg\sup_{\bs{\vartheta}_{M_0 + 1}  \in \widehat{\bs\Psi}_h^* }   \ell^{M_0+1}_n (\bs{\vartheta}_{M_0+1}).
\]

Under $H_{0}: M=M_0$, $\bs\Psi_h^*$ contains a set of  parameters $ \Theta_{\bs \vartheta_{M_0+1},1h}^*$ defined in (\ref{eq:Upsilon_1}) such that $g_{M_0 + 1}(\bs{w}; \bs{\vartheta}_{M_0 + 1})$ is equal to the true density $g_{M_0}(\bs{w}; \bs{\vartheta}_{M_0 }^*)$ for any $\bs{\vartheta}_{M_0 +1}\in \Theta_{\bs \vartheta_{M_0+1},1h}^*$. Then, by the analogous argument  in the proof of Proposition \ref{prop:consistency}, the local MLEs are consistent, i.e., under  $H_0: M = M_0$,   for $h=1,2,...,M_0$, $\inf_{\bs{\vartheta}_{M_0+1} \in  \Theta_{\bs \vartheta_{M_0+1},1h}^*} |\widehat{\bs{\vartheta}}_{M_0+1}^h - \bs{\vartheta}_{M_0+1}| \overset{p}{\to} 0$.


Consider the local LRTS for testing $H_{0,1h}: \bs\vartheta_h = \bs\vartheta_{h+1}$ defined by
\begin{equation*}
LR^{M_0,h}_n := 2\left\{ \ell^{M_0+1}_n (\widehat{\bs{\vartheta}}_{M_0+1}^h) - \ell^{M_0}_n (\widehat{\bs{\vartheta}}_{M_0})\right\}\quad\text{for $h=1,2,...,M_0$},
\end{equation*}
where  $LR^{M_0,h}_n$ converges in distribution to the random variable   $(\widehat{\bs{t}}_{\bs\lambda}^h )\t \bs{\mathcal{I}}^h_{\bs\lambda,\bs\nu}  \widehat{\bs{t}}^h_{\bs\lambda}$, as defined in (\ref{eq:t_h}), which is analogous to
 $(\widehat{\bs{t}}_{\bs\lambda} )\t \bs{\mathcal{I}}_{\bs\lambda,\bs\nu}  \widehat{\bs{t}}_{\bs\lambda}$ in (\ref{eq:t_lambda_def}) for testing $H_0: M=1$.

Because $\cup_{h=1}^{M_0}\bs{\Psi}_h^*= \bar\Theta_{\bs\vartheta_{M_0 + 1}}(\bs c)$, the LRTS is identical to the maximum of the local LRTS over $h=1,2,...,M_0$:
\begin{align}
LR^{M_0}_n &:= 2\left\{ \ell^{M_0+1}_n (\hat{\bs{\vartheta}}_{M_0+1}) - \ell_n^{M_0}(\widehat{\bs{\vartheta}}_{M_0})\right\}=\max\{LR^{M_0,1}_n,...,LR_n^{M_0,M_0}\},\label{eq:LR_M0_max}
\end{align}
where $\hat{\bs{\vartheta}}_{M_0+1}$ is the MLE defined by (\ref{eq:MLE}).
Then, because $LR^{M_0,h}_n\overset{d}{\rightarrow}(\widehat{\bs{t}}_{\bs\lambda}^h )\t \bs{\mathcal{I}}^h_{\bs\lambda,\bs\nu}  \widehat{\bs{t}}^h_{\bs\lambda}$, the asymptotic distribution of $LR^{M_0}_n$ is characterized by the maximum of $M_0$ random variables.

\begin{assumption}\label{assumption:LRT2} (a) Assumption \ref{assumption:LRT1}(a)(b) hold, (b) $\bs{\tilde{\mathcal{I}}}$ defined in (\ref{eq:I_m0}) is finite and non-singular, (c)   The true parameter $\bs\vartheta_{M_0}^*$ is an interior of the parameter space $\bar\Theta_{\bs\vartheta_{M_0}}(\bs c)$. 
\end{assumption}

\begin{proposition}\label{prop:tm0_distribution}
 Suppose that Assumptions \ref{assumption:K}, \ref{assumption:consistency}, and \ref{assumption:LRT2} hold. Under the null hypothesis $H_0: M = M_0$,  for the model with the component-specific distribution (\ref{eq:fm}) with (\ref{eq:f1}) or  (\ref{eq:f1-mixture}), or   (\ref{eq:fm-dynamic}) with (\ref{eq:f1-dynamic}) or  (\ref{eq:f1-dynamic-mixture}), $LR_n^{M_0}  \overset{d}{\to} \max\{
(\widehat{\bs{t}}_{\bs\lambda}^1 )\t \bs{\mathcal{I}}^1_{\bs\lambda,\bs\nu}  \widehat{\bs{t}}^1_{\bs\lambda},...,(\widehat{\bs{t}}_{\bs\lambda}^{M_0} )\t \bs{\mathcal{I}}^{M_0}_{\bs\lambda,\bs\nu}  \widehat{\bs{t}}^{M_0}_{\bs\lambda}\}$,
where $(\widehat{\bs{t}}_{\bs\lambda}^h )\t \bs{\mathcal{I}}^h_{\bs\lambda,\bs\nu}  \widehat{\bs{t}}^h_{\bs\lambda}$ for $h=1,...,M_0$ is defined in (\ref{eq:t_h}).
	\end{proposition}
	
In practice, we implement parametric bootstrap to obtain the bootstrap p-value for  testing $H_0: M= m$ against $H_1: M= m+1$.

\section{Sequential hypothesis testing}\label{sec:sht}
 To estimate  the number of components, we sequentially test $H_0: M=m$ against $H_1: M=m+1$  starting from $m=1$, and then $m=2,\ldots,\overline  M$, where $\overline M$ is the  upper bound for the number of components, which is assumed to be known and larger than $M_0$.
 The first value for $r$ that leads to a nonrejection of $H_0$ gives our estimate for $M_0$.  \cite{Robin2000} develop a similar sequential hypothesis test for estimating the rank of a matrix.

For $m=1, \ldots, \overline M$, let $c^{m}_{1-q_n}$ denote the $100(1-q_n)$ percentile of the  cumulative distribution function of a random variable $ \max \{(\widehat{\bs{t}}^1_{\bs\lambda})^\top \bs{\mathcal{I}}^1_{\bs\lambda,\bs\eta}  \widehat{\bs{t}}^1_{\bs\lambda} ,   \ldots, (\widehat{\bs{t}}^{m}_{\bs\lambda})^\top \bs{\mathcal{I}}^{m}_{\bs\lambda,\bs\eta}  \widehat{\bs{t}}^{m}_{\bs\lambda}  \}$ for testing $H_0: M=m$ in Propositions \ref{prop:tm0_distribution}. Let $\widehat c^{m}_{1-q_n}$ be a consistent estimator of $c^{m}_{1-q_n}$. Then, our estimator based on sequential hypothesis testing (SHT, hereafter) is defined as
\begin{align}\label{M-hat}
\widehat M_{\text{LRT}} & = \min_{M \in \{0, \ldots ,\bar M\}} \{M: LR_n^m \geq \widehat c^m_{1-q_n} \text{ for } m=1, \ldots ,M-1, \text{ and } LR_n^M < \widehat c^M_{1-q_n} \}.
\end{align}
The estimator  $\widehat M_{\text{LRT}}$  depends on the choice of the significance level $q_n$. The following proposition states that $\widehat M_{\text{LRT}}$ converge to $M_0$ in probability as $n \rightarrow \infty$ when $-n^{-1}\log q_n=o(1)$ and  $q_n=o(1)$.

 Let
 \[
 Q_n^M( \bs{\vartheta}_M):=n^{-1} \sum_{i=1}^n \log g_M(\bs{W}_i; \bs{\vartheta}_M) \text{ and } Q^M( \bs{\vartheta}_M):= \mathbb{E}[\log g_M(\bs{W}_i; \bs{\vartheta}_M)].
 \] 

 \begin{assumption} \label{assumption:Q_M} For $M=1,...,M_0$,
(a)  $Q^M( \bs{\vartheta}_M)$ has a unique maximum at $\bs{\vartheta}_M^*$ in $\Theta_{\bs{\vartheta}_M}$;  (b) $\Theta_{\bs{\vartheta}_M}$ is compact;
(c) $\bs{\vartheta}_M^*$ is interior to $\Theta_{\bs{\vartheta}_M}$; (d) $B^M(\bs{\vartheta}_M^*):= \E \left\{\nabla_{\bs{\vartheta}_M}  \log g_M(\bs{W}_i; \bs{\vartheta}_M) \nabla_{\bs{\vartheta}_M^\top} \log g_M(\bs{W}_i; \bs{\vartheta}_M)\right\}$ is non-singular;  (e) $A^M(\bs{\vartheta}_M^*):= \E \left\{\nabla_{\bs{\vartheta}_M\bs{\vartheta}_M^\top}  \log g_M(\bs{W}_i; \bs{\vartheta}_M)\right\}$ has a constant rank in some open neighborhood of $\bs{\vartheta}_M^*$; and
(f)  $Q^{M+1}( \bs{\vartheta}_{M+1}^*)  - Q^{M}( \bs{\vartheta}_M^*) > 0$ for $M=1,...,M_0-1$.
\end{assumption}

\begin{proposition}\label{prop:sht}
Suppose that $M_0<\overline M$ and Assumptions \ref{assumption:K}, \ref{assumption:LRT1}, \ref{assumption:LRT2}, and \ref{assumption:Q_M} hold. If  we choose $q_n$ such that  $-n^{-1}\log q_n=o(1)$ and  $q_n=o(1)$, then $\widehat M_{\text{LRT}}-M_0=o_p(1)$.
\end{proposition}

Assumptions \ref{assumption:Q_M}(a)--(e) ensure the consistency and asymptotic normality of $\widehat{\bs{\vartheta}}_M$ for $M\leq M_0$, where (c)--(e) correspond to Assumption A6 of \cite{white82em}.  Per Assumption \ref{assumption:Q_M}(f), the Kullback--Leibler information criterion of the model relative to the true $M_0$-component model strictly decreases as the number of components $M$ increases for $M < M_0$.

%

Selecting the significance level $q_n$ in finite samples is challenging. Due to the complexity of analyzing the optimal choice of $q_n$, we leave this for future research and recommend presenting results across conventional levels (1\%, 5\%, 10\%).

To compare BIC and LRT performance across different $q_n$ values, we also recommend conducting a simulation exercise similar to a parametric bootstrap; researchers can generate datasets from the estimated model with $M$ components, apply methods with different values of $q_n$, and analyze the frequency distribution of estimated components as shown in Tables \ref{tab:empirical_test_normal}--\ref{tab:empirical_test_mixture} in Section \ref{sec:simulation}.

%

\section{Consistent estimation of the number of components by the penalized likelihood method}\label{sec:bic}


We may also estimate the number of components by the penalized maximum likelihood estimator
\[
\widehat{M}_{PL} = \arg\max_{M\in\{1,2,...,\overline{M}\}}  p\ell_{n}^M(\hat{\bs{\vartheta}}_{M}),
\]
where
\[
p\ell_{n}^M(\hat{\bs{\vartheta}}_{M}):=\ell_n^M(\hat{\bs{\vartheta}}_{M}) - p_{n,k_M},
\]
and $\hat{\bs{\vartheta}}_{M}$ is the MLE defined by (\ref{eq:MLE}) while $ p_{n,k_M}$ is a penalization term with $k_M:= \text{dim}(\bs{\vartheta}_{M})$ representing the number of estimated parameters for  mixture models (\ref{eq:fm}).  

The following proposition states that the penalized maximum likelihood estimator of the number of components is consistent under the regularity condition.
\begin{assumption}\label{assumption: penalty}
For all \( n > 1 \), the penalty function $p_{n,k}$ satisfies (a) $p_{n,k+1}\geq p_{n,k}>0$, (b) $\lim_{n \to \infty} p_{n,k} = \infty$,  (c) $p_{n,k} = o(n)$,
and (d)
$\lim_{n \to \infty} \frac{p_{n,k'}}{p_{n,k}} > 1$ for $k<k'\leq \overline{M}$.
\end{assumption}
\begin{assumption}\label{assumption:bound}
For $M\in\{M_0+1,M_0+2,...,\overline M\}$, $\ell_{n}^M(\hat{\bs\vartheta}_{M})-\ell_n^{M_0}(\hat{\bs\vartheta}_{M_0})=O_p(1)$.
\end{assumption}

\begin{proposition}[Consistent estimation by the maximum penalized likelihood estimator]\label{proposition:bic}
Suppose that Assumptions \ref{assumption:Q_M}-\ref{assumption:bound} hold. Then, $\widehat{M}_{PL} \overset{p}{\rightarrow} M_0$.
\end{proposition}


Assumption \ref{assumption: penalty} corresponds to Assumption (C1) of \cite{keribin00sankhya}, specifying conditions for penalty functions. The BIC penalty function, given by \( p_{n,k} = \frac{k}{2}\log(n) \), satisfies this assumption, whereas the AIC penalty function, \( p_{n,k} = k\), does not satisfy Assumption \ref{assumption: penalty}(b).


Assumption \ref{assumption:bound} is a key high-level condition preventing overestimation. In Appendix \ref{app:bic}, we provide explicit conditions under which Assumption \ref{assumption:bound} holds by extending the previous asymptotic analysis, originally developed for testing $H_{0}: M = M_0$ against $H_{1}: M = M_0 + 1$, to a more general setting of testing $H_{0}: M = M_0$ against $H_{1}: M = M_1$, where $M_1 > M_0 + 1$. Under this extension, we derive conditions ensuring that $\ell_{n}^{M_1}(\hat{\bs\vartheta}_{M_1}) - \ell_n^{M_0}({\bs{\vartheta}}_{M_0}^*) = O_p(1)$ for $M>M_0+1$. The primary challenge encountered in generalizing our earlier analysis is the increased severity of the singularity in the Fisher Information matrix, which necessitates expansions of order higher than second-order in the asymptotic analysis.

To address the singularity of the Fisher information matrix, we reparameterize the model and expand the density ratio
$
 {g_{M}(\bs{w}; \bs\vartheta_{M})}/{g_{M_0}(\bs{w}; \bs\vartheta_{M_0}^*)} - 1$
beyond the second order. This expansion yields a quadratic-form approximation of the log-likelihood difference \( \ell_n^{M}(\hat{\bs{\vartheta}}_{M}^j) - \ell^{M_0}_n( {\bs{\vartheta}}_{M_0}^*) \) that captures higher-order terms necessary, building on a generalized version of Le Cam’s differentiable in quadratic mean (DQM) expansion  \citep{liushao03as,kasahara2018arXiv}. We then demonstrate that \( \ell_n^{M}(\hat{\bs{\vartheta}}_{M}^j)-\ell^{M_0}_n({\bs{\vartheta}}_{M_0}^*)=O_p(1) \), analogously to arguments used in the proof of Proposition \ref{prop:tm0_distribution}. Appendix \ref{app:bic} contains the detailed discussion and Assumption \ref{assumption:BIC} presents the conditions for obtaining $\ell_n^{M}(\hat{\bs{\vartheta}}_{M}^j)-\ell^{M_0}_n({\bs{\vartheta}}_{M_0}^*)=O_p(1)$.   Assumption \ref{assumption:BIC} (b) is a key condition requiring that the unique elements of  ${\nabla_{\bs\theta_h} f(\bs{w}; \bs\theta_h^*)}/{g_{M_0}(\bs{w}; \bs\vartheta_{M_0}^*)}$,
${\nabla_{\bs\theta_h\otimes \bs\theta_h} f(\bs{w}; \bs\theta_h^*)}/{g_{M_0}(\bs{w}; \bs\vartheta_{M_0}^*)}$, ...,
${\nabla_{\bs\theta_h^{\otimes  p_h}} f(\bs{w}; \bs\theta_h^*)}/{g_{M_0}(\bs{w}; \bs\vartheta_{M_0}^*)}$ are linearly independent and their expectation is finite for $h=1,2,...,M_0$, where the value of $p_h$ indicates a necessary order for  identification.

The following proposition formalizes the results derived in Appendix \ref{app:bic}.

\begin{proposition}\label{prop:bound} Suppose that Assumptions \ref{assumption:K}, \ref{assumption:consistency}, \ref{assumption:Q_M}, \ref{assumption: penalty}, and \ref{assumption:BIC} holds.  Then, Assumption \ref{assumption:bound} holds, and $\widehat{M}_{PL} \overset{p}{\rightarrow} M_0$.
\end{proposition}

  \begin{corollary}[Consistency of BIC and Inconsistency of AIC] \label{cor:bic}
Suppose Assumptions \ref{assumption:K}, \ref{assumption:consistency}, \ref{assumption:Q_M}, and \ref{assumption:BIC} hold. Then, (a)  If $p_{n,k}=\frac{k}{2}\log(n)$ (BIC penalty), we have $\widehat{M}_{PL} \overset{p}{\rightarrow} M_0$. (b)  If $p_{n,k}=k$ (AIC penalty), then $\lim_{n\rightarrow\infty} \Pr(\widehat{M}_{PL} > M_0) > 0$.
\end{corollary}

\cite{leroux92as} established that, under conditions similar to Assumption \ref{assumption: penalty}, the maximum penalized likelihood method produces an estimator that asymptotically does not underestimate the number of components. Building upon the work of \cite{dacunha99as}, \cite{keribin00sankhya} derived regularity conditions necessary for the consistency of the maximum penalized likelihood estimator. Our Assumption \ref{assumption:BIC}(b) corresponds to condition (P2) in \cite{keribin00sankhya}. However, Keribin's condition (P2) is not satisfied when the number of components \( M \) is moderately larger than the true number \( M_0 \), as it relies on the sufficiency of a second-order expansion of the the density ratio for identification. In contrast, our Assumption \ref{assumption:BIC}(b) accommodates settings where parameter identification requires a higher-order rank condition, thereby extending the applicability of the consistency result.

\section{Nonparametric estimation of the lower bound of the number of components by the  rank test }\label{sec:rank}

We also estimate the lower bound of the number of components without imposing the parametric assumption on error distributions by extending a method proposed by \cite{Kasahara2014} which in turn is based on the rank test of \cite{Kleibergen06}.  

We partition the support of $Y_{t}\in\mathcal{Y}_t$ into $|\Delta_{t}|$ mutually exclusive and exhaustive subsets $\Delta_{t}=\{\delta_{1}^t, \ldots ,\delta_{|\Delta_{t}|}^t\}$ so that $\mathcal{Y}_t=\cup_{i=1}^{|\Delta_{t}|}\delta_{i}$ and $\delta_{i}\cap \delta_{j}=\emptyset$ for $i\neq j$. 

For each $k\in \mathcal{T}:=\{1,2,...,T\}$, define $\mathcal{T}_{-k}:=\{s \in \mathcal{T}: s\neq t\}=\{1,..,t-1,t+1,...,T\}$ and let $\mathbf{Y}_{\mathcal{T}_{-k}}=(Y_1,..,Y_{k-1},Y_{k_1},...,Y_T)^\top$ be  a vector of $Y_t$s in the group $\mathcal{T}_{-k}$. Partition the support of $\mathbf{Y}_{\mathcal{T}_{-k}}$ into $|\Delta_{\mathcal{T}_{-k}}|$ mutually exclusive and exhaustive subsets $\Delta_{\mathcal{T}_{-k}}=\{\bs\delta_{1}^{\mathcal{T}_{-k}}, \ldots ,\bs\delta_{|\Delta_{\mathcal{T}_{-k}}|}^{\mathcal{T}_{-k}}\}$. In particular, we construct this partition by unfolding the tensor product $\otimes_{t\neq k} \Delta_{t}$ using the  the Khatri-Rao product denoted by $\odot$ as $\Delta_{\mathcal{T}_{-k}}=\odot_{t\neq k}   \Delta_{t}= \Delta_{1}\odot \cdots \odot \Delta_{k-1}\odot \Delta_{k+1}\odot \cdots\odot \Delta_T$ so that $|\Delta_{\mathcal{T}_{-k}}|=|\Delta_t|^{T-1}$.

 For each $k\in   \mathcal{T}$, we construct  a $|\Delta_{t}| \times |\Delta_{\mathcal{T}_{-k}}|$ bivariate probability matrix $\bs P_{k}$ by arranging $\Pr(Y_k \in \delta_{a},\bs Y_{-k} \in \bs\delta_{b}^{\mathcal{T}_{-k}})$ for partition level $(a,b)=(1,1),\ldots,(|\Delta_{t}|,| \Delta_{{\mathcal{T}_{-k}}}|)$ as
\begin{equation}
\bs P_k =  \left[
\begin{array}{ccc}
\Pr(Y_k \in \delta_{1},\bs Y_{\mathcal{T}_{-k}} \in \bs\delta_{1}^{\mathcal{T}_{-k}}) & \cdots & \Pr(Y_{k} \in \delta_{1},\bs Y_{\mathcal{T}_{-k}} \in \bs \delta_{|\Delta_{\mathcal{T}_{-k}}|}^{\mathcal{T}_{-k}})  \\
\vdots  & \ddots  & \vdots  \\
\Pr( Y_{k} \in \delta_{|\Delta_{t}|},\bs Y_{\mathcal{T}_{-k}} \in \bs \delta_{1}^{\mathcal{T}_{-k}}) & \cdots &\Pr(  Y_{k} \in \delta_{|\Delta_{t}|},\bs Y_{-k} \in \bs\delta_{|\Delta_{\mathcal{T}_{-k}}|}^{\mathcal{T}_{-k}})\end{array}
\right]. \label{P*_defn}
\end{equation}

Collect the marginal probability distribution of ${Y}_{k}$ and $\mathbf{Y}_{-k}$ over $\Delta_{t}$ and $\Delta_{\mathcal{T}_{-k}}$ conditional on being from the $j$-th component into a vector as
\begin{align*}
\bs p_k^j :&= (\Pr({Y}_{k}\in \delta_1|D=j),\ldots,\Pr({Y}_{k}\in \delta_{|\Delta_{t}|}|D=j))^\top\quad\text{and}\\
\bs q_k^j & := (\Pr(\mathbf{Y}_{-k}\in \bs\delta_1^{\mathcal{T}_{-k}}|D=j),\ldots,\Pr(\mathbf{Y}_{k}\in \bs\delta_{|\Delta_{\mathcal{T}_{-k}}|}^{\mathcal{T}_{-k}}|D=j))^\top,\quad\text{respectively}.
 \end{align*}

Then, under the conditional independence assumption  as in the mixture model (\ref{eq:fm}) but without imposing parametric restrictions,  we may represent $\bs P_k$ as
\[
\bs P_k  = \sum_{j=1}^M \alpha^j \bs p_k^j (\bs q_k^j)^\top.
\]
\cite{Kasahara2014} shows the rank of $\bs P_k$ identifies the lower bound of the number of components and develop a sequential hypothesis testing procedure for estimating the rank of $\bs P_k$ when the empirical quantile of the $Y_t$'s are used to construct the partition. Because  there are $T$ possible ways to pick different $k$'s out of $\{1,...,T\}$, we test the maximum of the ranks of $\bs P_k$ across $k=1,...,T$.



We first develop a rk-statistic of  \cite{Kleibergen06} for testing the null hypothesis of $\text{rank}(\bs P_k)=r$. Write
the eigenvalue decomposition of  ${\bs P}_k$ as
\[
{\bs P}_k = \mathbf{U}^k\mathbf{S}^k(\mathbf{U}^k)^\top = \begin{bmatrix}
\mathbf{U}^k_{11} & \mathbf{U}^k_{12}\\
\mathbf{U}^k_{21} & \mathbf{U}^k_{22}
\end{bmatrix}
\begin{bmatrix}
\mathbf{S}^k_1 & 0\\
0 & \mathbf{S}^k_2
\end{bmatrix}
\begin{bmatrix}
\mathbf{U}^k_{11} & \mathbf{U}^k_{12}\\
\mathbf{U}^k_{21} & \mathbf{U}^k_{22}
\end{bmatrix}^\top,
\]
where $\mathbf{U}^k$ is a $|\Delta_t| \times |\Delta_t|$ orthonormal matrix  and $\mathbf{S}^k$ is a diagonal matrix containing the eigenvalues of $\mathbf{P}_k$ in decreasing order. In the partition of $\mathbf{U}^k$ and $\mathbf{S}^k$ on the right-hand side, $\mathbf{U}^k_{11}$ and $\mathbf{S}^k_1$ are $r \times r$,
and the dimensions of the other submatrices are defined conformably. Then, the null hypothesis
$\mathcal{H}_0: \mathrm{rank}(\mathbf{Q}_k) = r$ is equivalent to $\mathcal{H}_0: \mathbf{S}^k_2 = 0$. 
The statistic of \cite{Kleibergen06} is based on an orthogonal transformation of $\mathbf{S}^k_2$
given by $\mathbf{\Lambda}^k_r = (\mathbf{A}^k_{r })^\top\mathbf{Q}_k\mathbf{A}^k_{r }$, where
\[
\mathbf{A}^k_{r } = \begin{bmatrix}
\mathbf{U}^k_{12}\\
\mathbf{U}^k_{22}
\end{bmatrix}
(\mathbf{U}^k_{22} )^{-1}(\mathbf{U}^k_{22}(\mathbf{U}^k_{22})^\top)^{1/2}.
\]

Let $\widehat{\bs P}_k$ be a sample analogue estimator for $\bs P_k$ for which we have $\sqrt{n}\text{vec}(\widehat{\bs P}_k-\bs P_k)\overset{d}{\rightarrow} N(0,\bs\Sigma_k)$.  The following proposition follows from Theorem 1 of \cite{Kleibergen06}.
\begin{proposition}\label{rk}
Suppose that $\sqrt{n}\text{vec}(\widehat{\bs P}_k-\bs P_k)\overset{d}{\rightarrow} N(0,\bs\Sigma_k)$ and that $\boldsymbol{\Omega}_r^k : = (\mathbf{A}_{r }^k \otimes \mathbf{A}_{r }^k)^\top \boldsymbol{\Sigma}_k (\mathbf{A}_{r }^k \otimes \mathbf{A}_{r }^k)$ is non-singular. If $\mathrm{rank}(\mathbf{P}_k) =  r$, then
$\sqrt{n} \widehat{\bs{\lambda}}_r^k \rightarrow_d
\mathcal{N}(0, \boldsymbol{\Omega}_r^k)$ as $n \rightarrow \infty$, where $\widehat{\bs{\lambda}}_r^k =\text{diag}((\widehat{\mathbf{A}}^k_{r })^\top \widehat{\mathbf{P}}_k \widehat{\mathbf{A}}_{r }^k)$. 
\end{proposition}
 \cite{Kleibergen06} proposed the statistic called the rk-statistic:
\[
\text{rk}^k(r) = n (\widehat{\bs{\lambda}}_r^k)^\top (\widehat{\boldsymbol{\Omega}}^k_r) ^{-1}\widehat{\bs{\lambda}}_r^k,
\]
where $\widehat{\boldsymbol{\Omega}}^k_r$ is a consistent estimator for $\boldsymbol{\Omega}^k_r$. If the assumptions of proposition \ref{rk} hold, $\mathrm{rk}(r)$
converges in distribution to a $\chi^2(|\Delta_t|-r)$ random variable
as $n \rightarrow \infty$.

 When $T \geq 3$, we test the null hypothesis that $\text{rank}(\bs P_k) \leq r$ for each $k = 1,\dots,T$. By selecting partitions such that $|\Delta_t| = r + 1$, we define  the following ave- and  max-rk test statistics:
\begin{align*}
\text{ave-rk}(r) = \frac{1}{T} \sum_{k=1}^T \mathrm{rk}^k(r)\quad\text{and}\quad
\text{max-rk}(r) = \max\{ \mathrm{rk}^1(r),...,\mathrm{rk}^T(r)\}.
\end{align*}
See Section 3.4 of \cite{Kasahara2014} for the asymptotic distribution of these statistics.


In practice, we use the Bayesian bootstrap to obtain the bootstrap p-value for testing the rank, rejecting the null hypothesis $\text{rank}(\bs P_k)\leq r$ for all $k=1,\dots,T$ at significance level $\alpha$ if the bootstrap p-value is strictly less than $\alpha$; see Appendix \ref{sec:bootstrap}. We estimate the lower bound for the number of components by applying the sequential hypothesis testing procedure described in Section 3.2 of \cite{Kasahara2014}.

\section{Simulation}\label{sec:simulation}

This section evaluates the relative performance of several approaches—rank tests based on ave-rk and max-rk statistics, LRT , AIC, and BIC—in estimating the number of components across different models through simulation studies. For our simulation and empirical analyses, we set the constraint that $\alpha_j, \tau_{jk} \geq 0.05$ and $\sigma_j \geq 0.05 \hat{\sigma}_0$, effectively setting $c_1=0.05$ and $c_2=0.05$. We estimate the parameter by the MLE under these constraints using the EM algorithm as described in Appendix \ref{sec:em}. 


\begin{table}[h!]
\centering
\begin{threeparttable}
\caption{Rejection frequencies (Size) for testing $H_0: M=2$ against $H_1: M=3$ by the ave-rk, the max-rk, and the LRTs at the 5\% significance level and selection frequencies by the AIC and the BIC with \textbf{data generated from two-component models}}
\label{tab:size_test_2}
\begin{tabular*}{0.9\linewidth}{@{\extracolsep{\fill}}ll|ccccc@{}}
\toprule
$n$ & $(\mu_1,\mu_2)$ & \textbf{ave-rk} & \textbf{max-rk} & \textbf{LR} & \textbf{AIC} & \textbf{BIC} \\ \midrule
$200$ &  $(-1, 1)$    & 2.6          & 2.4          & 4.6          & 40.8         & 0.2          \\
$200$ & $(-0.5, 0.5)$ & 0.0            & 0.2          & 5.4          & 36.8         & 0.0            \\
$400$ & $(-1, 1)$  & 5.2          & 6.0            & 6.2          & 46.0           & 0.2          \\
$400$ & $(-0.5, 0.5)$ & 0.0            & 0.0            & 4.8          & 50.2         & 0.2          \\ \bottomrule
\end{tabular*}

\begin{tablenotes}
\footnotesize
 \textbf{Notes:}  Based on 500 simulation repetitions, each simulated dataset consists of $(n,T)$ panel observations with $T=3$. The dataset is generated from the two-component model (\ref{eq:fm}) with (\ref{eq:f1}) without covariates, varying the mean parameters $(\mu_1,\mu_2)$ while holding the mixing proportion fixed at $\alpha=0.5$ and standard deviations constant at $(\sigma_1,\sigma_2) = (0.8, 1.2)$ across all simulations. The reported values represent rejection frequencies (in percentages) at the 5\% significance level for the ave-rk, max-rk, and LRTs under the null hypothesis $H_0: M=2$. For the AIC and BIC methods, we report the percentages of simulations in which these criteria selected the three-component model over the true two-component model.
\end{tablenotes}
\end{threeparttable}
\end{table}

\begin{table}[h!]
\centering
\begin{threeparttable}
\caption{Rejection frequencies (Power) for testing $H_0: M=2$ against $H_1: M=3$ by the ave-rk, the max-rk, and the LRTs  at the 5\% significance level and selection frequencies by the AIC and the BIC with \textbf{data generated from three-component models}}
\label{tab:power_test_2}
\begin{tabular*}{0.9\linewidth}{@{\extracolsep{\fill}}ll|ccccc@{}}
\toprule
$n$ & $(\mu_1,\mu_2,\mu_3)$ & \textbf{ave-rk} & \textbf{max-rk} & \textbf{LR} & \textbf{AIC} & \textbf{BIC} \\ \midrule
$200$ & $(-0.5, 0, 1.5)$ & 0.0   & 0.6  & 11.2 & 46.8 & 0.6  \\
$200$ & $(-1.5, 0, 1.5)$ & 53.2 & 36.4 & 99.8 & 100  & 95.4 \\
$400$ & $(-0.5, 0, 1.5)$ & 3.2  & 3.0    & 13.6 & 60.6 & 0.4  \\
$400$ & $(-1.5, 0, 1.5)$ & 91.8 & 78   & 100  & 100  & 100  \\ \bottomrule
\end{tabular*}

\begin{tablenotes}
\footnotesize
 \textbf{Notes:} Based on 500 simulation repetitions, each simulated dataset consists of $(n,T)$ panel observations with $T=3$.  The dataset is generated from the $3$-component model (\ref{eq:fm}) with (\ref{eq:f1}) without covariates, using various values of $(\mu_1,\mu_2,\mu_3)$ while keeping $(\alpha_1,\alpha_2,\alpha_3)=(1/3, 1/3, 1/3)$ and $(\sigma_1,\sigma_2,\sigma_3) =(1, 1, 1)$ fixed across all simulations.   The reported values represent the rejection frequencies (in percentages) at the 5\% significance level for each of the ave-rk, max-rk, and LRTs under the null hypothesis $H_0: M=2$.  For the AIC and BIC methods, we report the percentages of simulations in which these criteria selected the true three-component model over the two-component model.

\end{tablenotes}
\end{threeparttable}
\end{table}

Tables \ref{tab:size_test_2} and \ref{tab:power_test_2} summarize size and power performance of statistical methods for testing $H_0: M=2$ against $H_1: M=3$ in models with conditionally independent errors without covariates, as defined by equations (\ref{eq:fm}) and (\ref{eq:f1}), based on $500$ simulations with $T=3$. Results compare ave-rk, max-rk, LR, AIC, and BIC methods, where we report the selection frequencies between two- and three-components models for AIC and BIC.

Table \ref{tab:size_test_2} presents rejection frequencies at the 5\% significance level, using data from two-component models without covariates. The ave-rk and max-rk tests achieve near-nominal levels for larger samples ($N=400$) and distinct means ($\mu=(-1,1)$) but underperform with smaller samples or closer means ($\mu=(-0.5,0.5)$). The LRT consistently maintains near-nominal performance. AIC often overestimates component counts, while BIC remains conservative, rarely overestimating.


Table \ref{tab:power_test_2} evaluates power using data from three-component models. The LRT exhibits excellent power (100\%) with clearly separated means ($\mu=(-1.5,0,1.5)$) for both sample sizes ($N=200,400$), outperforming ave-rk and max-rk tests even when mean separations narrow ($\mu=(-0.5,0,1.5)$).  While the BIC closely matches LRT performance under clear mean separation, it becomes overly conservative as the means become closer, selecting the true three-component model less frequently than the LRT. The AIC chooses the three-component models more often than other methods, reflecting its tendency to  overestimate the number of components.

\begin{table}[h!]
    \centering
    \begin{threeparttable}
         \caption{Selection frequencies when data generated from an estimated three-component model with \textbf{normal error density} ($M_0=3$, $\bs{\mathcal{K}=1}$, $T=3$)}
            \label{tab:empirical_test_normal}
            \begin{tabular*}{0.9\linewidth}{@{\extracolsep{\fill}}l | cccccc@{}}
                \toprule
             \multicolumn{7}{c}{ \textbf{Panel A:\  Sample Size of $\bs{n=50}$}}\\  \midrule
               \textbf{Methods ($q_n$, Error Density) }  & \textbf{M=1} & \textbf{M=2} & \textbf{M=3} & \textbf{M=4} & \textbf{M=5} & \textbf{M=6} \\
                \midrule
                \textbf{AIC (Normal)} & 0 & 4 & 84 & 12 & 0 & 0 \\
                \textbf{BIC (Normal)} & 0 & 25 & 73 & 2 & 0 & 0 \\ 
                \textbf{LR ($0.01$, Normal)} & 0 & 23 & 74 & 3 & 0 & 0 \\
                \textbf{LR ($0.05$, Normal)} & 0 & 11 & 82 & 7 & 0 & 0 \\
                \textbf{LR ($0.10$, Normal)} & 0 & 9 & 83 & 8 & 0 & 0 \\  \midrule
                \textbf{AIC (Mixture)} & 0 & 4 & 63 & 22 & 9 & 2 \\
                \textbf{BIC (Mixture)} & 0 & 58 & 40 & 2 & 0 & 0 \\
                \textbf{LR ($0.01$,  Mixture)} & 1 & 74 & 25 & 0 & 0 & 0 \\
                \textbf{LR ($0.05$,  Mixture)} & 0 & 54 & 46 & 0 & 0 & 0 \\
                \textbf{LR ($0.10$,  Mixture)} & 0 & 34 & 64 & 2 & 0 & 0 \\  \midrule
                \textbf{ave-rk ($0.05$)} & 10 & 90 & 0 & 0 & 0 & 0 \\
                \textbf{max-rk ($0.05$)} & 11 & 89 & 0 & 0 & 0 & 0 \\  \midrule  \midrule
             \multicolumn{7}{c}{ \textbf{Panel B: \ Sample Size of $\bs{n=225}$}}\\  \midrule
               \textbf{Methods ($q_n$, Error Density) }  & \textbf{M=1} & \textbf{M=2} & \textbf{M=3} & \textbf{M=4} & \textbf{M=5} & \textbf{M=6} \\
                \midrule 
                \textbf{AIC (Normal)} & 0 & 0 & 81 & 19 & 0 & 0 \\
                \textbf{BIC (Normal)} & 0 & 0 & 100 & 0 & 0 & 0 \\
                \textbf{LR ($0.01$, Normal)} & 0 & 0 & 99 & 1 & 0 & 0 \\
                \textbf{LR ($0.05$, Normal)} & 0 & 0 & 94 & 6 & 0 & 0 \\
                \textbf{LR ($0.10$, Normal)} & 0 & 0 & 89 & 11 & 0 & 0 \\ \midrule
                \textbf{AIC (Mixture)} & 0 & 0 & 70 & 27 & 3 & 0 \\
                \textbf{BIC (Mixture)} & 0 & 0 & 100 & 0 & 0 & 0 \\
                \textbf{LR ($0.01$, Mixture)} & 0 & 0 & 100 & 0 & 0 & 0 \\
                \textbf{LR ($0.05$, Mixture)} & 0 & 0 & 95 & 5 & 0 & 0 \\
                \textbf{LR ($0.10$, Mixture)} & 0 & 0 & 90 & 10 & 0 & 0 \\  \midrule
                \textbf{ave-rk ($0.05$)} & 0 & 69 & 31 & 0 & 0 & 0 \\
                \textbf{max-rk ($0.05$)} & 0 & 80 & 20 & 0 & 0 & 0 \\
                \bottomrule
            \end{tabular*}
            \begin{tablenotes}
                \footnotesize
 \textbf{Notes:}    Based on 100 simulation repetitions, each simulated dataset consists of $(n,T)$ panel observations with $n=50$ or $225$ and $T=3$. The datasets are generated from the three-component model (\ref{eq:fm}) with normal error density (\ref{eq:f1}), without covariates, using parameter values estimated from Chilean fabricated metal products industry data: mixing probabilities $\bs\alpha = [0.352, 0.402, 0.245]$, means $\bs\mu = [-1.01, -0.557, -0.242]$, and standard deviations $\bs\sigma = [0.464, 0.187, 0.195]$. The true number of components is $M_0=3$. The reported values represent the percentages of simulations in which each criterion—AIC, BIC, ave-rk, max-rk, and LRTs—selected a given number of components.
        \end{tablenotes}

    \end{threeparttable}
\end{table}

\begin{table}[h!]
\centering
    \begin{threeparttable}
         \caption{Selection frequencies when data generated from an estimated three-component model with \textbf{normal mixture error density} ($M_0=3$, $\bs{\mathcal{K}=2}$,  $T=3$)}
    \label{tab:empirical_test_mixture}
    \begin{tabular*}{0.9\linewidth}{@{\extracolsep{\fill}}l | cccccc@{}}
        \toprule
                     \multicolumn{7}{c}{ \textbf{Panel A:\  Sample Size of $\bs{n=50}$}}\\  \midrule
           \textbf{Methods ($q_n$, Error Density) } & \textbf{M=1} & \textbf{M=2} & \textbf{M=3} & \textbf{M=4} & \textbf{M=5} & \textbf{M=6} \\  \midrule
            \textbf{AIC (Normal)} & 0 & 9 & 80 & 10 & 1 & 0 \\
            \textbf{BIC (Normal)} & 0 & 32 & 68 & 0 & 0 & 0 \\
            \textbf{LR ($0.01$, Normal)} & 0 & 33 & 66 & 1 & 0 & 0 \\
            \textbf{LR ($0.05$, Normal)} & 0 & 17 & 79 & 3 & 1 & 0 \\
            \textbf{LR ($0.10$, Normal)} & 0 & 13 & 79 & 7 & 1 & 0 \\  \midrule
            \textbf{AIC (Mixture)} & 0 & 16 & 44 & 32 & 5 & 3 \\
            \textbf{BIC (Mixture)} & 1 & 44 & 50 & 5 & 0 & 0 \\
            \textbf{LR ($0.01$,  Mixture)} & 1 & 75 & 24 & 0 & 0 & 0 \\
            \textbf{LR ($0.05$,  Mixture)} & 0 & 52 & 44 & 4 & 0 & 0 \\
            \textbf{LR ($0.10$,  Mixture)} & 0 & 37 & 56 & 7 & 0 & 0 \\  \midrule
            \textbf{ave-rk ($0.05$)} & 11 & 89 & 0 & 0 & 0 & 0 \\
            \textbf{max-rk ($0.05$)} & 10 & 90 & 0 & 0 & 0 & 0 \\
              \midrule  \midrule
                                 \multicolumn{7}{c}{ \textbf{Panel B:\  Sample Size of $\bs{n=225}$}}\\  \midrule
           \textbf{Methods ($q_n$, Error Density) } & \textbf{M=1} & \textbf{M=2} & \textbf{M=3} & \textbf{M=4} & \textbf{M=5} & \textbf{M=6} \\
           \midrule
            \textbf{AIC (Normal)} & 0 & 0 & 44 & 54 & 2 & 0 \\
            \textbf{BIC (Normal)} & 0 & 0 & 96 & 4 & 0 & 0 \\
            \textbf{LR ($0.01$, Normal)} & 0 & 0 & 84 & 16 & 0 & 0 \\
            \textbf{LR ($0.05$, Normal)} & 0 & 0 & 68 & 32 & 0 & 0 \\
            \textbf{LR ($0.10$, Normal)} & 0 & 0 & 54 & 46 & 0 & 0 \\ \midrule
            \textbf{AIC (Mixture)} & 0 & 0 & 56 & 37 & 6 & 1 \\
            \textbf{BIC (Mixture)} & 0 & 0 & 96 & 4 & 0 & 0 \\
            \textbf{LR ($0.01$,  Mixture)} & 0 & 0 & 89 & 11 & 0 & 0 \\
            \textbf{LR ($0.05$,  Mixture)} & 0 & 0 & 82 & 18 & 0 & 0 \\
            \textbf{LR ($0.10$,  Mixture)} & 0 & 0 & 76 & 23 & 1 & 0 \\ \midrule
            \textbf{ave-rk ($0.05$)} & 0 & 58 & 42 & 0 & 0 & 0 \\
            \textbf{max-rk ($0.05$)} & 0 & 74 & 26 & 0 & 0 & 0 \\
            \bottomrule
            \end{tabular*}
            \begin{tablenotes}
                \footnotesize
                                \textbf{Notes:}    Based on 100 simulation repetitions, each simulated dataset consists of $(n,T)$ panel observations with $n=50$ or $225$ and $T=3$.  The datasets are generated from the three-component model (\ref{eq:fm}) with normal mixture error density (\ref{eq:f1-mixture}), without covariates, using parameter values estimated from Chilean fabricated metal products industry data: mixing probabilities $\bs\alpha =  [0.275, 0.247, 0.478]$,  normal mixture error distribution mixing probabilities $\bs\tau=       \{(0.818, 0.182), (0.621, 0.379), (0.105, 0.895)\}$, means $\bs\mu= \{(-1.099, -1.04), (-0.249, -0.232), (-0.992, -0.539)\}$, standard deviations $\bs\sigma=[0.465, 0.197, 0.178]$.  The simulation imposes that $\alpha$ (mixing probabilities) is bounded between 0.05 and 0.95, and $\tau$ (within-component mixing probabilities) is bounded between 0.05 and 0.95. The true number of components is $M_0=3$. The reported values represent the percentages of simulations in which each criterion—AIC, BIC, ave-rk, max-rk, and LRTs—selected a given number of components.
            \end{tablenotes}
    \end{threeparttable}
\end{table}

We also evaluate sequential testing procedures under realistic conditions by simulating datasets from three-component models estimated from real data. Specifically, we consider different sample sizes and various specifications for density functions and stochastic processes governing the error terms. This approach enables us to assess how model misspecification impacts the accuracy of estimating the number of components.

Table \ref{tab:empirical_test_normal} summarizes the selection frequencies (in percentages) for identifying the number of components (from 1 to 6) using the AIC, BIC, LRT, and ave-rk and max-rk tests in simulations for the sample sizes of $n=50$ and $225$ with $T=3$. Data are generated from a three-component model without covariates under \textit{normal error density}, as defined by equations (\ref{eq:fm}) and (\ref{eq:f1}), estimated using Chilean fabricated metal products industry data. The model features mixing probabilities $\bs\alpha = [0.352, 0.402, 0.245]$, means $\bs\mu = [-1.01, -0.557, -0.242]$, and standard deviations $\bs\sigma = [0.464, 0.187, 0.195]$. Results are reported for correctly specified normal errors ("Normal") and over-specified two-component normal mixtures ("Mixture"), with LRTs evaluated at significance levels $q_n=0.01$, $0.05$, and $0.10$ while the ave- and max-rk tests implemented at 5\% significance level. The parameters are estimated by the MLE under the constraint $\alpha_j, \tau_{jk} \geq 0.05$ and $\sigma_j \geq 0.05 \hat{\sigma}_0$ but the simulation results are not so sensitive to changing these tuning parameters as shown in Tables \ref{tab:sensitivity-1}--\ref{tab:sensitivity-3} in Appendix \ref{sec:additional}.


Panel A of Table \ref{tab:empirical_test_normal} reports the results for a small sample size of $n=50$, showing that both the BIC and LR tests effectively identify the correct number of components when the model is correctly specified with normal error density. Specifically, the LR tests select $M = 3$ most frequently, with accuracy improving as the significance level increases: 74\% at $q_n = 0.01$, 82\% at $q_n = 0.05$, and 83\% at $q_n = 0.10$, whereas BIC selects the correct model ($M = 3$) in 73\% of cases. Thus, BIC tends to underestimate the number of components relative to the LR tests under correct specification. On the other hand, AIC selects the correct number of components in 84\% of cases, although it overestimates the number in 12\% of cases.

Thus, for the correctly specified models, the LR test outperforms BIC with similar component means and small samples, as BIC underestimates component numbers. While BIC is simpler to implement, bootstrapping makes the LR test equally easy. Using all three methods provides a useful range: BIC as the lower bound, AIC as the upper bound, and LR test as a middle ground. Since rank tests typically underestimate component numbers, findings of $M\geq 2$ provide strong evidence for heterogeneity.

 Under the over-specified scenario with two-component normal mixture error density specification, both BIC and the LR test select $M=2$ over $M=3$ more frequently compared to the correctly specified normal-density scenario. The LR test is particularly conservative at lower significance levels, choosing $M=2$ in 74\% of cases at $q_n=0.01$, but its accuracy for selecting $M=3$ improves to 64\% at $q_n=0.10$. The BIC selects $M=2$ in 58\% of cases and $M=3$ in 40\%. AIC consistently favours more complex models, selecting $M=4$ or higher in 33\% of cases.

The ave-rk and max-rk tests are highly conservative, selecting $M=2$ in nearly all cases (90\% and 89\%, respectively), and rarely identifying the true number of components.

Overall, the LR test and BIC perform comparably, though BIC sometimes underestimates component numbers when mixture components have similar means and sample sizes are small. While BIC is simpler to implement, bootstrapping makes the LR test equally easy. Using all three methods provides a useful range, especially with small samples: BIC as the lower bound, AIC as the upper bound, and the LR test as a middle ground. Furthermore, since rank tests typically underestimate component numbers, findings of $M\geq 2$ also provide strong evidence for heterogeneity.

Panel B of Table \ref{tab:empirical_test_normal}  presents results for the sample size matches the actual data at $n=225$ instead of $n=50$, where we observe a substantial improvement in the performance of BIC and LRT as the sample size increases from $n=50$ to $n=225$. Both methods correctly select $M=3$ in approximately 90\% or more of cases, even under the over-specified mixture density scenario. In contrast, the AIC continues to favor more complex models than BIC or LRT, while the ave-rk and max-rk tests tend to select $M=2$ over $M=3$ in the majority of cases.

Table \ref{tab:empirical_test_mixture} examines the case in which the true data-generating process is a three-component model with a \textit{normal mixture error density}, as described in (\ref{eq:fm}) with  (\ref{eq:f1-mixture}), rather than a simple normal error density. Data are simulated using parameter values estimated from the Chilean fabricated metal products industry, as detailed in the table notes, for sample sizes of $n=50$ and $225$.

Similar to the results obtained with the normal error density DGP, both BIC and the LR test tend to frequently select $M=2$ over $M=3$  under the normal mixture error density specification, compared to a normal error density scenario at the small sample size of $n=50$. As the sample size increases from $n=50$ to $n=225$, the accuracy of BIC and LR tests significantly improves, particularly for BIC and the LR test at $q_n=0.01$. At $n=225$, however, the LR test under the incorrect normal error density specification tends to select $M=4$ over $M=3$ more frequently than the correctly-specified normal mixture scenario, suggesting that incorrectly imposing a normal error density may lead to overestimating the number of components with the sufficiently large sample size.  The AIC has a higher tendency to select $M=4$ over $M=3$ than the BIC or the LRT while the rank-based tests continue to under-estimate the number of components.

In the Appendix, Tables~\ref{tab:empirical_test_normal_food_textile}–\ref{tab:empirical_test_mixture_food_textile} present simulation results based on estimated parameters for the food products and textiles industries, revealing patterns similar to those in Tables~\ref{tab:empirical_test_normal}–\ref{tab:empirical_test_mixture}.
 
 Appendix   \ref{sec:additional} also presents the simulation results for dynamic panel mixture models when the error terms follow AR(1) processes whose innovations are drawn from either a normal distribution or a two-component normal mixture, as defined by equations (\ref{eq:fm-dynamic}) with (\ref{eq:f1-dynamic}) or (\ref{eq:f1-dynamic-mixture}), respectively, where both BIC and LRT performs well under the realistic DGP based on an estimated model from  the Chilean fabricated metal products industry. See Tables \ref{tab:sequential_test_ar1_normal}-\ref{tab:sequential_test_ar1_mixture}.\footnote{Furthermore, simulation results using data generated from the estimated models with labor-augmented technological changes are presented in Tables~\ref{tab:sequential_test_ar1_lat_normal} and~\ref{tab:sequential_test_ar1_lat_mixture}.}




\section{Empirical analysis: plant heterogeneity in production function}\label{sec:empirics}

In this section, we investigate heterogeneity beyond Hicks-neutral technological differences using finite mixture specifications (\ref{eq:spec1}), (\ref{eq:spec2}), and (\ref{eq:spec3}), as illustrated in Examples \ref{example-1}, \ref{example-2}, and \ref{example-3}. Specifically, we estimate the number of latent technology types characterized by variations in input elasticities and labor-augmented technological changes through the application of the Likelihood Ratio Test (LRT), Akaike Information Criterion (AIC), Bayesian Information Criterion (BIC), and the nonparametric rank test based on ave-rk statistics.

  \begin{table}[h!]
  \centering
  \caption{Descriptive Statistics for Chilean Plants}
  \label{tab:sum}
   \begin{threeparttable}
    \begin{tabularx}{\textwidth}{@{\extracolsep{\fill}} lcccccccc}
      \toprule
      \multicolumn{9}{c}{\textbf{Chilean Plants}} \\
      \midrule
      \multicolumn{9}{c}{\textbf{Panel A: \ $\bs{T=3}$}} \\
      \midrule
      \textbf{Industry} & \textbf{NT} & \textbf{N} & import share & $\frac{P_{V,t} V_{it}}{P_{Y,t} Y_{it}}$ & $\frac{P_{V,t} V_{it}}{P_{L,t} L_{it}}$ & $\log K_{it}$ & $\log L_{it}$ & $\log Y_{it}$ \\
      \midrule
      Fabricated Metal  & 675 & 225 & 0.100 & 0.573 & 2.464 & 7.986 & 3.966 & 12.926 \\
         &     &     & (0.197) & (0.216) & (2.905) & (1.857) & (0.975) & (1.380) \\
      Food Products & 2382 & 794 & 0.023 & 0.603 & 4.033 & 7.420 & 3.997 & 12.685 \\
      &      &     & (0.083) & (0.145) & (4.525) & (2.296) & (1.035) & (1.656) \\
      Textiles & 588 & 196 & 0.176 & 0.572 & 3.568 & 8.280 & 3.944 & 12.914 \\
         &     &     & (0.246) & (0.213) & (5.120) & (1.880) & (1.098) & (1.344) \\
      \midrule
      \multicolumn{9}{c}{\textbf{Panel B: \ $\bs{T=5}$}} \\
      \midrule
      Fabricated Metal  & 515 & 103 & 0.086 & 0.555 & 2.768 & 7.775 & 3.880 & 12.766 \\
         &     &     & (0.182) & (0.219) & (4.097) & (1.909) & (0.992) & (1.364) \\
      Food Products & 2195 & 439 & 0.026 & 0.602 & 4.878 & 7.610 & 4.034 & 12.834 \\
      &      &     & (0.096) & (0.149) & (5.460) & (2.361) & (1.008) & (1.693) \\
      Textiles & 465 & 93 & 0.209 & 0.563 & 4.314 & 8.417 & 3.932 & 12.947 \\
       &     &     & (0.265) & (0.215) & (5.288) & (2.163) & (1.164) & (1.448) \\
      \bottomrule
    \end{tabularx}

    \begin{tablenotes}
      \footnotesize
      \item \textbf{Notes:} \textit{Panel A:} Data are from 1994 to 1996. \textit{Panel B:} Data are from 1992 to 1996. Summary statistics are based on Chilean plant-level data. Only plants with continuous data for all $T$ years are included (balanced panel). Observations with $\log (V_{it} / Y_{it}) \leq -3$ or $\log (V_{it} / Y_{it}) > \log(2)$ are excluded.
      NT: Number of total observations; N: Number of unique plants.
      Means are reported; standard deviations are in parentheses below each mean.
      ``Import share'' is the share of imported materials in total materials.
      $\frac{P_{V,t} V_{it}}{P_{Y,t} Y_{it}}$ is the revenue share of material input; $\frac{P_{V,t} V_{it}}{P_{L,t} L_{it}}$ is the ratio of material input to labor input, both using average prices at time $t$. $K_{it}$, $L_{it}$, and $Y_{it}$ denote capital, labor, and output, respectively.
    \end{tablenotes}
    \end{threeparttable}
\end{table}

For our analysis, we use plant-level panel data from Chilean manufacturing plants covering the period 1992--1996.\footnote{Please refer to \cite{KASAHARA2008} for details on the dataset.} We focus on the three largest industries---Food Products, Fabricated Metal Products, and Textiles---and include observations of plants with continuous data entry for either $T=3$ or $T=5$ consecutive years within each industry. We use data from 1994--1996 for estimating the model with conditionally independent errors, while data from 1992--1996 are used for dynamic panel mixture models. The selected panel lengths correspond to the minimum length required for non-parametric identification of the respective panel mixture models.\footnote{Non-parametric identification requires a panel length of $T=3$ for models with conditionally independent errors \citep{Kasahara2009} and $T=5$ for dynamic panel mixture models \citep{Hu2012,Kasahara2022esri}.} Table~\ref{tab:sum} presents summary statistics for the revenue share of materials and the log of gross output among other variables in these industries. The note section of Table~\ref{tab:sum} further explains the sample selection criteria.

 \subsection{Plant heterogeneity in output elasticity of material input}

Empirical studies often assume a Cobb-Douglas production technology
$\log O_{it} = \gamma_0 + \gamma_v V_{it} + \gamma_\ell \log L_{it} + \gamma_k \log K_{it} +\gamma_z \log Z_{it} + \omega_{it}$, where  $O_{it}$, $V_{it}$,  $L_{it}$, $K_{it}$, and $Z_{it}$ represent output, material input, labor,  capital, and other observed characteristics, respectively.  Because  input elasticities  are constant under Cobb-Douglass specification,  the first-order condition for profit maximization (\ref{y-foc}) in Example \ref{example-1} implies
\begin{equation}\label{foc-cd}
\frac{P_{V,t} V_{it}}{P_{O,t} O_{it}}=\gamma_v,
\end{equation} i.e., the ratios of material input to output is constant, where  $P_{O,t}$ and $P_{V,t}$ are output and material input prices, respectively.

The implication of Equation (\ref{foc-cd}) that material input shares are identical across plants can be empirically tested. Figure \ref{fig:empirical_all} displays histograms illustrating the distribution of plant-level material input shares in the Food Products, Fabricated Metal Products, and Textiles industries, clearly revealing substantial variability across plants within each industry. In Column 4 of Table \ref{tab:sum}, standard deviations of the revenue shares attributable to material costs are considerable, ranging from $0.145$ to $0.216$, even within narrowly defined industries. Consequently, the hypothesis that the ratio of material input to output is constant across plants is overwhelmingly rejected by the data.

While the substantial variability in material input shares across plants suggests potential plant heterogeneity in the output elasticity of material input, the implication derived from equation (\ref{foc-cd}) is strictly valid only under the Cobb-Douglas assumption. Under more general production function specifications, such as CES or translog forms, output elasticities with respect to inputs naturally depend on the levels of materials, labor, capital, and other observable plant characteristics, even in the absence of technological heterogeneity. Additionally, the relationship between material input-to-output ratios and the coefficient $\gamma_v$ may be influenced by idiosyncratic shocks.\footnote{As \cite{Gandhi2020} demonstrate, if the production function is specified as
$\log O_{it} = \gamma_0 + \gamma_v V_{it} + \gamma_\ell \log L_{it} + \gamma_k \log K_{it} +\gamma_z \log Z_{it} + \omega_{it}+\epsilon_{it}$, where $\epsilon_{it}$ is an i.i.d. mean zero random variable whose realization is unknown when the intermediate input $V$ is selected, then the first-order condition for profit maximization yields $\frac{P_{V,t} V_{it}}{P_{O,t} O_{it}}=\gamma_v \mathbb{E}[e^{\epsilon_{it}}]e^{-\epsilon_{it}}$. Consequently, the relationship between material input-to-output ratios and the coefficient $\gamma_v$ is subject to idiosyncratic shocks.}

To investigate persistent plant heterogeneity in output elasticities of material input within a more general production function framework while accounting for transitory idiosyncratic shocks, we specify that the logarithm of the output elasticity of material input is linearly related to observed plant characteristics ($\mathbf{x}_{it}$), such as the log of capital and other plant-specific attributes, conditional on the latent technological type of the plant, denoted by $D_i = j$, as
\begin{equation}\label{spec-4}
\log\left(\frac{P_{V,t} V_{it}}{P_{O,t} O_{it}}\right) = \mathbf{x}_{it}^\top \boldsymbol{\beta}_j + \epsilon_{it}\quad\text{for $j=1,2,...,M$},
\end{equation}
where $\boldsymbol{\beta}_j$ represents parameters specific to the $j$-th latent technology type, and $\epsilon_{it}$ denotes an idiosyncratic shock to plant-level input elasticities. The distribution of $\epsilon_{it}$ is type-specific and can be represented as either a single-component normal distribution or a mixture of normals:
\[
\epsilon_{it} \overset{iid}{\sim} \mathcal{N}(\mu_{j},\sigma_j^2) \quad \text{or} \quad \epsilon_{it} \overset{iid}{\sim} \sum_{k=1}^{{\cal{K}}} \tau_{jk}\, \mathcal{N}(\mu_{jk},\sigma_j^2)\quad\text{given $D_i = j$},
\]
where ${\cal{K}}$ represents the number of components in normal mixture error distributions.

This formulation leads directly to the finite mixture model given in equation (\ref{eq:fm}) with component density functions as outlined in equations (\ref{eq:f1}) or (\ref{eq:f1-mixture}). By estimating the number of latent technology types ($M$), we can quantify the degree of heterogeneity in output elasticities of material input among plants, even after controlling for observable plant characteristics.

\begin{figure}[ht!]
	\centering
	\subfigure[Fabricated Metal Products Industry]{\includegraphics[width=0.32\textwidth]{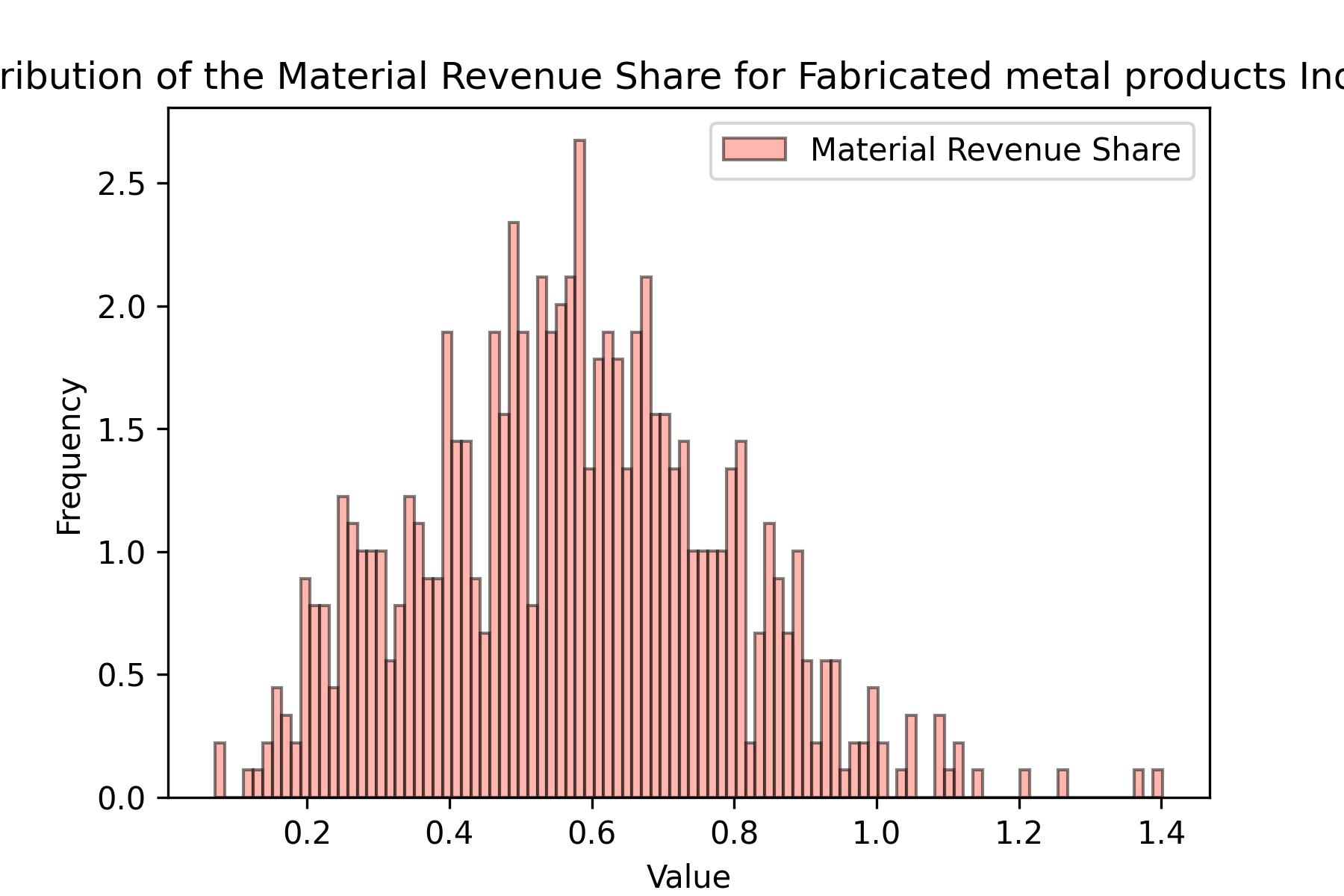 }\label{fig:metal_empirical}
	}
	\subfigure[Food Products Industry]{\includegraphics[width=0.32\textwidth]{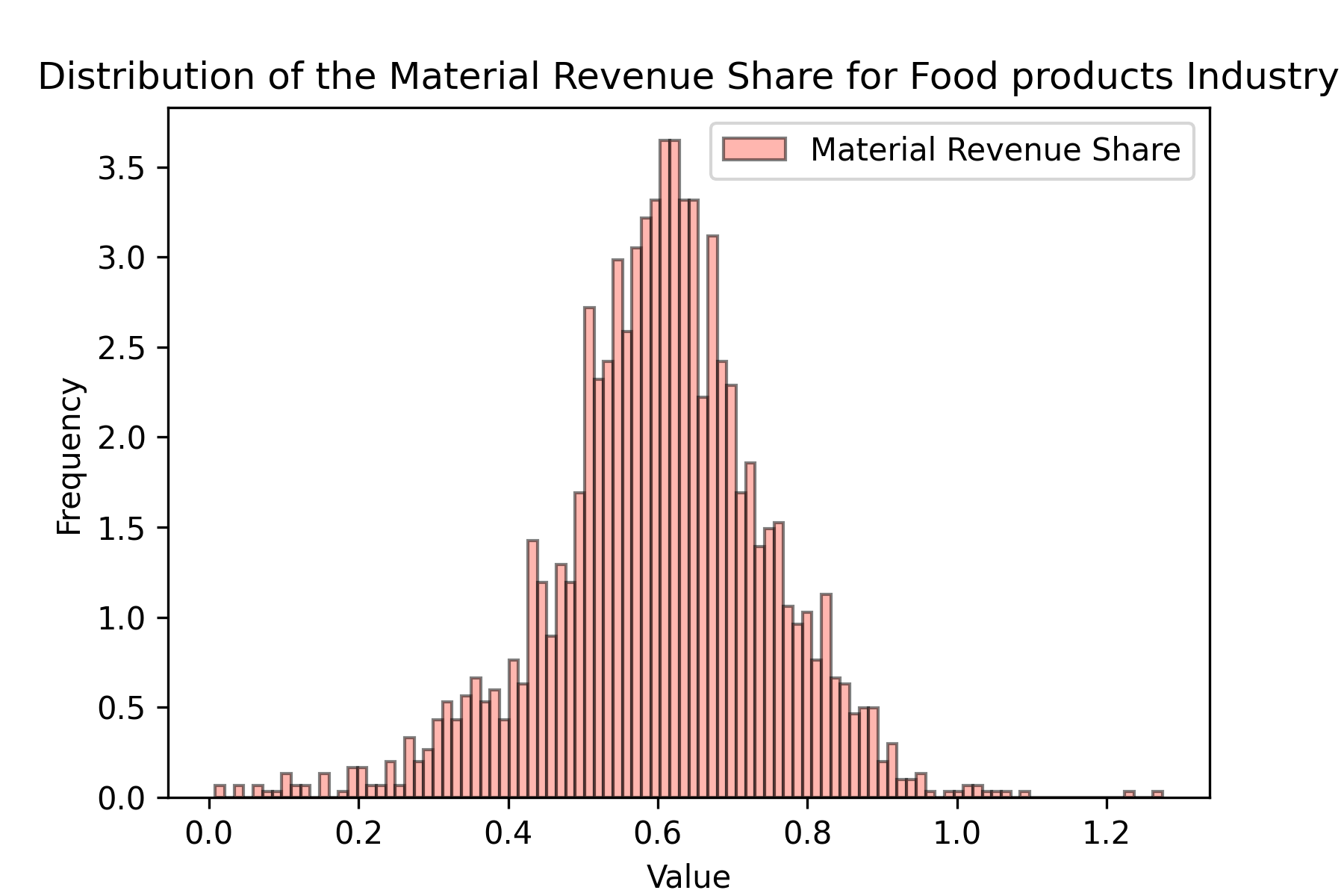  }\label{fig:food_empirical}
	}
	\subfigure[Textiles Industry]{\includegraphics[width=0.32\textwidth]{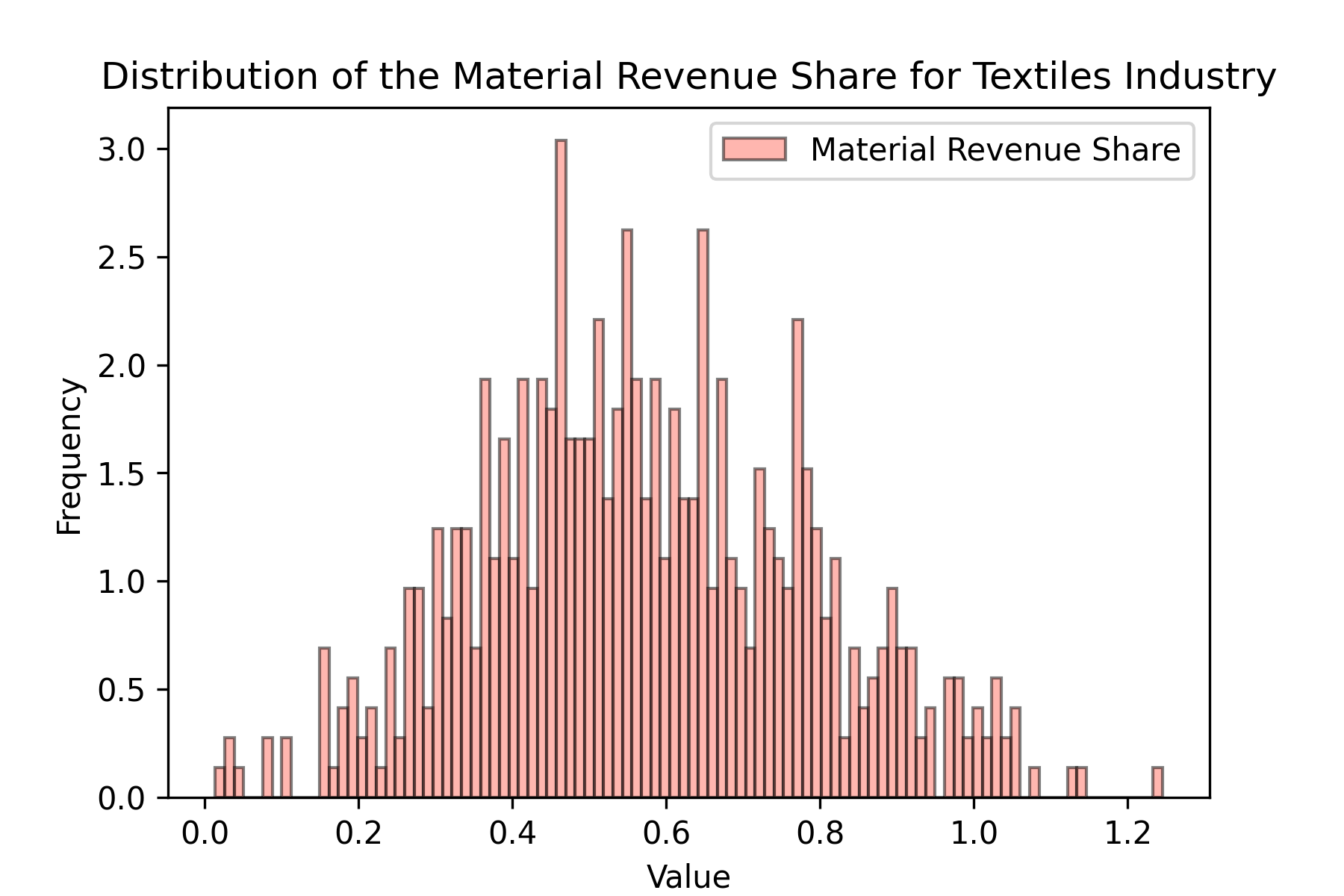}\label{fig:textiles_empirical}
	}
	\caption{Histograms of the Material Revenue Share across three industries (Metal, Food, Textiles).}
	\label{fig:empirical_all}
\end{figure}


%
%

To determine the number of latent technology types, we conduct the sequential hypothesis test using the LRT, as outlined in Section \ref{sec:sht}, setting $q_n=0.01$ guided by the simulation results presented in Section \ref{sec:simulation} although the results are not  sensitive to an alternative choice of $q_n=0.05$.  Additionally, we apply the AIC, the BIC, and the nonparametric rank test  in determining the number of technology types.

We conduct analyses both without covariates and with covariates including the log of capital ($\log K$), imported material shares (import), and 4-digit industry classification dummies (CIIU4). To allow for flexible error structures and relax the normality assumption, we model the error terms using normal distributions as well as two-component and three-component normal mixtures.

Panel A of Table \ref{tab:model_comparison} provides a comparative analysis of the estimated number of technology types across the three industries, utilizing the AIC, BIC, LRT, and nonparametric rank test for models with conditionally independent errors as specified in equation (\ref{eq:fm}). As the flexibility of the error term structure increases from normal to multi-component mixtures, the number of components estimated generally decreases for the AIC, BIC, and LRT. For example, without covariates and assuming normally distributed error terms, the LRT estimates a high number of components—8 for Metal, 10+ for Food, and 9 for Textiles. However, when employing a more flexible three-component normal mixture error density specification, the number of components reduces to 4, 8, and 6 for Metal, Food, and Textiles, respectively. This is consistent with our simulation results reported in Table \ref{tab:empirical_test_mixture}, indicating that a less flexible parametric assumption on the error term distribution may lead to over-estimation of the number of components. Including covariates such as $\log K$, imported material shares, and controls for 4-digit industry classifications (CIIU4) further decreases the estimated number of components for the LRT to 4, 5, and 5 for Metal, Food, and Textiles, respectively. 

Comparing the results of the LRT with those from the AIC and the BIC, we find that the BIC estimates a number of types similar to the LRT, whereas the AIC typically estimates a higher number. The number of technology types suggested by both the LRT and the BIC remains high for models with covariates and normal mixture error densities. Furthermore, as reported in the first row of Table \ref{tab:model_comparison}, the lower bound of the number of components estimated by the nonparametric rank test is 3 for all industries. Overall, these results provide strong evidence of substantial and persistent plant heterogeneity in material input elasticities across years.


Figure \ref{fig:distribution_all} presents histograms of estimated error terms alongside component-specific error distributions (solid lines) for models incorporating covariates that control for differences in $\log K$, import shares, and 4-digit industry dummies. The models employ iid error distributions specified as three-component normal mixtures ($\mathcal{K}=3$). Plant observations are classified into distinct technology types based on posterior probabilities calculated using Bayes' theorem, with each type indicated by a unique color. Following the LRT model selection criterion, we present a 4-component model for Fabricated Metal, a 5-component model for Food Products, and a 5-component model for Textiles. The results demonstrate clear evidence of non-normality in component-specific error distributions, while the flexible normal mixture specification effectively captures this non-normality within each component, illustrating the advantages of our modeling approach.


To better understand the nature of heterogeneity, Figures \ref{fig:parameter-normal} and \ref{fig:parameter-mixture} present the estimated parameters and the 95\% confidence intervals (CI) for two-component models with covariates, where the error densities are modeled using normal and two-component normal mixture distributions, respectively.  The CI is computed based on the robust standard errors computed via the sandwich variance estimator to
account for potential misspecification. We specifically focus on the two-component models because confidence intervals for some parameters become excessively wide in models with three or more components.\footnote{Parameter estimates for three-component models are provided in Figures \ref{fig:parameter-normal-3} and \ref{fig:parameter-mixture-3} in Appendix \ref{sec:additional}.}

In Figure \ref{fig:parameter-normal}, under the normality assumption, the two latent types differ significantly in their estimated mixing proportions, means, and variances, whereas no statistically significant differences are found in the estimated coefficients for $\log K_{it}$ and import shares. Specifically, compared to the first type, the second type exhibits higher mixing proportions, higher means, and lower variances consistently across all three industries. The absence of statistical evidence for differences in $\beta_{\log K}$ and $\beta_{\text{import}}$ suggests that the Cobb-Douglas specification with two-component mixtures adequately captures the observed persistence in material shares.

Figure \ref{fig:parameter-mixture} presents parameter estimates for two-component models under the normal mixture density, reporting averaged mean parameters and average standard deviation parameters as $\hat\mu_j = \sum_{k=1}^2 \hat\tau_{jk}\hat\mu_{jk}$ and $\widehat{Var}(\epsilon_{it}| D_i=j)=\hat\sigma_j^2+\hat\tau_{j1}(1-\hat\tau_{j1}) (\hat\mu_{j1}-\hat\mu_{j2})^2$, respectively, where bootstrap is used to construct the confidence intervals for these two parameters. Mixing proportions and mean differences between latent types resemble those under normal density, while $\beta_{\log K}$ and $\beta_{\text{import}}$ are insignificant for both types. Wider confidence intervals reflect model complexity.  

\begin{table}
  \centering
  \caption{The estimated number of technology types using the AIC, the BIC, the LRT, and rank test}
  \label{tab:model_comparison}
  \begin{adjustbox}{width=\textwidth}
    \begin{tabular}{lccccccccc}
      \toprule
  \textbf{Model Specification }		& \multicolumn{3}{c}{\textbf{AIC}} & \multicolumn{3}{c}{\textbf{BIC}} & \multicolumn{3}{c}{\textbf{LRT} / \textbf{Rank Test} } \\
      \cmidrule(lr){2-4} \cmidrule(lr){5-7} \cmidrule(lr){8-10}
  \textbf{Covariates / Error Distribution Specification}	& \textbf{Metal} & \textbf{Food} & \textbf{Textiles}
      & \textbf{Metal} & \textbf{Food} & \textbf{Textiles}
      & \textbf{Metal} & \textbf{Food} & \textbf{Textiles} \\
      \midrule\midrule
      \textbf{Panel A: Mixture Models under Conditional Independence (\ref{spec-4})} \\
      \midrule
      Nonparametric Rank  Test (ave-rk)& - & - & - & - & - & - & 3 & 3 & 3 \\
      \midrule
            \textbf{Covariates / Error Distribution Specification} \\
      No Covariate / Normal & 8 & 10 & 9 & 5 & 8 & 6 & 8 & 10+ & 9 \\
      No Covariate / two-component normal mixture & 6 & 10 & 10 & 5 & 8 & 6 & 5 & 8 & 6 \\
      No Covariate / three-component normal mixture & 5 & 10 & 8 & 4 & 8 & 4 & 4 & 8 & 6 \\

      \midrule
      $\log K$, Import, CIIU4 / Normal   & 7 & 10 & 9 & 4 & 8 & 6 & 4 & 10+ & 7  \\
      $\log K$, Import, CIIU4 / two-component normal   mixture & 4 & 8 & 9 & 4 & 6 & 6 & 4 & 5 & 5  \\
      $\log K$, Import, CIIU4 / three-component normal   mixture  & 4 & 7 & 6 & 4 & 6 & 6 & 4 & 5 & 5  \\\midrule\midrule

      \textbf{Panel B: Mixture Models under AR(1) specification (\ref{model-ar1})} \\
        \midrule
          \textbf{Covariates / Innovation Distribution Specification} \\
              \midrule
              No Covariate / Normal & 9 & 10 & 10 & 4 & 6 & 4 & 4 & 7 & 4 \\
              No Covariate / two-component normal   mixture & 5 & 7 & 3 & 5 & 6 & 3 & 1 & 2 & 3 \\
              \midrule
              $\log K$, Import, CIIU4 / Normal    & 4 & 7 & 5 & 3 & 4 & 4 & 3 & 5 & 4 \\
              $\log K$, Import, CIIU4 / two-component normal   mixture  & 3 & 5 & 3 & 2 & 5 & 3 & 2 & 2 & 2 \\
              \bottomrule
              \end{tabular}
  \end{adjustbox}

  \begin{tablenotes}
    \footnotesize
  \textbf{Notes:}  The table reports the selected number of components by AIC, BIC, LRT, and rank test for different model specifications for Fabricated Metal Products, Food Products, and Textile industries in Chile.
  \end{tablenotes}
\end{table}

%

\begin{figure}[ht!]
    \centering
    \caption{Histograms of the estimated component-specific error distributions using three-component models ($\mathcal{K}=3$) with $\log K$, Import, and CIIU4 as covariates}\vspace{-0.5cm}  \label{fig:distribution_all}
    \subfigure[Fabricated Metal Products ]{\includegraphics[width=0.32\textwidth]{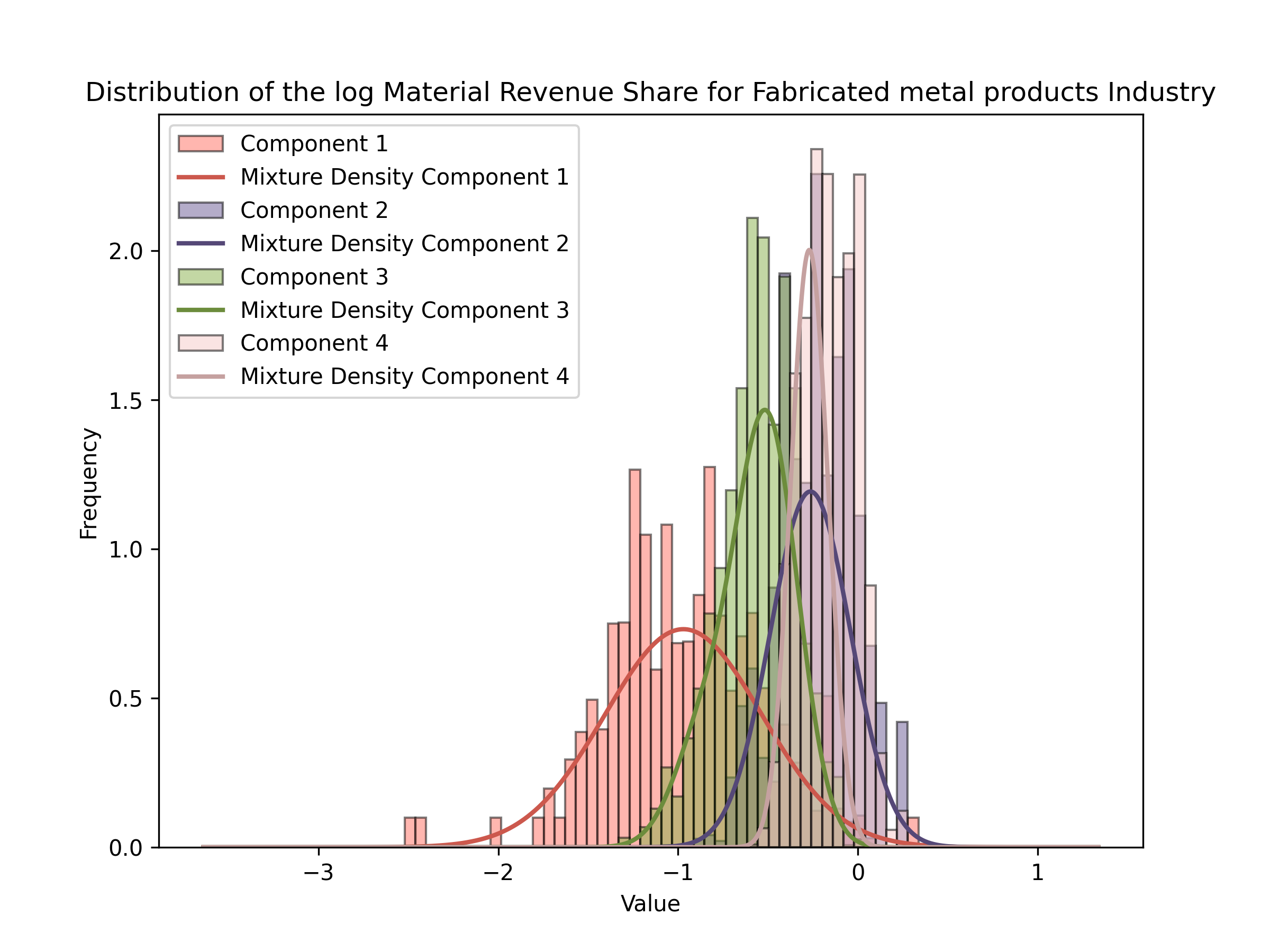}}
    \subfigure[Food Products ]{\includegraphics[width=0.32\textwidth]{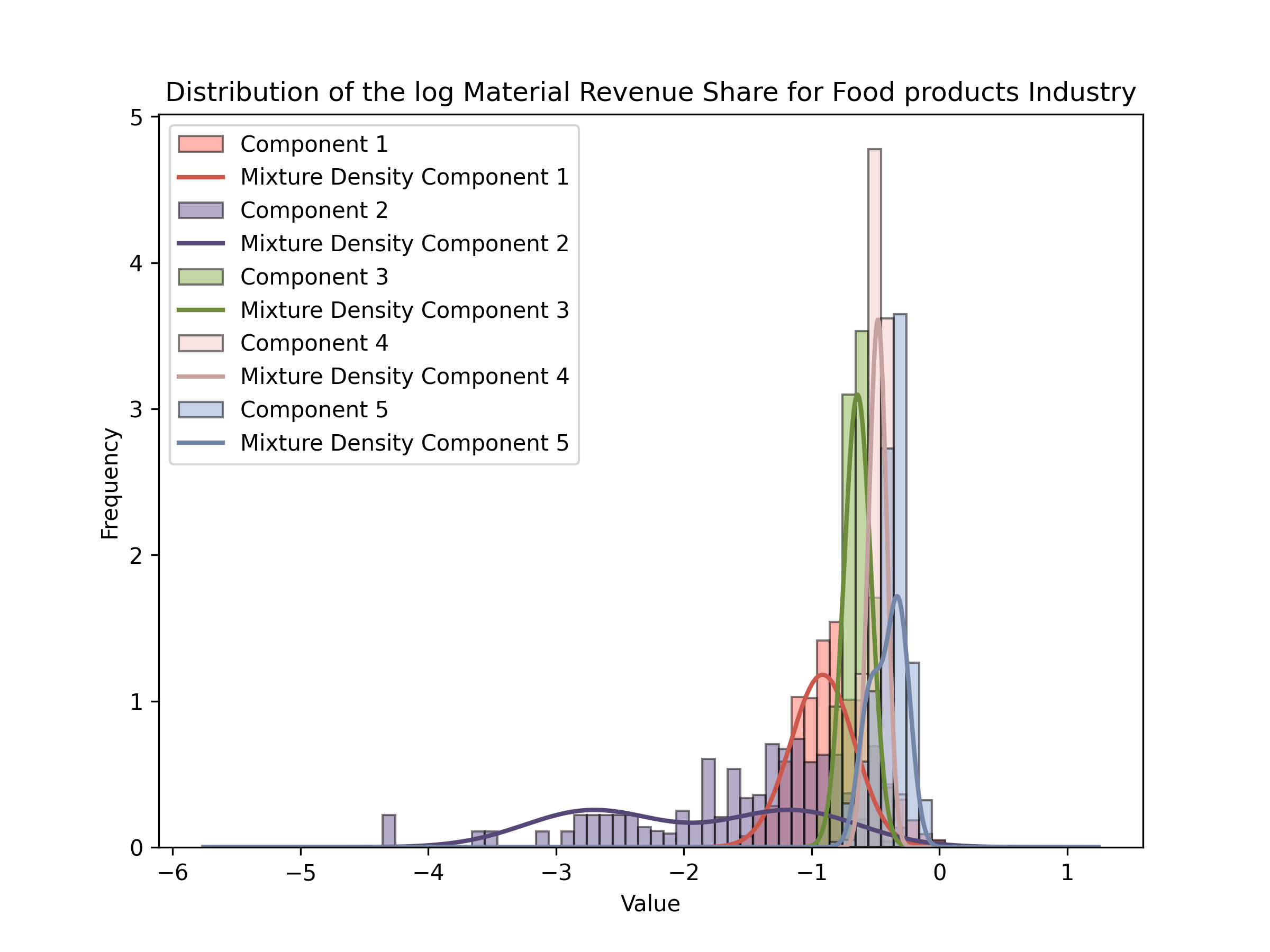}}
    \subfigure[Textiles  ]{\includegraphics[width=0.32\textwidth]{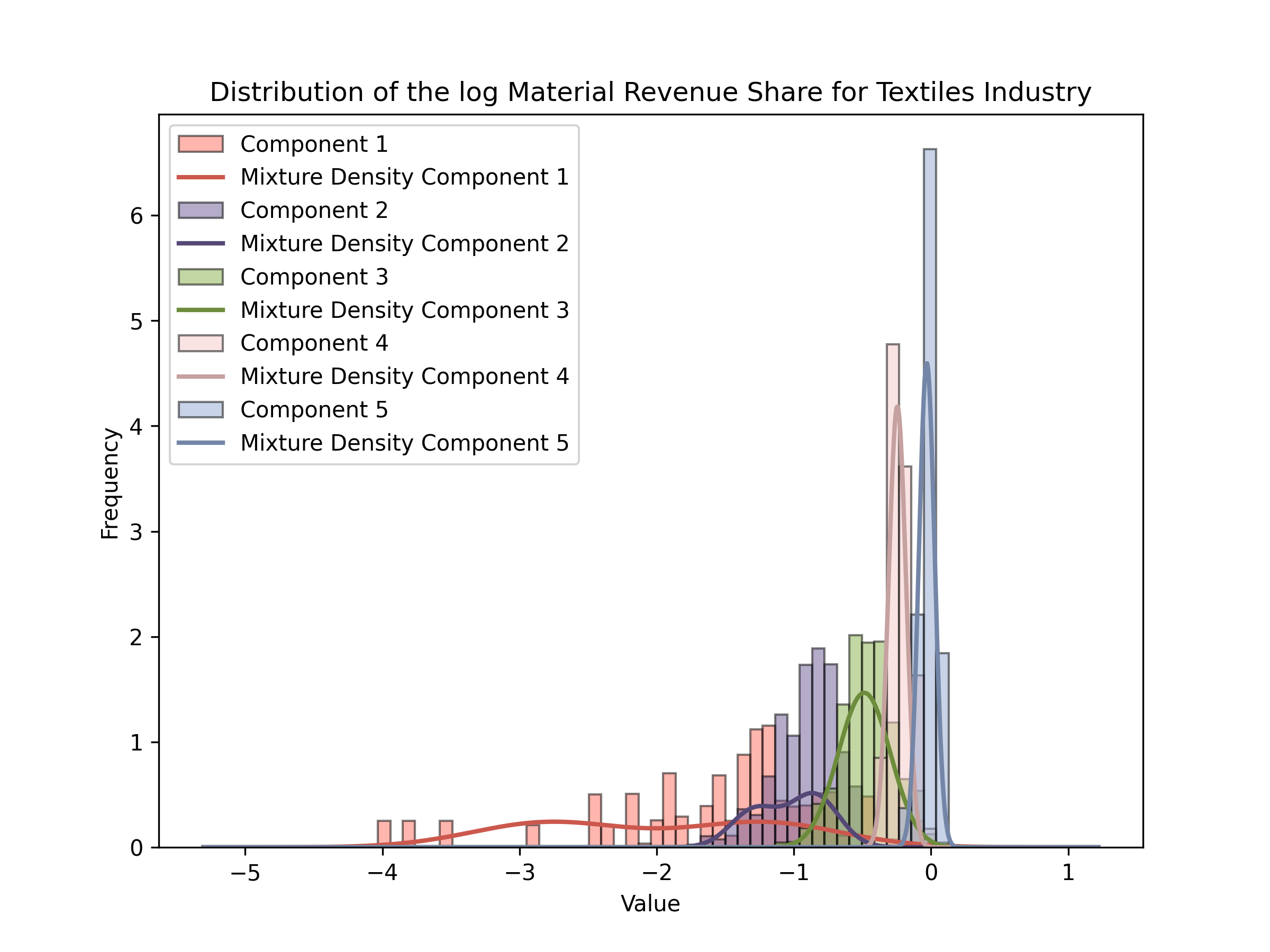}}
\end{figure}

\begin{figure}[ht!]
    \centering
        \caption{Estimated parameter values for two-component mixture models under conditional independence (\ref{spec-4}) with normal error density and covariates  ($M=2, \mathcal{K}=1$)  } \label{fig:parameter-normal}
    \subfigure[$\hat\alpha_j$]{\includegraphics[width=0.32\textwidth]{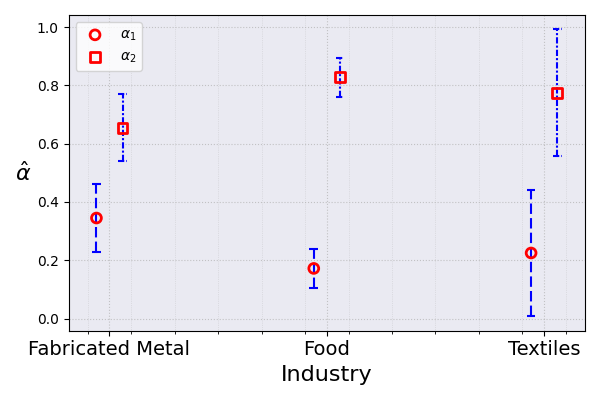}}
    \subfigure[$\hat\mu_j$]{\includegraphics[width=0.32\textwidth]{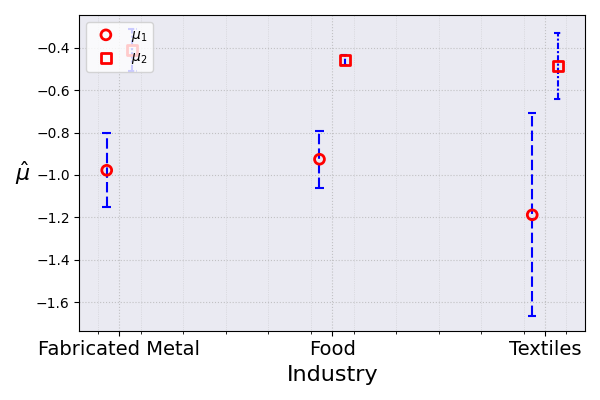}}
    \subfigure[$\hat\sigma_j$]{\includegraphics[width=0.32\textwidth]{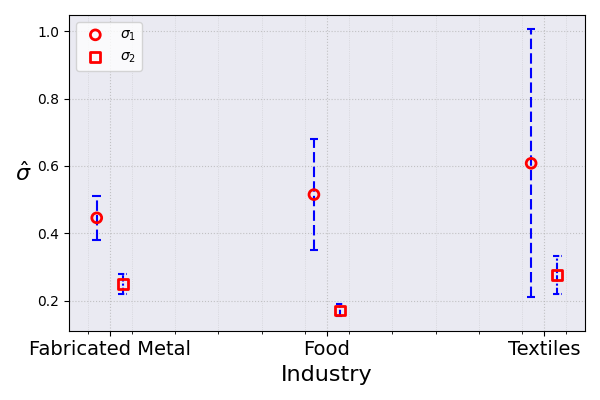}}

    \subfigure[$\hat{\beta}_{\log K, j}$]{\includegraphics[width=0.32\textwidth]{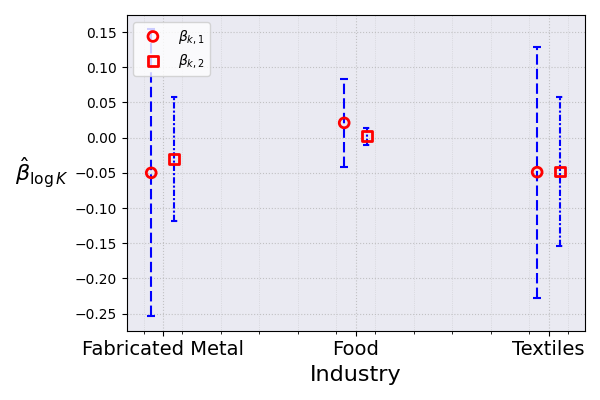}}
    \subfigure[$\hat{\beta}_{\text{Import},j}$]{\includegraphics[width=0.32\textwidth]{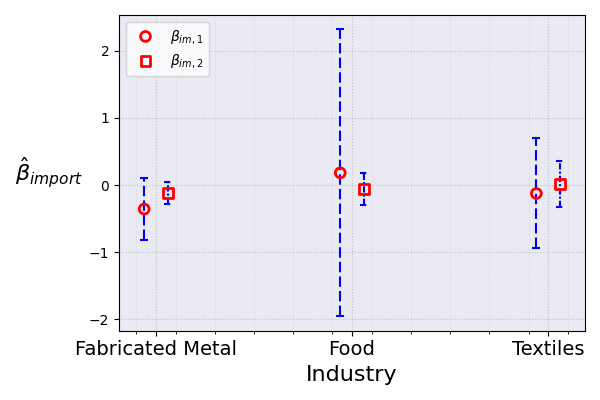}}
 \end{figure}


\begin{figure}[ht!]
    \centering
        \caption{Estimated parameter values for two-component mixture models under conditional independence (\ref{spec-4}) with two-components normal mixture error density and covariates  ($M=2, \mathcal{K}=2$) }\label{fig:parameter-mixture}
    \subfigure[$\hat\alpha_j$]{\includegraphics[width=0.32\textwidth]{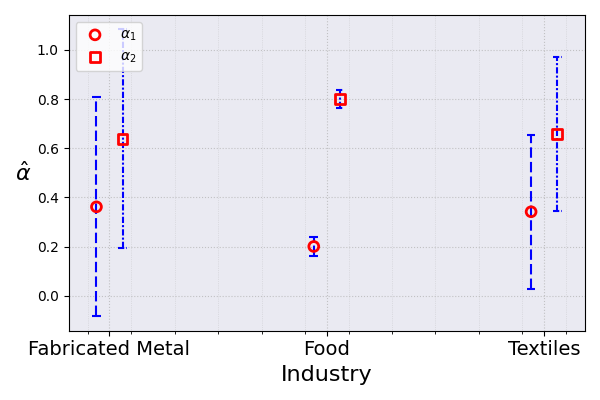}}
    \subfigure[$\hat{\mu}_j$]{\includegraphics[width=0.32\textwidth]{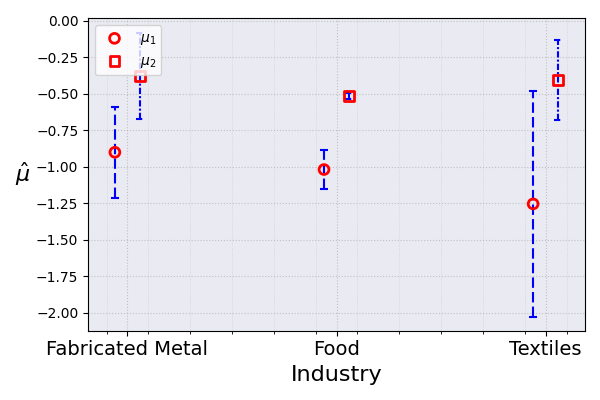}}
    \subfigure[$\widehat{Var}(\epsilon_{it}| D_i=j)$]{\includegraphics[width=0.32\textwidth]{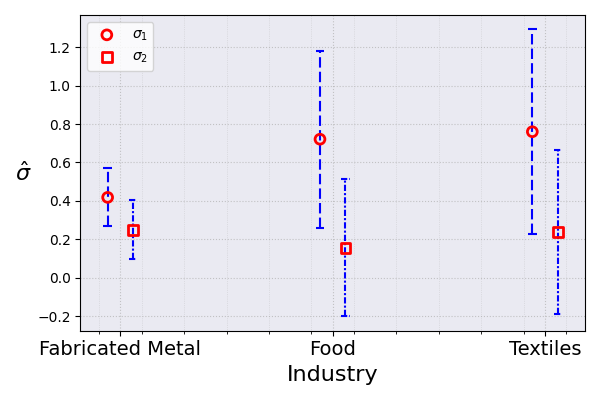}}

    \subfigure[$\hat{\beta}_{\log K, j}$]{\includegraphics[width=0.32\textwidth]{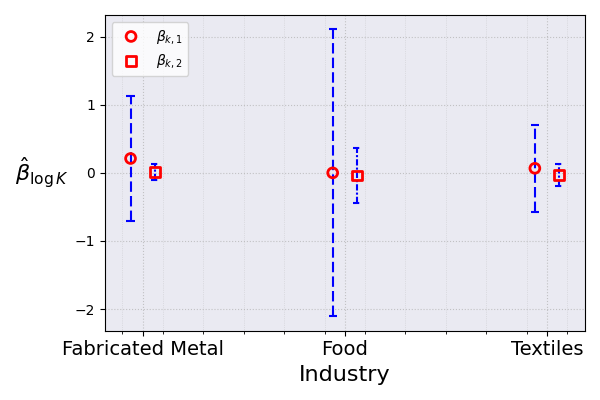}}
    \subfigure[$\hat{\beta}_{\text{Import},j}$]{\includegraphics[width=0.32\textwidth]{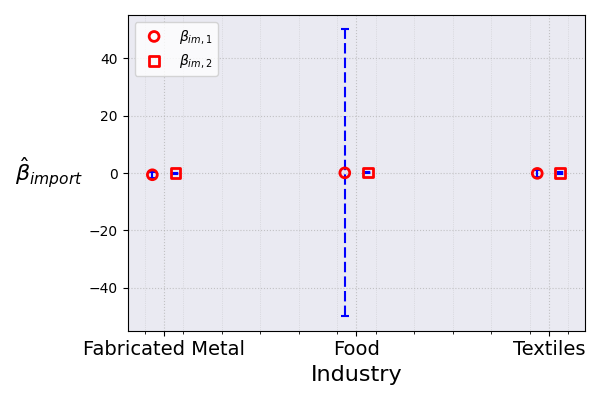}}
    \end{figure}

%
%

 \begin{figure}[ht!]
    \centering
        \caption{ Estimated parameter values for two-component mixture models under AR(1) specification with normal mixture error density and covariates  ($M=2, \mathcal{K}=2$).  }
        \label{fig:parameter-ar1-mixture}
    \subfigure[$\hat\alpha_j$]{\includegraphics[width=0.32\textwidth]{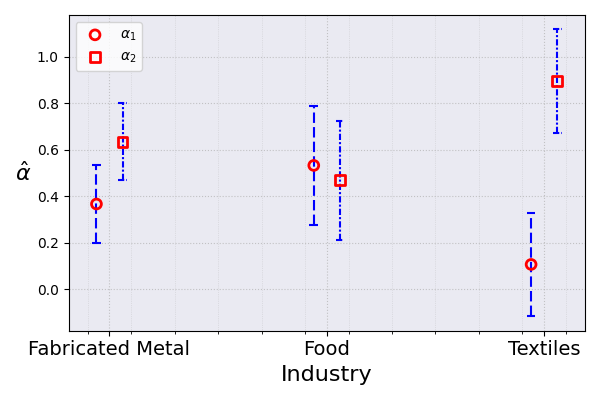}}
      \subfigure[$\hat\mu_j$]{\includegraphics[width=0.32\textwidth]{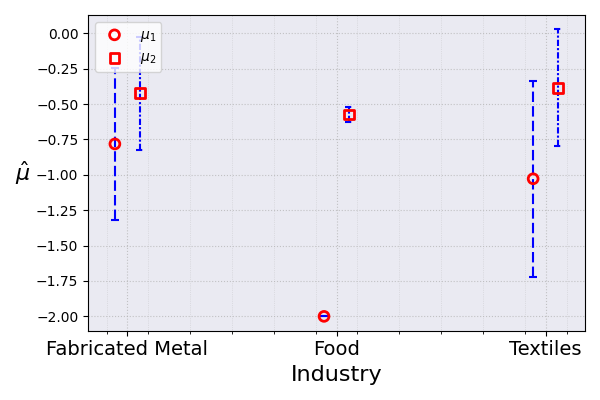}}     \subfigure[$\widehat{Var}(\epsilon_{it}| D_i=j)$]{\includegraphics[width=0.32\textwidth]{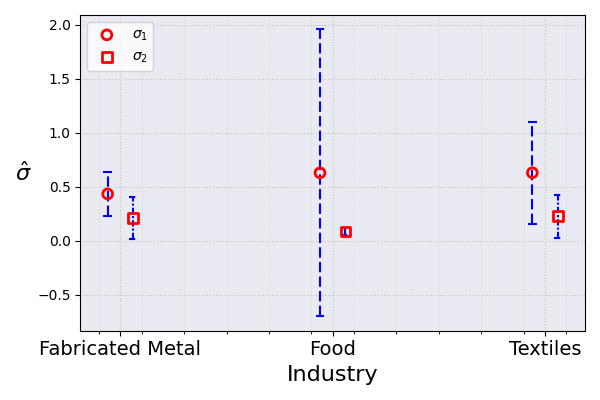}}

    \subfigure[$\hat\rho_j$]{\includegraphics[width=0.32\textwidth]{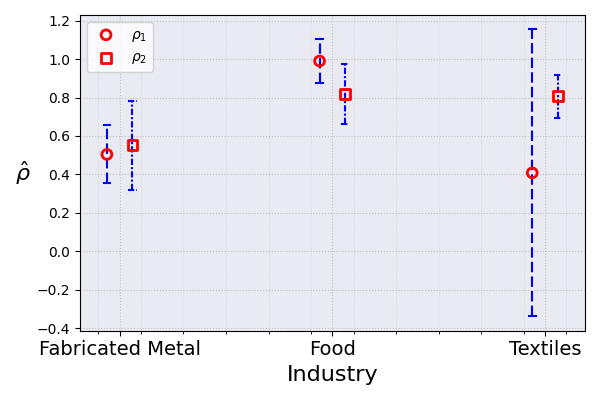}}
    \subfigure[$\hat{\beta}_{\log K, j}$]{\includegraphics[width=0.32\textwidth]{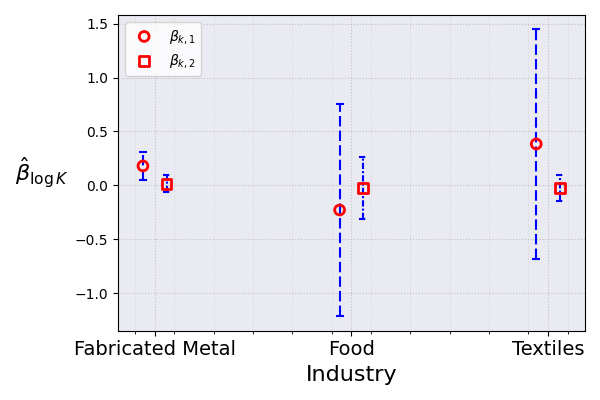}}
    \subfigure[$\hat{\beta}_{\text{Import},j}$]{\includegraphics[width=0.32\textwidth]{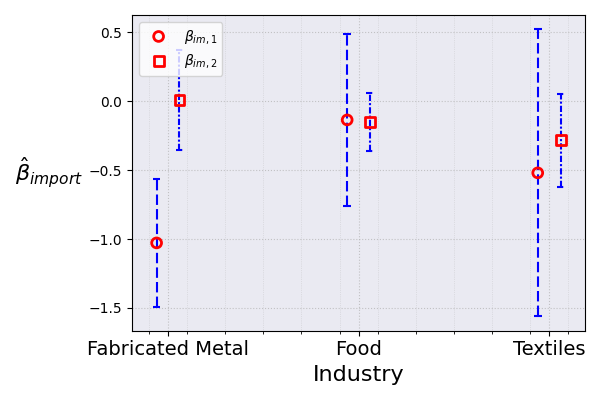}}
\end{figure}

%

As discussed in Example \ref{example-2}, we also consider the case where the error term $\epsilon_{it}$ in equation (\ref{spec-4}) follows an AR(1) process with $\epsilon_{it}=\rho_j \epsilon_{it-1}+\xi_{it}$ to capture potential persistence in shocks to output elasticities of material input. This leads to the following specification:
\begin{align}\label{model-ar1}
\log\left(\frac{P_{V,t} V_{it}}{P_{O,t} O_{it}}\right)
=   \rho_j\log\left(\frac{P_{V,t-1} V_{it-1}}{P_{O,t-1} O_{it-1}}\right)
+ \bs{\beta}_j^\top \bs{x}_{it} -\rho_j\bs\beta_j^\top \bs{x}_{it-1} + \xi_{it}\quad\text{for $t=2,3,...,T$,}
\end{align}
where $\xi_{it}$ is drawn from either a normal distribution or multi-component normal mixtures as specified in equation (\ref{xi}). By further specifying the initial distribution of $\log\left(\frac{P_{V,t} V_{it}}{P_{O,t} O_{it}}\right)$ at $t=1$ as:
\begin{equation}\label{model-ar1-2}
\log\left(\frac{P_{V,1} V_{i1}}{P_{O,1} O_{i1}}\right)
=    \bs{\beta}_{1,j}^\top \bs{x}_{i1}+\epsilon_{i1},
\end{equation}
where $\epsilon_{i1} \overset{iid}{\sim} \mathcal{N}(\mu_{1,j},\sigma_{1,j}^2)$ or $\sum_{k=1}^{{\cal{K}}} \tau_{1,jk}\, \mathcal{N}(\mu_{1,jk},\sigma_{1,j}^2)$ given latent technology type $D_i = j$, this specification yields the conditional density of $\{y_t\}_{t=1}^T$ and $\{\bs{x}_{t}\}_{t=1}^T$ in the form of equation (\ref{eq:fm}) with (\ref{eq:fm-dynamic}), where the component-specific densities are given by either equation (\ref{eq:f1-dynamic}) or (\ref{eq:f1-dynamic-mixture}). In this framework, each component-specific technology type represents a distinct stochastic process governing the output elasticities of material input. 

Panel B of Table \ref{tab:model_comparison} reports the estimated number of components selected by the AIC, BIC, and LRT for models specified in equations (\ref{model-ar1})–(\ref{model-ar1-2}).\footnote{Parameters are estimated using the EM algorithm detailed in Appendix \ref{sec:em-algorithm-dynamic}, explicitly incorporating the constraint implied by specification (\ref{model-ar1}).} Due to the incorporation of persistence in the error term $\epsilon_{it}$ within the AR(1) framework, the results indicate fewer estimated components compared to models assuming conditionally independent errors. Nonetheless, there remains clear evidence of heterogeneity in the stochastic processes governing the error term $\epsilon_{it}$ especially for Food and Textiles industries; under the specification with covariates and two-component normal mixture innovation densities, the LRT identifies 2, 2, and 2 components for the Metal, Food, and Textiles industries, respectively, whereas the BIC selects 2, 5, and 2 components. 

Figure \ref{fig:parameter-ar1-mixture} reports parameter estimates for the AR(1) specification with covariates and two-component normal mixture innovation densities.\footnote{For estimating models with an AR(1) specification, we imposed a restriction that the parameter $\mu_j$ lies within the interval $[-2, 2]$, as the estimated value of $\mu_1$ for the food industry was implausibly low when the model was estimated without this bound.}  For Fabricated Metal Products and Textile industries, relative to the first type, the second type exhibits higher mixing proportions, higher means, lower variances, and higher persistence (as captured by the AR(1) coefficient $\rho_j$). Reflecting model complexity, the confidence intervals are generally wide, and differences between latent types are often not statistically significant.

%

 \subsection{Plant heterogeneity in the stochastic process of factor-augmented technological changes }


Factor-augmented technological changes are widely used to capture heterogeneity in production functions beyond the standard Hicks-neutral approach  \citep{doraszelski2018measuring,Zhang2019,Raval2019}.
The heterogeneity in the stochastic process of factor-augmented technological changes  can be effectively modeled with finite mixture dynamic panel data models.

 As discussed in Example \ref{example-3},  we extend the production function (\ref{prod}) by incorporating labor-augmented technological change \(\delta_{it}\) \citep[cf.,][]{Demirer2022}:
\[
O_{it}=e^{\omega_{it}}\bar F_{jt}(\chi_{jt}(V_{it}, e^{\delta_{it}} L_{it}), K_{it}, Z_{it},\epsilon_{it})\ \text{ with }\
\chi_{jt}(V_{it}, e^{\delta_{it}} L_{it}) := \left[\alpha_{V,jt} V_{it}^{1/\varsigma_j} + \alpha_{L,jt} (e^{\delta_{it}} L_{it})^{1/\varsigma_j}\right]^{\varsigma_j},
\]
and \(\delta_{it}\) evolves according to an AR(1) process:
\[
\delta_{it} = \rho^{\delta}_j \delta_{it-1} + \eta^{\delta}_{it},
\]
where \(\eta^{\delta}_{it}\) follows either a normal or a multi-component normal mixture distribution.

Assuming both material and labor inputs are flexibly chosen after observing relevant prices and shocks, profit maximization implies a dynamic relationship for log input ratios:
\begin{equation}\label{eq:factor-augmenting}
\log(V_{it}/L_{it})\ = \mu^{\delta}_{jt} + \rho^{\delta}_j \log(V_{it-1}/L_{it-1})\ + \tilde{\eta}^{\delta}_{it},
\end{equation}
where  $\tilde\eta_{it}^\delta = \eta^{\delta}_{it}/(1-\varsigma_j)$ and $\mu^{\delta}_{jt}$ depends on $\alpha_{V,jt}$, $\alpha_{L,jt}$, $\varsigma_j$, and the input prices. By specifying the density functions of  the initial values \( \log(V_{it}/L_{it}) \) at $t=1$,
this specification leads to conditional densities as described in equations (\ref{eq:fm}) and (\ref{eq:fm-dynamic}), with component-specific densities detailed in either (\ref{eq:f1-dynamic}) or (\ref{eq:f1-dynamic-mixture}).

		Table \ref{tab:model_comparison_mlshare} presents the results of model selection for the dependent variable $\log(V/L)$ (the share of material inputs over labor) across three industries—Metal, Food, and Textiles—using AIC, BIC, and Likelihood Ratio (LR) metrics. The table compares various model specifications, including those without covariates, those with conditioning variables ($\log K$, import shares, and 4-digit industry classification dummies). Both static models and models with autocorrelated error terms (AR1) are analyzed, and the error term structures range from normal distributions to more flexible two-component, three-component, and 4-component normal mixtures.


   For the model without mixture components or covariates, AIC overestimates the number of components, predicting $10$ components for all industries. In contrast, BIC is more conservative, predicting $3$ components for Metal and Textiles, and $6$ for Food.
	More flexible mixture error structures generally lead to fewer predicted components, particularly under BIC and LR metrics.
	Introducing mixture error structures reduces the number of predicted components across all metrics.
   For example, under the Plain Mixture 2-Component model, AIC predicts $6$ components for Metal, $2$ for Food, and $2$ for Textiles, while BIC predicts $3$, $2$, and $2$ components, respectively. LR aligns more closely with BIC, predicting $2$, $4$, and $2$ components, respectively.
   Including covariates ($\log K$) and CIIU4 controls in the models reduces the predicted number of components in most cases, particularly when using BIC and LR.
  For example, under the $lnK$, Import, 2-Component Control CIIU4 model, LR predicts $1$ components for Metal, $3$ for Food, and $1$ for Textiles, compared to $2$, $4$, and $2$ under the corresponding Plain Mixture model.



			\begin{table}[ht]
			\centering
			\caption{The estimated number of factor-augmented technological change AR(1) processes for model (\ref{eq:factor-augmenting}), selected by AIC, BIC, LR, and rank tests}
			\label{tab:model_comparison_mlshare}
			\begin{adjustbox}{width=\textwidth}
				\begin{tabular}{lccccccccc}
					\toprule
					\textbf{Model Type} & \multicolumn{3}{c}{\textbf{AIC}} & \multicolumn{3}{c}{\textbf{BIC}} & \multicolumn{3}{c}{\textbf{LR}} \\
					\cmidrule(lr){2-4} \cmidrule(lr){5-7} \cmidrule(lr){8-10}
			\textbf{Covariates / Innovations Distribution}			& \textbf{Metal} & \textbf{Food} & \textbf{Textiles}
					& \textbf{Metal} & \textbf{Food} & \textbf{Textiles}
					& \textbf{Metal} & \textbf{Food} & \textbf{Textiles} \\ \midrule
          No Covariate / Normal & 10 & 10 & 10 & 3 & 6 & 3 & 3 & 6 & 3 \\
          No Covariate  / two-component normal   mixture & 6 & 2 & 2 & 3 & 2 & 2 & 2 & 4 & 2 \\

          \midrule

          $\log K$, Import, CIIU4 / Normal  & 5 & 3 & 4 & 2 & 3 & 2 & 3 & 3 & 3 \\
          $\log K$, Import, CIIU4  / two-component normal   mixture  & 10 & 10 & 4 & 3 & 3 & 3 & 1 & 3 & 1 \\
					\bottomrule
				\end{tabular}
			\end{adjustbox}
		\end{table}

\section{Conclusion} 
This paper develops statistical methods for determining the number of components in panel data finite mixture regression models with normally distributed or more flexible normal mixture errors. We establish that panel data structures eliminate higher-order degeneracy problems present in cross-sectional normal mixture models while retaining issues of unbounded likelihood and infinite Fisher information. We address these challenges and derive the asymptotic null distribution of the LRT statistic, demonstrate the consistency of BIC and inconsistency of AIC, and propose a sequential hypothesis testing approach for consistent component selection.

The empirical application to Chilean manufacturing data provides compelling evidence of substantial plant heterogeneity in production technology. Using finite mixture models with conditionally independent normal mixture errors, we identify 4 or 5 distinct technology types across three major industries: Fabricated Metal Products, Food Products, and Textiles. Our analysis reveals significant variation in output elasticities of material inputs across plants within narrowly defined industries, with nonparametric rank tests consistently identifying at least three components across all industries. Dynamic panel analysis incorporating lagged dependent variables as covariates suggests that while some heterogeneity reflects transitory variations, substantial permanent technological differences persist across plants.

These findings contrast sharply with standard production function estimation methods that impose homogeneous coefficients across plants. Our results highlight the importance of explicitly accounting for unobserved plant heterogeneity beyond Hicks-neutral technological differences, with significant implications for policy analysis where homogeneous assumptions may mischaracterize productivity distributions.

The methodological framework extends beyond production function analysis to other contexts with latent group structures in panel data. Future research could extend the analysis to more general production function specifications, develop methods for simultaneous component and distributional selection, and investigate relationships between estimated technological types and observable plant characteristics. Our work contributes both methodologically and empirically to understanding plant heterogeneity, providing  useful tools for uncovering latent structures while demonstrating substantial technological diversity that challenges common assumptions in empirical production analysis.

\newpage

\bibliographystyle{apalike}
\bibliography{references}

\newpage

\appendix
\section{Appendix}

\renewcommand{\thefigure}{\Alph{section}\arabic{figure}}
\setcounter{figure}{0}

\renewcommand{\thetable}{\Alph{section}\arabic{table}}
\setcounter{table}{0}

\subsection{Likelihood ratio test for testing $H_{02}: \alpha(1-\alpha)=0$} \label{sec:LRT3}

We derive the asymptotic distribution of the LRT statistic for testing the hypothesis $H_{02}: \alpha(1 - \alpha) = 0$ for testing $H_0: M=1$ against $H_1: M=2$. Our analysis focuses on the null hypothesis $\alpha = 0$, noting that the case $\alpha = 1$ follows by symmetry.

Consider the following restricted parameter space in which $\bs\theta_1$ and $\bs\theta_2$ are bounded away from each other while $\alpha$ can take the value of $0$ and $1$:
\[
 {\widetilde\Theta}_{\bs\vartheta_2}(\bs c) := \{ \bs\vartheta_2 \in  \Theta_{\bs\vartheta_2} :  ||\bs\theta_1-\bs\theta_2||\geq c_1, \ \min_j \sigma_j^2\geq c_2 \sum_{k=1}^2  \alpha_k\sigma_k^2,\   \sum_{k=1}^2 \alpha_k \sigma_k^2\geq c_3 \}.
 \]
Define the MLE $\tilde{\bs\vartheta}_2$ in this restricted parameter space  by
\begin{equation}\label{eq:MLE-2}
\ell_n^2(\tilde{\bs\vartheta}_2) = \max_{{\bs{\vartheta}}_2 \in {\widetilde\Theta}_{\bs\vartheta_2}(c)} \ell_n^2({\bs\vartheta}_2).
\end{equation}
Following the proof of  Proposition (\ref{prop:consistency}), $\tilde{\bs\vartheta}_2$   can be shown to be consistent under  Assumptions \ref{assumption:K} and \ref{assumption:consistency}.

 To simplify the asymptotic representation and regularity conditions, we use the parameter ${\bs\lambda}:=\bs\theta_1-\bs\theta_2$  and reparameterize $(\bs\theta_1, \bs\theta_2)$ to $( {\bs\lambda},\bs \theta_2)$ so that the model parameter is $\bs\psi=(\bs\theta_2,\bs\lambda,\alpha)\in\Theta_{\bs\psi}$. Under $H_0: M=1$, the two-components model replicates the true one-component model when $\bs\theta_2=\bs\theta^*$ and $\alpha=0$. Given that $\bs\lambda$ is unidentified when $\alpha=0$, we follow \citet{andrews01em}  by deriving the limit of the LRT statistic for each $ {\bs\lambda}:=\bs\theta_1-\bs\theta_2$ in $\Theta_{\ {\bs\lambda}}(c_1):=\{ {\bs\lambda}\in  \Theta_{ {\bs\lambda}}: || {\bs\lambda}||\geq c_1\}$ for $c_1>0$.

Define the reparameterized log-density as $\log g_2(\bs w;\bs\theta_2,\bs\lambda,\alpha)=\log( \alpha f(\bs w;\bs\theta_2+\bs\lambda)+(1-\alpha) f(\bs w;\bs\theta_2))$. Collect the partial derivative of $\log g_2(\bs w;\bs\theta_2,\bs\lambda,\alpha)$ with respect to $\bs\theta_2$ and its right partial derivative with respect to $\alpha$ evaluated at $(\bs\theta_2,\bs\lambda,\alpha)=(\bs\theta^*,\bs\lambda,0)$ as
\begin{align}\label{eq:s_1}
	 &\bs{s}(\bs{w};\bs\lambda): = \begin{pmatrix}
		\bs{s}_{\bs{\bs\theta_2}}(\bs w)   \\
		 {s}_{\alpha}(\bs w;\bs\lambda)
\end{pmatrix}:=
\begin{pmatrix}
\nabla_{\bs\theta_2} \log g_2(\bs w;\bs\theta^*,\bs\lambda,0)\\
\nabla_{\alpha} \log g_2(\bs w;\bs\theta^*,\bs\lambda,0)\\
\end{pmatrix}
=
\begin{pmatrix}
 \frac{\nabla_{\bs{\theta}} f(\bs w;\bs\theta^*) }{f(\bs w;\bs\theta^*)} \\
  \frac{f(\bs w;\bs\theta^*+\bs\lambda) - f(\bs w;\bs\theta^*)}{f^*(\bs w;\bs\theta^*)}
\end{pmatrix}.
\end{align}

Define $\bs{\mathcal{I}}(\bs\lambda) := E[\bs{s}(\bs{W};\bs\lambda)\bs{s}(\bs{W};\bs\lambda)\t]$. Analogously to (\ref{eq:information}),  define
\begin{equation}\label{eq:information-2}
\begin{split}
\bs{\mathcal{I}}(\bs\lambda) =
\begin{pmatrix}
 \bs{\mathcal{I}}_{\bs\theta_2}  & \bs{\mathcal{I}}_{\bs\theta_2\alpha}(\bs\lambda)  \\
 \bs{\mathcal{I}}_{\alpha\bs\theta_2}(\bs\lambda)  & \bs{\mathcal{I}}_{\alpha}(\bs\lambda)
\end{pmatrix},
\quad \bs{\mathcal{I}}_{\bs\theta_2}  = \mathbb{E}[\bs s_{\bs\theta_2}(\bs W)  \bs s_{\bs\theta_2}(\bs W)\t], \quad \bs{\mathcal{I}}_{\alpha\bs\theta_2}(\bs\lambda)  = \mathbb{E}[\bs s_{\alpha }(\bs W;\bs\lambda)  \bs s_{\bs\theta_2}(\bs W)\t ],  \\
\bs{\mathcal{I}}_{\bs\theta_2\alpha}(\bs\lambda)  = \bs{\mathcal{I}}_{\alpha\bs\theta_2}(\bs\lambda) \t, \quad \bs{\mathcal{I}}_{\alpha}(\bs\lambda)  =\mathbb{E}[\bs s_{\alpha }(\bs W)^2 ],  \ \text{and}\ \ \bs{\mathcal{I}}_{\alpha,\bs\theta_2}(\bs\lambda)  = \bs{\mathcal{I}}_{\alpha}(\bs\lambda)  - \bs{\mathcal{I}}_{\alpha\bs\theta_2}(\bs\lambda)  \bs{\mathcal{I}}_{\bs\theta_2} ^{-1} \bs{\mathcal{I}}_{\bs\theta_2\alpha} (\bs\lambda).  \quad\quad
\end{split}
\end{equation}
Let $\{\bs{S}(\bs\lambda)=(\bs S_{\bs\theta_2},\bs S_\alpha(\bs\lambda)): \bs\lambda\in\widetilde\Theta_{\bs\lambda}(c_1)\}$ be a mean zero vector-valued Gaussian process such that  $\mathbb{E}[\bs{S}( \bs\lambda) \bs{S}(\bs\lambda)^\top]=\bs{\mathcal{I}}(\bs\lambda)$, where $\bs S_{\bs\theta_2}$ is independent of $\bs\lambda$, $\bs S_\alpha(\bs\lambda)$ is $1\times 1$, and $\bs{\mathcal{I}}(\bs\lambda)$ is defined in (\ref{eq:information-2}).
Let $\bs{S}_{\alpha,\bs\theta_2}(\bs\lambda) := \bs S_\alpha(\bs\lambda)-\bs{\mathcal{I}}_{\alpha\bs\theta_2}(\bs\lambda) \bs{\mathcal{I}}_{\bs\theta_2} ^{-1} \bs S_{\bs\theta_2}$.

Define the LRT statistics for testing $H_{02}:\alpha(1-\alpha)=0$ by
\begin{align}\label{eq:LR_def-2}
	\widetilde{LR}^2_n(\bs c) &:=   2 \left\{\ell_n^2(\widetilde{\bs{\vartheta}}_2) - \ell_n^1(\hat{\bs\theta}_0)  \right\} =
	2\left\{ \max_{(\bs\theta_2,\bs\lambda,\alpha)\in \Theta_{\bs\theta} \times \tilde\Theta_{\bs\lambda}(c_1)\times [0,1/2]} \ell_n^2(\bs\theta_2+\bs\lambda,\bs\theta_2,\alpha) -  \ell_n^1(\hat{\bs\theta}_0)  \right\}.
\end{align}

\begin{assumption}\label{assumption:h02} (a) $\theta^*$ is in the interior of $\Theta_{\theta}$. (b) $f(\bs w;\bs\theta)$ is twice continuosly differentiable on $\Theta_{\bs\theta}$.
 (c) $\bs{\mathcal{I}}(\bs\lambda)$ defined in (\ref{eq:information-2}) is finite and positive definite uniformly in $\Theta_{\bs\lambda}(c_1)$. (d) Assumption \ref{assumption:consistency} holds when we replace ${\bar \Theta}_{\bs\vartheta_{2}}(\bs c)$ with ${\tilde \Theta}_{\bs\vartheta_{2}}(\bs c)$.

\end{assumption}

\begin{proposition}\label{prop:h02}
 Assumptions \ref{assumption:K} and \ref{assumption:h02} hold. Then, (a)  $\inf_{\boldsymbol{\vartheta}_2 \in \Theta_{\bs\vartheta_2}^*} ||\widetilde{\boldsymbol{\vartheta}}_2 - \boldsymbol{\vartheta}_2|| \to 0$ almost surely, and (b)
$\widetilde{LR}^2_n(\bs c) \xrightarrow{d} \sup_{\bs\lambda\in\Theta_{\bs\lambda}(c_1)} \left(\max\{0, \bs{\mathcal{I}}_{\alpha,\bs\theta_2}(\bs\lambda)^{-1/2} \bs{S}_{\alpha,\bs\theta_2}(\bs\lambda)\}\right)^2$.
\end{proposition}

Assumption \ref{assumption:h02}(b) can be verified for a class of mixture models we consider in this paper.

A necessary condition for Assumption \ref{assumption:h02}(c) is $\sup_{\bs\lambda \in\widetilde\Theta_{\bs\lambda}(c_1)} \mathbb{E}[\nabla_{\alpha} \log g_2(\bs w;\bs\theta^*,\bs\lambda,0)^2]<\infty$, which is violated for the finite mixture normal regression models when $\sigma_j^2> 2\sigma^{*2}$.


One approach is to impose a restriction that $\sigma_{1}^2 \leq 2 \sigma^{*2}$. However, a difficulty arises in imposing the constraint $\sigma_{1}^2 \leq 2 \sigma^{*2}$ when estimating a two-component model, given that $\sigma^{*2}$ is unknown. A possible solution is to use the estimated variance from the one-component model, setting:
$$
\sigma_j^2 \in \left[0 ,\, 2\hat\sigma_0^2-c\right] \text{for $j=1,2$},
$$
where $\hat\sigma_0^2$ is the estimated variance from the one-component model and $c$ is a small positive constant. This ensures the asymptotic validity of the condition $\sigma_j^2 < 2\sigma^{*2}$. However, in finite samples, there remains a positive probability that values of $\sigma_j^2$ within this random interval may violate the constraint. 


\subsection{Asymptotic analysis of the penalized maximum likelihood estimator in Section \ref{sec:bic}} \label{app:bic}

In this appendix, we extend the previous asymptotic analysis for testing \( H_{0}: M = M_0 \) against \( H_{1}: M = M_0 + 1 \) to the more general setting of testing \( H_{0}: M = M_0 \) against \( H_{1}: M = M_1 \), where \( M_1 > M_0+1, \) to derive conditions for \( \ell_{n}^{M}(\hat{\bs\vartheta}_{M}) - \ell_n^{M_0}({\bs{\vartheta}}_{M_0}^*) = O_p(1) \).

There are two primary challenges in extending our earlier analysis. First, while there are exactly \( M_0 \) ways for the \( (M_0 + 1) \)-component model to replicate the \( M_0 \)-component model, the number of ways for an \( M_1 \)-component model to replicate an \( M_0 \)-component model can be substantially larger when $M_1 > M_0+1$. Second, the degree of singularity in the Fisher Information Matrix may become more severe, necessitating the higher-order expansions beyond the second-order for an asymptotic analysis.

To deal with the first challenge, given any integer $M$ that is greater than $M_0$, we extend the definition of the restricted parameter space (\ref{Omega}) by partitioning the \( M \)  component-specific parameters into \( M_0 \) ordered, non-empty subsets. Explicitly, we consider a partition \( \mathcal{T} = (T_1, T_2, \ldots, T_{M_0}) \) of the index set \( \{1, 2, \ldots, M\} \), subject to the following conditions:
(i) \( T_h \cap T_{h'} = \emptyset \) for \( h \neq h' \), and \( |T_h| \geq 1 \),
(ii) \( \bigcup_{h=1}^{M_0} T_h = \{1, 2, \ldots, M\} \),
(iii) If \( j \in T_h \) and \( k \in T_{h'} \) with \( h < h' \), then \( j < k \). The total number of such ordered, non-empty, consecutive partitions is \( J_{M} := \binom{M - 1}{M_0 - 1} \). Let \( \{\mathcal{T}^j\}_{j=1}^{J_{M}} \) denote the set of all these partitions, where each \( \mathcal{T}^j = (T_1^j, T_2^j, \ldots, T_{M_0}^j) \).

Additionally, consider a partition \( \{\Theta_{\bs\theta,h}^*\}_{h=1}^{M_0} \) of the parameter space \( \Theta_{\bs\theta} \), defined by (\ref{theta-partition}), such that each \( \Theta_{\bs\theta,h}^* \) forms a neighborhood containing \( \bs\theta_h^{*} \) but excluding \( \bs\theta_j^{*} \) for all \( j \neq h \). Then, for each partition \( \mathcal{T}^j \), \( j = 1, 2, \dots, J_{M} \), the restricted parameter space is defined as:
\[
\bs{\Psi}_{\mathcal{T}^j}^* = \left\{ \bs{\vartheta}_{M} \in \bar\Theta_{\bs{\vartheta}_{M}}(\bs{c}) : \sum_{j=1}^{M} \alpha_j = 1, \quad \bs{\theta}_j \in \Theta_{\bs{\theta},h}^* \text{ for all } j \in T_h^j, \quad h = 1, \dots, M_0 \right\}.
\]
Note that each \( \bs{\Psi}_{\mathcal{T}^j}^* \) uniquely represents a distinct configuration in which the components from an \( M \)-component model are distributed across \( M_0 \) components. Furthermore, the union of  \( \bs{\Psi}_{\mathcal{T}^j}^* \) over $j=1,2,...,J_M$ covers the parameter space $ \bar\Theta_{\bs\vartheta_{M}}(\bs c)$.

Define the \textit{local} MLE that maximizes the log-likelihood function of the $M$-component model under the constraint that $\bs\vartheta_{M}\in  {\bs\Psi}_{\mathcal{T}^j}^*$ by
\[
 \ell^{M}_n(\widehat{\bs{\vartheta}}_{M}^j)   := \arg\sup_{\bs{\vartheta}_{M}  \in \bs {\Psi}_{\mathcal{T}^j}^*}   \ell_n^{M}(\bs{\vartheta}_{M}).
\]
Because $\cup_{j=1}^{J_{M}}\bs {\Psi}_{\mathcal{T}^j}^*= \bar\Theta_{\bs\vartheta_{M}}(\bs c)$, \begin{equation}
\ell_n^{M}(\hat{\bs{\vartheta}}_{M}) - \ell^{M_0}_n( {\bs{\vartheta}}_{M_0}^*) =\max_{j=1,2,...,J_{M}} \{\ell_n^{M}(\hat{\bs{\vartheta}}_{M}^j) - \ell^{M_0}_n( {\bs{\vartheta}}_{M_0}^*) \}. \label{eq:M0_M1_max}
\end{equation}
In view of (\ref{eq:M0_M1_max}), $\ell_n^{M}(\hat{\bs{\vartheta}}_{M}) - \ell^{M_0}_n( {\bs{\vartheta}}_{M_0}^*)=O_p(1)$ follows if we are able to show that $\ell_n^{M}(\hat{\bs{\vartheta}}_{M}^j) - \ell^{M_0}_n( {\bs{\vartheta}}_{M_0}^*)=O_p(1)$ for any $j=1,...,J_{M}$.

%

To deal with the singularity of the Fisher Information Matrix for the $M$-component density function  $g_{M}(\bs w;\bs\vartheta_M)$, we consider the following one-to-one reparameterization of $\bs\vartheta_M$. Specifically, for $j = 1, 2, \dots, J_{M}$,  we reparameterize \( \bs\alpha = (\alpha_1, \dots, \alpha_{M-1}) \) by \( \beta_h^j = \sum_{k \in T_h^j} \alpha_k \), $\bs\tau_h^j = (\tau_{h,1}^j,...,\tau^j_{h,|T_h^j|-1})\t$, and \( \tau_{h,k}^j = \alpha_k / \beta_h^j \) for \( k \in T_h^j \)  and  \( h = 1, \dots, M_0 \). Given the index set \( \mathcal{T}^j = (T_1^j, \dots, T_{M_0}^j) \) with \( T_h^j = (k_h, k_h+1, \dots, k_h+|T_h^j|-1) \), and applying an additional reparameterization of \( \{\bs\theta_k\}_{k=1}^{M} \) defined by (\ref{eq:repara-2}), the reparameterized density function of the \( M \)-component model is expressed as:
\begin{align}\label{density-M}
 g_{M}(\bs w;\bs\psi_{M}^j,\bs\tau^j) &= \sum_{h=1}^{M_0} \beta_h^j \left( \sum_{k \in T_h^j} \tau_k^j f(\bs w;\tilde{\bs\theta}_k^j) \right),
\end{align}
where
\begin{equation}\label{eq:repara-2}
\begin{pmatrix}
\tilde{\bs\theta}_{k_h}^j \\
\tilde{\bs\theta}_{k_h+1}^j \\
\vdots \\
\tilde{\bs\theta}_{k_h+|T_h^j|-1}^j
\end{pmatrix}
=
\begin{pmatrix}
\bs\nu_h^j - \sum_{\ell=2}^{|T_h^j|} \tau_{h,\ell}^j \bs\lambda_{h,\ell}^j \\
\bs\nu_h^j + \tau_{h,1}^j \bs\lambda_{h,2}^j \\
\vdots \\
\bs\nu_h^j + \tau_{h,1}^j \bs\lambda_{h,|T_h^j|-1}^j
\end{pmatrix}, \quad h=1,\dots,M_0.
\end{equation}

Collecting the reparameterized parameters, we set \( \bs\psi_{M}^j := (\bs\eta^j, \bs\lambda^j)^\top \), \( \bs\tau^j := (\bs\tau_1^j, \dots, \bs\tau_{M_0}^j)^\top \), \( \bs\eta^j := (\bs\beta^j, \bs\nu^j) \), \( \bs\beta^j := (\beta_1^j, \dots, \beta_{M_0-1}^j)^\top \), \( \bs\nu^j := (\bs\nu_1^j, \dots, \bs\nu_{M_0}^j)^\top \), and \( \bs\lambda^j := (\bs\lambda_1^j, \dots, \bs\lambda_{M_0}^j)^\top \), where \( \bs\nu_h^j := (\bs\nu_{h,2}^j, \dots, \bs\nu_{h,|T_h^j|}^j)^\top \), \( \bs\lambda_h^j := (\bs\lambda_{h,2}^j, \dots, \bs\lambda_{h,|T_h^j|}^j)^\top \), and \( \bs\tau_h^j := (\tau_{h,1}^j, \dots, \tau_{h,|T_h^j|-1}^j) \) for \( h=1, \dots, M_0 \). When the data is generated from the $M_0$-components model density  \( g_{M_0}(\bs w;\bs\vartheta_{M_0}^*) := \sum_{h=1}^{M_0} \alpha_h^* f(\bs w;\bs\theta_h^*) \), we have \( \tilde{\bs\theta}_k^{j*} = \bs\theta_h^* \) for all \( k \in T_h^j \), hence \( \bs\nu_h^{j*} = \bs\theta_h^* \) and \( \bs\lambda_h^{j*} = \bs{0} \) while \( \bs\tau^j \) is not identified. In this scenario, since \( \beta_h^{j*} = \alpha_h^* \) and \( \sum_{k \in T_h^j} \tau_k^j f(\bs w;\tilde{\bs\theta}_k^{j*}) = f(\bs w;\bs\theta_h^*) \) in (\ref{density-M}), the density function \( g_{M}(\bs w;\bs\psi_{M}^{j*},\bs\tau^j) \) coincides with the true \( M_0 \)-component model density \( g_{M_0}(\bs w;\bs\vartheta_{M_0}^*)\).

With this reparameterization, the first-order derivatives satisfy:
\begin{equation}\label{eq:g_nabla}
\nabla_{\bs\nu_h^j} \log g_{M}(\bs{w};\bs{\psi}_{M}^{j*},\bs\tau^j) = \frac{\nabla_{\bs\theta_h} f(\bs{w}; \bs\theta_h^*)}{g_{M_0}(\bs{w}; \bs\vartheta_{M_0}^*)}, \quad
\nabla_{\bs\lambda_h^j} \log g_{M}(\bs{w};\bs{\psi}_{M}^{j*},\bs\tau^j) = \bs{0}, \quad h=1,\dots,M_0.
\end{equation}
Because  $\nabla_{\bs{\lambda}_h^j} \log g_{M}(\bs{w};\bs{\psi}_{M}^{j*},\bs\tau^j)   = \bs{0}$, the Fisher information matrix is singular, and the standard quadratic approximation fails. As analyzed in Section \ref{sec:LRT1}, when  $M=M_0+1$, the unique elements of $\nabla_{\bs\lambda_h^j\bs\lambda_h^j}\log g_{M}(\bs{w};\bs{\psi}^{j*}_{M},\bs\tau^j)$ plays the role of score function in (\ref{eq:s_1}) to identify  $\bs{\lambda}_h^j$.  However, when $M$ is much larger than $M_0$,  we need higher order derivatives beyond the second order derivatives to identify ${\bs\lambda}_h^j$.

Denote the density ratio by
\[
l_{\bs\psi_{M}^j\bs\tau^j,i} :=\frac{ g_{M}(\bs W_i;\bs\psi_{M}^j,\bs\tau^j) }{g_{M_0}(\bs W_i;\bs\vartheta_{M_0}^*)}
\]
so that $\ell_n(\bs\psi_{M}^j,\bs\tau^j)-\ell_n(\bs\psi_{M}^{j*},\bs\tau^j)=\sum_{i=1}^n \log l_{\bs\psi_{M}^j\bs\tau^j,i}$.

Let $\bs\lambda_{h,\ell_1}^j\otimes\bs\lambda_{h,\ell_2}^j\otimes \cdots \otimes\bs\lambda_{h,\ell_p}^j$ denote the tensor containing all interactions among the elements of the vectors $\bs\lambda_{h,\ell_1}^j$, $\bs\lambda_{h,\ell_2}^j$, $\cdots$, $\bs\lambda_{h,\ell_p}^j$,  where \(\otimes\) denotes the Kronecker product.
We use the higher-order derivatives of $l_{\bs\psi_{M}^j\bs\tau^j,i} $ with respect to $\bs\lambda_h^j $ evaluated at $\bs\psi_M=\bs\psi^*_M$ to identify $ \bs\lambda_h^j $: for $h=1,2...,M_0$ and for $p\geq 2$,
 \begin{align*}
\nabla_{\bs\lambda_{h,\ell_1}^j\otimes\bs\lambda_{h,\ell_2}^j\otimes \cdots \otimes\bs\lambda_{h,\ell_p}^j} l_{\bs\psi_{M}^j\bs\tau^j,i} &=  \gamma_{h,\ell_1\ell_2\cdots \ell_p}^j(\bs\tau_h^j)\frac{\alpha_h^* \nabla_{ \bs\theta^{\otimes  p}} f^*_{h,i}}{g^*_{M_0,i}} \quad \text{for $(\ell_1,...,\ell_p)\in\{2,...,|T_h^j|\}^p$},
\end{align*}
where $\bs\theta^{\otimes  p}$ represents $\bs\theta\otimes \bs\theta\cdots \otimes \bs\theta$ (p times), $f^*_{h,i}:=f(W_i;\bs\theta_h^*)$, and $g^*_{M_0,i}:=g_{M_0}(\bs W_i;\bs\vartheta_{M_0}^*)$. For $p=2,3$, we have $\gamma_{h,\ell_1\ell_2}^j(\bs\tau_h^j)=\tau_{h,1}^j\tau_{h,\ell_1}^j(\tau_{h,\ell_2}^j+ \delta_{\ell_1\ell_2})$ and $\gamma_{h,\ell_1\ell_2\ell_3}^j(\bs\tau_h^j)= \tau_{h,1}^j\tau_{h,\ell_1}^j(\delta_{\ell_1\ell_2}\delta_{\ell_2\ell_3}(\tau_{h,1}^j)^2-\tau_{h,\ell_2}^j \tau_{h,\ell_3}^j )$  with $\delta_{ij}$ being Kronecker delta.
The elements of $\nabla_{\bs\lambda_{h,\ell_1}^j\otimes\bs\lambda_{h,\ell_2}^j\otimes \cdots \otimes\bs\lambda_{h,\ell_p}^j} l_{\bs\psi_{M}^j\bs\tau^j,i}$ are ave-zero random variables.

Let $q:=\text{dim}(\bs{\theta})$. Let $\widetilde{\text{vech}}_p({\nabla_{\bs\theta^{\otimes  p}} f^*}/{g^*})$ extract all \textit{unique} elements of $q$-dimensional symmetric array ${\nabla_{\bs\theta^{\otimes  p}} f^*}/{g^*}$, multiply by its frequency, and stacks them into a vector.   For example,  for $p=3$ and $q\geq 3$,  $\widetilde{\text{vech}}_3\Bigl({\nabla_{\bs\theta\otimes  \bs\theta\otimes  \bs\theta} f^*}/{g^*}\Bigr)$ with $\bs\theta=(\theta_1,...,\theta_q)^\top$ will contain (i) \(q\) elements of the form \({\nabla_{\theta_i^3} f^*}/{g^*}\), (ii) $q(q-1)$ elements of the form $3{\nabla_{\theta_i^2\theta_j} f^*}/{g^*}\) for $i\neq j$, and (iii) $\binom{q}{3}$ of the form $6{\nabla_{\theta_i\theta_j\theta_k} f^*}/{g^*}\) for $i\neq j\neq k$.

 Choose $p_h^j$ so that the dimension of the unique elements of $\frac{\nabla_{\bs\theta\otimes  \bs\theta} f_{h,i}^*}{g_{M_0,i}^*}$, $\frac{\nabla_{ \bs\theta\otimes  \bs\theta\otimes \bs\theta}  f_{h,i}^*}{g_{M_0,i}^*}$, ...., $\frac{\nabla_{ \bs\theta^{\otimes  p_M}} f_{h,i}^*}{g_{M_0,i}^*} $ is at least as large as the number of elements in $\bs\lambda_h^j $ but no larger than necessary:
\begin{align}
& \dim\left(\widetilde{\textnormal{vech}}_2\left(\frac{\alpha_h^* \nabla_{\bs\theta\otimes  \bs\theta}  f_{h,i}^*}{g_{M_0,i}^*}\right)\right)+ \cdots + \dim\left(\widetilde{\textnormal{vech}}_{p_h^j}\left(\frac{\alpha_h^* \nabla_{\bs\theta^{\otimes  {p_h^j}}}  f_{h,i}^*}{g_{M_0,i}^*}\right) \right)\ \geq\ \dim(\bs\lambda_h^j )\nonumber\\
&\quad> \
 \dim\left(\widetilde{\textnormal{vech}}_2\left(\frac{\alpha_h^* \nabla_{\bs\theta\otimes  \bs\theta}  f_{h,i}^*}{g_{M_0,i}^*}\right) \right)+ \cdots + \dim\left(\widetilde{\textnormal{vech}}_{p_h^j-1}\left(\frac{\alpha_h^* \nabla_{\bs\theta^{\otimes  (p_h^j-1)}}  f_{h,i}^*}{g_{M_0,i}^*}\right) \right).\label{eq:lambda}
\end{align}

The value of $p_h^j$ indicates a necessary order of local expansions for identifying $\bs\lambda_h^j$.
For example, consider  the case of testing $H_0: M=1$ against $H_1: M=M_1$ for any $M_1>1$. Then, $g_{M_0,i}^*=f^*_i$ and $J_{M_1}=1$. In this case,  for $p=2$, because $ \dim\left(\widetilde{\textnormal{vech}}_2\left(\frac{\nabla_{\bs\theta\otimes  \bs\theta} f_i^*}{f_i^*}\right)\right)=\frac{q(q+1)}{2}$ and $\dim({\bs\lambda})=q(M_1-1)$,
(\ref{eq:lambda}) is satisfied if $\frac{(q+1)}{2}\geq M_1-1$.  For $p=3$ and $q\geq 3$, the condition is
$\frac{q(q+1)}{2} + \left(q + q(q-1) + \frac{q(q-1)(q-2)}{6}\right) \geq q(M_1 - 1)> \frac{q(q+1)}{2}$, which simplifies to
$\frac{(q+5)(q+1)}{6}+1  \geq M_1>\frac{q+3}{2}$.
For $p=4$ and $q\geq 4$,  the condition is given by $\frac{(q+5)(q+1)}{6}+ \frac{q^3 - 5q^2 + 10q - 4}{2}+1\geq M_1>\frac{(q+5)(q+1)}{6}+1$, where ${\bs\lambda}$ is identified for up to $M_1=18$ if $q=4$.

For $h=1,...,M_0$, let
 \[
 \bs v_{{\bs\lambda}_h^j\bs\tau_h^j}:=
 \begin{pmatrix}
 \text{vech}_2\left( \sum_{\ell_1=2}^{|T_h^j|}\sum_{\ell_2=2}^{|T_h^j|} \gamma_{h,\ell_1\ell_2}^j(\bs\tau_h^j) \bs\lambda_{h,\ell_1}^j\otimes \bs\lambda_{h,\ell_2}^j\right) \\
 \vdots\\
 \text{vech}_{p_h^j}\left( \sum_{\ell_1=2}^{|T_h^j|} \cdots \sum_{\ell_{p_M}=2}^{|T_h^j|} \gamma_{h,\ell_1\ell_2\cdots \ell_{p_{M}}}^j(\bs\tau_h^j)\bs\lambda_{h,\ell_1}^j \otimes \bs\lambda_{h,\ell_2}^j  \otimes  \cdots \otimes \bs\lambda_{h,\ell_{p_h^j}}^j\right)
 \end{pmatrix},
 \]
 where  $\text{vech}_p(\cdot)$ is defined similarly to $\widetilde{\text{vech}}_p(\cdot)$ as an operator that extracts all unique elements of $q$-dimensional symmetric array, but without multiplying each element by its frequency, so that the elements of $\text{vech}_p(\bs\lambda_{h,\ell_1} \otimes \cdots \otimes \bs\lambda_{h,\ell_p})$ are conformable to those of $\widetilde{\text{vech}}_p({\nabla_{\bs\theta^{\otimes  p}} f_h^*}/{g_{M_0}^*})$.

 Collect the relevant normalized reparameterized parameters and define $ \bs{t}_{\bs{\psi}_{M}^j\bs\tau^j}$ as
\begin{equation}
 \label{eq:t_bic}
 \bs{t}_{\bs{\psi}_{M}^j\bs\tau^j} :=   \begin{pmatrix}
  \bs{\bs\eta} - \bs{\bs\eta}^*\\
   \bs v_{{\bs\lambda}_1^j\bs\tau_1^j}\\
   \vdots\\
   \bs v_{{\bs\lambda}_{M_0}^j\bs\tau_{M_0}^j}
 \end{pmatrix}.
\end{equation}

For $j=1,...,J_{M}$, define the vector $\bs{s}^{j}(\bs{W})$ as
\begin{align}\label{eq:s_bic}
	 &\bs{s}^{j}(\bs{W}) = \begin{pmatrix}
		\bs{s}_{\bs{\bs\eta}}^{j}(\bs W)   \\
		\bs{s}_{\bs{\lambda}}^{j}(\bs W)
\end{pmatrix}\quad \text{with}  \quad \bs{s}_{\bs{\bs\eta} }^{j}(\bs W) :=
\begin{pmatrix}
\bs{s}_{\bs{\alpha}}^{j}(\bs W) \\
\bs{s}_{\bs{\nu}}^{j}(\bs W)
\end{pmatrix}\quad\text{and}\quad \bs{s}_{\bs{\lambda}}^{j}(\bs W) :=
\begin{pmatrix}
\bs{s}_{\bs{\lambda}_1^j}^{j}(\bs W) \\
\vdots\\
\bs{s}_{\bs{\lambda}_{M_0}^j}^{j}(\bs W)
\end{pmatrix},
\end{align}
where
\begin{align}
\bs{s}_{\bs{\alpha}}^{j}(\bs W) =
\begin{pmatrix}
\frac{f_{1}(\bs W;\bs\theta_1^*)-f_{M_0}(\bs W;\bs\theta_{M_0}^*)}{g_{M_0}(\bs W;\bs\vartheta_{M_0}^*)}\\
\vdots\\
\frac{f_{M_0-1}(\bs W;\bs\theta_{M_0-1}^*)-f_{M_0}(\bs W;\bs\theta_{M_0}^*)}{g_{M_0}(\bs W;\bs\vartheta_{M_0}^*)}
\end{pmatrix},\quad
\bs{s}_{\bs{\nu}}^{j}(\bs W) =
\begin{pmatrix}
\frac{\alpha_1^* \nabla_{  \bs\theta} f(\bs W;\bs\theta_1^*)}{g_{M_0}(\bs W;\bs\vartheta_{M_0}^*)}\\
\vdots\\
\frac{\alpha_{M_0}^* \nabla_{ \bs\theta} f(\bs W;\bs\theta_{M_0}^*)}{g_{M_0}(\bs W;\bs\vartheta_{M_0}^*)}\\
\end{pmatrix},\quad\text{and}
\end{align}
\begin{align}
\bs{s}_{\bs{\lambda}_h^j}^{j}(\bs W)=
\begin{pmatrix}
\widetilde{\textnormal{vech}}_2\left(\frac{\alpha_h^* \nabla_{\bs\theta\otimes  \bs\theta}  f(\bs W;\bs\theta_h^*)}{g_{M_0}(\bs W;\bs\vartheta_{M_0}^*)}\right)\\
\vdots\\
\widetilde{\textnormal{vech}}_{p_h^j}\left(\frac{\alpha_h^* \nabla_{\bs\theta^{\otimes  {p_h^j}}}   f(\bs W;\bs\theta_h^*)}{g_{M_0}(\bs W;\bs\vartheta_{M_0}^*)}\right)
\end{pmatrix}\quad\text{for $h=1,...,M_0$}.
\end{align}

We  denote $\bs s_i^j:=\bs s^j(\bs W_i)$. We assume that $l_{\bs\psi_{M}^j\bs\tau^j,i}$ can be expanded around $l_{\bs\psi_{M}^j\bs\tau^j,i}=1$ as follows. Let $ P_n({\bs s}^j_i({\bs s}^j_i)\t):=n^{-1}\sum_{i=1}^n {\bs s}^j_i({\bs s}^j_i)\t$ and $\nu_n( {\bs s}^j_i):= n^{-1/2}\sum_{i=1}^n [ {\bs s}^j_i- \mathbb{E}_{\bs\theta^*}[ {\bs s}^j_i]]$.

For \( \epsilon > 0 \), define the neighborhood of the set \( \bs\psi_{M}^j \) corresponding to the null hypothesis \( H_0: M = M_0 \) as follows:
\[
\mathcal{N}_\epsilon = \{\bs\psi_{M}^j \in \Theta_{\bs\psi_{M}^j} : \sup_{\bs\tau^j \in \Theta_{\bs\tau^j,c_1}} \|\bs{t}_{\bs{\psi}_{M}^j \bs\tau^j}\| < \epsilon\},
\]
where \( \Theta_{\bs\tau^j,c_1} \) denotes the parameter space of \( \bs\tau^j \), as implied by the parameter space \( \Theta_{\bs\alpha,c_1} \) of \( \bs\alpha \), which restricts the values of \( \tau_h^j \) to be bounded away from 0 and 1.

\begin{assumption}\label{assumption:BIC} For $M=M_0+1,...,\bar M$, the following assumption holds when the data is generated from the $M_0$-component model.
For all $j=1,...,J_{M}$ and
$i=1,...,n$, $l_{\bs\psi_{M}^j\bs\tau^j,i}-1$ admits an expansion
\begin{equation}\label{density-ratio}
l_{\bs\psi_{M}^j\bs\tau^j,i} -1 = \left(  \bs{t}_{\bs{\psi}_{M}^j\bs\tau^j}\right)\t \bs s_i^j + r_{\bs\psi_{M}^j\bs\tau^j,i},
\end{equation}
where  $\bs{t}_{\bs{\psi}_{M}^j\bs\tau^j}$ satisfies $\bs\psi_{M}^j\rightarrow \bs\psi_{M}^{j*}$ if $\sup_{\bs\tau^j\in \Theta_{\bs\tau^j,c_1}}|| \bs{t}_{\bs{\psi}_{M}^j\bs\tau^j}||\rightarrow 0$, and $(\bs s_i^j, r_{\bs\psi_{M}^j\bs\tau^j,i})$ satisfy, for some $C\in (0,\infty)$, $\delta>0$, $\epsilon>0$, and $\rho\in(0,1)$, (a) $\mathbb{E}_{\bs\vartheta_{M_0}^*}||\bs s_i^j||^{2+\delta}<C$, (b) $||P_n(\bs{s}_i^j(\bs{s}_i^j)\t) - \bs{\mathcal{I}}^j|| = o_p(1)$, where $\bs{\mathcal{I}}^j:= \mathbb{E}[\bs{s}^j_i (\bs{s}^j_i) ^\top]$ is finite and non-singular, (c) $\mathbb{E}_{\bs\theta^*}[\sup_{{\bs\psi_{M}^j} \in \mathcal{N}_\varepsilon} |r_{\bs\psi_{M}^j\bs\tau^j,i}/(||\bs{t}_{\bs{\psi}_{M}^j\bs\tau^j}||||\bs\psi_{M}^j-\bs\psi_{M}^{j*}||)|^2] < \infty$, (d) $\sup_{{\bs\psi_{M}^j} \in \mathcal{N}_\varepsilon} [ \nu_n(r_{\bs\psi_{M}^j\bs\tau^j,i})/(||{\bs{t}_{\bs{\psi}_{M}^j\bs\tau^j}}||||\bs\psi_{M}^j-\bs\psi_{M}^{j*}||)] = O_p(1)$,  and  (e) $\sup_{{\bs\psi_{M}^j} \in \mathcal{N}_{\varepsilon}} ||\nu_n(\bs{s}_i^j)|| =O_p(1)$.

\end{assumption}

Assumption \ref{assumption:BIC}(b) is a key regularity condition that requires   the unique elements of ${\nabla_{\bs\theta} f_h^*}/{g_{M_0}^*}$, ${\nabla_{\bs\theta\otimes  \bs\theta} f_h^*}/{g_{M_0}^*}$, ...., ${\nabla_{ \bs\theta^{\otimes  p_h}} f_h^*}/{g_{M_0}^*} $ are linearly independent and their expectation is finite.

\begin{proposition} \label{Ln_thm1}
Suppose that Assumption \ref{assumption:BIC} holds. Then,  for any $c>0$,
\begin{equation*}
\sup_{\bs\psi_{M}^j \in  \mathcal{N}_{c/\sqrt{n}}}\left|\ell_n(\bs\psi_{M}^j,\bs\tau^j)-\ell_n(\bs\psi_{M}^{j*},\bs\tau^j) - \sqrt{n} \bs t^{\top}_{\bs\psi_{M}^j\bs\tau^j} \nu_n (\bs s_i^j) + n \bs t^{\top}_{\bs\psi_{M}^j\bs\tau^j} \bs{\mathcal{I}}^j \bs t_{\bs\psi_{M}^j\bs\tau^j}/2 \right| = o_{p }(1).
\end{equation*}
\end{proposition}

The following proposition expands $\ell_n(\bs\psi_{M}^j,\bs\tau^j)$ in $A_{\varepsilon_n}(\eta) := \{\bs\psi_{M} \in \mathcal{N}_{\varepsilon_n}: \ell_n(\bs\psi_{M}^j,\bs\tau^j)-\ell_n(\bs\psi_{M}^{j*},\bs\tau^j) \geq -\eta \}$ for some $\eta \in [0,\infty)$ and for $\varepsilon_n\rightarrow 0$ slowly to ensure $\Pr( \hat{\bs\psi}_M \in \mathcal{N}_{\varepsilon_n})=1$. 
\begin{proposition} \label{Ln_thm2}
Suppose that Assumption \ref{assumption:BIC} holds.
Then, for all $j=1,...,J_{M}$, for any $\eta>0$, and for any $\{\varepsilon_n: n=1,2,...\}$ such that $\varepsilon_n\rightarrow 0$ but $\Pr( \hat{\bs\psi}_M \in \mathcal{N}_{\varepsilon_n})=1$, \\(a) $ \sup_{\bs\psi_M \in A_{\varepsilon_n}(\eta)} \left|\bs{t}_{\bs{\psi}_{M}^j\bs\tau^j}\right| = O_{p   }(n^{-1/2})$,  (b) for any $c>0$,
\begin{equation*}
\sup_{\bs\psi_M \in A_{\varepsilon_n}(\eta) \cup \mathcal{N}_{c/\sqrt{n}} }\left|\ell_n(\bs\psi_{M}^j,\bs\tau^j)-\ell_n(\bs\psi_{M}^{j*},\bs\tau^j) - \sqrt{n} \bs{t}_{\bs{\psi}_{M}^j\bs\tau^j} \nu_n (\bs s_i^j) + n \bs{t}_{\bs{\psi}_{M}^j\bs\tau^j} \bs{\mathcal{I}}^j \bs{t}_{\bs{\psi}_{M}^j\bs\tau^j}/2 \right| = o_{p }(1),
\end{equation*}
and (c) $\sup_{\bs\psi_M \in A_{\varepsilon_n}(\eta)} \left|\ell_n(\bs\psi_{M}^j,\bs\tau^j)-\ell_n(\bs\psi_{M}^{j*},\bs\tau^j) \right| =O_p(1)$.
\end{proposition}

Because a consistent MLE is in $ A_{\varepsilon_n}(\eta)$ by definition, Proposition \ref{Ln_thm2} implies that $\ell_n(\hat{\bs\psi}_{M}^j,\bs\tau^j)-\ell_n(\bs\psi_{M}^{j*},\bs\tau^j)  =O_p(1)$ for a consistent MLE.

\subsection{Bootstrap procedure for the ave- and max-rk statistic}\label{sec:bootstrap}

We consider the following Bayesian bootstrap by drawing $\{\omega^{(b)}_i\}_{i=1}^n$ for $b=1,...,B$, where  $\omega_{i}^{(b)} = \tilde\omega_i^{(b)}/\sum_{j=1}^n \tilde\omega_j^{(b)}$ with
$\tilde\omega_i^{(b)}\overset{iid}{\sim} Exp(1)$ for $i=1,...,n$, and compute the bootstrap version of $\widehat{\bs P}_k$ as
\[
\widehat{\bs P}_k^{(b)} =  \left[
\begin{array}{ccc}
 \sum_{i=1}^n \omega_i^{(b)}  \mathbb{I}(Y_k \in \delta_{1},\bs Y_{\mathcal{T}_{-k}} \in \bs\delta_{1}^{\mathcal{T}_{-k}}) & \cdots &  \sum_{i=1}^n \omega_i^{(b)}  \mathbb{I}(Y_{k} \in \delta_{1},\bs Y_{\mathcal{T}_{-k}} \in \bs \delta_{|\Delta_{\mathcal{T}_{-k}}|}^{\mathcal{T}_{-k}})  \\
\vdots  & \ddots  & \vdots  \\
 \sum_{i=1}^n \omega_i^{(b)} \mathbb{I}( Y_{k} \in \delta_{|\Delta_{t}|},\bs Y_{\mathcal{T}_{-k}} \in \bs \delta_{1}^{\mathcal{T}_{-k}}) & \cdots & \sum_{i=1}^n \omega_i^{(b)} \mathbb{I}(  Y_{k} \in \delta_{|\Delta_{t}|},\bs Y_{-k} \in \bs\delta_{|\Delta_{\mathcal{T}_{-k}}|}^{\mathcal{T}_{-k}})\end{array}
\right] 
\]
for $k=1,...,T$. Note that we use the same weights across different $k$s to accomodate the dependencies across $k$s so that the bootstrapped correlation between $\widehat{\bs P}_k^{(b)}$ and $\widehat{\bs P}_\ell^{(b)}$ for $k\neq \ell$ correctly captures the corresponding sample correlation between $\widehat{\bs P}_k$ and $\widehat{\bs P}_\ell$.

With this bootstrapped $\{\widehat{\bs P}_k^{(b)} \}_{k=1}^T$ for $b=1,...,B$, we construct the bootstrapped rk-statistics as
\begin{equation}\label{bootstrap-1}
\text{rk}^{k(b)}(r) = n (\widehat{\bs{\lambda}}_r^{k(b)} -\widehat{\bs{\lambda}}_r^{k} )^\top (\widehat{\boldsymbol{\Omega}}^{k(b)}_r) ^{-1}(\widehat{\bs{\lambda}}_r^{k(b)}-\widehat{\bs{\lambda}}_r^{k})\quad\text{for $k=1,...,T$},
\end{equation}
where $\widehat{\bs{\lambda}}_r^{k(b)}$ and $\widehat{\boldsymbol{\Omega}}^{k(b)}_r$ are computed as described in Proposition \ref{rk} but we compute $\widehat{\mathbf{A}}_{r }^k$ and $\widehat{\bs\Sigma}_k$ using $\widehat{\bs P}_k^{(b)}$ in place of $\widehat{\bs P}_k$.

Based on $\{\text{rk}^{k(b)}(r)  \}_{k=1}^T$ derived as above, we can construct bootstrapped ave- and max-rk-statistics as:
\begin{align*}
\text{ave-rk}^{(b)}(r) = \frac{1}{T} \sum_{k=1}^T \mathrm{rk}^{k(b)}(r)\ \text{ and }\ \text{max-rk}^{(b)}(r) = \max\{ \mathrm{rk}^{1(b)}(r),...,\mathrm{rk}^{T(b)}(r)\} \ \text{ for $b=1,2,...,B$.} 
\end{align*}

We then compute the bootstrap p-value as the empirical proportion of the bootstrapped test statistics \(\text{max-rk}^{(b)}(r)\) that equal or exceed the observed statistic \(\text{max-rk}(r)\): $\text{bootstrap p-value} = \frac{1}{B} \sum_{b=1}^{B} \mathbb{I}\left(\text{max-rk}^{(b)}(r) \geq \text{max-rk}(r)\right)$.
We reject the null hypothesis that \(\text{rank}(\bs P_k)\leq r\) for all \(k=1,\dots,T\) at significance level \(\alpha\) if the bootstrap p-value is strictly less than \(\alpha\).

\subsection{Asymptotic distribution under local alternatives}\label{sec:local}
We derive the asymptotic distribution of the LRTS under local alternatives. For brevity, we focus on testing $H_0: M=1$ against $H_A: M=2$. Consider the following local alternative to the homogeneous model $f(\bs w;\bs\gamma^*,\bs\theta^*)$ with $\bs\theta^*=(\mu^*,\sigma^{*2},(\bs{\beta}^*)\t)\t$. For brevity, we omit the common parameter $\bs \gamma$ in this section.
In a reparameterized parameter, $\bs{\Psi}^* = ((\bs\nu^*)\t, (\bs{\lambda}^*)\t)\t$.
For $\alpha^*\in (0,1)$  and a local parameter $\bs h=(\bs {h}_{\bs\nu}\t,\bs {h}_{\bs\lambda}\t)\t$ with $\bs h_{\bs \lambda}\in v(\bs\Theta_{\bs\lambda})$,  we consider a sequence of contiguous local alternatives
$(\alpha_n,\bs{\psi}_n\t)\t = (\alpha_n,\bs\nu_n\t,\bs\lambda_n\t)\in  \bs{\Theta}_\alpha\times \bs{\Theta}_{\bs\nu}\times\bs{\Theta}_{\bs\lambda}$ such that, with $\bs{t}_{\bs\lambda}(\bs\lambda,\alpha)$ given by (\ref{eq:t_1}),
\begin{equation}\label{eq:h}
\bs {h_\nu}= \sqrt{n}(\bs\nu_n -\bs\nu^*),\quad \bs {h_\lambda} =\sqrt{n} \bs{t}_{\bs\lambda}(\bs\lambda_n,\alpha_n),\quad \text{and} \quad \alpha_n = \alpha^* + o(1).
\end{equation}
Equivalently, the non-reparameterized contiguous local alternatives are given by
\begin{equation}\label{eq:local-alt}
\bs\theta_{1,n} =\bs \nu_n + (1-\alpha_n) \bs\lambda_n\quad\text{and}\quad \bs\theta_{2,n} = \bs\nu_n - \alpha_n \bs\lambda_n
\end{equation}
for  $\bs\nu_n = \bs \nu^* + n^{-1/2} \bs{h}_{\nu}$ and
$\bs\lambda_n= (\lambda_{1,n},\lambda_{2,n},....,\lambda_{q,n})\t$ with
\begin{align*}
\lambda_{j,n} = n^{-1/4} (\alpha_n(1-\alpha_n))^{-1/2} h_{\lambda,j}\quad \text{for $j=1,...,q$},
\end{align*}
where $\bs{h}_{\bs\lambda}=(h_{\lambda,1}^2,2h_{\lambda,1}h_{\lambda,2},h_{\lambda,1}^2,...,2h_{\lambda,q-1}h_{\lambda,q},h_{\lambda,q}^2)\t$. The local alternatives are of order $n^{1/4}$ rather than $n^{1/2}$. See
  the discussion following Proposition \ref{prop:t2_distribution}.

The following proposition provides the asymptotic distribution of the LRT  test statistics under contiguous local alternatives.
\begin{proposition}\label{prop:local-power}
Suppose that the assumptions in Proposition \ref{prop:tm0_distribution} hold for $M_0=1$. Consider a sequence of contiguous local alternatives $\bs\vartheta_{2,n} = (\alpha_n,\bs\theta_{1,n}^\top,\bs\theta_{2,n}^\top)\t$ given in (\ref{eq:local-alt}), where  $\alpha_n$ and $\bs\lambda_n$ satisfy (\ref{eq:h}). Then, under $H_{1,n}: \bs\vartheta=  \bs\vartheta_{2,n}$, we have $LR_n\overset{d}{\rightarrow} (\tilde{\bs{t}}_{\bs\lambda} )\t \bs{\mathcal{I}}_{\bs\lambda,\bs\eta}  \tilde{\bs{t}}_{\bs\lambda}$,
where $\tilde{\bs t}_{\bs\lambda}$ has the same distribution as $\widehat {\bs t}_{\bs\lambda}$ in Proposition \ref{prop:t2_distribution} but $\bs{G}_{\bs{\lambda,\bs{\eta}}}$ is replaced with $ ( \bs{\mathcal{I}}_{\bs\lambda,\bs\eta}  )^{-1} \bs{S}_{\bs\lambda,\bs\eta} + \bs h_{\bs\lambda}$.
\end{proposition}
Importantly, a set of contiguous local alternatives considered in (\ref{eq:local-alt}) excludes a sequence such that $\alpha_n\rightarrow 0$ or $1$.

\subsection{Proof of Propositions}

\begin{proof}[Proof of Proposition \ref{prop:unbounded_likelihood}]

	We first consider a model with an intercept parameter and a variance parameter but without covariates with  $\mathbf W_i = \{ y_{it} \}_{t=1}^T$ for the model (\ref{eq:fm}) with normal density (\ref{eq:f1}).

		Define
	\[
	s_i^2 = \frac{1}{T-1}\sum_{t=1}^T (Y_{it}-\bar Y_i)^2 \quad\text{with}\quad \bar Y_i = \frac{1}{T} \sum_{t=1}^T Y_{it},
	\]
	where $(T-1)s_i^2/\sigma^{*2}$ follows a chi-square distribution with $T-1$ degrees of freedom.
Let $i^* = \arg\min_{i=1,\ldots,n} \{s_i^2\}$ so that  $s_{i^*}^2  = \min\{s_1^2,\ldots,s_n^2\}$ is the minimum of $s_i^2$ across all values of $i$.
We consider a sequence of parameters $\bs\vartheta_{2,n}=(\alpha_n,\bs\theta_{1,n}\t,\bs{\theta}_{2,n}\t)\t$ with $\alpha_n = 1/n$,  $\bs\theta_{1,n}=(\mu_{1,n},\sigma_{1,n}^2)\t=(\bar Y_{i^*},s_{i^*}^2)\t$, and  $\bs\theta_{2,n}=\bs\theta^*=(\mu^*,\sigma^*)\t$ for all $n$. 
It suffices to show that $LR_n^*(\bs\vartheta_{2,n})$ is unbounded in probability.

	Define \[
	\ell(\bs{W}_i;\bs\theta) : = \log f(\bs{W}_i;\bs\theta) = -\frac{T}{2}\log \sigma^2 - \frac{T}{2}\log(2\pi) - \frac{1}{2}\sum_{t=1}^T \left(\frac{Y_{it}-\mu}{\sigma}\right)^2.
	\]
	Then, the LRT statistic for a two-component mixture is written as
		\begin{align}
		LR_n^*(\bs\vartheta_{2,n}) & = 2 \left\{\sum_{i=1}^n \log\left(\alpha_n \prod_{t=1}^T \frac{1}{\sigma_{1,n}}\phi\left(   \frac{Y_{it}-\mu_{1,n}}{\sigma_{1,n}}\right)
	+(1-\alpha_n) \prod_{t=1}^T \frac{1}{\sigma^*}\phi\left(  \frac{Y_{it}-\mu^*}{\sigma^*}\right)\right)-  \sum_{i=1}^n\ell(\bs{W}_i;\bs\theta^*) \right\}\nonumber\\
	& =
	2\sum_{i\neq i^*} \left\{\log\left(   \exp(\log\alpha_n+\ell(\bs{W}_i;{\bs\theta}_{1,n})) + \exp(\log (1-\alpha_n)+\ell(\bs{W}_i;{\bs\theta}^*))\right)-\ell(\bs{W}_i;\bs\theta^*)\right\} \nonumber \\
	&\quad +2 \left\{\log\left(   \exp(\log\alpha_n+\ell(\bs{W}_{i^*};{\bs\theta}_{1,n})) + \exp(\log (1-\alpha_n)+\ell(\bs{W}_{i^*};{\bs\theta}^*))\right)-\ell(\bs{W}_{i^*};\bs\theta^*)\right\} \label{lr-1}.
	\end{align}

	The first term on the right-hand side of (\ref{lr-1}) can be rewritten as
	\begin{align*}
	 = 2(n-1) \log \left(\frac{n-1}{n}\right)
	+ 2	\sum_{i\neq i^*}\log\left(1+\frac{1}{n-1} \exp(\ell(\bs{W}_i;{\bs\theta}_{1,n})-\ell(\bs{W}_i;{\bs\theta}^*))\right),
	\end{align*}
	which is bounded from below by $-2$ as $n\rightarrow \infty$ because  $\lim_{n\rightarrow \infty} 2(n-1) \log \left(\frac{n-1}{n}\right)  =-2$ and $\log\left(1+\frac{1}{n-1} \exp(\ell(\bs{W}_i;{\bs\theta}_{1,n})-\ell(\bs{W}_i;{\bs\theta}^*))\right)\geq 0$ for all $n$.

The second term  on the right-hand side of (\ref{lr-1}) is written as
\begin{align}
2\{ - \log n + \ell(\bs{W}_{i^*};{\bs\theta}_{1,n})\} + 2 \log\left(1 + (n-1)  \exp(\ell(\bs{W}_{i^*};{\bs\theta}^*)-\ell(\bs{W}_{i^*};{\bs\theta}_{1,n}))\right) - 2 \ell(\bs{W}_{i^*};\bs\theta^*),\label{lr-2}
\end{align}
where  $2\{ - \log n + \ell(\bs{W}_{i^*};{\bs\theta}_{1,n})\}$ diverges to infinity as $n\rightarrow
\infty$ by   Lemma \ref{lemma:unbounded_likelihood}, the second term in (\ref{lr-2}) is bounded below from zero, and the third term is bounded in probability because $\ell(\bs{W}_{i^*};\bs\theta^*)=O_p(1)$. Therefore,    for any $M<\infty$, we have
	$\Pr \Big(LR_n^*(\bs\vartheta_{2,n}) \le M \Big) \to 0$ as $n \to \infty$. 

For a model with covariates for the model (\ref{eq:fm}) with normal density (\ref{eq:f1}), we can consider a sequence of parameters $\bs\vartheta_{2,n}=(\alpha_n,\bs\theta_{1,n}\t,\bs{\theta}_{2,n}\t)\t$ with $\alpha_n = 1/n$,  $\bs\theta_{1,n}=(\mu_{1,n},\sigma_{1,n}^2,\bs\beta_{1,n}\t)\t=(\bar Y_{i^*}, s_{i^*}^2,\bs 0\t)\t$ with    $\bs\theta_{2,n}=\bs\theta^*=(\mu^*,\sigma^*,(\bs\beta^*)\t)\t$. Then, repeating the above argument, the stated result follows.

For the model (\ref{eq:fm}) with normal mixture density (\ref{eq:f1-mixture}), suppose that $\mu_{j1}<\mu_{j2}$ for $j=1,2$. Then, in veiw of (\ref{bound-mixture-density}),  we can bound the log-likelihood function from below as
\begin{align*}
&\sum_{i=1}^n \log\left(\alpha \prod_{t=1}^T \sum_{k=1}^{{\cal{K}}} \tau_{1k} \frac{1}{\sigma_{1}}\phi\left(   \frac{Y_{it}-\mu_{1k}}{\sigma_{1}}\right)
	+(1-\alpha) \prod_{t=1}^T \sum_{k=1}^{{\cal{K}}} \tau_{2k}\frac{1}{\sigma_2}\phi\left(  \frac{Y_{it}-\mu_{2k}}{\sigma_2}\right)\right) \\
	\geq &\sum_{i=1}^n \log\left(\alpha \prod_{t=1}^T  \frac{1}{\sigma_{1}}\phi\left(   \frac{Y_{it}-\mu_{11}}{\sigma_{1}}\right)
	+(1-\alpha) \prod_{t=1}^T  \frac{1}{\sigma_2}\phi\left(  \frac{Y_{it}-\mu_{21}}{\sigma_2}\right)\right).
\end{align*}
Therefore, the LRT statistic for a two-component mixture, $LR_n^*(\bs\vartheta_{2,n})$, is bounded below by the right hand side of (\ref{lr-1}).
Then, repeating the argument following (\ref{lr-1}),  the stated result follows from  Lemma \ref{lemma:unbounded_likelihood}. This proves part (i).

For part (ii), consider a model without covariates for the dynamic panel model (\ref{eq:fm-dynamic}) with (\ref{eq:f1-dynamic}). Let $Y_{11}$ denote the first observation at time $t=1$.
Given the realized data of a fixed sample size $n$, consider a sequence of parameters $\{(\alpha_\ell,\bs\theta_{1,\ell}\t,\bs{\theta}_{2,\ell}\t)\t\}_{\ell=1}^\infty$ with $\alpha_\ell = 1/2$, $\bs\theta_{1,\ell}=(\mu_{11,\ell},\mu^*,\rho^*,\sigma_{1,\ell}^{2},\sigma^{*2})\t$, and  $\bs\theta_{2,\ell}=\bs\theta^*=(\mu_{1}^*,\mu^*,\rho^*,\sigma_1^{*2},\sigma^{*2})\t$, where $\{\mu_{11,\ell},\sigma_{1,\ell}^2\}_{\ell=1}^\infty$ satisfies $(\mu_{11,\ell},\sigma_{1,\ell}^2)\rightarrow (Y_{11},0)$ as $\ell\rightarrow\infty$ such that
\[
\frac{
\mu_{11,\ell} - Y_{11}}{\sigma_{1,\ell}} \rightarrow C \quad\text{as $\ell\rightarrow\infty$,}
\]
where $C$ is some constant. Then, for any positive and large constant $B>0$, we have $LR_n^*(\bs\vartheta_{2,\ell})>B$ by taking a sufficiently large value of $\ell$ because
\[
\lim_{\ell\rightarrow\infty} \frac{1}{\sigma_{1,\ell}} \phi\left(\frac{
\mu_{11,\ell} - Y_{11}}{\sigma_{1,\ell}} \right)  = \phi\left(C\right) \times \lim_{\sigma_{1,\ell}\rightarrow 0}  \frac{1}{\sigma_{1,\ell}}= \infty.
\]
The analogous argument applies to the model (\ref{eq:fm-dynamic}) with normal mixture density (\ref{eq:f1-dynamic-mixture}), and the stated result of part (ii) follows.

\end{proof}

\begin{proof}[Proof of Proposition \ref{prop:linear_independence}]
Note that, for $Z=\tfrac{Y-\mu}{\sigma}$,
\[\frac{\partial}{(\partial\mu)}
\Bigl[\tfrac{1}{\sigma}\,\phi(Z)\Bigr]
\;=\;
\tfrac{Z}{\sigma}\;\frac{1}{\sigma}\,\phi(Z),
\quad
\frac{\partial^{2}}{(\partial\mu)^{2}}
\Bigl[\tfrac{1}{\sigma}\,\phi(Z)\Bigr]
\;=\;
\tfrac{Z^{2}-1}{\sigma^{2}}\;\frac{1}{\sigma}\,\phi(Z), 
\]
\[\text{and}\quad
\frac{\partial}{\partial\sigma^{2}}
\Bigl[\tfrac{1}{\sigma}\,\phi(Z)\Bigr]
\;=\;
\tfrac{1}{2\,\sigma^{2}}\,(Z^{2}-1)\;\frac{1}{\sigma}\,\phi(Z).
\]

Hence, for either \(j=1,2\), a straightforward calculation gives
\begin{equation} \label{eq:prop-2} 
\left.\frac{\partial^{2}g_{2}}{(\partial\mu_{j})^{2}}\right|_{(\bar\alpha.\theta^{*},\theta^*)}
\;=\bar\alpha\;
f^{*}\,\frac{1}{\sigma^{*2}}\,
\Bigl[\,(Z_{1}+\cdots+Z_{T})^{2}-T\Bigr]\quad\text{and}\quad
\left.\frac{\partial\,g_{2}}{\partial\sigma_{j}^{2}}\right|_{(\bar\alpha.\theta^{*},\theta^*)}
\;=\bar\alpha\;
f^{*}\,\frac{1}{2\,\sigma^{*2}}
\sum_{t=1}^{T}\bigl(Z_{t}^{2}-1\bigr),
\end{equation}
where 
\[
Z_{t}:=\frac{Y_{t}-\mu^{*}}{\sigma^{*}},
\qquad
t=1,\dots,T.
\]

Suppose, for contradiction, that real constants \(a,b\) exist such that
\[
\left.\frac{\partial^{2}g_{2}}{(\partial\mu_{j})^{2}}\right|_{(\bar\alpha.\theta^{*},\theta^*)}
\;=\;
a
\;+\;
b\,
\left.\frac{\partial\,g_{2}}{\partial\sigma_{j}^{2}}\right|_{(\bar\alpha.\theta^{*},\theta^*)}
\]
with positive probability.  Since \(\bar\alpha  f^{*}/\sigma^{*2}>0\) almost surely, substituting (\ref{eq:prop-2}) into the above equation, and dividing both sides by \(\bar\alpha f^{*}/\sigma^{*2}\) gives
\[
(Z_{1}+\cdots+Z_{T})^{2}\;-\;T
\;=\;
\frac{a\,\sigma^{*2}}{\bar\alpha f^{*}}
\;+\;
\frac{b}{2}\sum_{t=1}^{T}\bigl(Z_{t}^{2}-1\bigr),
\]
or equivalently
\[
(Z_{1}+\cdots+Z_{T})^{2}
\;=\;
\frac{b}{2}\sum_{t=1}^{T}Z_{t}^{2}
\;+\;\Bigl(\frac{a\,\sigma^{*2}}{\bar\alpha f^{*}}-\tfrac{b\,T}{2}\Bigr)
\;+\;T.
\]
But
\[
(Z_{1}+\cdots+Z_{T})^{2}
\;=\;
\sum_{t=1}^{T}Z_{t}^{2}
\;+\;2\sum_{1\le i<j\le T}Z_{i}Z_{j},
\]
so this would force
\[
\Bigl(1-\tfrac{b}{2}\Bigr)\sum_{t=1}^{T}Z_{t}^{2}
\;+\;
2\sum_{1\le i<j\le T}Z_{i}Z_{j}
\;=\;
C,
\quad
C \;:=\; \frac{a\,\sigma^{*2}}{\bar\alpha f^{*}} \;-\;\tfrac{b\,T}{2} \;+\; T,
\]
as a polynomial identity in \((Z_{1},\dots,Z_{T})\).  Since the coefficient of each cross‐term \(Z_{i}Z_{j}\) on the left is \(2\), but on the right must be \(0\), no choice of \(b\) can satisfy this identity except on a set of probability zero.  Hence no such constants \(a,b\) exist, and
\[
\Pr\Bigl[\, \left.
\tfrac{\partial^{2}g_{2}}{(\partial\mu_{j})^{2}}\right|_{(\bar\alpha.\theta^{*},\theta^*)}
\;=\;
a
\;+\; \left.
b\,\tfrac{\partial\,g_{2}}{\partial\sigma_{j}^{2}}\right|_{(\bar\alpha.\theta^{*},\theta^*)}
\Bigr]
\;=\;0.
\] 
\end{proof}

\begin{proof}[Proof of Proposition \ref{prop:consistency}]

Our proof closely follows the proof of Theorem 3.3 in \citet{hathaway85as} by verifying Assumptions 1, 2, 3, and 5 of \cite{kieferwolfowitz56ams}.

We first consider a model with the component-specific density function (\ref{eq:fm}) with (\ref{eq:f1}).
Because our model has additional free parameters $\bs\beta_j$s, as in the proof of \cite{kasaharashimotsu15jasa},
we consider the joint density of $m_{q}:=M(q+1)$ observations instead of $M+1$ observations in \citet[][p.\ 798]{hathaway85as}, where $q:=\text{dim}(\bs\beta)$. The joint density function of $m_{q}$ observations is itself a mixture of $M^{m_{q}}$ components, where each component is given by $\prod_{j = 1}^{m_{q}}\prod_{t=1}^T f(Y_{jt}; \mu_{i_j} + \bs{X}_{jt}^{\top}\bs{\beta}_{i_j} ,\sigma_{i_j})$ for some choices $i_j\in \{1,\ldots,M\}$, with the density of $N(\mu,\sigma^2)$ denoted by $f(y; \mu,\sigma): = \frac{1}{\sigma}\phi((y-\mu)/\sigma)$.


Assumptions 1, 2, and 3 of \citet{kieferwolfowitz56ams} are easily verified for the joint density of $m_{q}$ observations. We verify Assumption 5 of \citet{kieferwolfowitz56ams} for the joint density function of $m_{q}$ observations by showing that
\begin{equation}
E\left[\log \prod_{j = 1}^{m_{q}}\prod_{t=1}^T f(Y_{jt}; \mu_{i_j}^* + \bs{X}_{jt}^{\top}\bs{\beta}_{i_j}^*,\sigma^*_{i_j})\right] > -\infty \label{H3.1}
\end{equation}
for $\bs{\vartheta}^*_M\in\Theta_M^*$ and that
\begin{equation}
E \sup_{\bs{\vartheta}_M \in \Theta_{\bs{\vartheta}_M}(\bs c)} \left[\log \prod_{j=1}^{m_{q}}\prod_{t=1}^T f(Y_{jt}; \mu_{i_j} + \bs{X}_{jt}^{\top}\bs{\beta}_{i_j},\sigma_{i_j})\right]
<\infty \label{H3.2}
\end{equation}
for all component choices $i_j\in \{1,\ldots,M\}$,
which correspond to equations (3.1) and (3.4) in \citet{hathaway85as}, respectively. (\ref{H3.1}) follows from the argument in the proof of Theorem 3.3 of \citet{hathaway85as}.

For (\ref{H3.2}),  because $\bs{\vartheta}_M\in {\bar \Theta}_{\bs\vartheta_M}(\bs c)$, there exists $c\in(0,1]$ such that $\min_{j,k}\sigma_j/\sigma_k>c$.
Proceeding as in \citet[][pp.\ 798--799]{hathaway85as}, we can show that
\begin{align*}
\sup_{\bs{\vartheta}_M \in \Theta_{\bs{\vartheta}_M}(\bs c)}\log \left[ \prod_{j=1}^{m_{q}}\prod_{t=1}^T f(Y_{jt}; \mu_{i_j} + \bs{X}_{jt}^{\top}\bs{\beta}_{i_j},\sigma_{i_j})\right]
\end{align*}
is no larger than, for some $\ell\in\{1,...,M\}$ and $j_1,j_2,...,j_{q+1}\in \{1,...,m_k\}$,
\begin{equation}\label{H3.5}
\sup_{\mu_\ell,\bs\beta_\ell,\sigma_\ell} \log \left[ \delta(\sigma_\ell)
\prod_{r=1}^{q+1}\prod_{t=1}^T f(Y_{j_rt}; \mu_{\ell} + \bs{X}_{j_rt}^{\top}\bs{\beta}_{\ell},\sigma_{\ell})\right],
\end{equation}
where $\delta(\sigma_\ell)=(2\pi)^{-T(M-1)(q+1)/2}(c\sigma_\ell)^{-T(M-1)(q+1)}$, because $f(Y_{jt}; \mu_{i_j} + \bs{X}_{jt}^{\top}\bs{\beta}_{i_j},\sigma_{i_j})=(2\pi)^{-1/2}(\sigma_{i_j})^{-1}\exp(-\{ Y_{jt}-( \mu_{i_j} + \bs{X}_{jt}^{\top}\bs{\beta}_{i_j})\}^2/2\sigma_{i_j}^2 )\leq (2\pi)^{-1/2}(\sigma_{i_j})^{-1}\leq (2\pi)^{-1/2}(c\sigma_\ell)^{-1}$ for $j \notin \{j_1,j_2,...,j_{q+1}\}$.

Note that $\prod_{q=1}^{q+1}\prod_{t=1}^T f(Y_{j_qt}; \mu_{\ell} + \bs{X}_{j_qt}^{\top}\bs{\beta}_{\ell},\sigma_{\ell})$ is the likelihood function of a linear Gaussian model. Therefore,  the maximized value of (\ref{H3.5}) equals $C_1 - C_2\log(SSR)$,
where $C_1$ and $C_2 $ are a finite constant that depends only on $M$, $k$, and $T$ while $SSR$ is the sum of squared residuals obtained from regressing $\{\{Y_{j_qt}\}_{t=1}^T\}_{q=1}^{q+1}$ on $\{\{1,\bs X_{j_qt}\}_{t=1}^T\}_{q=1}^{q+1}$. Because we have $(T-1)(q+1)$ more observations than the number of parameters, the SSR is distributed as $\sigma_\ell^{*2} \chi^2((T-1)(q+1))$.  Since $E\log(\chi^2((T-1)(q+1))) < \infty$, the expected value of (\ref{H3.5}) is finite, and (\ref{H3.2}) holds. This verifies Assumption 5 of \citet{kieferwolfowitz56ams}, and the stated consistency result follows.

When the component-specific density function is given by (\ref{eq:fm}) with (\ref{eq:f1-mixture}), where the number of components  in (\ref{eq:f1-mixture}) is $K$, the joint density function of $m_{q}$ observations can be written as  a mixture of $M^{m_{q}}\times K^{m_qT}$ components,
where each component is given by $\prod_{j = 1}^{m_{q}}\prod_{t=1}^T f(Y_{jt}; \mu_{i_{j}i_k} + \bs{X}_{jt}^{\top}\bs{\beta}_{i_j} ,\sigma_{i_j})$ for some choices $i_j\in \{1,\ldots,M\}$ and $i_{k}\in \{1,...,K\}$. Then, repeating the same argument as above, we can show that the expected value of
$\sup_{\bs{\vartheta}_M \in \Theta_{\bs{\vartheta}_M}(\bs c)}\log \left[ \prod_{j=1}^{m_{q}}\prod_{t=1}^T f(Y_{jt}; \mu_{i_ji_k} + \bs{X}_{jt}^{\top}\bs{\beta}_{i_j},\sigma_{i_j})\right]$ is bounded, and Assumption 5 of \citet{kieferwolfowitz56ams} holds. Because Assumptions 1, 2, and 3 of \citet{kieferwolfowitz56ams} are also easily verified, the consistency result follows.

The proof for the dynamic model with the component-specific density function (\ref{eq:fm-dynamic}) using either (\ref{eq:f1-dynamic}) or (\ref{eq:f1-dynamic-mixture}) follows analogously and is therefore omitted.

\end{proof}

\begin{proof}[Proof of Proposition \ref{prop:t2_distribution}]

The proof  is similar to that of Proposition 3 in \cite{kasaharashimotsu15jasa}.

Given the value of $\bs c$ and $\alpha\in [c_1,1-c_1]$, define the space for reparameterized parameters as
$\bs{\psi} := (\bs{\nu}\t,\bs{\lambda}\t)\t \in \Theta_{\bs\psi_\alpha}$,
where $\Theta_{\bs{\psi}_\alpha} = \{ \bs{\psi}=(\bs\nu\t,\bs\lambda\t)\t:   (\alpha,\bs{\nu} + ( 1 - \alpha) \bs{\lambda}, \bs{\nu} - \alpha \bs{\lambda}) \in \Theta_{\bs\vartheta_2}(\bs c)\}.$  Under the null hypothesis $H_{01}: \bs{\theta}_1  = \bs{\theta}_2 = \bs{\theta}^*$, we have $\bs{\lambda} =\bs 0$ and $\bs{\nu} = \bs{\theta}^*$. We rewrite the reparameterized parameters under the null hypothesis as $\bs{\psi}^{*} = ((\bs{\theta}^*)\t,\bs 0\t)\t$. We denote the reparameterized  density function and its logarithm as
\begin{align}\label{eq:repar}
	g(\bs{w};\bs{\psi},\alpha) & := \alpha f(\bs{w}; \bs{\nu}  + (1 - \alpha) \bs{\lambda}) + (1 - \alpha)  f(\bs{w}; \bs{\nu} - \alpha \bs{\lambda})\ \text{ and }\  l(\bs{w};\bs{\psi},\alpha)    = \log g(\bs{w};\bs{\psi},\alpha).
\end{align}

Let $L_n(\bs{\psi},\alpha):= \sum_{i=1}^n  l(\bs{W}_i;\bs{\psi},\alpha)$ be the reparameterized log-likelihood function.
For each $\alpha\in [c_1,1-c_1]$, define the reparameterized MLE as
\begin{equation}
	  \label{eq:mle}
\hat{\bs{\psi}}_\alpha= \arg \max_{\bs{\psi} \in \Theta_{\bs\psi_\alpha}}  L_n(\bs{\psi},\alpha).
\end{equation}

Collect the relevant normalized reparameterized parameters and define $\bs{t}(\bs{\psi},\alpha)$ as
\begin{equation}
 \label{eq:t_1}
 \bs{t}(\bs{\psi},\alpha)  = \begin{pmatrix}
 \bs{t}_{\bs{\nu}} \\
 \bs{t}_{\bs{\lambda} }(\bs{\lambda},\alpha)
 \end{pmatrix}= \begin{pmatrix}
	 \bs{\nu} - \bs{\nu}^* \\
	  \alpha ( 1- \alpha)  \bs{v} (\bs{\lambda} )
 \end{pmatrix},
\end{equation}
where $v(\bs{\lambda})$  is  given by (\ref{eq:v}).

As discussed in the proof of Lemma \ref{lemma:expansion}, taking the fourth-order Taylor expansion of $L_n(\bs{\psi},\alpha)$  around $(\bs\psi^*,\alpha)$, we may write $2\{L_n(\bs{\psi},\alpha)- L_n(\bs{\psi}^*,\alpha)\}$ as a quadratic function of $\sqrt{n}\bs{t}(\bs{\psi},\alpha)$ as
\begin{align}
2\{L_n(\bs{\psi},\alpha)- &L_n(\bs{\psi}^*,\alpha)\} = 2(\sqrt{n}\bs{t}(\bs{\psi},\alpha))\t \bs S_n - (\sqrt{n}\bs{t}(\bs{\psi},\alpha))\t \bs {\mathcal{I}}_n(\sqrt{n}\bs{t}(\bs{\psi},\alpha)) + R_n(\bs{\psi},\alpha)  \label{eq:LR0} \\
 & \quad=  \bs G_n\t \bs{\mathcal{I}}_n \bs G_n -  \left[ \sqrt{n}\bs{t}(\bs{\psi},\alpha)- \bs G_n\right]\t \bs{\mathcal{I}}_n  \left[ \sqrt{n}\bs{t}(\bs{\psi},\alpha)- \bs G_n\right] + R_n(\bs{\psi},\alpha),\label{eq:LR1}
\end{align}
where   $\bs S_n := n^{-1/2} \sum_{i=1}^n {\bs s}(\bs W_i)$ and $\bs G_n:=\bs{\mathcal{I}}_n^{-1}\bs S_n$, where $\bs {\mathcal{I}}_n$ is the negative of the sample Hessian defined in the proof of Lemma \ref{lemma:expansion}. 


Noting that $L_n(\bs{{\psi}}^*,\alpha)= L_{0,n}(\bs{\gamma}^*_0,\bs{\theta}^*_0)$,
write
\begin{align}
LR_n & =   \max_{\alpha \in [c_1, 1- c_2]} 2 \{  L_n(\bs{\hat{\psi}}_\alpha,\alpha)  -   L_n(\bs{{\psi}}^*,\alpha)  \} -  2\{ L_{0,n}(\hat{\bs{\gamma}}_0,\hat{\bs{\theta}}_0)   -   L_{0,n}(\bs{\gamma}^*_0,\bs{\theta}^*_0)   \}. \label{eq:split}
\end{align}

Define \begin{equation*}
\bs{S}_n = \begin{pmatrix}
\bs{S}_{\bs{\eta} n } \\
\bs{S}_{\bs{\lambda} n }
\end{pmatrix}:=  \begin{pmatrix}
n^{-1/2} \sum_{i=1}^n{\bs s}_{\bs\eta}(\bs W_i)\\
n^{-1/2} \sum_{i=1}^n{\bs s}_{\bs\lambda\bs\lambda}(\bs W_i)
\end{pmatrix}, \quad \begin{matrix}
\bs{S}_{\bs{\lambda}, \bs{\eta} n }  := \bs{S}_{\bs{\lambda} n } - \bs{\mathcal{I}}_{\bs{\lambda} \bs{\eta}}  \bs{\mathcal{I}}_{\bs{\eta}}^{-1} \bs{S}_{\bs{\eta} n } , \quad  \bs{G}_{\bs{\lambda}, \bs{\eta} n } := \bs{\mathcal{I}}_{\bs{\lambda},\bs{\eta}}^{-1} \bs{S}_{\bs{\lambda}, \bs{\eta} n }, \\
\bs{t}_{\bs{\eta}, \bs{\lambda} }   := \bs{t}_{\bs{\eta} } -
\bs{\mathcal{I}}_{\bs{\eta} } \bs{\mathcal{I}}_{\bs{\eta} \bs{\lambda} }^{-1} \bs{t}_{ \bs{\lambda} }(\bs{\lambda},\alpha),
\end{matrix}
\end{equation*}
and split the quadratic form in (\ref{eq:LR0}) to obtain
\begin{align}\label{eq:LR2}
2\{L_n(\bs{\psi},\alpha)-  L_n(\bs{\psi}^*,\alpha)\} =  B_n (\sqrt{n} \bs{t}_{\bs{\eta},\bs{\lambda}} )  + C_n (\sqrt{n} \bs{t}_{\bs{\lambda}}(\bs{\lambda},\alpha) ) +R_n(\bs\psi,\alpha),
\end{align}
where \begin{equation} \label{eq:split2}
\begin{split}
B_n ( \bs{t}_{\bs{\eta},\bs{\lambda}} )  & = 2 \bs{t}_{\bs{\eta},\bs{\lambda}}^\top \bs{S}_{\bs{\eta} n} - \bs{t}_{\bs{\eta},\bs{\lambda}}^\top \bs{\mathcal{I}}_{\bs{\eta}} \bs{t}_{\bs{\eta},\bs{\lambda}},\\
C_n ( \bs{t}_{\bs{\lambda}} ) & = 2 \bs{t}_{\bs{\lambda}}^\top \bs{S}_{\bs{\lambda},\bs{\eta} n} - \bs{t}_{\bs{\lambda}} ^\top \bs{\mathcal{I}}_{\bs{\lambda},\bs{\eta}} \bs{t}_{\bs{\lambda}}\\
& = \bs{G}_{\bs{\lambda}, \bs{\eta} n}^\top \bs{\mathcal{I}}_{\bs{\lambda},\bs{\eta}} \bs{G}_{\bs{\lambda}, \bs{\eta} n} - (\bs{t}_{\bs{\lambda}} - \bs{G}_{\bs{\lambda}, \bs{\eta} n})^\top \bs{\mathcal{I}}_{\bs{\lambda},\bs{\eta}}  (\bs{t}_{\bs{\lambda}}  - \bs{G}_{\bs{\lambda}, \bs{\eta} n}),
\end{split}
\end{equation}
with
$\bs{G}_{\bs{\lambda},\bs{\eta}n}\overset{d}{\rightarrow} \bs{G}_{\bs{\lambda},\bs{\eta}}=(\bs{\mathcal{I}}_{ \bs{\lambda},\bs{\eta}} )^{-1} \bs{S}_{\bs{\lambda},\bs{\eta}  } $ and $\bs{S}_{\bs{\lambda}, \bs{\eta} n } \overset{d}{\to} \bs{S}_{\bs{\lambda},\bs{\eta}  } \sim N(\bs 0, \bs{\mathcal{I}}_{ \bs{\lambda},\bs{\eta}} )$.
In addition, $R_n(\hat{\boldsymbol{\psi}}_\alpha,\alpha) = o_p(1)$ holds from  Lemma \ref {lemma:expansion}(a) and $\sqrt{n}\bs{t}(\hat{\bs\psi}_\alpha,\alpha) = O_p(1)$.

Because $\Delta_{(\gamma,\theta)}f(x; \hat{\bs{\gamma}}_0^*,\hat{\bs{\theta}}_0^*)$ is identical to $\Delta_{\bs\eta} f(x; \bs\psi^*,\alpha)$, a standard analysis gives $2 [L_{0,n}(\hat{\bs{\gamma}}_0,\hat{\bs{\theta}}_0)   -   L_{0,n}(\bs{\gamma}^*_0,\bs{\theta}^*_0)]  = \max_{\bs{t}_{\bs{\eta} }}  B_n (\sqrt{n} \bs{t}_{\bs{\eta} } ) + o_p(1)$.
Note that the possible values of both $\sqrt{n}\bs{t}_{\bs{\eta}}$ and $\sqrt{n}\bs{t}_{\bs{\eta},\bs{\lambda}}$ approach $\R^{q}$. Therefore,  in view of (\ref{eq:LR2}) and (\ref{eq:split2}), we can write equation (\ref{eq:split}) as
\begin{equation}\label{eq:LRTS_expansion_homo_repar}
LR_n  = \max_{\alpha\in [c_1,1-c_1]} C_n (\sqrt{n} \bs{t}_{\bs{\lambda}}(\hat{\bs{\lambda}}_\alpha, \alpha))  + o_p(1),
\end{equation}
where $\hat{\bs{\lambda}}_\alpha$ is as defined in (\ref{eq:mle}).

%
The asymptotic distribution of $LR_n$ follows from applying Theorem 3(c)  of \cite[][p. 1362]{Andrews1999} to (\ref{eq:LR2}) and (\ref{eq:LRTS_expansion_homo_repar}).
First,   Assumption 2 of \cite{Andrews1999} holds because Assumption 2* of \cite{Andrews1999} holds because of Proposition 2(a).
Second, Assumption 3 of \cite{Andrews1999} holds with $B_T = n^{1/2}$  and $T=n$ because $\bs{S}_{\bs{\lambda}, \bs{\eta} n } \overset{d}{\to} \bs{S}_{\bs{\lambda},\bs{\eta}  } \sim N(\bs 0, \bs{\mathcal{I}}_{ \bs{\lambda},\bs{\eta}} )$
 and $\bs{\mathcal{I}}_{ \bs{\lambda},\bs{\eta}} $ is non-singular.
Assumption 4 of \cite{Andrews1999} holds from part (a).
Assumption 5 of \cite{Andrews1999} follows from Assumption 5* and Lemma 3 of \cite{Andrews1999} with $b_T=n^{1/2}$ because $\alpha(1-\alpha)v(\Theta_{\bs\lambda})$ is locally equal to $\Lambda_{\bs{\lambda}}$.
Therefore, it follows from Theorem 3(c) of \cite{Andrews1999} that $C_n (\sqrt{n} \bs{t}_{\bs{\lambda}}(\hat{\bs{\lambda}},\alpha) )   \overset{d}{\to} (\hat{\bs{t}}_{\bs\lambda} )^\top \bs{\mathcal{I}}_{\bs\lambda,\bs\eta} \hat{\bs{t}}_{\bs\lambda}$,
where $\hat{\bs{t}}_{\bs\lambda}$ is defined by (\ref{eq:t_lambda_def}).

\end{proof}

\begin{proof}[Proof of Proposition \ref{prop:infinite_fisher}]
	Under $H_{2,0}$,  we obtain $\vartheta_{M_0+1} \in \Upsilon_{2h}^*$,
	\begin{equation}
		\begin{split}
		& \E [\{ \nabla_{\alpha_h}  \log f_{M_0+1}  (\bs W_i, \vartheta_{M_0+1} )\}^2  ]  \\
		& = \int  \frac{ \{ f( \bs w;   \bs \theta_{h} )   - f( \bs w;  \bs \theta_{M_0}^* )  \}^2 }{\sum_{j=1}^{M_0} \alpha_j^* f( \bs w;  \bs \theta_j^* ) } d \bs w \\
		& = \int  \frac{ \{ f( \bs w; \bs \theta_{h} ) \}^2 }{\sum_{j=1}^{M_0} \alpha_j^* f( \bs w;  \bs \theta_j^* ) } d \bs w  + \int  \frac{ \{  f( \bs w;  \bs \theta_{M_0}^* )  \}^2 }{\sum_{j=1}^{M_0} \alpha_j^*
		 f( \bs w;  \bs \theta_j^* ) } d \bs w -   2  \int  \frac{  f( \bs w;  \bs \theta_{h} )  f( \bs w;  \bs \theta_{M_0}^* )   }{\sum_{j=1}^{M_0} \alpha_j^* f( \bs w;  \bs \theta_j^* ) } d \bs w .
	\end{split} \label{eq:infinite_fisher}
 \end{equation}
The latter two terms on the right-hand side of (\ref{eq:infinite_fisher}) are bounded because
$f( \bs w;  \bs \theta_{M_0}^* )  / \sum_{j=1}^{M_0} \alpha_j^* f( \bs w;  \bs \theta_j^* )  \le (1 / \alpha_{M_0}^*)$ for any $\bs w$ and $f(\bs w; \bs \theta)$ integrates to one.
Therefore, the left-hand side of (\ref{eq:infinite_fisher}) goes to infinity if and only if the first term on the right-hand side of (\ref{eq:infinite_fisher}) goes to infinity.

Because $\max_j \alpha_j \le \sum_{j}^{M_0} \alpha_j \le M_0 \max_j \alpha_j$, we obtain
\[  \frac{1}{M_0}  \frac{ \{  f( \bs w;  \bs \theta_{h} )  \}^2 }{ \max_j  \{ \alpha_j^* f( \bs w;  \bs \theta_j^* )  \} }  \le   \frac{ \{  f( \bs w;  \bs \theta_{h} )  \}^2 }{ \sum_{j=1}^{M_0} \alpha_j^* f( \bs w;  \bs \theta_j^* ) }  \le   \frac{ \{  f( \bs w;  \bs \theta_{h} )  \}^2 }{ \max_j  \{ \alpha_j^* f( \bs w; \bs \theta_j^* )  \} }. \]

Without loss of generality, we assume that $\sigma_{M_0}^* = \max \{ \sigma_1^*,\ldots, \sigma_{M_0}^* \}$ and that the maximum is unique.
We focus on models without covariates because the law of iterated expectations implies that the stated result also holds for models with covariates if it holds for models without covariates.

We first prove the case for the component-specific density function $f(\boldsymbol{w};\boldsymbol{\theta})=\prod_{t=1}^Tf(y_t;\bs\theta)$, where $f(y_t;\bs\theta)$ is given by (\ref{eq:f1}) or  (\ref{eq:f1-mixture}) with $\bs\beta_j=\bs 0$.  Under the normal density function (\ref{eq:f1}) without covariates, there exists a sufficiently large but finite positive constant $B$, such that $\max_j \{ \alpha_j^* f (\bs w,  \mu_j^*,  \sigma_j^2) \} = \alpha_{M_0}^* f (\bs w,   \mu_{M_0}^*,  \sigma_{M_0}^2)$ when $| y_t| > B$ for all $t = 1,\ldots, T$.
Note that
\begin{equation}\label{bound}
\begin{split}
\frac{ \{  f( \bs w;  \mu_{h},  \sigma_{h} )  \}^2 }{ f( \bs w;  \mu_{M_0}^*,   \sigma_{M_0}^* )  } & = \prod_{t=1}^T \frac{ \sigma_{M_0}^{*} }{(2 \pi )^{1/2} \sigma_h^2 } \exp \left\{ -\frac{1}{\sigma_h^2} (y_t - \mu_h)^2 + \frac{1}{2 (\sigma_{M_0}^{*})^2 } (y_t - \mu_{M_0}^{*})^2 \right\} \\
& = \left(\frac{\sigma_{M_0}^* }{(2 \pi )^{1/2} \sigma_h^2 } \right)^T\exp \left\{ -  \frac{1}{\sigma_h^2}\sum_{t=1}^T( y_t - \mu_h)^2 + \frac{1}{2 (\sigma_{M_0}^{*})^2} \sum_{t=1}^T(y_t - \mu_{M_0}^{*})^2 \right\}.
\end{split}
\end{equation}
Then, the integral of the right-hand side of (\ref{bound})  over $|y_t| \ge B$ for $t=1,...,T$ is infinite if  $\sigma_h^2/ \sigma_{M_0}^{2*} > 2$, and the stated result holds.

When  $f(y_t;\bs\theta)$ is given by normal mixture density (\ref{eq:f1-mixture}), suppose that $\mu_{j1}<\mu_{j2}< \cdots<\mu_{{\cal{K}}}$ for all $j=1,....,M_0$. Then, there exists a constant $B$ such that, when $y_t>B$  for all $t = 1,\ldots, T$, we have $\max_j \{ \alpha_j^* f (\bs w,  \bs\theta^*_j) \} = \alpha_{M_0}^* f (\bs w,\bs\theta^*_{M_0})$ and
\begin{equation}\label{bound-mixture-density}
\frac{1}{\sigma_j}\exp\left(-\frac{1}{2}\left(\frac{y_t-\mu_{{\cal{K}}}}{\sigma_j}\right)^2\right)\geq \sum_{k=1}^{{\cal{K}}}\tau_{jk}\frac{1}{\sigma_j}\exp\left(-\frac{1}{2}\left(\frac{y_t-\mu_{jk}}{\sigma_j}\right)^2\right)\geq \frac{1}{\sigma_j}\exp\left(-\frac{1}{2}\left(\frac{y_t-\mu_{j1}}{\sigma_j}\right)^2\right).
\end{equation}
Then, it follows that
\[
\frac{ \{  f( \bs w;  \bs\theta_{h} )  \}^2 }{ f( \bs w;  \bs\theta_{M_0}^* )   }  \geq  \left(\frac{\sigma_{M_0}^* }{(2 \pi )^{1/2} \sigma_h^2 } \right)^T\exp \left\{ -\frac{1}{\sigma_h^2}\sum_{t=1}^T (y_t - \mu_{1h})^2 + \frac{1}{2 (\sigma_{M_0}^{*})^2 } \sum_{t=1}^T(y_t - \mu_{M_0K}^{*})^2 \right\},
\]
where  the integral of the right-hand side  over $y_t>B$ for $t=1,...,T$  is infinite if  $\sigma_h^2/ \sigma_{M_0}^{2*} > 2$ as in (\ref{bound}).

Next, for dynamic panel models,  suppose that $f(\boldsymbol{w};\boldsymbol{\theta})=f_1(y_1;\bs\theta)\prod_{t=2}^Tf(y_t|y_{t-1};\bs\theta)$ with $f_1(y_1;\bs\theta)$ and $f(y_t|y_{t-1};\bs\theta)$ given in
(\ref{eq:f1-dynamic}) with $\bs\beta_{1,j}=\bs\beta_j=\bs 0$.  Fix the value of $\{y_t\}_{t=1}^{T-1}$. Then,
\begin{align*}
\frac{ \{  f( \bs w;  \bs\theta_{h} )  \}^2 }{ f( \bs w;  \bs\theta_{M_0}^* )   } &= \left(\frac{f_1(y_1;\bs\theta_h)^2\prod_{t=2}^{T-1}f(y_t|y_{t-1};\bs\theta_h)^2}{f_1(y_1;\bs\theta_{M_0}^*)\prod_{t=2}^{T-1}f(y_t|y_{t-1};\bs\theta_{M_0}^*)}\right) \\
&\quad \times \frac{\sigma_{M_0}^* }{(2 \pi )^{1/2} \sigma_h^2 }  \exp \left\{ -  \frac{1}{\sigma_h^2} ( y_T - \mu_h-\rho_h y_{T-1})^2 + \frac{1}{2 (\sigma_{M_0}^{*})^2} (y_T - \mu_{M_0}^{*}-\rho_{M_0}^{*} y_{T-1})^2 \right\},
\end{align*}
where the integral of the right-hand side of this equation over $|y_T| \ge B$  is infinite for a sufficiently large $B$ if  $\sigma_h^2/ \sigma_{M_0}^{2*} > 2$. Therefore, $E\left[\left.\frac{ \{  f( \bs W;  \bs\theta_{h} )  \}^2 }{ f( \bs W;  \bs\theta_{M_0}^* )   } \right|  \{Y_t\}_{t=1}^{T-1}=\{y_t\}_{t=1}^{T-1}  \right]=\infty$, and the stated result follows from the law of iterated expectations. The case for $f_1(y_1;\bs\theta)$ and $f(y_t|y_{t-1};\bs\theta)$ given in
(\ref{eq:f1-dynamic-mixture}) can be proven analogously, and hence we omit it.


\end{proof}

\begin{proof}[Proof of Proposition \ref{prop:tm0_distribution}]

Collect the score vector for testing $H_{0,1h}$ for $h = 1,\ldots, M_0$ into one vector as
\begin{equation}
    \label{eq:s_tilde}
\tilde{\bs{s}}(\bs{W}) = \begin{pmatrix}
\tilde{\bs{s}}_{\bs{\eta}} (\bs{W})  \\
\tilde{\bs{s}}_{\bs{\lambda\lambda}} (\bs{W})
\end{pmatrix}, \ \text{ where } \underset{(M_0 +  p + q +1 ) \times 1}{\tilde{\bs{s}}_{\bs{\eta}}(\bs{W})}  = \begin{pmatrix}
\bs{s}_{\bs{\alpha}}(\bs{W})  \\
\bs{s}_{ (\bs{\nu} ) } (\bs{W})
\end{pmatrix}\quad \text{ and } \tilde{\bs{s}}_{\bs{\lambda\lambda}} (\bs{W}) = \begin{pmatrix}
\bs{s}^1_{\bs{\lambda\lambda}} (\bs{W}) \\
\vdots \\
\bs{s}^{M_0}_{\bs{\lambda\lambda}} (\bs{W})
\end{pmatrix},
\end{equation}
where
\begin{equation}
 \begin{split}
    \bs{s}_{\bs{\alpha}} (\bs{W}) & = \begin{pmatrix}
    f(\bs{W};\bs{\theta}^*_1) -    f(\bs{W};\bs{\theta}^*_{M_0}) \\
    \vdots \\
       f(\bs{W};\bs{\theta}^*_{M_0-1}) -    f(\bs{W};\bs{\theta}^*_{M_0})
    \end{pmatrix} \Bigg / f_{M_0}(\bs{W}; \bs{\vartheta}_{M_0}^*) ,  \\
    \bs{s}_{ \bs{\nu}  } (\bs{W})   & = \sum_{j=1}^{M_0} \alpha_j^*\frac{
    \nabla_{ \bs{\nu}  }  f(\bs{W};\bs{\theta}^*_{j})}{f_{M_0}(\bs{W}; \bs{\vartheta}_{M_0}^*)},\ \ \text{and}\ \   \bs{s}^h_{\bs{\lambda\lambda}} (\bs{W})   = \frac{\widetilde{ \nabla}_{\bs\theta_h\bs\theta_h\t}  f(\bs{W};\bs{\theta}^*_h) }{ f_{M_0}(\bs{W}; \bs{\vartheta}_{M_0}^*)}\quad\text{for $h=1,2,...,M_0$} ,
    \end{split}
\end{equation}
with
$\widetilde{ \nabla}_{\bs\theta_h \bs\theta\t_h}  f(\bs{W};\bs{\theta}^*_h)  := (c_{11} \nabla_{\theta_{h1}\theta_{h1}} f^*,...,c_{qq}\nabla_{\theta_{hq}\theta_{hq}} f^*,c_{12}\nabla_{\theta_{h1}\theta_{h2}} f^*,...,c_{(q-1)q}\nabla_{\theta_{hq-1}\theta_{hq}} f^*)\t$
for $\bs{\theta}_h:=(\theta_{h1},\theta_{h2},\theta_{h3},...,\theta_{hq})\t:=(\mu_h,\sigma_h^2,\beta_{h1},...,\beta_{hq-2})\t$ and  $c_{jk}=1/2$ for $j\neq k$ and $c_{jk}=1$ for $j=k$.
 Define \begin{equation}
  \label{eq:I_m0}
\begin{split}
    \tilde{\bs{\mathcal{I}}} := \E[\bs{\tilde{s}}(\bs{W}) \bs{\tilde{s}}(\bs{W})^\top], \quad  \tilde{\bs{\mathcal{I}}}_{\bs{\eta}} := \E[\bs{\tilde{s}}_{\bs{\eta}}(\bs{W})  \bs{\tilde{s}}_{\bs{\eta}}(\bs{W})^\top ], \quad \tilde{\bs{\mathcal{I}}}_{\bs{\lambda} \bs{\eta}} := \E[\bs{\tilde{s}}_{\bs{\lambda\lambda}}(\bs{W}) \bs{\tilde{s}}_{\bs{\eta}}(\bs{W})^\top ], \\
    \tilde{\bs{\mathcal{I}}}_{\bs{\eta} \bs{\lambda}}  := \tilde{\bs{\mathcal{I}}}_{\bs{\lambda} \bs{\eta}}^\top, \quad \tilde{\bs{\mathcal{I}}}_{\bs{\lambda\lambda}} := \E[\bs{\tilde{s}}_{\bs{\lambda\lambda}}(\bs{W}) \bs{\tilde{s}}_{\bs{\lambda\lambda}}(\bs{W})^\top ], \quad \tilde{\bs{\mathcal{I}}}_{\bs{\bs\lambda,\bs\eta}} := \tilde{\bs{\mathcal{I}}}_{\bs{\lambda \lambda}} - \tilde{\bs{\mathcal{I}}}_{\bs{\lambda \eta}} \tilde{\bs{\mathcal{I}}}_{\bs\eta}^{-1} \tilde{\bs{\mathcal{I}}}_{\bs{\eta \lambda}}.
\end{split}
\end{equation}
Then, the asymptotic distribution of the normalized score function is given by
\[
 \tilde{\bs S}_n: = \frac{1}{\sqrt{n}} \sum_{i=1}^n \tilde{\bs s}(\bs W_i) \overset{d}{\to} \tilde {\bs S} \sim N(\bs 0,   \tilde{\bs{\mathcal{I}}}),
\]
where, in view of (\ref{eq:s_tilde}),  $ \tilde {\bs S} $ may be partitioned as $ \tilde {\bs S} =( \tilde {\bs S}_{\bs\eta}\t, \tilde {\bs S}_{\bs\lambda\bs\lambda}\t)\t$ with $n^{-1/2} \sum_{i=1}^n \tilde{\bs s}_{\bs\eta}(\bs W_i)\overset{d}{\to} \tilde {\bs S}_{\bs\eta}$ and $n^{-1/2} \sum_{i=1}^n \tilde{\bs s}_{\bs\lambda\bs\lambda}(\bs W_i)\overset{d}{\to} \tilde {\bs S}_{\bs\lambda\bs\lambda}$ .

Let $\bs{\tilde{S}}_{\bs\lambda,\bs\eta} := (\bs{{S}}_{\bs\lambda,\bs\eta}^1, \ldots, \bs{{S}}_{\bs\lambda,\bs\eta}^{M_0} )^\top:= \bs{\tilde{S}}_{\bs\lambda\bs\lambda} -
\tilde{\bs{\mathcal{I}}}_{\bs\lambda\bs\eta} \tilde{\bs{\mathcal{I}}}_{\bs\eta} ^{-1} \bs{\tilde{S}}_{\bs\eta}
 \sim N(0,   \tilde{\bs{\mathcal{I}}}_{\bs\lambda,\bs\eta})$ be a $\R^{M_0 (q)(q+1)/2}$-valued random vector. For $h=1,2,...,M_0$,  define $\tilde{\bs{\mathcal{I}}}_{\bs\lambda,\bs\eta}^h := \E[\bs{{S}}_{\bs\lambda,\bs\eta}^h ( \bs{{S}}_{\bs\lambda,\bs\eta}^h)^\top]$ and $\bs{{G}}_{\bs\lambda,\bs\eta}^h := ({\bs{\mathcal{I}}}_{\bs\lambda,\bs\eta}^h)^{-1} \bs{{S}}_{\bs\lambda,\bs\eta}^h$.

Define $\hat{\bs t}^h_{\bs\lambda}  $ analogously to $\hat{\bs t}_{\bs\lambda} $ as
\begin{equation}
     \label{eq:t_h}
\begin{split}
r^h_{\bs\lambda} (\hat{\bs{t}}^h_{\bs{\lambda}} ) = \inf_{{\bs{t}}^h_{\bs{\lambda}}\in \Lambda_{\bs\lambda}} r^h({\bs{t}}^h_{\bs{\lambda}}); \quad
r^h_{\bs\lambda} ({\bs{t}}^h_{\bs{\lambda}}) := ({\bs{t}}^h_{\bs{\lambda}} - \bs{{G}}^h_{\bs\lambda,\bs\eta} )^\top  {\bs{\mathcal{I}}}_{\bs\lambda,\bs\eta}^h({\bs{t}}^h_{\bs{\lambda}} - \bs{{G}}^h_{\bs\lambda,\bs\eta} )\quad\text{for $h=1,2,...,M_0$}.
\end{split}
\end{equation}

The local quadratic-form approximation of the log-likelihood function  $LR^{M_0,h}_n$ around $\Theta^*_{\bs\vartheta_{M_0+1},1h} \subset \Theta_{\vartheta_{M_0 + 1}}$  has an identical structure to the approximation that we derive in Section \ref{sec:LRT1} in testing $H_{01}$ in the test of homogeneity.  Consequently, we can show that $LR^{M_0,h}_n\overset{d}{\to} (\hat{\bs{t}}^h_{\bs\lambda})^\top \bs{\mathcal{I}}^h_{\bs\lambda,\bs\eta}  \hat{\bs{t}}^h_{\bs\lambda}$.  Then, given (\ref{eq:LR_M0_max}), the asymptotic null distribution of the LRTS for testing $H_{01}$ is given by the maximum over $(\hat{\bs{t}}^h_{\bs\lambda})^\top \bs{\mathcal{I}}^h_{\bs\lambda,\bs\eta}  \hat{\bs{t}}^h_{\bs\lambda}$s for $h=1,2,..., M_0$. We now proceed with the proof.

 For $h = 1,\ldots, M_0$,  let $\mathcal{N}_h^* \subset \Theta_{\bs{\vartheta}_{M_0 + 1}}(\bs c) $ be a sufficiently small closed neighborhood of $\Theta^*_{\bs\vartheta_{M_0+1},1h}$ such that $\alpha_h,\alpha_{h+1} > 0$ holds and $\Theta^*_{\bs\vartheta_{M_0+1},1k} \not\subset \mathcal{N}_h^*$ if $k \neq h$. Consider the following one-to-one reparameterization from the $(M_0 + 1)$-component model parameter $\bs{\vartheta}_{M_0 + 1} = (\alpha_1,\ldots,\alpha_{M_0},\bs{\theta}^\top_1,\ldots,\bs{\theta}^\top_h,\bs{\theta}^\top_{h+1},\ldots,\bs{\theta}^\top_{M_0 + 1})^\top$.
     Similar to (\ref{eq:m0_repar2}), the one-to-one reparameterization  for testing the null hypothesis $H_{0,1h}$ is given by
    \begin{equation*}
    \begin{pmatrix}
    \bs{\lambda}_h \\
    \bs{\nu}_h
    \end{pmatrix} := \begin{pmatrix}
    \bs{\theta}_{h} - \bs{\theta}_{h+1} \\
    \tau \bs{\theta}_{h} + (1 - \tau) \bs{\theta}_{h+1}
    \end{pmatrix} \text{ so that }
    \begin{pmatrix}
    \bs{\theta}_{h} \\
    \bs{\theta}_{h+1}
    \end{pmatrix} = \begin{pmatrix}
    \bs{\nu} + ( 1- \tau) \bs{\lambda} \\
    \bs{\nu} - \tau \bs{\lambda}
    \end{pmatrix},
    \end{equation*} and  $\alpha_j$ is reparameterized for $j=1,2,...,M_0$ as
    \begin{align*}
     (\pi_1,\ldots,\pi_{h-1},\pi_h,\pi_{h+1},\ldots , \pi_{M_0 - 1}) &= (\alpha_1,\ldots,\alpha_{h-1}, (\alpha_h + \alpha_{h+1}),\alpha_{h+2},\ldots,\alpha_{M_0})\\
 \tau  &=  {\alpha_h}/({\alpha_h + \alpha_{h+1}})
      \end{align*}
    so that $\pi_h  = \alpha_h + \alpha_{h+1}$ and  $\pi_{M_0} = 1 - \sum_{j=1}^{M_0 - 1} \pi_j$.

    Collect the reparameterized parameters except $\tau$ as
\begin{equation*}
 \bs{\psi}_{h,\tau}  = (\bs\eta\t,\bs{\lambda}^\top_h)^\top\quad\text{with}\quad\bs \eta= (\pi_1,\ldots,\pi_{M_0-1},\bs{\theta}^\top_1,\ldots,\bs{\theta}^\top_{h-1},\bs{\nu}^\top_h,\bs{\theta}^\top_{h+2},\ldots, \bs{\theta}^\top_{M_0+1})^\top.
\end{equation*}
 In the reparameterized model, the null restriction $\bs{\theta}_h  = \bs{\theta}_{h+1}$ implied by $H_{0,1h}$ holds if and only if $\bs{\lambda}_h = 0$. Under $H_{0,1h}$, we have $\lambda_h^*=0$ and
 $\eta^*=(\alpha_1^*,...,\alpha_{M_0-1}^*, (\bs{\theta}^*_1)^\top,\ldots, (\bs{\theta}^*_{M_0})^\top)^\top$.
  Define the log-likelihood under the reparameterized parameters as
    \begin{equation*}
    f_{M_0+1}^{h}(\bs{w};\bs{\psi}_{h,\tau},\tau )  = \pi_h g^h(\bs{w},\bs{\psi}_{h,\tau}, \tau) +  \sum_{j=1}^{h-1} \pi_j f(\bs{w};\bs{\theta}_j) +  \sum_{j=h}^{M_0} \pi_{j+1} f(\bs{w};\bs{\theta}_{j+1}),
    \end{equation*}
    where $g^h(\bs{w},\bs{\psi}_{h,\tau}, \tau)$ is defined similarly to (\ref{eq:repar}) as
    \begin{equation}
    g^h(\bs{w},\bs{\psi}_{h,\tau}, \tau) = \tau f(\bs{w};\bs{\nu}_h + (1 - \tau) \bs{\lambda}_h) + (1 - \tau) f(\bs{w};\bs{\nu}_h - \tau \bs{\lambda}_h).
    \end{equation}

 Define the local  MLE of $\bs{\psi}_{h,\tau}$ by
 \begin{equation}\label{eq:psi_mle}
 \bs{\hat{\psi}}_{h,\tau} := \argmax_{\bs{\psi}_{h,\tau}  \in \mathcal{N}_h^*}  L^h_n(\bs{\psi}_{h,\tau}, \tau),
 \end{equation}
 where $L_n^h(\bs{\psi}_{h,\tau}, \tau) := \sum_{i=1}^N \log g^h(\bs{W}_i;\bs{\psi}_{h,\tau}, \tau) $.
 Because $\bs{\psi}_{h,\tau}^*$ is the only parameter value in $\mathcal{N}_h^*$ that generates the true density, $\bs{\hat{\psi}}_{h,\tau} - \bs{\psi}_{h,\tau}^* = o_p(1)$ follows. 

Define the LRTS for testing $H_{0,1h}$ as $LR_{n}^{M_0,h}  := \max_{\tau \in [c_1,1-c_1]} 2 ( L_n^h(\hat{\bs{\psi}}_{h,\tau}, \tau)  - L_{0,n}( \hat{\bs{\vartheta}}_{M_0}))$. Then, in view of (\ref{eq:LR_M0_max}), the stated result holds if
 \begin{equation}
   \label{eq:joint_LR}
(LR_{n}^{M_0,1},\ldots,LR_{n}^{M_0,M_0})^\top \overset{d}{\to} (\hat{\bs{t}}^1_{\lambda})^\top \bs{\mathcal{I}}^1_{\eta,\lambda} (\hat{\bs{t}}^1_{\lambda}) ,   \ldots, (\hat{\bs{t}}^{M_0}_{\lambda})^\top \bs{\mathcal{I}}^{M_0}_{\eta,\lambda} (\hat{\bs{t}}^{M_0}_{\lambda}) )^\top.
\end{equation}
Observe that $L^h_n(\bs{\psi}_{h,\tau}, \tau)  -  L^h_n(\bs{\psi}^*_{h,\tau}, \tau) $ admits the same expansion as
 $L_n(\bs{\hat{\psi}},\alpha)  -   L_n(\bs{{\psi}}^*,\alpha)  $ in (\ref{eq:LR1}) and (\ref{eq:LR2}) when
  $(\alpha,\bs{t}(\bs{\psi}, \alpha), \bs{t}_{\bs\lambda}(\bs{\lambda}, \alpha), \bs S_n, \bs G_n, \bs{\mathcal{I}}_n,R_n(\bs\psi,\alpha))$ is replaced with  \\ $(\tau,\bs{t}^h(\bs{\psi}^h, \tau), \bs{t}_{\bs\lambda}^h(\bs{\lambda}^h, \tau),  \bs S_n^h, \bs G_n^h,\bs{\mathcal{I}}_n^h,R_n^h(\bs\psi^h,\tau))$, where $(\bs S_n^h, \bs{\mathcal{I}}_n^h)$ is defined similarly to $(\bs S_n, \bs{\mathcal{I}}_n)$ but
  $(\bs{s}_{\bs{\eta}},\bs{s}_{\bs{\lambda\lambda}})$ is replaced with $(\tilde{\bs{s}}_{\bs{\eta}}, {\bs{s}}^h_{\bs{\lambda\lambda}})$ and $\bs G_n^h := (\bs{\mathcal{I}}_n^h)^{-1} \bs S_n^h$.  Applying the proof of Lemma \ref{lemma:expansion}, we have ${\bs S}_n^h\overset{d}{\to} {\bs S}^h \sim N(\bs 0, \bs{\mathcal{I}}^h)$ and $\bs{\mathcal{I}}_n^h\overset{p}{\to} \bs{\mathcal{I}}^h$. Then, (\ref{eq:joint_LR}) follows from the proof of Lemma \ref{lemma:expansion} and Proposition \ref{prop:t2_distribution} for each local  MLE when $(\bs G_n, \hat{\bs t}_{\bs\lambda},\bs{\mathcal{I}}_{\bs\lambda,\bs\eta})$ is replaced with  $(\bs G_n^h, \hat{\bs t}_{\bs\lambda}^h,\bs{\mathcal{I}}_{\bs\lambda,\bs\eta}^h)$ and the results are collected; note that $(\bs S_n^1,...,\bs S_n^{M_0})\overset{d}{\to} (\bs S^1,...,\bs S^{M_0})$.
\end{proof}

\begin{proof}[Proof of Proposition \ref{prop:sht}]
We first prove that when $M<M_0$, $\Pr( LR_n^M> \hat c^M_{1-q_n})\rightarrow 1$ as $n\rightarrow \infty$. Under Assumption \ref{assumption:Q_M}, by the standard consistency proof \citep[e.g., Theorem 2.5 of][]{Newey1994} give $\hat{\bs{\vartheta}}_M\overset{p}{\rightarrow} \bs{\vartheta}_M^*$ for $M\leq M_0$. Furthermore,   it follows from the argument in Theorem 3.2 of \cite{white82em} that,  for $M\leq M_0$,
 \begin{equation}\label{eq:asynormal}
 \sqrt{n} (\hat{\bs{\vartheta}}_M- \bs{\vartheta}_M^*) \overset{p}{\rightarrow} N(0, A^M(\bs{\vartheta}_M^*)^{-1} B^M(\bs{\vartheta}_M^*)A^M(\bs{\vartheta}_M^*)^{-1}).
 \end{equation}

Then, from  (\ref{eq:asynormal}) and the mean value expansion,  we have $Q_n^M(\hat{\bs{\vartheta}}_M)- Q^M({\bs{\vartheta}}_M^*) = O_p(n^{-1/2})$ and
\begin{align*}
\frac{ LR_n^M}{n} &:= 2\left( Q_n^{M+1}(\hat{\bs{\vartheta}}_{M+1}) -  Q_n^M(\hat{\bs{\vartheta}}_M) \right)= 2\left(Q^{M+1}({\bs{\vartheta}}_{M+1}^*) -Q^M({\bs{\vartheta}}_M^*)\right)  + o_p(1)
\end{align*}
for $M=1,2,...,M_0-1$.

Because $Q^{M+1}({\bs{\vartheta}}_{M+1}^*) -Q^M({\bs{\vartheta}}_M^*)>0$ by Assumption \ref{assumption:Q_M}(f), ${LR_n^M}/{n}\rightarrow \infty$ as $n\rightarrow \infty$. By Lemma \ref{lemma: sht}, $- n^{-1} \log q_n=o(1)$ and $\hat c^M_{1-q_n}-c^M_{1-q_n}=o_p(1)$ implies that $n^{-1} \hat c^M_{1-q_n}=o_p(1)$. Therefore,  when $M<M_0$,
we have $\Pr(  LR_n^M>\hat  c^M_{1-q_n})=\Pr(  LR_n^M/n>\hat c^M_{1-q_n}/n) \rightarrow 1$ as $n\rightarrow \infty$.

When $M=M_0$, because $  LR_n^{M_0} =O_p(1)$ by Proposition \ref{prop:tm0_distribution} and $\hat c^{M}_{1-q_n}\rightarrow\infty$ by $q_n=o(1)$, $ \Pr(   LR_n^{M_0} >\hat c^{M_0}_{1-q_n}) \rightarrow 0$ as $n\rightarrow \infty$.
 \end{proof}

\begin{proof}[Proof of Proposition \ref{proposition:bic}] 

We first prove that $\lim_{n\rightarrow \infty} \Pr(\widehat{M}_{PL}<M_0)=0$.  To show \(\Pr(\hat{M}_{PL} < M_0) \to 0\), it suffices to prove that for each \(M<M_0\),
\[
\Pr\bigl(p\ell_n^M > p\ell_n^{M_0} \bigr) \;\to\; 0.
\]
By definition,
\begin{equation}\label{eq:W}
\frac{p\ell_n^M-p\ell_n^{M_0}}{n} = ( Q_n^M(\hat{\bs{\vartheta}}_{M}) - Q_n^{M_0}(\hat{\bs{\vartheta}}_{M_0}) )- \frac{ p_{n,k_M}- p_{n,k_{M_0}}}{n}.
\end{equation}
From $p_{n,k} = o(n)$ in Assumption \ref{assumption: penalty}(c),
\begin{equation}
 \frac{ p_{n,k_M}- p_{n,k_{M_0}}}{n} \rightarrow 0. \label{eq:conv_p}
 \end{equation}

By the union bound and the triangle inequality, for any $\epsilon > 0$,
\begin{align*}
\Pr\Bigl(&\left| Q_n^M(\hat{\boldsymbol\vartheta}_{M}) - Q_n^{M_0}(\hat{\boldsymbol\vartheta}_{M_0}) - \Bigl(Q^{M}({\boldsymbol\vartheta}_{M}^*) - Q^{M_0}({\boldsymbol\vartheta}_{M_0}^*)\Bigr)\right| > 4\epsilon\Bigr)
\end{align*}
is at most
\begin{align}
&\Pr\Bigl(\left|Q^M(\hat{\boldsymbol\vartheta}_{M}) - Q^{M}({\boldsymbol\vartheta}_{M}^*)\right| > \epsilon\Bigr)
+ \Pr\Bigl(\left|Q^{M_0}(\hat{\boldsymbol\vartheta}_{M_0}) - Q^{M_0}({\boldsymbol\vartheta}_{M_0}^*)\right| > \epsilon\Bigr) \label{eq:bound-1}\\[1mm]
&\quad +\, \Pr\Bigl(\left|Q_n^M(\hat{\boldsymbol\vartheta}_{M}) - Q^{M}(\hat{\boldsymbol\vartheta}_{M})\right| > \epsilon\Bigr)
+ \Pr\Bigl(\left|Q_n^{M_0}(\hat{\boldsymbol\vartheta}_{M_0}) - Q^{M_0}(\hat{\boldsymbol\vartheta}_{M_0})\right| > \epsilon\Bigr). \label{eq:bound-2}
\end{align}

The consistency of $\hat{\boldsymbol\vartheta}_{M}$ and $\hat{\boldsymbol\vartheta}_{M_0}$, together with the Continuous Mapping Theorem, implies that the two terms in \eqref{eq:bound-1} converge to $0$ as $n\rightarrow \infty$.

For small $\delta>0$, let $B(\delta,\boldsymbol\vartheta_M^*)$ and $B(\delta,\boldsymbol\vartheta_{M_0}^*)$ be open balls in the parameter spaces $\Theta_{\boldsymbol\vartheta_M}$ and $\Theta_{\boldsymbol\vartheta_{M_0}}$, respectively. Since $\hat{\boldsymbol\vartheta}_{M}$ and $\hat{\boldsymbol\vartheta}_{M_0}$ are consistent, with probability approaching one, the two terms in \eqref{eq:bound-2} are bounded by
\[
\Pr\Bigl(\sup_{\boldsymbol\vartheta_M\in B(\delta,\boldsymbol\vartheta_M^*)}
\left|Q_n^M(\boldsymbol\vartheta_M) - Q^M(\boldsymbol\vartheta_M)\right| > \epsilon\Bigr)
\quad\text{and}\quad
\Pr\Bigl(\sup_{\boldsymbol\vartheta_{M_0}\in B(\delta,\boldsymbol\vartheta_{M_0}^*)}
\left|Q_n^{M_0}(\boldsymbol\vartheta_{M_0}) - Q^{M_0}(\boldsymbol\vartheta_{M_0})\right| > \epsilon\Bigr).
\]
By applying Lemma 2.4 of \cite{Newey1994} under Assumption \ref{assumption:Q_M}(b), and noting that $\log g_M(\bs w;\bs\vartheta_M)$ is continuous at each $\bs\vartheta_M$ and the envelope function $\sup_{\bs\vartheta_{M}\in B(\delta,\bs\vartheta_M^*)}|\log g_M(\bs w;\bs\vartheta_M)|$ has finite expectation for $M\leq M_0$, these two probabilities converge to 0 by the Uniform Law of Large Numbers. Thus, we obtain
\begin{equation}
Q_n^M(\hat{\boldsymbol\vartheta}_{M}) - Q_n^{M_0}(\hat{\boldsymbol\vartheta}_{M_0})
\overset{p}{\longrightarrow} Q^M(\theta_M^*) - Q^{M_0}(\theta_{M_0}^*). \label{eq:conv_q}
\end{equation}

Therefore, from \eqref{eq:W}, \eqref{eq:conv_p}, and \eqref{eq:conv_q}, and in view of Assumptions \ref{assumption:Q_M}(f) and \ref{assumption: penalty}(c), we have
\[
\frac{p\ell_n^M - p\ell_n^{M_0}}{n}
\overset{p}{\longrightarrow}
Q^M(\theta_M^*) - Q^{M_0}(\theta_{M_0}^*)
< 0.
\]
Hence, $\Pr(p\ell_n^M > p\ell_n^{M_0}) \to 0$ for $M<M_0$, which proves that
\begin{equation}\label{eq:underestimate}
\lim_{n\rightarrow \infty} \Pr(\hat{M}_{PL} < M_0) = 0.
\end{equation}

We proceed to prove that $\lim_{n\rightarrow \infty} \Pr(\widehat{M}_{PL}>M_0)=0$ by showing that, for $M=M_0+1,...,\bar M$,
\[
\Pr(p\ell_n^M > p\ell_n^{M_0} )\rightarrow 0\quad\text{as $n\rightarrow \infty$}.
\]
By definition of $p\ell_n^M$, we have
\[
\Pr( p\ell_n^M > p\ell_n^{M_0} ) = \Pr\left(
\frac{ \ell_n^M(\hat{\boldsymbol\vartheta}_{M}) - \ell_n^{M_0}(\hat{\boldsymbol\vartheta}_{M_0})}{p_{n,M_0}}
+1- \frac{p_{n,k_M}}{p_{n,k_{M_0}}}>0\right).
\]
   For any $\epsilon>0$, $\lim_{n\rightarrow \infty} \Pr\left(
\frac{ \ell_n^M(\hat{\boldsymbol\vartheta}_{M}) - \ell_n^{M_0}(\hat{\boldsymbol\vartheta}_{M_0})}{p_{n,M_0}}
> \epsilon\right)=0$ because $p_{n,M_0}\rightarrow \infty$  and $\ell_n^M(\hat{\boldsymbol\vartheta}_{M}) - \ell_n^{M_0}(\hat{\boldsymbol\vartheta}_{M_0})=O_p(1)$ by  Assumptions \ref{assumption: penalty}(b) and   \ref{assumption:bound}, respectively. Therefore,
\[
\frac{ \ell_n^M(\hat{\boldsymbol\vartheta}_{M}) - \ell_n^{M_0}(\hat{\boldsymbol\vartheta}_{M_0})}{p_{n,M_0}}+1
- \frac{p_{n,k_M}}{p_{n,k_{M_0}}}\overset{p}{\rightarrow} 1- \lim_{n\rightarrow\infty } \frac{p_{n,k_M}}{p_{n,k_{M_0}}},
\]
which is strictly negative by Assumption \ref{assumption: penalty}(d). Thus, $\Pr(p\ell_n^M > p\ell_n^{M_0} )\rightarrow 0\quad\text{as $n\rightarrow \infty$}$, and $\lim_{n\rightarrow \infty} \Pr(\widehat{M}_{PL}>M_0)=0$ follows.

\end{proof}

 \begin{proof}[Proof of Proposition \ref{prop:bound}]
Because \( \ell_n^{M_0}(\hat{\bs\vartheta}_{M_0}) - \ell_n^{M_0}({\bs{\vartheta}}_{M_0}^*) = O_p(1) \) under standard regularity conditions in Assumption \ref{assumption:Q_M}, Assumption \ref{assumption:bound} holds when \( \ell_{n}^{M}(\hat{\bs\vartheta}_{M}) - \ell_n^{M_0}({\bs{\vartheta}}_{M_0}^*) = O_p(1) \) in view of $\ell_{n}^M(\hat{\bs\vartheta}_{M})-\ell_n^{M_0}(\hat{\bs\vartheta}_{M_0})= ( \ell_{n}^{M}(\hat{\bs\vartheta}_{M}) - \ell_n^{M_0}({\bs{\vartheta}}_{M_0}^*) )-(\ell_n^{M_0}(\hat{\bs\vartheta}_{M_0}) - \ell_n^{M_0}({\bs{\vartheta}}_{M_0}^*) )$.
Because a consistent MLE is in $ A_{\varepsilon_n}(\eta)$ defined in Appendix \ref{app:bic}, Propositions \ref{prop:consistency} and \ref{Ln_thm2} imply that $\ell_n(\hat{\bs\psi}_{M}^j,\bs\tau^j)-\ell_n(\bs\psi_{M}^{j*},\bs\tau^j)  =O_p(1)$, and the stated result holds.
 \end{proof}

\begin{proof}[Proof of Corollary \ref{cor:bic}]
It is straightforward to verify that the penalty function $p_{n,k}=\frac{k}{2}\log(n)$ satisfies Assumption \ref{assumption: penalty}. Thus, result (a) directly follows from Proposition \ref{prop:bound}.
For (b), Assumption \ref{assumption: penalty}(c) holds for AIC penalty. Then, repeating  the argument in the proof of Proposition \ref{proposition:bic} up to (\ref{eq:underestimate}) proves that $ p\ell^{M_0}_n(\hat{\bs{\vartheta}}_{M_0}) >\max\left\{p\ell_n^{1}(\widehat{\bs{\vartheta}}_{1}),...,p\ell_n^{M_0-1}(\widehat{\bs{\vartheta}}_{M_0-1})\right\}$ holds with probability approaching one as $n\rightarrow\infty$. This implies that
\begin{align}
\lim_{n\rightarrow\infty} \Pr\left(\widehat{M}_{PL}>M_0\right) & \geq \lim_{n\rightarrow\infty}\Pr\left( p\ell^{M_0+1}_n(\hat{\bs{\vartheta}}_{M_0+1}) >\max\left\{p\ell_n^{1}(\widehat{\bs{\vartheta}}_{1}),...,p\ell_n^{M_0}(\widehat{\bs{\vartheta}}_{M_0})\right\}\right)\nonumber\\
&=  \lim_{n\rightarrow\infty}\Pr\left( p\ell^{M_0+1}_n(\hat{\bs{\vartheta}}_{M_0+1}) > p\ell_n^{M_0}(\widehat{\bs{\vartheta}}_{M_0}) \right)\nonumber\\
&=\lim_{n\rightarrow\infty}\Pr\left(2\left[ \ell^{M_0+1}_n(\hat{\bs{\vartheta}}_{M_0+1}) - \ell_n^{M_0}(\widehat{\bs{\vartheta}}_{M_0})\right] > 2(k_{M_0+1}-k_{M_0})\right),
\label{eq:cor-bic}
\end{align}
where
 $k_{M_0}$ and $k_{M_0+1}$ are the number of parameters for $M_0$ and $(M_0+1)$ components model.   Since the term $2\left[\ell^{M_0+1}_n(\hat{\bs{\vartheta}}_{M_0+1}) - \ell_n^{M_0}(\widehat{\bs{\vartheta}}_{M_0})\right]$ is $O_p(1)$ and converges in distribution as specified in Proposition \ref{prop:tm0_distribution}, and because $k_{M_0}$ and $k_{M_0+1}$ are finite, it follows  from (\ref{eq:cor-bic}) that
$\lim_{n \to \infty} \Pr\left(\widehat{M}_{PL}>M_0\right) > 0,$
and part (b) follows.
\end{proof}

 \begin{proof}[Proof of Proposition \ref{prop:h02}]
The part  (a) can be proven similar to Proposition \ref{prop:consistency}, and thus omitted.

The proof for part (b) closely follows Theorem 2(b) in \citet{andrews01em}. For each $\bs\lambda \in \Theta_{\bs\lambda}(c_1)$, we approximate the log-likelihood function $\ell_n^2(\bs\theta_2, \bs\lambda, \alpha)$ around the true parameter values $(\bs\theta^*, 0)$ using the partial derivative with respect to $\bs\theta_2$ and the right partial derivative with respect to $\alpha$ to obtain
\begin{align}
2\{\ell_n^2(\bs\theta_2, \bs\lambda, \alpha)&- \ell_n^2(\bs\theta^*, \bs\lambda,0)\} =   \bs G_n(\bs\lambda)\t \bs{\mathcal{I}}_n(\bs\lambda) \bs G_n(\bs\lambda)\nonumber\\
& -  \left[ \sqrt{n}\bs{t}(\bs\theta_2,\alpha)- \bs G_n(\bs\lambda)\right]\t \bs{\mathcal{I}}_n  \left[ \sqrt{n}\bs{t}(\bs\theta_2,\alpha)- \bs G_n(\bs\lambda)\right] + R_n(\bs\theta_2, \bs\lambda, \alpha),\label{LR22}
\end{align}
where $R_n(\bs\theta_2, \bs\lambda, \alpha)$ is a remainder term, and $\bs {\mathcal{I}}_n(\bs\lambda)$, $ \bs G_n(\bs\lambda)$, and $\bs{t}(\bs\theta_2,\alpha)$ are defined as
\begin{align*}
\bs {\mathcal{I}}_n(\bs\lambda):=n^{-1} \sum_{i=1}^n {\bs s}(\bs W_i;\bs\lambda){\bs s}(\bs W_i;\bs\lambda)\t,\
\bs G_n(\bs\lambda):=\bs{\mathcal{I}}_n(\bs\lambda)^{-1}\bs  S_n(\bs\lambda),\ \text{and}\
\bs{t}(\bs\theta_2,\alpha) :=
\begin{pmatrix}
\bs\theta_2-\bs\theta^*\\
\alpha
\end{pmatrix}
\end{align*}
with   $\bs S_n(\bs\lambda)=n^{-1/2}\sum_{i=1}^n {\bs s}(\bs W_i;\bs\lambda)$ and ${\bs s}(\bs W_i;\bs\lambda)$ defined in (\ref{eq:s_1}).  $(\theta,\pi)$ and $(B_T,D \ell_T(\theta_0,\pi),\mathcal{J}_{T \pi},Z_{T \pi})$ in \citet{andrews01em} correspond to our $((\bs\theta_2\t,\alpha)\t,\bs\lambda)$ and $(n^{1/2},n^{1/2}\bs S_n(\bs\lambda),\bs{\mathcal{I}}_n(\bs\lambda),\bs G_n(\bs\lambda))$.

We prove the stated result by applying Theorem 2(b) of \citet{andrews01em} to (\ref{LR22}). $(\beta,\delta,\pi)$ and $(B_T,G_\pi,\mathcal{J}_\pi,Z_\pi,Z_{\beta\pi})$ in \citet[pp.\ 697-699]{andrews01em} correspond to our $(\alpha,\bs\theta_2,\bs\lambda)$ and $(n^{1/2},\bs S(\bs\lambda),\bs{\mathcal{I}}( \bs\lambda),\bs G( \bs\lambda), \bs{G}_{\alpha,\bs\theta_2}(\bs\lambda))$,  where $\bs G( \bs\lambda):=\bs{\mathcal{I}}( \bs\lambda)^{-1}\bs S( \bs\lambda)$, $\bs G_{\alpha,\bs\theta_1}( \bs\lambda):=\bs{\mathcal{I}}_{\alpha,\bs\theta_2}( \bs\lambda)^{-1}\bs S_{\alpha,\bs\theta_2}( \bs\lambda)$, and $\psi$ in \citet[pp.\ 697-699]{andrews01em} does not exist in our setting. The stated result then follows because $\bs s_{\bs\theta_2}(\bs w)$ is identical to the score of the one-component model and $\hat  { \lambda}_{\beta \pi}'(H \mathcal{J}_{*\pi}^{-1}H')^{-1}\hat{ \lambda}_{\beta \pi}$ in Theorem 2(b) of \citet{andrews01em} is distributed as $(\max\{0,\bs{\mathcal{I}}_{\alpha.\bs\theta_2}( \bs\lambda)^{-1/2} \bs S_{\alpha.\bs\theta_2}( \bs\lambda)\})^2$. We proceed to verify the assumptions of Theorem 2(b) of \citet{andrews01em} (hereafter, A-Assumptions $2^{2^*}$, 3-5, 7, and 8). A-Assumption $2^{2^*}$(a)(b) follow from our Assumption \ref{assumption:h02}(a)(b). A-Assumption $2^{2^*}$(c) holds because our Assumption \ref{assumption:h02}(c)(d) and the uniform law of large numbers imply that $\sup_{ \bs\lambda\in\Theta_{ \bs\lambda}(c_1)}|| \bs{\mathcal{I}}_n( \bs\lambda) - \bs{\mathcal{I}}( \bs\lambda)|| \rightarrow_p 0$ and $\bs{\mathcal{I}}( \bs\lambda)$ is continuous.

A-Assumption 3 follows from Theorem 10.2 of \citet{pollard90book} if (i) $\widetilde\Theta_{\bs\lambda}$ is totally bounded, (ii) the finite dimensional distributions of $\bs G_n(\cdot)$ converge to those of $\bs G(\cdot)$, and (iii) $\{\bs G_n(\cdot): n\geq 1\}$ is stochastically equicontinuous. Condition (i) follows from the compactness of $\widetilde\Theta_{\bs\lambda}$ while conditions (ii) follows from Assumption \ref{assumption:h02}(b)(c) and the multivariate CLT. Condition (iii) can be verified by our Assumption \ref{assumption:h02}(b)(c) and Theorem 2 of \citet{andrews94hdbk} because $\nabla_{(\bs\theta_2\t,\alpha)\t} \log g_2(\cdot; \bs\theta_2,\bs\lambda,\alpha)$ are Lipschitz functions indexed by a finite dimensional parameter $\bs\lambda$ by Assumption  \ref{assumption:h02}(b).

A-Assumption 4 follows from Lemma 1 of \citet{andrews01em} because, for each $ \bs\lambda\in \widetilde\Theta_{ \bs\lambda}(c_1)$, $(\tilde{\bs\theta}_2( \bs\lambda),\tilde \alpha( \bs\lambda))= \arg\max_{(\bs\theta_2,\alpha)\in\Theta_{\bs\theta_2}\times[0,1/2]} \ell_n^2(\bs\theta_2, \bs\lambda, \alpha)$ converges to $(\bs\theta_2^*,0)$ in probability from the standard consistency proof. A-Assumption 5 holds because (i) the set $[0,1]$ equals a nonnegative half-line locally around $0$, and (ii) $\Theta_{\bs\theta_2}-\bs\theta_2^*$ is locally equal to ${\mathbb{R}}^{\text{dim}(\bs\theta_2)}$. A-Assumption 7(a) is not relevant for our problem. A-Assumptions 7(b) and 8 follow from our proof of A-Assumption 5.   

\end{proof}

\begin{proof}[Proof of Proposition \ref{Ln_thm1}]
For  notational brevity, we drop the superscript $j$ from $\bs\psi_M^j$, $\bs\tau^j$, $\bs s_i^j$ and  $\bs{\mathcal{I}}^j$.
 By using the Taylor expansion of $2 \log(1+x) = 2x - x^2(1+o(1))$ for small $x$, we have uniformly $\bs\psi_M  \in \mathcal{N}_{c/\sqrt{n}}$,
\begin{equation}\label{ell_appn}
\ell_n(\bs\psi_M,\bs\tau) - \ell_n(\bs\psi_M^*,\bs\tau) = 2 \sum_{k=1}^n \log(1+h_{\bs\psi_M\bs\tau,i}) = n P_n(2h_{\bs\psi_M\bs\tau,i} - [1+o_p(1)] h_{\bs\psi_M\bs\tau,i}^2),
\end{equation}
where $h_{\bs\psi_M\bs\tau,i} := \sqrt{l_{\bs\psi_M\bs\tau,i}}-1$.
The stated result holds if we show that
\begin{align}
&  \sup_{ \bs\psi_M \in \mathcal{N}_{c/\sqrt{n}}} \left| nP_n(h_{\bs\psi_M\bs\tau,i}^2) - n \bs t^{\top}_{\bs\psi_M\bs\tau} \bs{\mathcal{I}} \bs t_{\bs\psi_M\bs\tau}/4 \right| = o_p(1)\quad\text{and}\label{hk_appn}\\
& \sup_{ \bs\psi_M \in \mathcal{N}_{c/\sqrt{n}}} | nP_n(h_{\bs\psi_M\bs\tau,i}) - \sqrt{n} \bs t^{\top}_{\bs\psi_M\bs\tau} \nu_n(\bs s_i)/2 + n\bs t_{\bs\psi_M\bs\tau} \bs{\mathcal{I}} \bs t^{\top}_{\bs\psi_M\bs\tau}/8 | = o_p(1),\label{hk_appn2}
\end{align}
because then the right-hand side of (\ref{ell_appn}) is equal to $\sqrt{n} \bs t^{\top}_{\bs\psi_M\bs\tau} \nu_n (\bs s_i) - \bs t_{\bs\psi_M\bs\tau} \bs{\mathcal{I}} \bs t^{\top}_{\bs\psi_M\bs\tau}/2 + o_p(1)$ uniformly in   $ \bs\psi_M \in \mathcal{N}_{c/\sqrt{n}}$.

We first show (\ref{hk_appn}). Let \[
m_{ \bs\psi_M \bs\tau i} :=l_{\bs\psi_M\bs\tau,i} -1= \bs t^{\top}_{\bs\psi_M\bs\tau} \bs s_i +  r_{ \bs\psi_M \bs\tau,i}. \] Observe that
\begin{equation} \label{B2}
\max_{1\leq i \leq n} \sup_{ \bs\psi_M \in \mathcal{N}_{c/\sqrt{n}}} |m_{ \bs\psi_M \bs\tau i}| = \max_{1\leq i \leq n} \sup_{ \bs\psi_M \in \mathcal{N}_{c/\sqrt{n}}} |\bs t^{\top}_{\bs\psi_M\bs\tau} \bs s_i +  r_{ \bs\psi_M \bs\tau,i}| = o_p(1),
\end{equation}
from Assumptions \ref{assumption:BIC}(a) and (c) and Lemma \ref{max_bound}. Write $4P_n ( h_{\bs\psi_M\bs\tau,i}^2)$ as
\begin{equation}\label{B0}
4P_n ( h_{\bs\psi_M\bs\tau,i}^2) = P_n \left(\frac{4(l_{\bs\psi_M\bs\tau,i}-1)^2}{(\sqrt{l_{\bs\psi_M\bs\tau,i}} + 1)^2}\right) = P_n(l_{\bs\psi_M\bs\tau,i}-1)^2 - P_n \left( (l_{\bs\psi_M\bs\tau,i}-1)^3 \frac{(\sqrt{l_{\bs\psi_M\bs\tau,i}}+3)}{(\sqrt{l_{\bs\psi_M\bs\tau,i}} + 1)^3}\right).
\end{equation}
From (\ref{density-ratio}),
\begin{equation} \label{lk_lln}
\begin{aligned}
P_n(l_{\bs\psi_M\bs\tau,i}-1)^2 
&= \bs t^{\top}_{\bs\psi_M\bs\tau}P_n(\bs s_i\bs s_i')\bs t_{\bs\psi_M\bs\tau} +   \zeta_{ \bs\psi_M \bs\tau n  },
\end{aligned}
\end{equation}
where $ \zeta_{ \bs\psi_M \bs\tau n  }:=2 \bs t^{\top}_{\bs\psi_M\bs\tau} P_n[\bs s_i  r_{ \bs\psi_M \bs\tau,i}  ] + P_n( r_{ \bs\psi_M \bs\tau,i}  )^2$. It follows from Assumptions \ref{assumption:BIC}(a)(b)(c) and $(E|XY|)^2 \leq E|X|^2E|Y|^2$ that, uniformly in $ \bs\psi_M \in \mathcal{N}_\varepsilon$,
\begin{equation}\label{zeta}
 |  \zeta_{ \bs\psi_M \bs\tau n  }|=O_p(||\bs t_{\bs\psi_M\bs\tau}||^2||\bs\psi_M-\bs\psi_M^*||)+O_p(n^{-1}||\bs t_{\bs\psi_M\bs\tau}||||\bs\psi_M-\bs\psi_M^*||) + O_p(n^{-1}||\bs\psi_M-\bs\psi_M^*||^2).
 \end{equation} Then, (\ref{hk_appn}) holds because $ || P_n(\bs s_i\bs s_i') - \bs{\mathcal{I}}|| = o_p(1)$ and the second term on the right-hand side of (\ref{B0}) is bounded by
\begin{align*}
& \mathcal{C}  \sup_{ \bs\psi_M \in \mathcal{N}_{c/\sqrt{n}}} P_n \left[ |m_{ \bs\psi_M \bs\tau i}|^3   \right]    \leq o_p(1)  \sup_{ \bs\psi_M \in \mathcal{N}_{c/\sqrt{n}}} P_n \left[ m_{ \bs\psi_M \bs\tau i}^2  \right]  = o_p(n^{-1}),
\end{align*}
from $\sup_{\bs\psi_M \in \mathcal{N}_{c/\sqrt{n}}}\frac{(\sqrt{l_{\bs\psi_M\bs\tau,i}}+3)}{(\sqrt{l_{\bs\psi_M\bs\tau,i}} + 1)^3}\leq  \mathcal{C}$,  (\ref{B2}), and  $P_n( m_{ \bs\psi_M \bs\tau i}^2) = \bs t^{\top}_{\bs\psi_M\bs\tau}\bs{\mathcal{I}} \bs t_{\bs\psi_M\bs\tau} + o_p(||\bs t_{\bs\psi_M\bs\tau}||^2)$.

We proceed to show (\ref{hk_appn2}). Consider the following expansion of $h_{\bs\psi_M\bs\tau,i}$:
\begin{equation} \label{hk_1}
h_{\bs\psi_M\bs\tau,i} = (l_{\bs\psi_M\bs\tau,i}-1)/2 - h_{\bs\psi_M\bs\tau,i}^2/2 = (\bs t^{\top}_{\bs\psi_M\bs\tau} \bs s_i +  r_{ \bs\psi_M \bs\tau,i} )/2 - h_{\bs\psi_M\bs\tau,i}^2/2.
\end{equation}
Then, (\ref{hk_appn2}) follows from (\ref{hk_appn}), (\ref{hk_1}), and Assumptions \ref{assumption:BIC}(d), and the stated result follows.
\end{proof}

\begin{proof}[Proof of Proposition \ref{Ln_thm2}]
For  brevity, we drop the superscript $M$ from $\bs s_i^M$ and  $\bs{\mathcal{I}}^M$.
Define $h_{\bs\psi_M\bs\alpha,i} := \sqrt{l_{\bs\psi_M\bs\alpha,i}}-1$.
For part (a), it follows from $\log(1+x) \leq x$ and $h_{\bs\psi_M\bs\alpha,i} = (l_{\bs\psi_M\bs\alpha,i}-1)/2 - h_{\bs\psi_M\bs\alpha,i}^2/2$ (see (\ref{hk_1})) that
\begin{equation} \label{Ln_ineq}
\ell_n(\bs\psi_M,\bs\alpha)-\ell_n(\bs\psi_M^*,\bs\alpha) = 2 \sum_{k=1}^n \log(1+h_{\bs\psi_M\bs\alpha,i}) \leq 2 n P_n(h_{\bs\psi_M\bs\alpha,i}) = \sqrt{n} \nu_n(l_{\bs\psi_M\bs\alpha,i}-1) - n P_n(h_{\bs\psi_M\bs\alpha,i}^2).
\end{equation}
Observe that $h_{\bs\psi_M\bs\alpha,i}^2 =(l_{\bs\psi_M\bs\alpha,i}-1)^2/(\sqrt{l_{\bs\psi_M\bs\alpha,i}} + 1)^2 \geq \mathbb{I}\{l_{\bs\psi_M\bs\alpha,i} \leq \kappa \} (l_{\bs\psi_M\bs\alpha,i}-1)^2/ (\sqrt{\kappa}+1)^2$ for any $\kappa>0$. Therefore, \begin{equation}\label{pn_lower}
P_n (h_{\bs\psi_M\bs\alpha,i}^2)\geq (\sqrt{\kappa}+1)^{-2}P_n \left( \mathbb{I}\{l_{\bs\psi_M\bs\alpha,i} \leq \kappa\} (l_{\bs\psi_M\bs\alpha,i}-1)^2\right).
\end{equation}
In view of (\ref{lk_lln}), (\ref{pn_lower}) is rewritten as
\begin{align} \label{pn_lower_2}
P_n (h_{\bs\psi_M\bs\alpha,i}^2)& \geq (\sqrt{\kappa}+1)^{-2} \left\{\bs t^{\top}_{\bs\psi_M\bs\alpha} \left[ P_n(\bs s_{i}\bs s_{i}') - P_n(\mathbb{I}\{l_{\bs\psi_M\bs\alpha,i} > \kappa\}\bs s_{i}\bs s_{i}') \right] \bs t_{\bs\psi_M\bs\alpha} + \tilde \zeta_{ \bs\psi_M \bs\alpha n  }\right\},
\end{align}
where $\tilde \zeta_{ \bs\psi_M \bs\alpha n  }$ satisfies the bound similar to  (\ref{zeta}).
From H\"older's inequality, we have $P_n(\mathbb{I}\{l_{\bs\psi_M\bs\alpha,i} > \kappa\}||\bs s_{i}||^2 ) \leq [P_n(\mathbb{I}\{l_{\bs\psi_M\bs\alpha,i} > \kappa\} )]^{\delta/(2+\delta)} [ P_n(||\bs s_{i}||^{2+\delta}) ]^{2/(2+\delta)}$, where the right-hand side is no larger than $\kappa^{-\delta/(2+\delta)} O_p(1)$ uniformly in  $\bs\psi_M \in \mathcal{N}_{\varepsilon_n}$ because (i) it follows from $\kappa \mathbb{I}\{l_{\bs\psi_M\bs\alpha,i} > \kappa\} \leq l_{\bs\psi_M\bs\alpha,i}$ that $P_n(\mathbb{I}\{l_{\bs\psi_M\bs\alpha,i} > \kappa\} ) \leq \kappa^{-1} P_n(l_{\bs\psi_M\bs\alpha,i})$ and $  \sup_{ \bs\psi_M \in \mathcal{N}_{\varepsilon_n}}|P_n(l_{\bs\psi_M\bs\alpha,i})-1| = o_p(1)$ from Assumptions \ref{assumption:BIC}(d)(e), and (ii) $P_n( ||\bs s_{i}||^{2+\delta})=O_p(1)$ from Assumption \ref{assumption:BIC}(a). Consequently, $\mathbb{P}(   \sup_{ \bs\psi_M \in \mathcal{N}_{\varepsilon_n}}P_n(\mathbb{I}\{l_{\bs\psi_M\bs\alpha,i} > \kappa\}||\bs s_{i}||^2 ) \geq \lambda_{\min}/4 ) \to 0$ as $\kappa \to \infty$, where $\lambda_{\min}$ is  the smallest eigenvalue of $\bs{\mathcal{I}}$. Hence,  we can write (\ref{pn_lower_2}) as $P_n (h_{\bs\psi_M\bs\alpha,i}^2) \geq \frac{\tau}{2} (1+o_p(1)) \bs t^{\top}_{\bs\psi_M\bs\alpha} \bs{\mathcal{I}} \bs t_{\bs\psi_M\bs\alpha} + \frac{\tau}{2}  \tilde\zeta_{ \bs\psi_M \bs\alpha n  }$ for $\tau:=2(\sqrt{\kappa}+1)^{-2}>0$ by taking $\kappa$ sufficiently large. Because $\sqrt{n} \nu_n(l_{\bs\psi_M\bs\alpha,i}-1) = \sqrt{n} \bs t^{\top}_{\bs\psi_M\bs\alpha}\nu_n (\bs s_{i}) + O_p(\sqrt{n}||{\bs t}_{\bs\psi_M\bs\alpha}||||\bs\psi_M-\bs\psi_M^*||)$ from Assumptions \ref{assumption:BIC}(d) and because $\tilde \zeta_{ \bs\psi_M \bs\alpha n  }$ satisfies the bound analogous to (\ref{zeta}), it follows from (\ref{Ln_ineq}) that, uniformly in $ \bs\psi_M \in \mathcal{N}_{\varepsilon_n}$,
\begin{align}
-\eta & \leq \ell_n(\bs\psi_M,\bs\alpha)-\ell_n(\bs\psi_M^*,\bs\alpha) \nonumber\\
&\leq  \sqrt{n} \bs t^{\top}_{\bs\psi_M\bs\alpha} \nu_n (\bs s_{i}) - \frac{\tau}{2} (1+o_p(1)) n \bs t^{\top}_{\bs\psi_M\bs\alpha} \bs{\mathcal{I}} \bs t_{\bs\psi_M\bs\alpha} + O_p(\sqrt{n}||{\bs t}_{\bs\psi_M\bs\alpha}||||\bs\psi_M-\bs\psi_M^*||) \nonumber\\
&\quad +O_p(n||\bs t_{\bs\psi_M\bs\alpha}||^2||\bs\psi_M-\bs\psi_M^*||)+O_p(||\bs t_{\bs\psi_M\bs\alpha}||||\bs\psi_M-\bs\psi_M^*||) + O_p(||\bs\psi_M-\bs\psi_M^*||^2).\label{rk_lower2}
\end{align}
Let $\bs T_n:= \bs{\mathcal{I}}^{1/2}\sqrt{n} \bs t_{\bs\psi_M\bs\alpha}$. From (\ref{rk_lower2}), Assumptions \ref{assumption:BIC}(b) and (e), and the fact $||\bs\psi_M-\bs\psi_M^*|| \to 0$ if $\bs t_{\bs\psi_M\bs\alpha} \to 0$, we obtain
\[
-\eta \leq \ell_n(\bs\psi_M,\bs\alpha)-\ell_n(\bs\psi_M^*,\bs\alpha) \leq
||\bs T_n|| O_p(1)-\frac{\tau}{2}||T_n||^2+o_p(||T_n||^2)+o_p(||T_n||) +o_p(1),
\]
uniformly in $ \bs\psi_M \in \mathcal{N}_{\varepsilon_n}$. By  rearranging the inequality  and choosing a sufficiently large constant $C$, we have $ ||\bs T_n|| C - (\tau/2) ||\bs T_n||^2 + C  \geq 0$ with an arbitrarily high probability. Namely,  for any $\delta>0$, there exist $C$, $n_0<\infty$ such that
\begin{equation}
\mathbb{P}\left( \inf_{ \bs\psi_M \in \mathcal{N}_{\varepsilon_n}} | |\bs T_n|| C - (\tau/2) ||\bs T_n||^2 + C  \geq 0 \right) \geq 1-\delta,\ \text{ for all }\ n > n_0.
\end{equation}
Rearranging the terms inside $\mathbb{P}(\cdot)$ gives $  \sup_{ \bs\psi_M \in \mathcal{N}_{\varepsilon_n}} (||\bs T_n||-(C/\tau))^2 \leq 2B/\tau+(C/\tau)^2$. Taking its square root gives $\mathbb{P}(  \sup_{ \bs\psi_M \in \mathcal{N}_{\varepsilon_n}}||\bs T_n|| \leq C_1) \geq 1-\delta$ for a constant $C_1$, and part (a) follows. Part (b) follows from part (a) and Proposition \ref{Ln_thm1}.
\end{proof}

 \subsection{Lemmas}

\begin{lemma}\label{lemma:unbounded_likelihood}
For any $M<\infty$, $ \Pr\left( -\log n + \ell(\bs W_{i^*};\bar Y_{i^*},s_{i^*}^2) < M\right) =\exp(-C_T n^{1/T})$  as $n\rightarrow \infty$, where $C_T$ for $T=2,3,..$ are some positive constant that depends on $M$ and $T$.\end{lemma}
\begin{proof}[Proof of Lemma \ref{lemma:unbounded_likelihood}]
Because $\sum_{t=1}^T \frac{\left(Y_{it}-\mu \right)^2}{s_{i^*}^2}= T-1$ when $i=i^*$, we have
\begin{align}
-\log n+\ell(\bs W_{i^*};\bar Y_{i^*},s_{i^*}^2)
&= -\log n - \frac{T}{2} \log s_{i^*}^2 -\frac{T}{2}\log(2\pi) - \frac{T-1}{2} \nonumber\\
&=-\log\left( B n  \left(\chi_{i^*,T-1}^2 \right)^{T/2} \right)
\end{align}
with  $B:=\left(\frac{2\pi\sigma^{*2}}{T-1}\right)^{T/2}\exp(\frac{T-1}{2})>0$ and
\[
\chi_{i^*,T-1}^2 := \frac{(T-1)s_{i^*}^2}{\sigma^{*2}}.
\]

Therefore, to prove the stated result, it suffices to show that for any $\epsilon=\exp(-M)/B>0$, $\Pr\left(n \left(\chi_{i^*,T-1}^2\right)^{T/2} >\epsilon\right)=\Pr\left(\chi_{i^*,T-1}^2 >\left({\epsilon}/{n}\right)^{2/T}\right)\rightarrow 0$ as $n\rightarrow\infty$. Given the property of the first-order statistic, the distribution  function of $\chi_{i^*,T-1}^2$ is given by $1-[1-F_{{T-1}}(t)]^n$, where $F_{{k}}(t)$ is the cumulative distribution function for chi-squared variables with  $k$ degree of freedom. It follows that
\[
\Pr\left(\chi_{i^*,T-1}^2 >\left(\frac{\epsilon}{n}\right)^{2/T}  \right) =\left[1-F_{{T-1}}\left( \left(\frac{\epsilon}{n}\right)^{2/T}  \right)\right]^n.
\]

For any finite $T\geq 2$,    write
\begin{equation}\label{eq-lemma1}
\Big[1-F_{{T-1}}\big( (\epsilon/n)^{2/T} \big) \Big]^n =  \Big\{ \Big[1-F_{{T-1}}\big( (\epsilon/n)^{2/T} \big) \Big]^{\frac{1}{  F_{{T-1}}( (\epsilon/n)^{2/ T} )} }\Big\}^{ n F_{{T-1}}( (\epsilon/n)^{2/T} )  }.
\end{equation}
Then, because  $(1 - F)^{\frac{1}{F}} \to \frac{1}{e}$ when $ F \to 0 $, the stated result follows from (\ref{eq-lemma1}) if we can show
\begin{equation}\label{eq-lemma1-2}
\frac{  F_{{T-1}}( (\epsilon x)^{2/T} )}{x} =O(x^{-1/T})\ \text{as}\ x\rightarrow 0
\end{equation}  for $x=1/n$, where (\ref{eq-lemma1})-(\ref{eq-lemma1-2}) imply that $\log \Big[1-F_{{T-1}}\big( (\epsilon/n)^{2/T} \big) \Big]^n\leq - C_T n^{1/T}$ for some positive constant $C_T$ when $n$ is sufficiently large.

For small \( t \), the CDF \( F_k(t) \) can be expanded as $F_k(t) = \frac{t^{k/2}}{2^{k/2} \Gamma(k/2 + 1)} + o(t^{k/2}) \text{ as } t \to 0$. Then, letting   $k=T-1$ and \( t =  (\epsilon x)^{2/T}=( \epsilon x)^{2/(k+1)} \), and simplifying the exponent $(x^{2/(k+1)})^{k/2} = x^{k/(k+1)}$,
 we have
\[F_{{T-1}}( (\epsilon x)^{2/T} )=  \frac{( \epsilon x)^{k/(k+1)} }{2^{k/2} \Gamma(k/2 + 1)} + o((\epsilon x)^{k/(k+1)}) \text{ as } x \to 0. \]
Therefore, given finite $\epsilon>0$, noting that $x^{k/(k+1) - 1} =   x^{-1/(k + 1)}$, we have
\[ \frac{F_{{T-1}}( (\epsilon x)^{2/T} )}{x} = \frac{(\epsilon x)^{-1/(k + 1)}}{2^{k/2} \Gamma(k/2 + 1)}  + o(x^{-1/(k + 1)})= O(x^{-1/(k + 1)}) \text{ as } x \to 0,\]
and the equation (\ref{eq-lemma1-2}) follows with $k+1=T$.

\end{proof}

\begin{lemma}\label{lemma:bic}
For any $M<\infty$, $ \Pr\left( -4\log n + \ell(\bs W_{i^*};\bar Y_{i^*},s_{i^*}^2) < M\right) \rightarrow 1$ as $n\rightarrow \infty$.

\end{lemma}

 \begin{proof}
Following the proof of Lemma \ref{lemma:unbounded_likelihood}, because $-4\log n+\ell(\bs W_{i^*};\bar Y_{i^*},s_{i^*}^2) =-\log\left( B n^4  \left(\chi_{i^*,T-1}^2 \right)^{T/2} \right)$, it suffices to show that, for any $\epsilon>0$, $\Pr\left(\chi_{i^*,T-1}^2 >\left({\epsilon}/{n^4}\right)^{2/T}\right)=\Big[1-F_{{T-1}}\big( (\epsilon/n^4)^{2/T} \big) \Big]^n\rightarrow 1$ as $n\rightarrow\infty$. For any finite $T\geq 2$,    write
\begin{equation}\label{eq-lemma-bic}
\Big[1-F_{{T-1}}\big( (\epsilon/n^4)^{2/T} \big) \Big]^n =  \Big\{ \Big[1-F_{{T-1}}\big( (\epsilon/n^4)^{2/T} \big) \Big]^{\frac{1}{  F_{{T-1}}( (\epsilon/n^4)^{2/ T} )} }\Big\}^{ n F_{{T-1}}( (\epsilon/n^4)^{2/T} )  }.
\end{equation}
Then, because  $(1 - F)^{\frac{1}{F}} \to \frac{1}{e}$ when $ F \to 0 $, the stated result follows from (\ref{eq-lemma-bic}) if we can show  $\frac{  F_{{T-1}}( (\epsilon x^4)^{2/T} )}{x} \to 0\text{ as }  x \to 0$ for $x=1/n$. By applying L'Hôpital's rule, we have
\[\begin{split}
\lim_{x \to 0} \frac{  F_{{T-1}}( (\epsilon x^4 )^{2/T} )}{x } = \lim_{x \to 0}  f_{T-1}( (\epsilon x^4 )^{2/T} ) \epsilon^{2/T} (8/T) x^{8/T-1}
= \lim_{x \to 0} \tilde C_{T,\epsilon} e^{- ((\epsilon x^4 )^{2/T} )/2} x^{ 3-\frac{4}{T}} =0,
\end{split}\]
where $\tilde C_{T,\epsilon} = \frac{\epsilon^{(T-1)/T}(8/T)}{2^{(T-1)/2} \Gamma((T-1)/2)} $ because $e^{- ((\epsilon x^4 )^{2/T} )/2}\to 1$  and $x^{ 3-\frac{4}{T}} \to 0$ as $x \to 0$ for any finite $T\geq 2$.
 \end{proof}

\begin{lemma}\label{lemma:expansion}
Suppose that Assumptions  \ref{assumption:K}, \ref{assumption:consistency}, and \ref{assumption:LRT1} hold. Then, under $H_0: M=1$, for $\alpha\in [c_1,1-c_1]$, (a) for any $\delta>0$, $\lim\sup_{n\rightarrow \infty} \Pr(\sup_{\bs\psi \in \Theta_{\bs\psi}: ||\bs\psi-\bs\psi^*||\leq \kappa} |R_n(\bs\psi,\alpha)| > \delta (1+||n\bs t(\bs\psi,\alpha)||^2)) \rightarrow 0$ as $\kappa\rightarrow 0$, (b) $ \bs S_n \overset{d}{\to} \bs S\sim N(0,\bs{\mathcal{I}})$, and (c) $\bs{\mathcal{I}}_n \overset{p}{\to} \bs{\mathcal{I}}$, where $\bs{\mathcal{I}}$ is finite and non-singular, and (d) $\bs t(\hat{\bs{\psi}}_\alpha,\alpha)=O_p(n^{-1/2})$ for any $\alpha\in [c_1,1-c_1]$.
\end{lemma}

\begin{proof}[Proof of Lemma \ref{lemma:expansion}]
The proof follows that of Proposition 2 in \cite{Kasahara2012}.
For a vector $\bs{x}$ and a function $f(\bs{x})$, let $\nabla_{\bs{x}^k}f(\bs{x})$ denote its $k$-th derivative with respect to $\bs{x}$, which can be a multidimensional array. Observe that for any finite $k$ and for a neighborhood $\mathcal{N}$ of $\bs{\psi}^*$, we obtain
\begin{equation}
\begin{aligned}
&E|| \nabla_{\bs{\psi}^k} g(\bs{W}_i;\bs{\psi}^*,\alpha)/ g(\bs{W}_i;\bs{\psi}^*,\alpha)||^2<\infty, \\
&E||\sup_{\bs{\psi}\in\Theta_{\bs{\psi}} \cap \mathcal{N}}\nabla_{\bs{\psi}^k} \log g(\bs{W}_i;\bs{\psi},\alpha)||^2<\infty
\end{aligned} \label{nabla_0}
\end{equation}
because each element of $\nabla_{\bs{\psi}^k} \log g(y|\bs{x},\bs{z};\bs{\psi},\alpha)$ is written as a sum of products of Hermite polynomials.
 Note also that  the following holds:
\begin{align}
& \nabla_{\eta \lambda_j }L_n(\bs\psi^*,\alpha)=0, \quad \nabla_{\lambda_i \lambda_j\lambda_k } L_n(\bs\psi^*,\alpha)=O_p(n^{1/2}), \label{nabla_1} \\
& \nabla_{\eta\eta \lambda_i }L_n(\bs\psi^*,\alpha) =O_p(n), \quad \nabla_{\eta\eta\eta}L_n(\bs\psi^*,\alpha)=O_p(n), \label{nabla_2}
\end{align}
where equation (\ref{nabla_1}) follows from Proposition \ref{lemma:KS2018_lemma7}{(a) and (c)} and (\ref{nabla_0}) and equation (\ref{nabla_2}) is a simple consequence of (\ref{nabla_0}).
Furthermore, for a neighborhood $\mathcal{N}$ of $\bs\psi^*$,
\begin{align}
& \sup_{\bs\psi\in\Theta_{\bs\psi} \cap \mathcal{N}} \left| n^{-1}\nabla^{(4)} L_n(\bs\psi ,\alpha)- E\nabla^{(4)}  \log g(\bs W_i;\bs\psi ,\alpha) \right| = o_p(1), \label{nabla_3}\\
& E\nabla^{(4)} g(\bs W_i;\bs\psi,\alpha) \text{ is continuous in }\psi  {\in \Theta_{\bs\psi} \cap \mathcal{N}}. \label{nabla_4}
\end{align}
  Equations (\ref{nabla_3}) and (\ref{nabla_4})
follow  from Lemma 2.4 of \cite{Newey1994} and the fact that $\nabla_{\bs{\psi}^k} \log g(\bs{w};\bs{\psi},\alpha)$ is written as a sum of products of Hermite polynomials.

Taking a fourth-order Taylor expansion of $L_n(\bs{\psi},\alpha)$ around $\bs{\psi}^*$ and using (\ref{nabla_0}) and (\ref{nabla_1}), we can write $L_n(\bs{\psi},{\alpha}) -L_n(\bs{\psi}^*,\alpha)$ as the sum of the relevant terms and the remainder term as follows:
\begin{align}
\lefteqn{ L_n(\bs{\psi},{\alpha}) -L_n(\bs{\psi}^*,\alpha) = } \nonumber\\
& \quad \nabla_{\bs{\eta}}L_n^* (\bs{\eta} - \bs{\eta}^*) + \frac{1}{2!} (\bs{\eta} - \bs{\eta}^*)^{\top}\nabla_{\bs{\eta} \bs{\eta}^{\top}}L_n^*(\bs{\eta} - \bs{\eta}^*) +\frac{1}{2!}  \sum_{i=1}^{q}\sum_{j=1}^{q} \nabla_{\lambda_i\lambda_j }L_n^*\lambda_i\lambda_j \label{LR0-score0} \\
&\quad + \frac{3}{3!} \sum_{i=1}^{q} \sum_{j=1}^{q} (\bs{\eta} - \bs{\eta}^*)\t\nabla_{\bs\eta \lambda_i\lambda_j} L_n^* \lambda_i\lambda_j    \label{LR0-hessian0} \\
 &\quad + \frac{1}{4!}  \sum_{i=1}^{q}\sum_{j=1}^{q}\sum_{k=1}^{q} \sum_{\ell=1}^{q} \nabla_{\lambda_i\lambda_j\lambda_k\lambda_\ell}L_n^*\lambda_i\lambda_j\lambda_k\lambda_\ell + R_{n}(\bs{\psi},\alpha), \label{LR0-hessian1}   \end{align}
where $\nabla L_n^*$ denotes the derivative of $L_n(\bs{\psi},\alpha)$ evaluated at $(\bs{\psi}^*,\alpha)$. In view of (\ref{nabla_1}) and (\ref{nabla_2}), the remainder term is written as
\begin{align}
\lefteqn{ R_n(\bs\psi,\alpha) =  O_p(n^{1/2}) \sum_{i=1}^{q}\sum_{j=1}^{q}\sum_{k=1}^{q} \lambda_i\lambda_j \lambda_k + O_p(n) \left(  \sum_{i=1}^{q} ||\bs{\eta}-\bs{\eta}^*||^2 \lambda_i +  ||\bs{\eta}-\bs{\eta}^*||^3 \right)} \label{Rn_2} \\
& +O_p(n)\sum_{i=1}^{q}\sum_{j=1}^{q}\sum_{k=1}^{q}\left( ||\bs{\eta}-\bs{\eta}^*||^4 + ||\bs{\eta}-\bs{\eta}^*||^3 |\lambda_i|+ ||\bs{\eta}-\bs{\eta}^*||^2 |\lambda_i \lambda_j| + ||\bs{\eta}-\bs{\eta}^*|| |\lambda_i \lambda_j \lambda_k| \right) \qquad \label{Rn_3} \\
& + \frac{1}{4!}  \sum_{i=1}^{q}\sum_{j=1}^{q}\sum_{k=1}^{q}\sum_{\ell=1}^{q}  \{\nabla_{\lambda_i \lambda_j \lambda_k \lambda_\ell} L_n(\bs{\psi}^\dag,\alpha) - \nabla_{\lambda_i \lambda_j \lambda_k \lambda_\ell} L_n(\bs{\psi}^*,\alpha) \} \lambda_i \lambda_j \lambda_k \lambda_\ell
\label{Rn_4}
\end{align}
with $\bs{\psi}^\dag$ being between $\bs{\psi}$ and $\bs{\psi}^*$. Because $||\sqrt{n}t(\bs{\psi},\alpha)||^2 = n||\bs{\eta} - \bs{\eta}^*||^2 + n \sum_{i=1}^{q}\sum_{j=1}^i \alpha^2(1-\alpha)^2 |\lambda_i\lambda_j|^2$, the right-hand side of (\ref{Rn_2}) and the terms in (\ref{Rn_3}) are bounded by $O_p(1)(||\sqrt{n}t(\bs{\psi},\alpha)||+||\sqrt{n}t(\bs{\psi},\alpha)||^2)(||\bs{\eta}-\bs\eta^*||+||\lambda||)$. In view of (\ref{nabla_3}) and (\ref{nabla_4}), (\ref{Rn_4}) is bounded by  $||\sqrt{n}t(\bs{\psi},\alpha)||^2[d(\bs{\psi}^\dagger)+o_p(1)]$ with $d(\bs{\psi}^\dagger)\rightarrow 0$ as $\bs{\psi}^\dagger \rightarrow \bs{\psi}^*$, where a function $d(\bs{\psi}^\dagger)$ corresponds to $n^{-1}\E[\nabla_{\lambda_i \lambda_j \lambda_k \lambda_\ell}L_n(\bs{\psi}^\dagger,\alpha) - \nabla_{\lambda_i \lambda_j \lambda_k \lambda_\ell}L_n(\bs{\psi}^*,\alpha)]$. Therefore, $R_n(\bs{\psi},\alpha)=(1+||\sqrt{n}t(\bs{\psi},\alpha)||)^2[d(\bs{\psi}^\dagger)+o_p(1)  + O_p(||\bs{\psi}-\bs{\psi}^*||)]$, and part (a) follows.

Part (b) follows from  Lemma \ref{lemma:KS2018_lemma7}(c) and (d), the Lindeberg--Levy central limit theorem, and the finiteness of $\bs{\mathcal{I}}$ in part (c).

For part (c), we first provide the formula of $\bs{\mathcal{I}}_n$. Partition $\bs{\mathcal{I}}_n$ as
\[
\bs{\mathcal{I}}_n = \left( \begin{array}{cc}
\bs{\mathcal{I}}_{\bs\eta n} & \bs{\mathcal{I}}_{\bs\eta\bs\lambda n}\\
\bs{\mathcal{I}}_{\bs\eta\bs\lambda n}\t & \bs{\mathcal{I}}_{\bs\lambda n}
\end{array}\right), \quad \bs{\mathcal{I}}_{\bs\eta n}: q\times q, \quad \bs{\mathcal{I}}_{\bs\eta\bs\lambda n}: q\times q_\lambda, \quad \bs{\mathcal{I}}_{\bs\lambda n}: q_\lambda \times q_\lambda,
\]
where $q_\lambda=q(q+1)/2$.
$\bs{\mathcal{I}}_{\bs\eta n}$ is given by $\bs{\mathcal{I}}_{\bs\eta n} = - n^{-1}\nabla_{\bs\eta\bs\eta\t}L_n(\psi^*,\alpha)$. For $\bs{\mathcal{I}}_{\bs\eta\bs\lambda n}$, let $A_{ij} = n^{-1}\nabla_{\bs\eta \lambda_i \lambda_j} L_n(\psi^*,\alpha)$  and write the  term in (\ref{LR0-hessian0})   as   $(n/2) \sum_{i=1}^{q} \sum_{j=1}^{q} (\bs\eta -\bs\eta^*)\t A_{ij} \lambda_i \lambda_j= n \sum_{i=1}^{q} \sum_{j=1}^i (\bs\eta -\bs\eta^*)\t A_{ij} c_{ij}  \lambda_i \lambda_j$, where $c_{ij}=1/2$ if $i=j$, and $c_{ij}=1$ if $i\neq j$. Then, by defining  $\bs{\mathcal{I}}_{\bs\eta\bs\lambda n} = -(A_{11}, \ldots,  A_{qq},  A_{12}, \ldots,  A_{q-1,q}) /\alpha(1-\alpha)$, the  term in (\ref{LR0-hessian0}) equals $- n(\bs\eta -\bs\eta^*)\t \bs{\mathcal{I}}_{\bs\eta\bs\lambda n}[\alpha(1-\alpha)v(\bs\lambda)]$. For $\bs{\mathcal{I}}_{\bs\lambda n}$, define $B_{ijk\ell} = n^{-1}(8/4!)\nabla_{\lambda_i \lambda_j\lambda_k \lambda_\ell} L_n(\bs\psi^*,\alpha)$ so that the first term in (\ref{LR0-hessian1}) is written as $(n/8)\sum_{i=1}^{q} \sum_{j=1}^{q} \sum_{k=1}^{q} \sum_{\ell=1}^{q} B_{ijk\ell} \lambda_i \lambda_j \lambda_k \lambda_\ell= (n/2) \sum_{i=1}^{q} \sum_{j=1}^i \sum_{k=1}^{q} \sum_{\ell=1}^k c_{ij}c_{k\ell}B_{ijk\ell}\lambda_i \lambda_j \lambda_k \lambda_\ell$. Define $\bs{\mathcal{I}}_{\bs\lambda n}$ such that the $(ij,k\ell)$ element of $\bs{\mathcal{I}}_{\bs\lambda n}$ is $-c_{ij}c_{k\ell} B_{ijk\ell}/\alpha^2(1-\alpha)^2$, where the values of $ij$ run over $\{(1,1),\ldots,(q,q),(1,2),\ldots,(q-1,q)\}$. Then,  the first term in (\ref{LR0-hessian1}) equals $-(n/2) [\alpha(1-\alpha)v(\bs\lambda)]'\bs{\mathcal{I}}_{\bs\lambda n}[\alpha(1-\alpha)v(\bs\lambda)]$. With this definition of $\bs{\mathcal{I}}_n$, the expansion (\ref{LR0-score0})-(\ref{LR0-hessian1}) is written as (\ref{eq:LR0}) in terms of $\sqrt{n}t(\psi,\alpha)$.

We now show that $\bs{\mathcal{I}}_n \rightarrow_p \bs{\mathcal{I}}$. $\bs{\mathcal{I}}_{\bs\eta n} \rightarrow_p \bs{\mathcal{I}}_{\eta}$ holds trivially. For $\bs{\mathcal{I}}_{\bs\eta\bs\lambda n}$, it follows from Lemma \ref{lemma:KS2018_lemma7}(c) and the law of large numbers that $A_{ij} \rightarrow_p -\E[\nabla_{\bs\eta} l(\bs W;\psi^*,\alpha) \nabla_{\lambda_i\lambda_j}l(\bs W;\psi^*,\alpha)]$, giving $\bs{\mathcal{I}}_{\bs\eta\bs\lambda n} \rightarrow_p E\left[\bs s_{\bs\eta} \bs s_{\bs\lambda\bs\lambda}\t/\alpha(1-\alpha)\right]=\bs{\mathcal{I}}_{\bs\eta \bs \lambda}$. For $\bs{\mathcal{I}}_{\bs\lambda n}$, Lemma \ref{lemma:KS2018_lemma7}(d) and the law of large numbers imply that  $\sum_{i=1}^{q} \sum_{j=1}^{q} \sum_{k=1}^{q} \sum_{\ell=1}^{q} B_{ijk\ell} \lambda_i \lambda_j \lambda_k \lambda_\ell \rightarrow_p \\ - \sum_{i=1}^{q} \sum_{j=1}^{q} \sum_{k=1}^{q} \sum_{\ell=1}^{q} E[\nabla_{\lambda_i \lambda_j}l(\bs W;\psi^*,\alpha)\nabla_{\lambda_k \lambda_\ell}l(\bs W;\psi^*,\alpha)]\lambda_i \lambda_j \lambda_k \lambda_\ell$, where the factor $(8/4!)=1/3$ in $B_{ijk\ell}$ and the three derivatives on the right-hand side of Lemma \ref{lemma:KS2018_lemma7}(d) cancel each other out. Therefore, we have  $\bs{\mathcal{I}}_{\bs\lambda n} \rightarrow_p E\left[\bs s_{\bs\lambda\bs\lambda}\bs s_{\bs\lambda\bs\lambda}\t  /\alpha^2(1-\alpha)^2\right]=\bs{\mathcal{I}}_{\bs\lambda}$, and $\bs{\mathcal{I}}_n \rightarrow_p \bs{\mathcal{I}}$ follows.

We complete the proof of part (c) by showing that $\bs{\mathcal{I}} = E[\bs{s}(\bs{W}) \bs{s}(\bs{W})\t]$ is finite and non-singular.
Note that $\bs s(\bs W) $ can be expressed in Hermite polynomials as in (\ref{eq:s_1_hermite}). Then, the finiteness of $\bs{\mathcal{I}}$ follows from  Assumption \ref{assumption:LRT1}(a) and the definition of Hermite polynomials.

To show that $\bs{\mathcal{I}}$ is positive definite, it suffices to show that there exists no multicollinearity in $\bs s(\bs w)$. Suppose, to the contrary, that $\bs {s}(\bs w)$ is multicollinear and that there exists a non-zero vector $\bs a$ that solves the  equation $\bs a\t \bs s(\bs w) = 0$ for all values of $\bs w $. Partition $ s(\bs w)$ as $\bs s(\bs w)=(\bs s_{(\mu)} \t, \bs s_{(\beta)}\t )\t$  with $\bs s_{(\mu)} = (s_{\mu},s_{\sigma},s_{\lambda_{\mu \mu }} ,s_{\lambda_{\mu \sigma }} ,s_{\lambda_{\sigma \sigma }})\t$  and $\bs s_{(\beta)} = (\bs s_{\bs\beta }\t, \bs s_{\lambda_{\mu\bs \beta }}\t,\bs s_{\lambda_{\sigma\bs \beta }}\t,\bs s_{\lambda_{\bs \beta \bs \beta }}\t)\t$,
where $\bs s(\bs w) $ is defined in (\ref{eq:s_1}) and  (\ref{eq:s_1_hermite}).
Similarly, partition $\bs a$ as $\bs a = (\bs a_{(\mu)}\t, \bs a_{(\beta)}\t )\t$ so that
\begin{equation}\label{eq:as}
\bs a\t \bs s(\bs w) = \bs a_{(\mu)}\t \bs s_{(\mu)} + \bs a_{(\beta)}\t \bs s_{(\beta)}.
\end{equation}
By Assumption \ref{assumption:LRT1}(b) and the property of Hermite polynomials, if $\bs a\t \bs s(\bs w) = \bs 0$ for all $\bs w$, then $\bs a_{(\beta)} = \bs 0$.

Then, in view of (\ref{eq:as}), the stated result follows if we can show that $\bs a_{(\mu)}\t  \bs s_{(\mu)} = \bs 0$ for all $\bs w$ implies $\bs a_{(\mu)}=0$.
Suppose that
\begin{align*}
	\bs a_{(\mu)}\t  \bs s_{(\mu)} & = a_{\mu } \sum_{t=1}^T  H^{1*}_{t} +  (a_{\sigma } + a_{\lambda_{\mu \mu }}) \sum_{t=1}^T  H^{2*}_{t} +    \frac{a_{\lambda_{\mu \mu }}}{2} \sum_{t=1}^T \sum_{s \neq t} H^{1*}_{t} H^{1*}_{ is}  \\
	& + a_{\lambda_{\mu \sigma }} \sum_{t=1}^T H^{3*}_{t} + a_{\lambda_{\mu \sigma }} \sum_{t=1}^T \sum_{s \neq t} H^{1*}_{t} H^{2*}_{ is} +
	3  a_{\lambda_{\sigma \sigma }}\sum_{t=1}^T H^{4*}_{t} +  \frac{a_{\lambda_{\sigma \sigma }}}{2} \sum_{t=1}^T \sum_{s \neq t} H^{2*}_{t} H^{2*}_{ is}  = 0
\end{align*}
for all $\bs w$,
where $H^{j*}_{t}$ for $j=1,2,3$ is defined in (\ref{eq:Hermite_polynomial})  in Appendix \ref{sec:appendix_score_1}.

Because the above equation holds for all values of $\bs w$, with the property of the Hermite polynomials, we have $a_{\mu } = 0, (a_{\sigma } + a_{\lambda_{\mu \mu }}) = 0, a_{\lambda_{\mu \mu }} = 0, a_{\lambda_{\mu \sigma }} = 0, a_{\lambda_{\sigma \sigma }} = 0$. This implies that $\bs a_{(\mu)}=0$.
Therefore,  no multicollinearity exists in $\bs s(\bs w)$ and $\bs{\mathcal{I}}$ is non-singular, proving part (c). 

The proof of part (d) closely follows the proof of Theorem 1 of \cite{Andrews1999}. Let $\boldsymbol{T}_{n}: = \boldsymbol{\mathcal{I}}_{n}^{1/2}\sqrt{n}\bs{t}(\hat{\bs\psi}_\alpha,\alpha)$. Then, in view of (\ref{eq:LR0}), we have
\begin{align*}
o_p(1)&\leq L_n(\hat{\boldsymbol{\psi}}_\alpha,\alpha) - L_n(\boldsymbol{\psi}^*,\alpha)\\
 &= \boldsymbol{T}_{n}' \boldsymbol{\mathcal{I}}_{n}^{-1/2} \boldsymbol{S}_{n} - \frac{1}{2} ||\boldsymbol{T}_{n}||^2 + R_n(\hat{\boldsymbol{\psi}}_\alpha,\alpha)\\
 &= O_p(||\boldsymbol{T}_{n}||) - \frac{1}{2} ||\boldsymbol{T}_{n}||^2 + (1 + || \boldsymbol{\mathcal{I}}_{n}^{-1/2} \boldsymbol{T}_{n}||)^2 o_p(1)\\
 &= ||\boldsymbol{T}_{n}||O_p(1) - \frac{1}{2} ||\boldsymbol{T}_{n}||^2 + o_p(||\boldsymbol{T}_{n}||) + o_p(||\boldsymbol{T}_{n}||^2) + o_p(1),
\end{align*}
where the third equality holds because $\boldsymbol{\mathcal{I}}_{n}^{-1/2} \boldsymbol{S}_{n} = O_p(1)$ and $R_n(\hat{\boldsymbol{\psi}}_\alpha,\alpha) = o_p((1 + ||\boldsymbol{\mathcal{I}}_{n}^{-1/2}\boldsymbol{T}_{n}||)^2)$ from Proposition \ref {lemma:expansion}. Rearranging this equation yields $||\boldsymbol{T}_n||^2 \leq 2 ||\boldsymbol{T}_n|| O_p(1)+o_p(1)$. Denote the $O_p(1)$ term by $\varsigma_{n}$. Then, $(||\boldsymbol{T}_{n}||-\varsigma_{n})^2\leq \varsigma_{n}^2 + o_p(1) = O_p(1)$; taking its square root gives $||\boldsymbol{T}_{n}|| \leq O_p(1)$. In conjunction with $\boldsymbol{\mathcal{I}}_{n} \rightarrow_p \boldsymbol{\mathcal{I}}$, we obtain $\sqrt{n}\bs{t}(\hat{\bs\psi}_\alpha,\alpha) = O_p(1)$, and part (d) follows.

\end{proof}

\begin{lemma}\label{lemma: sht}
Suppose that the assumptions in Proposition \ref{prop:sht} hold. If  $-n^{-1}\log q_n=o(1)$, then  $n^{-1} c^M_{1-q_n}=o(1)$.

\end{lemma}
\begin{proof}
For brevity of notation,
write $c_n=c^{M}_{1-q_n}$.
By Theorem 2.1 of \cite{Foutz77as}, $PLR_n(M)\overset{d}{\rightarrow} \sum_{j=1}^K b_j \chi_j^2$ for $0<b_j<\infty$ and $K$ is finite, where $\chi_1^2$, ..., $\chi_K^2$ are independent chi-square random variables with one degree of freedom.
Then, we have
\begin{align*}\label{eq:q_n}
q_n&= \Pr\left( \sum_{j=1}^K b_j \chi_j^2 \geq c_n\right)\leq \sum_{j=1}^K  \Pr\left(  \chi_j^2 \geq \frac{c_n}{ b_j}\right)\leq  \frac{K}{\sqrt{1-2t}}\exp\left(-t \frac{c_n}{ b^*}\right)\quad\text{for $0<t<\frac{1}{2}$}
\end{align*}
with $b^*=\arg\max\{b_1,...,b_K\}$, where the last inequality follows from a Chernoff bound: $\Pr\left(  \chi_j^2 \geq \frac{c_n}{ b^*}\right)\leq \frac{\E[\exp(t(\chi_j^2-1))]}{\exp(t(\frac{c_n}{ b^* }-1))}=\frac{1}{\sqrt{1-2t}}\exp\left(-t \frac{c_n}{ b^*}\right)$ for $0<t<\frac{1}{2}$. Therefore,
$-\frac{\log q_n}{n} \geq - \frac{1}{n}\log\left(\frac{K}{\sqrt{1-2t}}\right) +\frac{1}{2b_j^*} \frac{c_n}{n}$,
and the stated result follows.

\end{proof}

\begin{lemma}\label{lemma:KS2018_lemma7}
    Suppose that $g(\bs{w};\bs{\psi},\alpha)$ is defined as (\ref{eq:repar}), where $\bs{\psi} = (\bs{\eta}^\top,\bs{\lambda}^\top)^\top$. Let $g^*$, $\nabla g^*$, and $\nabla\log g^*$ denote $g(\bs{W};\bs{\psi},\alpha)$, $\nabla g(\bs{W};\bs{\psi},\alpha)$,  and $\nabla \log g(\bs{W};\bs{\psi},\alpha)$ evaluated at $(\bs{\psi}^*,\alpha)$, respectively. Let $\nabla f^*$ denote $\nabla f(\bs{W};\bs{\theta}^*)$.
    The following statements hold.
    \begin{enumerate}[label=(\alph*)]
        \item  For $l = 0,1,\ldots, \nabla_{ (\bs{\lambda} \otimes \bs{\eta}^{\otimes l})\t}g^* = 0$;
        \item $\nabla_{(\bs{\lambda}^{\otimes 2})\t } g^*  = \alpha (1 - \alpha) \nabla_{(\bs{\theta}^{\otimes 2})\t } f^*$;
\item $\nabla_{(\bs{\lambda}^{\otimes 2})\t } \log g^*  =  \alpha (1 - \alpha) \nabla_{(\bs{\theta}^{\otimes 2})\t } f^*/f^*$;
        \item $\E[\nabla_{\lambda_i\lambda_j}\log g^*]=0$, $\E[\nabla_{\lambda_i\lambda_j\lambda_k}\log g^*]=0$,
and  $\E[\nabla_{\eta\lambda_i\lambda_j }\log g^*]=-\E[\nabla_{\eta}\log g^*\nabla_{\lambda_i\lambda_j}\log g^*]$;
 \item $\E[\nabla_{\lambda_i\lambda_j\lambda_k\lambda_\ell}\log g^*]=
-\E[\nabla_{\lambda_i\lambda_j}\log g^*\nabla_{\lambda_k\lambda_\ell}\log g^*+\nabla_{\lambda_i\lambda_k}\log g^*\nabla_{\lambda_j\lambda_\ell}\log g^*+\nabla_{\lambda_i\lambda_\ell}\log g^*\nabla_{\lambda_j\lambda_k}\log g^*]$.    \end{enumerate}
\end{lemma}

\begin{proof}[Proof of Lemma \ref{lemma:KS2018_lemma7}]
    Recall that $$g(\bs{w};\bs{\psi},\alpha)  = \alpha f(\bs{w};\bs{\nu}  + (1 - \alpha) \bs{\lambda}) + (1 - \alpha)  f(\bs{w};\bs{\nu} - \alpha \bs{\lambda}). $$
    First, we show that for $l=0$ holds for (a),
    $\nabla_{\bs{\lambda}} g^* = \alpha ( 1 - \alpha) \nabla_{\bs{\theta}} f^* -  \alpha  ( 1 - \alpha) \nabla_{\bs{\theta}} f^*  = 0 $.
    For $l > 0$, by Fubini's theorem, we have
    \[
    \begin{split}
    \nabla_{ (\bs{\lambda} \otimes \bs{\eta}^{\otimes l})\t} g & = \left. \nabla_{\bs{\lambda}}\Big( \alpha \nabla_{ (\bs{\eta}^{\otimes l})\t} f(\bs{w};\bs{\nu}  + (1 - \alpha)  \bs{\lambda}) + (1 - \alpha) \nabla_{ (\bs{\eta}^{\otimes l})\t}f(\bs{w};\bs{\nu} - \alpha \bs{\lambda}) \right|_{\bs{\nu} = \bs{\theta}^*, \bs{\lambda} = \bs{0} } \Big) \\
    & = \left. \Big( \alpha ( 1- \alpha)\nabla_{ (\bs{\lambda} \otimes\bs{\eta}^{\otimes l})\t}f(\bs{w};\bs{\nu}  + (1 - \alpha)  \bs{\lambda}) - \alpha (1 - \alpha) \nabla_{ (\bs{\lambda} \otimes\bs{\eta}^{\otimes l})\t}f(\bs{w};\bs{\nu} - \alpha \bs{\lambda}) \right|_{\bs{\nu} = \bs{\theta}^*, \bs{\lambda} = \bs{0} } \Big) \\
    & = 0.
    \end{split}
    \]
    To show part $(b)$, note that \begin{equation*}
    \begin{split}
    \nabla_{(\bs{\lambda}^{\otimes 2})\t } g & =  \nabla_{\bs{\lambda}} \Big( \alpha ( 1 - \alpha) \nabla_{\bs{\lambda}^\top } f(\bs{w};\bs{\nu}  - \alpha (1 - \alpha)  \bs{\lambda})
    -\alpha (1 - \alpha) \nabla_{\bs{\lambda}^\top} f(\bs{w};\bs{\nu} - \alpha \bs{\lambda}) \Big) \\
    & = \left.   \alpha ( 1 - \alpha)^2 \nabla_{(\bs{\lambda}^{\otimes 2})\t } f(\bs{w};\bs{\nu}  - (1 - \alpha)  \bs{\lambda})
    + \alpha^2(1 - \alpha) \nabla_{(\bs{\lambda}^{\otimes 2})\t } f(\bs{w};\bs{\nu} - \alpha \bs{\lambda}) \right|_{\bs{\nu} = \bs{\theta}^*, \bs{\lambda} = \bs{0} } \\
    & = \nabla_{(\bs{\lambda}^{\otimes 2})\t } f^*.
    \end{split}
    \end{equation*}

    For part (c),  $\nabla_{\bs{\lambda}\bs{\lambda} \t } \log g^*  =  \nabla_{\bs{\lambda}\bs{\lambda} \t } g^*/g^* -( \nabla_{\bs\lambda } \log g^* )( \nabla_{\bs\lambda } \log g^*)=\alpha(1-\alpha)\nabla_{\bs{\lambda}\bs{\lambda} \t } f^*/f^*$ because $\nabla_{\bs{\lambda}\bs{\lambda} \t } g^*/g^*=\alpha(1-\alpha)\nabla_{\bs{\theta}\bs{\theta} \t } f^*/f^*$  from part (a) and $ \nabla_{\bs\lambda } \log g^* =\nabla_{\bs\lambda } g^* /g^*=0$ from part (b).

 For  parts (d) and (e), observe that $\int  \nabla_{\lambda_i}  \log g(\bs w;\psi,\alpha)g(\bs w;\psi,\alpha)dx=0$ holds for any $\psi$ in the interior of $\Theta_\psi$, and differentiating this equation w.r.t.\ $\lambda_j$ gives
\begin{equation}
\int \{ \nabla_{\lambda_i\lambda_j} \log g(\bs w;\psi,\alpha)  +  \nabla_{\lambda_i} \log g(\bs w;\psi,\alpha)  \nabla_{\lambda_j} \log g(\bs w;\psi,\alpha) \} g(\bs w;\psi,\alpha) dx=0.\label{IME}
\end{equation}
Evaluating (\ref{IME}) at $\psi=\psi^*$ in conjunction with  part (a) gives the first equation in  part (d). Differentiating (\ref{IME}) w.r.t.\ $\lambda_k$ or $\eta$ and evaluating  at $\psi=\psi^*$ gives the latter two equations in part (d).  Part (e) follows from differentiating (\ref{IME}) w.r.t.\ $\lambda_k$ and $\lambda_\ell$ and evaluating at $\psi=\psi^*$ in conjunction with  parts (a) and (d).

\end{proof}

\begin{lemma}\label{max_bound}[Lemma 2.1 of \cite{liushao03as}]
Suppose $X_1, \ldots, X_n$ are i.i.d. random variables with $\max_{1 \leq i \leq n} \mathbb{E}|X_i|^{q+\delta} < C$ for some $\delta>0$, $q > 0$,  and $C \in (0,\infty)$. Then, $\max_{1 \leq i \leq n} |X_i| = o_p(n^{1/q})$.
\end{lemma}

\begin{proof}
For any $\varepsilon > 0$, we have
\begin{align*}
P\left(\max_{1 \leq i \leq n} |X_i| > \varepsilon n^{1/q}\right) &\leq \sum_{1 \leq i \leq n} P(|X_i| > \varepsilon n^{1/q})\leq \varepsilon^{-q} \mathbb{E}(|X_1|^q\mathbb{I}{\{|X_1|>\varepsilon n^{1/q}\}})
\end{align*}
by a version of the Markov inequality and the i.i.d. assumption. As $n \to \infty$, the right-hand side tends to $0$ by the dominated convergence theorem.
\end{proof}

\subsection{Score function for testing $H_0:m = 1$ against $H_A:m=2$}\label{sec:appendix_score_1}
$H^j(\cdot)$ is defined as the $j$-th order Hermite polynomial. $H^1(t) = t$, $H^2(t) = t^2 - 1$ , $H^3(t) = t^3 - 3t$, and $H^4(t) = t^4 - 6t^2 + 3$.
As shown in the supplementary material of \cite{kasaharashimotsu15jasa}, the derivative of $\{ \frac{1}{\sigma} \phi(\frac{t}{\sigma}) \}$ is  $$\frac{\nabla_{\mu^m}\nabla_{ (\sigma^2)^{\ell}} \{ \frac{1}{\sigma} \phi(\frac{t}{\sigma}) \} }{\{ \frac{1}{\sigma} \phi(\frac{t}{\sigma}) \}}
 = \left(\frac{1}{2}\right)^\ell \left(\frac{1}{\sigma}\right)^{m+2\ell} H^{m+2\ell}\left(\frac{t}{\sigma}\right). $$
Let \( f^* = f(\bs{W}; \theta^*) = \prod_{t=1}^T f^*_t (\bs{W}; \theta^*) \), where \( f^*_t \) denotes the component density for observation \( t \), which may be either a normal or a mixture density. Let \( \nabla f^* = \nabla f(\bs{W}; \theta^*) \) denote its gradient.

Define the Hermite polynomials evaluated at \( Y_t - \bs{X}_t^\top \bs{\beta}^* - \mu^* \) as follows:
\begin{equation}\label{eq:Hermite_polynomial}
\begin{split}
  H^{j*}_{t} = \frac{1}{(\sigma^*)^j} H^{j}\left(\frac{Y_{t} - \bs{X}_{t}^\top \bs{\beta}^*  - \mu^*  }{\sigma^*} \right).
\end{split}
\end{equation}

\subsubsection{Score function for the model (\ref{eq:fm}) with normal density (\ref{eq:f1})}

Let
\(f^*_t = \frac{1}{\sigma^*} \phi\left( \frac{Y_{t} - \bs{X}_{t}^\top \bs{\beta}^*  - \mu^*  }{\sigma^*} \right) \).
Then, the first-order derivatives of the normal density function (\ref{eq:f1}) with respect to the parameters are given by:
\begin{align*}
  & \nabla_{\mu} f^*  =   f^* \sum_{t=1}^T H^{1*}_{t}, 
  & \nabla_{\sigma^2} f^*   =  f^*  \sum_{t=1}^T \frac{1}{2}  H^{2*}_{t}, \\
  & \nabla_{\bs{\beta}} f^*   =   f^*  \sum_{t=1}^T H^{1*}_{t} \bs{X}_{t}.
\end{align*}

The score function defined in (\ref{eq:s_1}) is then written in terms of the Hermite polynomials:
\begin{equation}\label{eq:s_1_hermite}
\begin{split}
\bs s_{\bs \eta }(\bs W)  = \begin{pmatrix}
s_{\mu } \\
s_{\sigma } \\
\bs s_{\bs\beta } 
\end{pmatrix}= \begin{pmatrix}
\sum_{t=1}^T  H^{1*}_{t} \\
\sum_{t=1}^T  H^{2*}_{t} \\
\sum_{t=1}^T  H^{1*}_{t}\bs X_{t} 
\end{pmatrix}\quad\text{and}\quad
\bs{s}_{\bs{\lambda} \bs{\lambda}}(\bs W) = \begin{pmatrix}
s_{\lambda_\mu\lambda_\mu}\\
s_{\lambda_\mu\lambda_\sigma}\\
s_{\lambda_\sigma\lambda_\sigma}\\
\bs{s}_{\lambda_\mu\lambda_{\bs\beta}}\\
\bs{s}_{\lambda_\sigma\lambda_{\bs\beta}}\\
\bs{s}_{\lambda_{\bs\beta}\lambda_{\bs\beta}}
  \end{pmatrix},
\end{split}
\end{equation}
where
\begin{equation}\label{eq:s-vector}
\begin{split}
 \begin{pmatrix}
s_{\lambda_{\mu}\lambda_{ \mu } }\\
s_{\lambda_{\mu }\lambda_{ \sigma }}\\
s_{\lambda_{\sigma }\lambda_{ \sigma }}\\
\bs{s}_{\lambda_{\mu}\lambda_{\bs\beta}}\\
\bs{s}_{\lambda_{\sigma}\lambda_{\bs\beta}}\\
\bs{s}_{\lambda_{\bs\beta}\lambda_{\bs\beta}}
\end{pmatrix} &=
\begin{pmatrix}
\sum_{t=1}^T \left[ H^{2*}_{t} + \sum_{s \neq t} H^{1*}_{t} H^{1*}_{s} \right] \\
\frac{1}{2} \sum_{t=1}^T \left[  H^{3*}_{t} + \sum_{s \neq t} H^{1*}_{t} H^{2*}_{s} \right] \\
 \frac{1}{4} \sum_{t=1}^T  \left[  H^{4*}_{t} + \sum_{s \neq t} H^{2*}_{t} H^{2*}_{s}  \right] \\
\sum_{t=1}^T \left[ H^{2*}_{t} + \sum_{s \neq t} H^{1*}_{t} H^{1*}_{s} \right] \boldsymbol{X}_{t}  \\
 \frac{1}{2}  \sum_{t=1}^T \left[  H^{3*}_{t} + \sum_{s \neq t} H^{1*}_{t} H^{2*}_{s} \right] \boldsymbol{X}_{t}  \\
 \sum_{t=1}^T \left[ H^{2*}_{t} \widetilde{\textnormal{vech}}(\boldsymbol{X}_{t} \boldsymbol{X}_{t}^\top) + \sum_{s \neq t} H^{1*}_{t} H^{1*}_{s} \widetilde{\textnormal{vech}}(\boldsymbol{X}_{t} \boldsymbol{X}_{s}^\top) \right]  
\end{pmatrix}.
\end{split}
\end{equation}

%

\subsubsection{Score function for the model (\ref{eq:fm}) with normal mixture density (\ref{eq:f1-mixture}) when $K=2$}
For brevity, we present the score function for the model (\ref{eq:fm}) with a two-component normal mixture (\ref{eq:f1-mixture}) when there are no covariates. The score function for the model with covariates can be analogously derived. Let
\[
f^* = \prod_{t=1}^T f^*_t = \prod_{t=1}^T \left( \tau^* f^*_{1t} + (1 - \tau^*) f^*_{2t}  \right),
\]
where \( f^*_{kt} = \frac{1}{\sigma^*} \phi\left(\frac{Y_t  - \mu_k^*}{\sigma^*} \right) \) for $k = 1, 2$.
In this model, we omit the covariates \( \bs{X}_t \) and the parameter \( \bs{\beta}^* \) for simplicity. The score function with covariates can be derived similarly to the previous section.
Define the $b$-th order normalized Hermite polynomial for $k$-th component evaluated at \( Y_t - \mu_k^* \) as follows:
\begin{equation}
\begin{split}
  H^{b*}_{kt} = \frac{1}{(\sigma^*)^b} H^{b}\left(\frac{Y_{t} - \mu_k^*}{\sigma^*} \right).
\end{split}
\end{equation}
Let
\[
\begin{aligned}
\gamma_{1t}:= \frac{ \tau^* f_{1t}^* }{\tau^* f_{1t}^* + (1 - \tau^*) f_{2t}^*},\quad \gamma_{2t}:=1-\gamma_{1t}.
\end{aligned}
\]
The score function defined in (\ref{eq:s_1}) is then written as:
 \begin{equation}\label{eq:s_1_hermite-mixture}
\begin{split}
\bs s_{\bs \eta }(\bs W)  = \begin{pmatrix}
s_{\tau }\\
s_{\mu_1 } \\
s_{\mu_2 } \\
s_{\sigma }
\end{pmatrix}= \begin{pmatrix}
 \sum_{t=1}^T \left( \frac{ \gamma_{1t} }{ \tau^* } - \frac{ \gamma_{2t} }{ 1 - \tau^* } \right) \\
 \sum_{t=1}^T \gamma_{1t} H^{1*}_{1t} \\
 \sum_{t=1}^T \gamma_{2t} H^{1*}_{2t} \\
\sum_{t=1}^T  \gamma_{1t} H^{2*}_{1t}  + \gamma_{2t} H^{2*}_{2t} 
\end{pmatrix}\quad\text{and}\quad
\bs{s}_{\bs{\lambda} \bs\lambda}(\bs W) = \begin{pmatrix}
s_{\lambda_{\tau}\lambda_{\tau} }\\
s_{\lambda_{\mu_1}\lambda_{\mu_1}}\\
s_{\lambda_{\mu_2}\lambda_{\mu_2}}\\
s_{\lambda_{\mu_1}\lambda_{\mu_2}}\\
s_{\lambda_{\tau}\lambda_{\mu_1} }\\
s_{\lambda_{\tau}\lambda_{\mu_2} }\\
s_{\lambda_{\mu_1}\lambda_{\sigma}}\\
s_{\lambda_{\mu_2}\lambda_{\sigma}}\\
s_{ \lambda_{\sigma} \lambda_{\sigma}}\\
s_{ \lambda_{\tau}\lambda_{\sigma} }
  \end{pmatrix},
\end{split}
\end{equation}
where
\begin{equation}\label{eq:s-vector-mixture}
\begin{pmatrix}
s_{\lambda_{\tau}\lambda_{\tau}} \\
s_{\lambda_{\mu_1}\lambda_{\mu_1}} \\
s_{\lambda_{\mu_2}\lambda_{\mu_2}} \\
s_{\lambda_{\mu_1}\lambda_{\mu_2}} \\
s_{\lambda_{\tau}\lambda_{\mu_1}} \\
s_{\lambda_{\tau}\lambda_{\mu_2}} \\
s_{\lambda_{\mu_1}\lambda_{\sigma}} \\
s_{\lambda_{\mu_2}\lambda_{\sigma}} \\
s_{\lambda_{\sigma} \lambda_{\sigma} } \\
s_{ \lambda_{\tau}\lambda_{\sigma} }
\end{pmatrix}
=
\begin{pmatrix}
 \sum_{t=1}^T \sum_{s\neq t} \left( \frac{ \gamma_{1t} }{ \tau^* } - \frac{ \gamma_{2t} }{ 1 - \tau^* } \right) \left( \frac{ \gamma_{1s} }{ \tau^* } - \frac{ \gamma_{2s} }{ 1 - \tau^* } \right) \\
\sum_{t=1}^T \left( \gamma_{1t} H^{2*}_{1t} + \sum_{s \neq t}\gamma_{1t} \gamma_{1s} H^{1*}_{1t}  H^{1*}_{1s} \right) \\
\sum_{t=1}^T \left( \gamma_{2t} H^{2*}_{2t} + \sum_{s \neq t}\gamma_{2t} \gamma_{2s} H^{1*}_{2t}  H^{1*}_{2s} \right) \\
\sum_{t=1}^T \left(  \sum_{s \neq t} \gamma_{1t} \gamma_{2s} H^{1*}_{1t} H^{1*}_{2s} \right) \\
\sum_{t=1}^T \left( \frac{ \gamma_{1t} }{ \tau^* }  H^{1*}_{1t} + \sum_{s \neq t} \gamma_{1t} H^{1*}_{1t} \left( \frac{ \gamma_{1s} }{ \tau^* } - \frac{ \gamma_{2s} }{ 1 - \tau^* } \right) \right) \\
\sum_{t=1}^T \left( - \frac{ \gamma_{2t} }{ 1 - \tau^* }  H^{1*}_{2t} + \sum_{s \neq t} \gamma_{2t} H^{1*}_{2t} \left( \frac{ \gamma_{1s} }{ \tau^* } + \frac{ \gamma_{2s} }{ 1 - \tau^* } \right) \right)  \\
\frac{1}{2} \sum_{t=1}^T \left( \gamma_{1t} H^{3*}_{1t} + \sum_{s \neq t} \gamma_{1t} H^{1*}_{1t} \left( \gamma_{1s} H^{2*}_{1s}  + \gamma_{2s} H^{2*}_{2s} \right) \right)\\
\frac{1}{2} \sum_{t=1}^T \left( \gamma_{2t} H^{3*}_{2t} + \sum_{s \neq t} \gamma_{2t} H^{1*}_{2t} \left( \gamma_{1s} H^{2*}_{1s}  + \gamma_{2s} H^{2*}_{2s} \right) \right)\\
\frac{1}{4}\sum_{t=1}^T \left( \left( \gamma_{1t} H^{4*}_{1t}  + \gamma_{2t} H^{4*}_{2t} \right) + \sum_{s \neq t} \left( \gamma_{1t} H^{2*}_{1t}  + \gamma_{2t} H^{2*}_{2t} \right) \left( \gamma_{1s} H^{2*}_{1s}  + \gamma_{2s} H^{2*}_{2s} \right) \right) \\
\frac{1}{2} \sum_{t=1}^T \left( \left( \frac{ \gamma_{1t} }{ \tau^* } -   \frac{ \gamma_{2t} }{ 1 - \tau^* }  \right) H^{2*}_{1t}  + \sum_{s \neq t} \left( \gamma_{1t} H^{2*}_{1t}  + \gamma_{2t} H^{2*}_{2t} \right) \left( \frac{ \gamma_{1s} }{ \tau^* } - \frac{ \gamma_{2s} }{ 1 - \tau^* } \right) \right)  
\end{pmatrix}.
\end{equation}

\subsection{Score function for testing $H_0:m = M_0$ against $H_A:m=M_0 + 1$}
We only present score function for the model (\ref{eq:fm}) with normal density (\ref{eq:f1}). The score functions for other models can be derived analogously.
Let $g^* = g(\bs{W};\vartheta_{M_0}^*) = \sum_{j=1}^{M_0} \alpha_j^* f_j^*$ denote the true $M_0$-component model as in equation (\ref{eq:fm0}), where $f_j^* = f(\bs{W}; \theta_j^*) = \prod_{t=1}^T \frac{1}{\sigma_j^*} \phi\left(\frac{Y_t - \mu_j^*}{\sigma_j^*}\right)$ is the $j$-th component density. 
In this section, we omit the covariate $\bs X_{it}$ and the parameters $\beta_j^*$ for simplicity; the derivation with covariates can be done analogously to the previous section.
Define $H^{b*}_{j, it}$ as a shorthand for the $b$-th order normalized Hermite polynomial evaluated at $\frac{Y_{it} - \mu_j^*}{\sigma_j^*}$, i.e., $H^{b*}_{j,t} = \frac{1}{(\sigma_j^*)^b} H^{b}\left(\frac{Y_{t} - \mu_j^*}{\sigma_j^*}\right)$. Define the weight $w_{i}^{j*}$ as
\(
w_{j}^{*}   =  \frac{\alpha_j^{*} f_j^*}{ \sum_{l} \alpha_l^{*} f_l^* }, \quad j=1,\ldots,M_0.
\)
The score functions are
\begin{align*}
&\bs s_{\bs\alpha}(\bs W)  = \begin{pmatrix}
\frac{w_1^{*}}{\alpha_1^*} - \frac{w_{M_0}^*}{\alpha_{M_0}^*} 
  \\
\vdots \\
\frac{w_{M_0-1}^*}{\alpha_{M_0-1}^*} - \frac{w_{M_0}^*}{\alpha_{M_0}^*} 
\end{pmatrix}, 
& \bs s_{\mu }(\bs W)  = \begin{pmatrix}
w_{1}^{*} \sum_{t=1}^T  H^{1*}_{1, t} \\
\vdots \\
w_{M_0}^*  \sum_{t=1}^T  H^{1*}_{M_0, t}
\end{pmatrix}, \\
& \bs s_{\sigma }(\bs W) = \begin{pmatrix}
w_{1}^*  \sum_{t=1}^T H^{2*}_{1, t} \\
\vdots \\
w_{M_0}^*  \sum_{t=1}^T H^{2*}_{M_0, t}
\end{pmatrix}, 
& \bs s_{\bs\lambda\bs\lambda}(\bs W) =
\begin{pmatrix}
\bs s_{\lambda_{\mu}\lambda_{ \mu } }\\
\bs s_{\lambda_{\mu }\lambda_{ \sigma }}\\
\bs s_{\lambda_{\sigma }\lambda_{ \sigma }}
\end{pmatrix},
\end{align*}
where
\begin{align*}
\bs s_{\lambda_{\mu}\lambda_{ \mu } } &= \begin{pmatrix}
w_1^* \sum_{t=1}^T \left[ H^{2*}_{1,t} + \sum_{s \neq t} H^{1*}_{1, t} H^{1*}_{1, s} \right] \\
\vdots \\
w_{M_0}^* \sum_{t=1}^T \left[ H^{2*}_{M_0,t} + \sum_{s \neq t} H^{1*}_{M_0, t} H^{1*}_{M_0, s} \right] 
\end{pmatrix}, 
\bs s_{\lambda_{\mu }\lambda_{ \sigma }} = \begin{pmatrix}
w_1^*  \frac{1}{2} \sum_{t=1}^T \left[  H^{3*}_{1,t} + \sum_{s \neq t} H^{1*}_{1,t} H^{2*}_{1, s} \right] \\
\vdots \\
w_{M_0}^*  \frac{1}{2} \sum_{t=1}^T \left[  H^{3*}_{M_0,t} + \sum_{s \neq t} H^{1*}_{M_0,t} H^{2*}_{M_0, s} \right] \\
\end{pmatrix}, \\[1em]
\bs s_{\lambda_{\sigma }\lambda_{ \sigma }} &= \begin{pmatrix}
w_1^* \frac{1}{4} \sum_{t=1}^T \left[ H^{4*}_{1,t} + \sum_{s \neq t} H^{2*}_{1,t} H^{2*}_{1,s} \right] \\
\vdots \\
w_{M_0}^* \frac{1}{4} \sum_{t=1}^T \left[ H^{4*}_{M_0,t} + \sum_{s \neq t} H^{2*}_{M_0,t} H^{2*}_{M_0,s} \right]
\end{pmatrix}.
\end{align*}

\subsection{EM Algorithm}\label{sec:em}
\subsubsection{EM Algorithm for the model (\ref{eq:fm}) with normal mixture density (\ref{eq:f1-mixture})}\label{sec:em-algorithm}

We introduce latent variables $D_i \in \{1,2,\ldots,M\}$ and $C_{it} \in \{1,2,\ldots,K\}$, where the top-level latent variable $D_i$ indicates which of the $M$ top-level components was chosen for observation $i$, and $C_{it}$ indicates which of the $K$ inner components was chosen for the $t$-th observation of unit $i$ under the selected top-level component. Then, the complete-data log-likelihood is given as
\[
\log L_c = \sum_{i=1}^n \sum_{j=1}^M \mathbb{I}(D_i=j)\left[\log\alpha_j + \sum_{t=1}^T \sum_{k=1}^{{\cal{K}}} \mathbb{I}(C_{it}=k)\bigl(\log \tau_{jk} + \log\phi(Z_{jk,it}) - \log \sigma_j\bigr)\right],
\]
where $Z_{jk,it}:=\frac{Y_{it}-\mu_{jk}-\bs X_{it}\t \bs\beta_j }{\sigma_j}$.

To obtain the Maximum Likelihood estimator $\hat{\bs\vartheta} = \arg\max_{\bs\vartheta} \sum_{i=1}^n \log g(\mathbf{W}_i;\bs\vartheta)$, we implment the following EM algorithm by alternating between the E-step and the M-step.

In E-step, given current estimates $(\alpha_j^{(m)},\bs\theta_j^{(m)})$, compute the posterior probabilities of the latent variables as
\begin{align*}
\pi_{i,j}^{(m)} & = P(D_i=j|\mathbf{W}_i,\{\alpha_j^{(m)},\bs\theta_j^{(m)}\})
= \frac{\alpha_j^{(m)} f(\mathbf{W}_i;\theta_j^{(m)})}{\sum_{r=1}^M \alpha_r^{(m)}f(\mathbf{W}_i;\theta_r^{(m)})}\text{ and }\\
\gamma_{jk,it}^{(m)} &= P(C_{it}=k|\mathbf{W}_i,D_i=j,\bs\theta_j^{(m)})
= \frac{\tau_{jk}^{(m)}\phi(z_{jk,it}^{(m)})}{\sum_{\ell=1}^{{\cal{K}}} \tau_{j\ell}^{(m)}\phi(z_{j\ell,it}^{(m)})}
\end{align*}
 for each $j=1,\ldots,M$ and $k=1,\dots,K$.

In M-step, we maximize  the expected complete-data log-likelihood $E[\log L_c|\mathbf{W},\alpha_j^{(m)},\bs\theta_j^{(m)}]$ as follows: for $j=1,...,M$,
\begin{enumerate}
\item  Update $\alpha_j$ as
\[
\alpha_j^{(m+1)} = \frac{\sum_{i=1}^n \pi_{i,j}^{(m)}}{n}, \quad \text{with } \sum_{j=1}^M \alpha_j^{(m+1)}=1.
\]
\item  Update $\tau_{jk}$ as
\[
\tau_{jk}^{(m+1)} = \frac{\sum_{i=1}^n \pi_{i,j}^{(m)}\sum_{t=1}^T \gamma_{jk,it}^{(m)}}{\sum_{i=1}^n \pi_{i,j}^{(m)} T}\quad\text{for $k=1,...,{\cal{K}}$}.
\]
\item   Update $\bs\mu_{j}$ as
   \[
   \mu_{jk}^{(m+1)} = \frac{\sum_{i,t} \pi_{i,j}^{(m)}\gamma_{jk,it}^{(m)}(Y_{it}-\bs{X}_{it}^\top \beta_j^{(m)})}{\sum_{i,t} \pi_{i,j}^{(m)}\gamma_{jk,it}^{(m)}}\quad\text{   for $k=1,...,K$.}
   \]

\item  Given  $\bs\mu_{j}^{(m+1)}$, update $\bs\beta_j$ as
\[
\bs\beta_j^{(m+1)}=\arg\min_{\bs\beta_j} \sum_{j=1}^M \sum_{k=1}^{{\cal{K}}} \sum_{i=1}^n \sum_{t=1}^T \pi_{i,j}^{(m)}\gamma_{jk,it}^{(m)} (Y_{it}-\mu_{jk}^{(m+1)}-\bs X_{it}\t \bs\beta_j )^2.
\]

\item Update $\sigma_{j}^2$
\[
\sigma_j^{(m+1)2} = \frac{\sum_{i=1}^n \pi_{i,j}^{(m)}\sum_{t=1}^T \sum_{k=1}^{{\cal{K}}} \gamma_{jk,it}^{(m)}(Y_{it}-\mu_{jk}^{(m+1)}-\bs X_{it}\t \bs\beta_j )^2}{\sum_{i=1}^n \pi_{i,j}^{(m)}T}.
\]
\end{enumerate}

We iterate between the E-step and the M-step until convergence.

\subsubsection{EM Algorithm for the model (\ref{eq:fm-dynamic}) with normal mixture density (\ref{eq:f1-dynamic-mixture})}\label{sec:em-algorithm-dynamic}

In a similar manner as above section, we denote latent variables $D_i \in \{1,2,\ldots,M\}$ and $C_{it} \in \{1,2,\ldots,K\}$, where the top-level latent variable $D_i$ indicates which of the $M$ top-level components was chosen for observation $i$, and $C_{it}$ indicates which of the $K$ inner components was chosen for the $t$-th observation of unit $i$ under the selected top-level component. Then, the complete-data log-likelihood is given as
\[
\begin{aligned}
\log L_c = \sum_{i=1}^n \sum_{j=1}^M \mathbb{I}(D_i = j) \Bigg[ \log\alpha_j
+  \Big(
& \sum_{k=1}^{\mathcal{K}} \mathbb{I}(C_{i1} = k) \big(  \log \tau_{jk} + \log\phi(Z_{jk,i1}) - \log \sigma_{1,j} \big) \\
& + \sum_{t=2}^T \sum_{k=1}^{\mathcal{K}} \mathbb{I}(C_{it} = k) \big( \log \tau_{jk} + \log\phi(Z_{jk,it}) - \log \sigma_j \big)
\Big) \Bigg]
\end{aligned}
\]
where $Z_{jk,it}:=\frac{Y_{it} - \rho Y_{i,t-1} - (1 - \rho) \mu_{jk}-\bs X_{it}\t \bs\beta_j +  \bs X_{i,t-1}\t \bs\beta_j \rho_j  }{\sigma_j}$ and $Z_{jk,i1} = \frac{Y_{i1} - \mu_{1,jk} - \bs X_{i1}\t \bs\beta_{1,j}}{\sigma_{1,j}}$.

To obtain the Maximum Likelihood estimator $\hat{\bs\vartheta} = \arg\max_{\bs\vartheta} \sum_{i=1}^n \log g(\mathbf{W}_i;\bs\vartheta)$, we implment the following EM algorithm by alternating between the E-step and the M-step.

In E-step, given current estimates $(\alpha_j^{(m)},\bs\theta_j^{(m)})$, compute the posterior probabilities of the latent variables as
\[
\pi_{i,j}^{(m)} = P(D_i=j|\mathbf{W}_i,\{\alpha_j^{(m)},\bs\theta_j^{(m)}\})
= \frac{\alpha_j^{(m)} f(\mathbf{W}_i;\theta_j^{(m)})}{\sum_{r=1}^M \alpha_r^{(m)}f(\mathbf{W}_i;\theta_r^{(m)})}
\]
\[
\gamma_{jk,it}^{(m)} = P(C_{it}=k|\mathbf{W}_i,D_i=j,\bs\theta_j^{(m)})
= \frac{\tau_{jk}^{(m)}\phi(Z_{jk,it}^{(m)})}{\sum_{\ell=1}^{{\cal{K}}} \tau_{j\ell}^{(m)}\phi(Z_{j\ell,it}^{(m)})} \qquad \text{for } T > 1,
\]
\[
\gamma_{jki,1}^{(m)} = P(C_{i1}=k|\mathbf{W}_i,D_i=j,\bs\theta_j^{(m)})
= \frac{\tau_{jk}^{(m)}\phi(Z_{jki,1}^{(m)})}{\sum_{\ell=1}^{{\cal{K}}} \tau_{j\ell}^{(m)}\phi(Z_{j\ell,i1}^{(m)})} \qquad \text{for } T = 1,
\]
for each $j=1,\ldots,M$ and $k=1,\dots,K$.

In M-step, we maximize the expected complete-data log-likelihood $E[\log L_c|\mathbf{W},\alpha_j^{(m)},\bs\theta_j^{(m)}]$ as follows: for $j=1,...,M$,
\begin{enumerate}
\item  Update $\alpha_j$ as
\[
\alpha_j^{(m+1)} = \frac{\sum_{i=1}^n \pi_{i,j}^{(m)}}{n}, \quad \text{with } \sum_{j=1}^M \alpha_j^{(m+1)}=1.
\]
\item  Update $\tau_{jk}$ as
\[
\tau_{jk}^{(m+1)} = \frac{\sum_{i=1}^n \pi_{i,j}^{(m)}\sum_{t=1}^T \gamma_{jk,it}^{(m)}}{\sum_{i=1}^n \pi_{i,j}^{(m)} T}\quad\text{for $k=1,...,{\cal{K}}$}.
\]
\item   Update $\bs\mu_{j}$ and $\bs\mu_{1,j}$ as
   \[
   \mu_{jk}^{(m+1)} = \frac{\sum_{i,t} \pi_{i,j}^{(m)}\gamma_{jk,it}^{(m)}( Y_{it} - \rho^{(m)} Y_{i,t-1} - \bs X_{it}\t \bs\beta_j^{(m)} +  \bs X_{i,t-1}\t \bs\beta_j^{(m)} \rho_j^{(m)}  )}{\sum_{i,t} \pi_{i,j}^{(m)}\gamma_{jk,it}^{(m)} (1 - \rho_j^{(m)}) }\quad\text{   for $k=1,...,K$.}
   \]
   and 
  \[
   \mu_{1,jk}^{(m+1)} = \frac{\sum_{i} \pi_{i,j}^{(m)}\gamma_{jk,i1}^{(m)}( Y_{i1} - \bs X_{i1}\t \bs\beta_{1,j}^{(m)}  )}{\sum_{i,t} \pi_{i,j}^{(m)}\gamma_{jk,i1}^{(m)}}\quad\text{   for $k=1,...,K$.}
   \]

\item  Given  $\bs\mu_{j}^{(m+1)}$, $\mu_{1,jk}^{(m+1)}$ and $\rho_{j}^{(m)}$, update $\bs\beta_j$ and $\bs\beta_{1,j}$ as
\begin{align*}
\bs\beta_j^{(m+1)} = \arg\min_{\bs\beta_j} \sum_{j=1}^M \sum_{k=1}^{\mathcal{K}} & \sum_{i=1}^n \sum_{t=2}^T 
\pi_{i,j}^{(m)} \gamma_{jk,it}^{(m)} \\
& \times
\Big[
 Y_{it} - \rho_{j}^{(m)} Y_{i,t-1} - (1 - \rho_{j}^{(m)}) \mu_{jk}^{(m+1)} - \left( \bs X_{it}^\top - \rho_j^{(m)} \bs X_{i,t-1}^\top \right) \bs\beta_j
\Big]^2
\end{align*}

and
\[
\bs\beta_{1,j}^{(m+1)} = \arg\min_{\bs\beta_{1,j}} \sum_{j=1}^M \sum_{k=1}^{\mathcal{K}} \sum_{i=1}^n 
\pi_{i,j}^{(m)} \gamma_{jk,i1}^{(m)} 
\left( Y_{i1} - \mu_{1,jk}^{(m+1)} - \bs X_{i1}^\top \bs\beta_{1,j} \right)^2.
\]

\item Given $\bs\beta_j^{(m+1)}$ and $\mu_{jk}^{(m+1)}$, update $\rho_j$ 
\begin{align*}
&\rho_j^{(m+1)} =  \arg\min_{\rho_j} \sum_{j=1}^M \sum_{k=1}^{\mathcal{K}} \sum_{i=1}^n \sum_{t=2}^T 
\pi_{i,j}^{(m)} \gamma_{jk,it}^{(m)} \\
&\qquad \times \left[
 \big(Y_{it} - \mu_{jk}^{(m+1)} - \bs X_{it}^\top  \bs\beta_j^{(m+1)}\big) - \rho_j \big(Y_{it-1} - \mu_{jk}^{(m+1)} - \bs X_{it-1}^\top  \bs\beta_j^{(m+1)}\big)
\right]^2.
\end{align*}

\item Update $\sigma_{j}^2$ and $\sigma_{1,j}^2$ As

\scalebox{0.85}{%
\parbox{\linewidth}{%
\begin{align*}
\sigma_j^{(m+1)2} =  
  \frac{
  \sum_{i=1}^n \sum_{t=2}^T \sum_{k=1}^{\mathcal{K}} \pi_{i,j}^{(m)} \gamma_{jk,it}^{(m)}
  \left(
    Y_{it} - Y_{i,t-1}\rho_j^{(m+1)}
    - \mu_{jk}^{(m+1)} (1 - \rho_j^{(m+1)})
    - \bs X_{it}^\top \bs\beta_j^{(m+1)} + \bs X_{i,t-1}^\top \bs\beta_j^{(m+1)} \rho_j^{(m+1)}
  \right)^2
  }{
  \sum_{i=1}^n \pi_{i,j}^{(m)} (T-1)
  }
\end{align*}
}%
}
and 
\[ \sigma_{1,j}^{(m+1)2} =  
  \frac{
  \displaystyle
  \sum_{i=1}^n \sum_{k=1}^{\mathcal{K}} \pi_{i,j}^{(m)} \gamma_{jk,i1}^{(m)}
  \left(
    Y_{i1} - \mu_{1,jk}^{(m+1)}
    - \bs X_{i1}^\top \bs\beta_{1,j}^{(m+1)}
  \right)^2
  }{
  \displaystyle
  \sum_{i=1}^n \sum_{k=1}^{\mathcal{K}} \pi_{i,j}^{(m)}
  }
\]
\end{enumerate}

We iterate between the E-step and the M-step until convergence.

\newpage

\section{Additional Tables and Figures}
\label{sec:additional}

%
%
\begin{table}[h!]
    \centering
    \caption{Selection frequencies for estimated number of components when data are generated from a two-component AR(1) mixture model with normal innovations   ($M_0=2$,  $\bs{\mathcal{K}=1}$, $\bs{n=225}$, $T=5$)}
    \label{tab:sequential_test_ar1_normal}
    
        \begin{tabular*}{0.9\linewidth}{@{\extracolsep{\fill}}l | cccccc@{}}
                    \toprule
                     \textbf{Methods ($q_n$, Error Density) }   & \textbf{M=1} & \textbf{M=2} & \textbf{M=3} & \textbf{M=4} & \textbf{M=5} & \textbf{M=6} \\
                        \midrule
                    \textbf{AIC (Normal)} & 0 & 4 & 22 & 20 & 23 & 31 \\
                    \textbf{BIC (Normal)} & 0 & 98 & 2 & 0 & 0 & 0 \\
                    \textbf{LR ($0.01$, Normal)} & 0 & 98 & 2 & 0 & 0 & 0 \\
                    \textbf{LR ($0.05$, Normal)} & 0 & 95 & 5 & 0 & 0 & 0 \\
                    \textbf{LR ($0.10$, Normal)} & 0 & 90 & 10 & 0 & 0 & 0 \\  \midrule
                    \textbf{AIC (Mixture)} & 0 & 29 & 38 & 23 & 7 & 3 \\
                    \textbf{BIC (Mixture)} & 1 & 80 & 19 & 0 & 0 & 0 \\
                    \textbf{LR ($0.01$, Mixture)} & 5 & 88 & 7 & 0 & 0 & 0 \\
                    \textbf{LR ($0.05$, Mixture)} & 2 & 90 & 8 & 0 & 0 & 0 \\
                    \textbf{LR ($0.10$, Mixture)} & 1 & 90 & 9 & 0 & 0 & 0 \\  \midrule
                    \textbf{ave-rk ($0.05$)} & 0 & 84 & 16 & 0 & 0 & 0 \\
                    \textbf{max-rk ($0.05$)} & 3 & 92 & 5 & 0 & 0 & 0 \\
                    \bottomrule
                    \end{tabular*}
                \begin{tablenotes}
                     \footnotesize
                        \textbf{Notes:}
                         Based on 100 simulation repetitions, each simulated dataset consists of $(n,T)$ panel observations with $n=225$ and $T=5$. The datasets are generated from the two-component model (\ref{eq:fm-dynamic}) with normal error density (\ref{eq:f1-dynamic})
                          without covariates, using parameter values estimated from Chilean fabricated metal products industry data: $\bs\alpha=[0.359, 0.641]$; $\bs\mu=[-0.883, -0.419]$; $\bs\rho=[0.472, 0.579]$; $\bs\sigma= [0.465, 0.208]$;  $\bs\mu_{0}=[-0.991, -0.499]$; $\bs\sigma_{0}=[0.476, 0.269]$. The simulation assumes that $\alpha$ (mixing probabilities) is bounded between 0.05 and 0.95, and $\tau$ (within-component mixing probabilities) is bounded between 0.05 and 0.95.  The true number of components is $M_0=2$. The reported values represent the percentages of simulations in which each criterion—AIC, BIC, ave-rk, max-rk, and LRTs—selected a given number of components.
                \end{tablenotes}
        
    \end{table}

\begin{table}
    \centering
    \caption{Selection frequencies for estimated number of components when data are generated from a two-component AR(1) mixture model with normal mixture innovations  ($M_0=2$,  $\bs{\mathcal{K}=2}$, $\bs{n=225}$, $T=5$)}
      \label{tab:sequential_test_ar1_mixture}
      \begin{tabular*}{0.9\linewidth}{@{\extracolsep{\fill}}l | cccccc@{}}
            \toprule
          \textbf{Methods ($q_n$, Error Density) } & \textbf{M=1} & \textbf{M=2} & \textbf{M=3} & \textbf{M=4} & \textbf{M=5} & \textbf{M=6} \\
          \midrule
          \textbf{AIC (Normal)} & 0 & 18 & 30 & 28 & 14 & 10 \\
          \textbf{BIC (Normal)} & 0 & 97 & 3 & 0 & 0 & 0 \\
          \textbf{LR ($0.01$, Normal)} & 0 & 97 & 3 & 0 & 0 & 0 \\
          \textbf{LR ($0.05$, Normal)} & 0 & 92 & 8 & 0 & 0 & 0 \\
          \textbf{LR ($0.10$, Normal)} & 0 & 87 & 12 & 1 & 0 & 0 \\  \midrule
          \textbf{AIC (Mixture)} & 0 & 45 & 38 & 15 & 1 & 1 \\
          \textbf{BIC (Mixture)} & 0 & 83 & 16 & 1 & 0 & 0 \\
          \textbf{LR ($0.01$, Mixture)} & 8 & 87 & 5 & 0 & 0 & 0 \\
          \textbf{LR ($0.05$, Mixture)} & 4 & 90 & 6 & 0 & 0 & 0 \\
          \textbf{LR ($0.10$, Mixture)} & 1 & 90 & 8 & 1 & 0 & 0 \\   \midrule
          \textbf{ave-rk ($0.05$)} & 0 & 76 & 24 & 0 & 0 & 0 \\
          \textbf{max-rk ($0.05$)} & 0 & 89 & 11 & 0 & 0 & 0 \\
          \bottomrule
        \end{tabular*}
        \begin{tablenotes}
             \footnotesize
          \textbf{Notes:}
           Based on 100 simulation repetitions, each simulated dataset consists of $(n,T)$ panel observations with $n=225$ and $T=5$. The datasets are generated from the two-component model (\ref{eq:fm-dynamic}) with normal mixture error density (\ref{eq:f1-dynamic-mixture})  without covariates, using parameter values estimated from Chilean fabricated metal products industry data: $\bs\alpha= [0.314, 0.686]$; $\bs\tau =  [[0.159, 0.841], [0.626, 0.374]]$ $\bs\mu=[[-1.369, -0.874], [-0.435, -0.409]]$; $\bs\rho= [0.342, 0.565]$; $\bs\sigma = [0.463, 0.217]$; $\bs\mu_0= [[-1.047, -1.038], [-0.512, -0.504]]$;  $\bs\sigma_0= [0.466, 0.281]$.  The simulation assumes that $\alpha$ (mixing probabilities) is bounded between 0.05 and 0.95, and $\tau$ (within-component mixing probabilities) is bounded between 0.05 and 0.95.  The true number of components is $M_0=2$. The reported values represent the percentages of simulations in which each criterion—AIC, BIC, ave-rk, max-rk, and LRTs—selected a given number of components.
        \end{tablenotes}
      \end{table}

      \begin{table}[h!]
    \centering
    \caption{Selection frequencies for estimated number of components when data are generated from a two-component factor-augmented technology AR(1) mixture model with normal innovations ($M_0=2$,  $\bs{\mathcal{K}=1}$, $\bs{n=103}$, $T=5$)}
    \label{tab:sequential_test_ar1_lat_normal}
    
        \begin{tabular*}{0.9\linewidth}{@{\extracolsep{\fill}}l | cccccc@{}}
                    \toprule
                     \textbf{Methods ($q_n$, Error Density) }   & \textbf{M=1} & \textbf{M=2} & \textbf{M=3} & \textbf{M=4} & \textbf{M=5} & \textbf{M=6} \\
                        \midrule
                    \textbf{AIC (Normal)} & 0 & 67 & 27 & 6 & 0 & 0 \\
                    \textbf{BIC (Normal)} & 0 & 98 & 2 & 0 & 0 & 0 \\
                    \textbf{AIC (Mixture)} & 0 & 4 & 6 & 20 & 13 & 57 \\
                    \textbf{BIC (Mixture)} & 0 & 55 & 33 & 10 & 2 & 0 \\
                    \textbf{LR ($0.01$, Normal)} & 0 & 99 & 1 & 0 & 0 & 0 \\
                    \textbf{LR ($0.05$, Normal)} & 0 & 93 & 7 & 0 & 0 & 0 \\
                    \textbf{LR ($0.10$, Normal)} & 0 & 88 & 12 & 0 & 0 & 0 \\
                    \textbf{LR ($0.01$, Mixture)} & 6 & 93 & 1 & 0 & 0 & 0 \\
                    \textbf{LR ($0.05$, Mixture)} & 4 & 95 & 1 & 0 & 0 & 0 \\
                    \textbf{LR ($0.10$, Mixture)} & 2 & 94 & 4 & 0 & 0 & 0 \\
                    \textbf{ave-rk ($0.05$)} & 95 & 5 & 0 & 0 & 0 & 0 \\
                    \textbf{max-rk ($0.05$)} & 100 & 0 & 0 & 0 & 0 & 0 \\
                    \bottomrule
                    \end{tabular*}
                    \begin{tablenotes}
                      \footnotesize
                      \textbf{Notes:}
                      Based on 100 simulation repetitions, each simulated dataset consists of $(n,T)$ panel observations with $n=103$ and $T=5$. The datasets are generated from the two-component model (\ref{eq:fm-dynamic}) with normal error density (\ref{eq:f1-dynamic}) without covariates, using parameter values: $\bs\alpha = [0.39, 0.61]$; $\bs\mu = [-0.442 -0.123]$; $\bs\rho = [ 0.582 , 0.771 ]$; $\bs\sigma = [0.419, 0.183]$; $\bs\mu_0 = [-0.947, -0.504]$; $\bs\sigma_0 = [0.538, 0.324]$. The simulation assumes that $\alpha$ (mixing probabilities) is bounded between 0.05 and 0.95, and $\tau$ (within-component mixing probabilities) is bounded between 0.05 and 0.95. The true number of components is $M_0=2$. The reported values represent the percentages of simulations in which each criterion—AIC, BIC, ave-rk, max-rk, and LRTs—selected a given number of components.
                    \end{tablenotes}
        
    \end{table}

\begin{table}
    \centering
    \caption{Selection frequencies for estimated number of components when data are generated from a two-component factor-augmented technology AR(1) mixture model with normal mixture innovations  ($M_0=2$,  $\bs{\mathcal{K}=2}$, $\bs{n=103}$, $T=5$)}
      \label{tab:sequential_test_ar1_lat_mixture}
      \begin{tabular*}{0.9\linewidth}{@{\extracolsep{\fill}}l | cccccc@{}}
            \toprule
          \textbf{Methods ($q_n$, Error Density) } & \textbf{M=1} & \textbf{M=2} & \textbf{M=3} & \textbf{M=4} & \textbf{M=5} & \textbf{M=6} \\
          \midrule
          \textbf{AIC (Normal)} & 0 & 13 & 51 & 28 & 8 & 0 \\
          \textbf{BIC (Normal)} & 2 & 27 & 64 & 7 & 0 & 0 \\
          \textbf{LR ($0.01$, Normal)} & 2 & 28 & 68 & 2 & 0 & 0 \\
          \textbf{LR ($0.05$, Normal)} & 2 & 23 & 66 & 9 & 0 & 0 \\
          \textbf{LR ($0.10$, Normal)} & 2 & 21 & 66 & 11 & 0 & 0 \\  \midrule
          \textbf{AIC (Mixture)} & 0 & 3 & 6 & 18 & 25 & 48 \\
          \textbf{BIC (Mixture)} & 2 & 14 & 43 & 33 & 8 & 0 \\
          \textbf{LR ($0.01$, Mixture)} & 16 & 58 & 25 & 1 & 0 & 0 \\
          \textbf{LR ($0.05$, Mixture)} & 14 & 54 & 26 & 6 & 0 & 0 \\
          \textbf{LR ($0.10$, Mixture)} & 14 & 43 & 33 & 10 & 0 & 0 \\   \midrule
          \textbf{ave-rk ($0.05$)} & 90 & 10 & 0 & 0 & 0 & 0 \\
          \textbf{max-rk ($0.05$)} & 100 & 0 & 0 & 0 & 0 & 0 \\     \bottomrule  \end{tabular*}
        \begin{tablenotes}
            \footnotesize
                   \textbf{Notes:}
                        Based on 100 simulation repetitions, each simulated dataset consists of $(n,T)$ panel observations with $n=103$ and $T=5$. The datasets are generated from the two-component model (\ref{eq:fm-dynamic}) with normal mixture error density (\ref{eq:f1-dynamic-mixture}) without covariates, using the following parameter values: $\bs\alpha= [0.951, 0.049]$; $\bs\tau = [[0.751, 0.249], [0.5, 0.5]]$; $\bs\mu=[[-0.143, -0.143], [-1.098, -1.098]]$; $\bs\sigma = [0.287, 0.692]$; $\rho = [0.781, 0.195]$; $\bs\mu_0 = [[-0.502, -1.082], [-3.189, -0.998]]$; $\bs\sigma_0 = [0.314, 0.33]$. The simulation assumes that $\alpha$ (mixing probabilities) is bounded between 0.05 and 0.95, and $\tau$ (within-component mixing probabilities) is bounded between 0.05 and 0.95. The true number of components is $M_0=2$. The reported values represent the percentages of simulations in which each criterion—AIC, BIC, ave-rk, max-rk, and LRTs—selected a given number of components.
        \end{tablenotes}
      \end{table}
      
        \begin{table}
          \centering 
      \caption{
        Sensitivity analysis: Selection frequencies for estimated number of components ($M$) in simulated data from a three-component normal mixture panel model ($M_0=3$, $\mathcal{K}=1$, $n=225$, $T=3$) \textbf{under different $\alpha$ bounds}
      } \label{tab:sensitivity-1}
      \begin{tabular*}{\textwidth}{@{\extracolsep{\fill}}l|cccccc@{}}
        \toprule
        \textbf{Methods ($q_n$, Error Density)} & \textbf{M=1} & \textbf{M=2} & \textbf{M=3} & \textbf{M=4} & \textbf{M=5} & \textbf{M=6} \\
        \midrule
        \multicolumn{7}{l}{\textbf{Panel A:} \textnormal{Parameter bounds: $\tau \in [0.05, 0.95]$, $\alpha \in [0.1, 0.9]$, $\sigma \geq 0.05 \times$ sample variance}} \\
        \midrule
          \textbf{AIC (Normal)}                & 0 & 0 & 86 & 12 & 2 & 0 \\
          \textbf{BIC (Normal)}                & 0 & 0 & 100 & 0 & 0 & 0 \\
          \textbf{LR (1\%)}           & 0 & 0 & 98 & 2 & 0 & 0 \\
          \textbf{LR (5\%)}           & 0 & 0 & 94 & 5 & 1 & 0 \\
          \textbf{LR (10\%)}          & 0 & 0 & 89 & 9 & 2 & 0 \\
          \midrule
          \textbf{AIC (Mixture)}      & 0 & 0 & 75 & 23 & 2 & 0 \\
          \textbf{BIC (Mixture)}      & 0 & 0 & 100 & 0 & 0 & 0 \\
          \textbf{LR (1\% Mixture)}   & 0 & 0 & 100 & 0 & 0 & 0 \\
          \textbf{LR (5\% Mixture)}   & 0 & 0 & 90 & 10 & 0 & 0 \\
          \textbf{LR (10\% Mixture)}  & 0 & 0 & 83 & 16 & 1 & 0 \\      \midrule
        \multicolumn{7}{l}{\textbf{Panel B:} \textnormal{Parameter bounds: $\tau \in [0.05, 0.95]$, $\alpha \in [0.01, 0.99]$, $\sigma \geq 0.05 \times$ sample variance}} \\
        \midrule
        \textbf{AIC (Normal)}         & 0 & 0 & 57 & 36 & 6 & 1 \\
        \textbf{BIC (Normal)}         & 0 & 0 & 100 & 0 & 0 & 0 \\
        \textbf{LR ($0.01$, Normal)}  & 0 & 0 & 97 & 3 & 0 & 0 \\
        \textbf{LR ($0.05$, Normal)}  & 0 & 0 & 97 & 3 & 0 & 0 \\
        \textbf{LR ($0.10$, Normal)}  & 0 & 0 & 94 & 6 & 0 & 0 \\
        \midrule
        \textbf{AIC (Mixture)}        & 0 & 0 & 54 & 36 & 9 & 1 \\
        \textbf{BIC (Mixture)}        & 0 & 0 & 99 & 1 & 0 & 0 \\
        \textbf{LR ($0.01$, Mixture)} & 0 & 0 & 99 & 1 & 0 & 0 \\
        \textbf{LR ($0.05$, Mixture)} & 0 & 0 & 97 & 3 & 0 & 0 \\
        \textbf{LR ($0.10$, Mixture)} & 0 & 0 & 88 & 12 & 0 & 0 \\
        \bottomrule  \end{tabular*}
      \begin{tablenotes}
        \footnotesize
        \textbf{Notes:} Based on 100 simulation repetitions. Each simulated dataset consists of $n=225$ units and $T=3$ periods, generated from the same DGP as Table~\ref{tab:empirical_test_normal}. Entries show the proportion of simulations in which each method selected $M$ components (in percent).
      \end{tablenotes}

      \end{table}

      \begin{table}
          \centering

      \caption{
        Sensitivity analysis: Selection frequencies for estimated number of components ($M$) in simulated data from a three-component normal mixture panel model ($M_0=3$, $\mathcal{K}=1$, $n=225$, $T=3$) \textbf{under different $\tau$ bounds}
      }\label{tab:sensitivity-2}
      \begin{tabular*}{\textwidth}{@{\extracolsep{\fill}}l|cccccc@{}}
        \toprule
        \textbf{Methods ($q_n$, Error Density)} & \textbf{M=1} & \textbf{M=2} & \textbf{M=3} & \textbf{M=4} & \textbf{M=5} & \textbf{M=6} \\
        \midrule
        \multicolumn{7}{l}{\textbf{Panel A:} \textnormal{Parameter bounds: $\tau \in [0.1, 0.9]$, $\alpha \in [0.05, 0.95]$, $\sigma \geq 0.05 \times$ sample variance}} \\
        \midrule
        \textbf{AIC (Normal)} & 0 & 0 & 81 & 17 & 2 & 0 \\
        \textbf{BIC (Normal)} & 0 & 0 & 100 & 0 & 0 & 0 \\
        \textbf{LR ($0.01$, Normal)} & 0 & 0 & 98 & 2 & 0 & 0 \\
        \textbf{LR ($0.05$, Normal)} & 0 & 0 & 97 & 3 & 0 & 0 \\
        \textbf{LR ($0.10$, Normal)} & 0 & 0 & 92 & 7 & 1 & 0 \\ \midrule
        \textbf{AIC (Mixture)} & 0 & 0 & 73 & 22 & 4 & 1 \\
        \textbf{BIC (Mixture)} & 0 & 0 & 100 & 0 & 0 & 0 \\
        \textbf{LR ($0.01$, Mixture)} & 0 & 0 & 100 & 0 & 0 & 0 \\
        \textbf{LR ($0.05$, Mixture)} & 0 & 0 & 93 & 7 & 0 & 0 \\
        \textbf{LR ($0.10$, Mixture)} & 0 & 0 & 87 & 13 & 0 & 0 \\
        \midrule
        \multicolumn{7}{l}{\textbf{Panel B:} \textnormal{Parameter bounds: $\tau \in [0.01, 0.99]$, $\alpha \in [0.05, 0.95]$, $\sigma \geq 0.05 \times$ sample variance}} \\
        \midrule
        \textbf{AIC (Normal)} & 0 & 0 & 80 & 18 & 2 & 0 \\
        \textbf{BIC (Normal)} & 0 & 0 & 100 & 0 & 0 & 0 \\
        \textbf{LR ($0.01$, Normal)} & 0 & 0 & 98 & 2 & 0 & 0 \\
        \textbf{LR ($0.05$, Normal)} & 0 & 0 & 98 & 2 & 0 & 0 \\
        \textbf{LR ($0.10$, Normal)} & 0 & 0 & 92 & 7 & 1 & 0 \\ \midrule
        \textbf{AIC (Mixture)} & 0 & 0 & 61 & 28 & 10 & 1 \\
        \textbf{BIC (Mixture)} & 0 & 0 & 99 & 1 & 0 & 0 \\
        \textbf{LR ($0.01$, Mixture)} & 0 & 0 & 100 & 0 & 0 & 0 \\
        \textbf{LR ($0.05$, Mixture)} & 0 & 0 & 97 & 3 & 0 & 0 \\
        \textbf{LR ($0.10$, Mixture)} & 0 & 0 & 91 & 9 & 0 & 0 \\
        \bottomrule
      \end{tabular*}
      \begin{tablenotes}
        \footnotesize
        \textbf{Notes:} Based on 100 simulation repetitions. Each simulated dataset consists of $n=225$ units and $T=3$ periods, generated from the same DGP as Table~\ref{tab:empirical_test_normal}. Entries show the proportion of simulations in which each method selected $M$ components (in percent).
      \end{tablenotes}

        \end{table}

         \begin{table}
        \centering 
          \caption{
      Sensitivity analysis: Selection frequencies for estimated number of components ($M$) in simulated data from a three-component normal mixture panel model ($M_0=3$, $\mathcal{K}=1$, $n=225$, $T=3$) \textbf{under different $\sigma$ bounds}
          }\label{tab:sensitivity-3}
          \begin{tabular*}{\textwidth}{@{\extracolsep{\fill}}l|cccccc@{}}
      \toprule
      \textbf{Methods ($q_n$, Error Density)} & \textbf{M=1} & \textbf{M=2} & \textbf{M=3} & \textbf{M=4} & \textbf{M=5} & \textbf{M=6} \\
      \midrule
      \multicolumn{7}{l}{\textbf{Panel A:} \textnormal{Parameter bounds: $\tau \in [0.05, 0.95]$, $\alpha \in [0.05, 0.95]$, $\sigma \geq 0.1 \times$ sample variance}} \\
      \midrule
      \textbf{AIC (Normal)} & 0 & 0 & 81 & 17 & 2 & 0 \\
      \textbf{BIC (Normal)} & 0 & 0 & 100 & 0 & 0 & 0 \\
      \textbf{LR ($0.01$, Normal)} & 0 & 0 & 98 & 2 & 0 & 0 \\
      \textbf{LR ($0.05$, Normal)} & 0 & 0 & 97 & 3 & 0 & 0 \\
      \textbf{LR ($0.10$, Normal)} & 0 & 0 & 92 & 7 & 1 & 0 \\ \midrule
      \textbf{AIC (Mixture)} & 0 & 0 & 67 & 28 & 5 & 0 \\
      \textbf{BIC (Mixture)} & 0 & 0 & 100 & 0 & 0 & 0 \\
      \textbf{LR ($0.01$, Mixture)} & 0 & 0 & 100 & 0 & 0 & 0 \\
      \textbf{LR ($0.05$, Mixture)} & 0 & 0 & 96 & 4 & 0 & 0 \\
      \textbf{LR ($0.10$, Mixture)} & 0 & 0 & 87 & 13 & 0 & 0 \\
      \midrule
      \multicolumn{7}{l}{\textbf{Panel B:} \textnormal{Parameter bounds: $\tau \in [0.05, 0.95]$, $\alpha \in [0.05, 0.95]$, $\sigma \geq 0.01 \times$ sample variance}} \\
      \midrule
      \textbf{AIC (Normal)} & 0 & 0 & 83 & 15 & 2 & 0 \\
      \textbf{BIC (Normal)} & 0 & 0 & 100 & 0 & 0 & 0 \\
      \textbf{LR ($0.01$, Normal)} & 0 & 0 & 98 & 2 & 0 & 0 \\
      \textbf{LR ($0.05$, Normal)} & 0 & 0 & 96 & 4 & 0 & 0 \\
      \textbf{LR ($0.10$, Normal)} & 0 & 0 & 93 & 4 & 3 & 0 \\ \midrule
      \textbf{AIC (Mixture)} & 0 & 0 & 63 & 30 & 7 & 0 \\
      \textbf{BIC (Mixture)} & 0 & 0 & 100 & 0 & 0 & 0 \\
      \textbf{LR ($0.01$, Mixture)} & 0 & 0 & 100 & 0 & 0 & 0 \\
      \textbf{LR ($0.05$, Mixture)} & 0 & 0 & 95 & 4 & 0 & 1 \\
      \textbf{LR ($0.10$, Mixture)} & 0 & 0 & 89 & 10 & 0 & 1 \\
      \bottomrule
          \end{tabular*}
          \begin{tablenotes}
      \footnotesize
      \textbf{Notes:} Based on 100 simulation repetitions. Each simulated dataset consists of $n=225$ units and $T=3$ periods, generated from the same DGP as Table~\ref{tab:empirical_test_normal}. Entries show the proportion of simulations in which each method selected $M$ components (in percent).
          \end{tablenotes}

      \end{table}

\begin{table}[h!]
    \centering
    \caption{Selection frequencies when data generated from an estimated three-component model with \textbf{normal error density} ($M_0=3$, $\bs{\mathcal{K}=1}$, $T=3$) for food and textile industries} \label{tab:empirical_test_normal_food_textile}
    \begin{threeparttable}
            
            \begin{tabular*}{0.9\linewidth}{@{\extracolsep{\fill}}l | cccccc@{}}
                \toprule
             \multicolumn{7}{c}{ \textbf{Panel A:\  Food industry}}\\  \midrule
               \textbf{Methods ($q_n$, Error Density) }  & \textbf{M=1} & \textbf{M=2} & \textbf{M=3} & \textbf{M=4} & \textbf{M=5} & \textbf{M=6} \\
                \midrule
                \textbf{AIC (Normal)} & 0 & 0 & 77 & 20 & 2 & 1 \\
                \textbf{BIC (Normal)} & 0 & 0 & 100 & 0 & 0 & 0 \\
                \textbf{LR ($0.01$, Normal)} & 0 & 0 & 99 & 0 & 1 & 0 \\
                \textbf{LR ($0.05$, Normal)} & 0 & 0 & 94 & 5 & 0 & 1 \\
                \textbf{LR ($0.10$, Normal)} & 0 & 0 & 84 & 13 & 2 & 1 \\
                \midrule
                \textbf{AIC (Mixture)} & 0 & 0 & 74 & 22 & 3 & 1 \\
                \textbf{BIC (Mixture)} & 0 & 0 & 100 & 0 & 0 & 0 \\
                \textbf{LR ($0.01$, Mixture)} & 0 & 0 & 98 & 2 & 0 & 0 \\
                \textbf{LR ($0.05$, Mixture)} & 0 & 0 & 95 & 5 & 0 & 0 \\
                \textbf{LR ($0.10$, Mixture)} & 0 & 0 & 89 & 10 & 0 & 1 \\ 
                \midrule
                \textbf{ave-rk ($0.05$)} & 0 & 93 & 7 & 0 & 0 & 0 \\
                \textbf{max-rk ($0.05$)} & 0 & 95 & 5 & 0 & 0 & 0 \\
                \midrule
             \multicolumn{7}{c}{ \textbf{Panel B: \ Textile industry}}\\  \midrule
               \textbf{Methods ($q_n$, Error Density) }  & \textbf{M=1} & \textbf{M=2} & \textbf{M=3} & \textbf{M=4} & \textbf{M=5} & \textbf{M=6} \\
                \midrule 
              \textbf{AIC (Normal)} & 0 & 0 & 79 & 18 & 3 & 0 \\
              \textbf{BIC (Normal)} & 0 & 0 & 99 & 1 & 0 & 0 \\
              \textbf{LR ($0.01$, Normal)} & 0 & 0 & 97 & 3 & 0 & 0 \\
              \textbf{LR ($0.05$, Normal)} & 0 & 0 & 96 & 3 & 1 & 0 \\
              \textbf{LR ($0.10$, Normal)} & 0 & 0 & 91 & 8 & 1 & 0 \\ \midrule
              \textbf{AIC (Mixture)} & 0 & 0 & 68 & 25 & 7 & 0 \\
              \textbf{BIC (Mixture)} & 0 & 0 & 100 & 0 & 0 & 0 \\
              \textbf{LR ($0.01$, Mixture)} & 0 & 2 & 97 & 1 & 0 & 0 \\
              \textbf{LR ($0.05$, Mixture)} & 0 & 0 & 92 & 8 & 0 & 0 \\
              \textbf{LR ($0.10$, Mixture)} & 0 & 0 & 88 & 12 & 0 & 0 \\  \midrule
              \textbf{ave-rk ($0.05$)} & 0 & 96 & 4 & 0 & 0 & 0 \\
              \textbf{max-rk ($0.05$)} & 0 & 96 & 4 & 0 & 0 & 0 \\
                \bottomrule
            \end{tabular*}
            \begin{tablenotes}
                \footnotesize
\textbf{Notes:}  Based on 100 simulation repetitions. Each simulated dataset consists of $(n,T)$ panel observations with $T=3$. For the food industry, $n=794$ and the true parameters are: mixing probabilities $\bs\alpha = [0.171, 0.457, 0.372]$, means $\bs\mu = [-0.917, -0.571, -0.341]$, and standard deviations $\bs\sigma = [0.533, 0.127, 0.125]$. For the textile industry, $n=196$ and the true parameters are: mixing probabilities $\bs\alpha = [0.081, 0.418, 0.5]$, means $\bs\mu = [-1.549, -0.835, -0.348]$, and standard deviations $\bs\sigma = [0.855, 0.248, 0.216]$. No covariates are included. The true number of components is $M_0=3$. The reported values represent the percentages of simulations in which each criterion—AIC, BIC, ave-rk, max-rk, and LRTs—selected a given number of components.
        \end{tablenotes}

    \end{threeparttable}
\end{table}

\begin{table}[h!]
    \centering
    \caption{Selection frequencies when data generated from an estimated three-component model with \textbf{mixture error density} ($M_0=3$, $\bs{\mathcal{K}=2}$, $T=3$) for food and textile industries}\label{tab:empirical_test_mixture_food_textile}
    \begin{threeparttable}
            
            \begin{tabular*}{0.9\linewidth}{@{\extracolsep{\fill}}l | cccccc@{}}
                \toprule
             \multicolumn{7}{c}{\textbf{Panel A:\  Food industry}}\\  \midrule
               \textbf{Methods ($q_n$, Error Density) }  & \textbf{M=1} & \textbf{M=2} & \textbf{M=3} & \textbf{M=4} & \textbf{M=5} & \textbf{M=6} \\
                \midrule
                \textbf{AIC (Normal)} & 0 & 0 & 43 & 50 & 7 & 0 \\
                \textbf{BIC (Normal)} & 0 & 0 & 95 & 5 & 0 & 0 \\
                \textbf{LR ($0.01$, Normal)} & 0 & 0 & 77 & 22 & 1 & 0 \\
                \textbf{LR ($0.05$, Normal)} & 0 & 0 & 65 & 34 & 1 & 0 \\
                \textbf{LR ($0.10$, Normal)} & 0 & 0 & 51 & 45 & 4 & 0 \\  \midrule 
                \textbf{AIC (Mixture)} & 0 & 0 & 80 & 18 & 2 & 0 \\
                \textbf{BIC (Mixture)} & 0 & 0 & 100 & 0 & 0 & 0 \\ 
                \textbf{LR ($0.01$, Mixture)} & 1 & 0 & 98 & 1 & 0 & 0 \\
                \textbf{LR ($0.05$, Mixture)} & 0 & 0 & 93 & 6 & 1 & 0 \\
                \textbf{LR ($0.10$, Mixture)} & 0 & 0 & 88 & 11 & 1 & 0 \\ 
                \midrule
                \textbf{ave-rk ($0.05$)} & 0 & 93 & 7 & 0 & 0 & 0 \\
                \textbf{max-rk ($0.05$)} & 0 & 93 & 7 & 0 & 0 & 0 \\
            \midrule
            \multicolumn{7}{c}{\textbf{Panel B: \ Textile industry}}\\  \midrule
               \textbf{Methods ($q_n$, Error Density) }  & \textbf{M=1} & \textbf{M=2} & \textbf{M=3} & \textbf{M=4} & \textbf{M=5} & \textbf{M=6} \\
                     \midrule 
                      \textbf{AIC (Normal)} & 0 & 0 & 76 & 22 & 2 & 0 \\
                      \textbf{BIC (Normal)} & 0 & 0 & 98 & 1 & 1 & 0 \\
                      \textbf{LR ($0.01$, Normal)} & 0 & 0 & 97 & 2 & 1 & 0 \\
                      \textbf{LR ($0.05$, Normal)} & 0 & 0 & 89 & 10 & 1 & 0 \\
                      \textbf{LR ($0.10$, Normal)} & 0 & 0 & 85 & 14 & 1 & 0 \\ \midrule
                      \textbf{AIC (Mixture)} & 0 & 0 & 76 & 20 & 4 & 0 \\
                      \textbf{BIC (Mixture)} & 0 & 0 & 99 & 1 & 0 & 0 \\
                      \textbf{LR ($0.01$, Mixture)} & 0 & 2 & 95 & 2 & 1 & 0 \\
                      \textbf{LR ($0.05$, Mixture)} & 0 & 0 & 94 & 5 & 1 & 0 \\
                      \textbf{LR ($0.10$, Mixture)} & 0 & 0 & 88 & 12 & 0 & 0 \\  \midrule
                      \textbf{ave-rk ($0.05$)} & 0 & 96 & 4 & 0 & 0 & 0 \\
                      \textbf{max-rk ($0.05$)} & 0 & 96 & 4 & 0 & 0 & 0 \\
                      \bottomrule
            \end{tabular*}
            \begin{tablenotes}
                \footnotesize
\textbf{Notes:}  Based on 100 simulation repetitions. Each simulated dataset consists of $(n,T)$ panel observations with $T=3$. For the food industry, $n=794$ and the true parameters are: mixing probabilities $\bs\alpha = [0.136, 0.451, 0.413]$, within-component mixing probabilities $\bs\tau = [[0.547, 0.453], [0.111, 0.889], [0.699, 0.301]]$, means $\bs\mu = [[-0.959, -0.943], [-0.844, -0.568], [-0.405, -0.229]]$, and standard deviations $\bs\sigma = [0.583, 0.113, 0.096]$. For the textile industry, $n=196$ and the true parameters are: mixing probabilities $\bs\alpha = [0.056, 0.433, 0.511]$, within-component mixing probabilities $\bs\tau = [[0.5, 0.5], [0.151, 0.849], [0.395, 0.605]]$, means $\bs\mu = [[-1.782, -1.665], [-1.262, -0.791], [-0.523, -0.242]]$, and standard deviations $\bs\sigma = [0.946, 0.208, 0.169]$. No covariates are included. The true number of components is $M_0=3$. The reported values represent the percentages of simulations in which each criterion—AIC, BIC, ave-rk, max-rk, and LRTs—selected a given number of components.
        \end{tablenotes}

    \end{threeparttable}
\end{table}

\begin{figure}
    \centering
        \caption{Estimated parameter values for two-component mixture models under conditional independence (\ref{spec-4}) with normal error density and covariates  ($M=3, \mathcal{K}=1$)  } \label{fig:parameter-normal-3}
    \subfigure[$\hat\alpha_j$]{\includegraphics[width=0.32\textwidth]{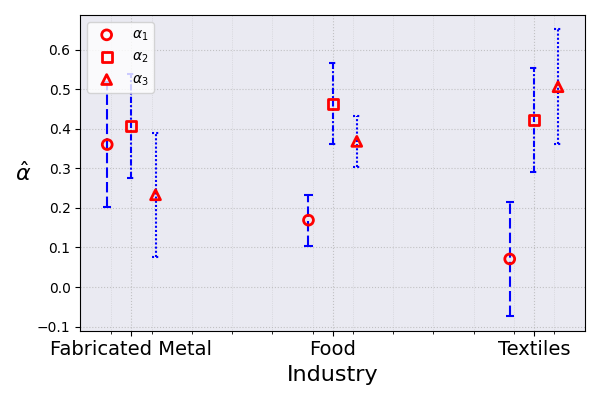}}
    \subfigure[$\hat\mu_j$]{\includegraphics[width=0.32\textwidth]{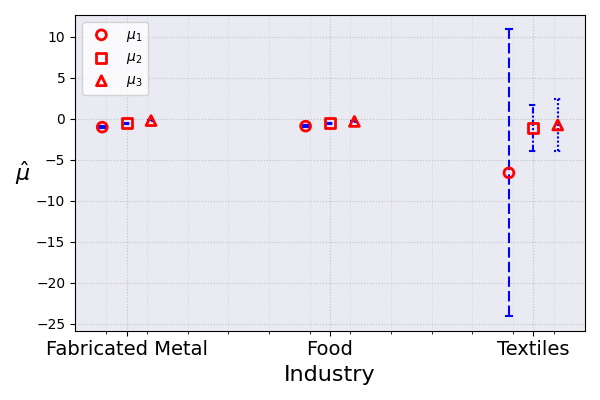}}
    \subfigure[$\hat\sigma_j$]{\includegraphics[width=0.32\textwidth]{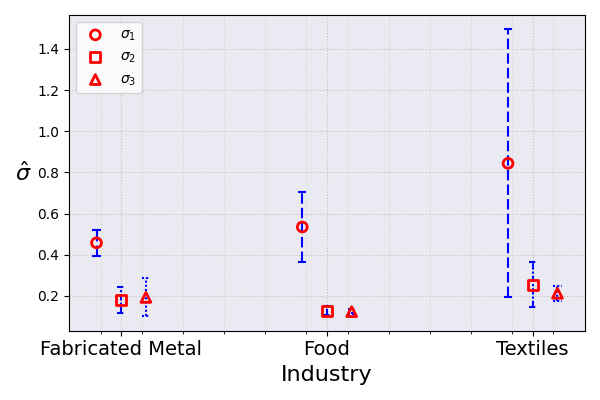}}

    \subfigure[$\hat{\beta}_{\log K, j}$]{\includegraphics[width=0.32\textwidth]{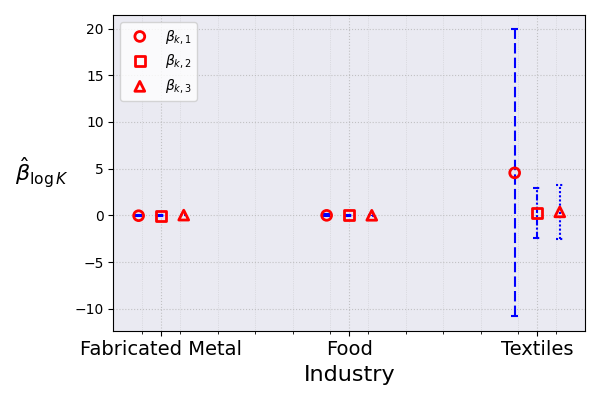}}
    \subfigure[$\hat{\beta}_{\text{Import},j}$]{\includegraphics[width=0.32\textwidth]{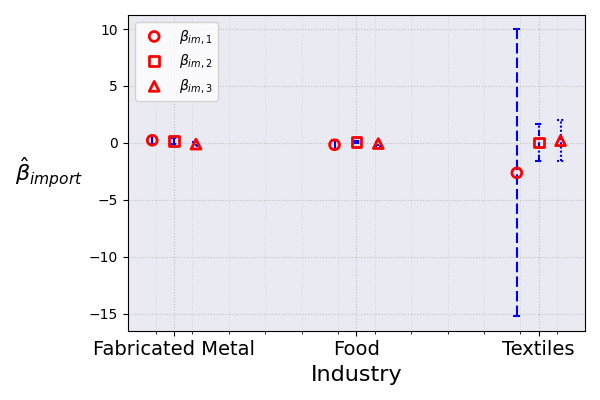}}
 \end{figure}


\begin{figure}
    \centering
        \caption{Estimated parameter values for two-component mixture models under conditional independence (\ref{spec-4}) with two-components normal mixture error density and covariates  ($M=3, \mathcal{K}=2$) }\label{fig:parameter-mixture-3}
    \subfigure[$\hat\alpha_j$]{\includegraphics[width=0.32\textwidth]{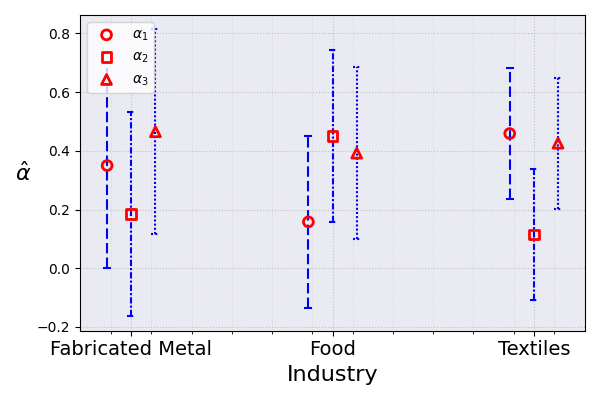}}
    \subfigure[$\hat{\mu}_j$]{\includegraphics[width=0.32\textwidth]{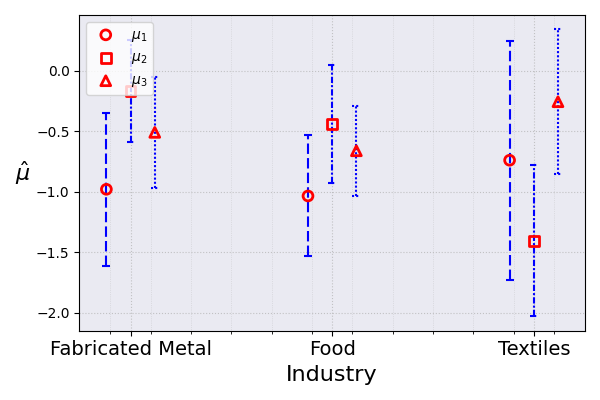}}
    \subfigure[$\hat\sigma_j$]{\includegraphics[width=0.32\textwidth]{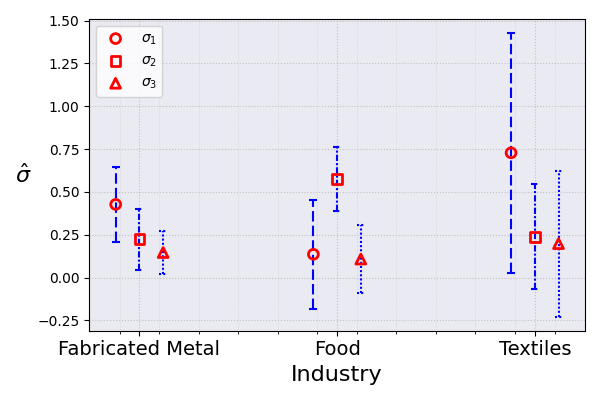}}

    \subfigure[$\hat{\beta}_{\log K, j}$]{\includegraphics[width=0.32\textwidth]{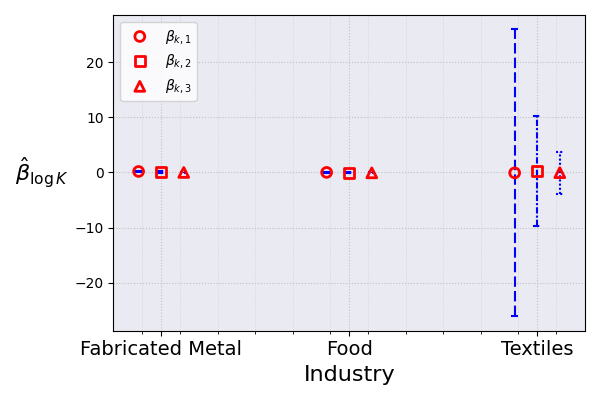}}
    \subfigure[$\hat{\beta}_{\text{Import},j}$]{\includegraphics[width=0.32\textwidth]{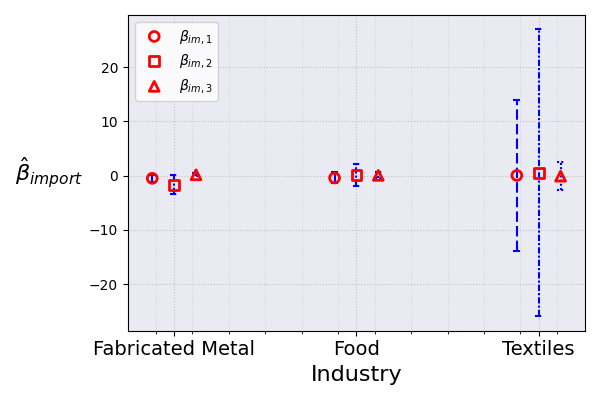}}
    \end{figure}

\end{document}